%% file: main.tex
\documentclass[11pt]{article}

\usepackage[dvipsnames]{xcolor}
\usepackage{amsfonts} 
\usepackage{amsthm}
\usepackage{graphicx} 
\usepackage{amsmath}
\numberwithin{equation}{section}
\usepackage{stmaryrd}
\usepackage{enumitem}
\usepackage{tikz}
\usetikzlibrary{decorations.pathmorphing, backgrounds,shapes.geometric, arrows, positioning, calc}
\usetikzlibrary{automata, positioning}
\usepackage{amssymb}
\usepackage{hyperref}
\hypersetup{
    colorlinks=true,
    linkcolor=blue,
    citecolor=blue
    }
\usepackage[capitalise]{cleveref}
\usepackage{etoolbox}
\usepackage{xspace}
\usepackage{upgreek}
\usepackage{halloweenmath}
\usepackage[colorinlistoftodos]{todonotes}
\usepackage{fancyhdr}
\usepackage{subcaption}
\usepackage{geometry}
\newif\ifproofs\proofstrue

\newtheorem{theorem}{Theorem}[section]
\newtheorem{definition}[theorem]{Definition}
\newtheorem{lemma}[theorem]{Lemma}

\newtheorem{observation}[theorem]{Observation}
\newtheorem{proposition}[theorem]{Proposition}
\newtheorem{problem}[theorem]{Problem}
\newtheorem{remark}[theorem]{Remark}
\newtheorem{corollary}[theorem]{Corollary}
\newtheorem{example}{Example}

\BeforeBeginEnvironment{proof}{%
  \color{violet!75!black}
}
\AfterEndEnvironment{proof}{%
  \color{black}
}

\newcommand{\myparagraph}[1]{\noindent \textbf{#1} \ }

\tikzstyle{block} = [rectangle, draw, fill=blue!20, 
    text centered, rounded corners, minimum height=1cm]
\tikzstyle{line} = [draw, -latex']

\input{macros}

\title{Determinization of Min-Plus Weighted Automata is Decidable}
\author{Shaull Almagor, Guy Arbel and Sarai Sheinvald}
\date{Technion}

\begin{document}

\maketitle

\begin{abstract}
We show that the determinization problem for min-plus (tropical) weighted automata is decidable, thus resolving this long-standing open problem.
In doing so, we develop a new toolbox for analyzing and reasoning about the run-structure of nondeterministic automata.
\end{abstract}

\section{Introduction}
Traditional automata accept or reject their input, and are therefore Boolean, in the sense that their language is a function $L:\Sigma^*\to \{0,1\}$.
\emph{Weighted Automata}~\cite{droste2009handbook,schutzenberger1961definition,chatterjee2010quantitative,Almagor2020Whatsdecidableweighted} (WFA, for short) extend Boolean automata to richer domains, so that the language of a WFA $\cA$ is a function $L_\cA:\Sigma^*\to D$ for some numeric domain $D$.
This is achieved by assigning \emph{weights} from $D$ to the transitions of the automaton, and defining two operations: one which aggregates weights along a run, and another that combines the weights along all runs to a value in $D$.

A natural class of domains over which WFAs are defined are \emph{semirings}: algebraic structures with two operations, $\otimes$ and $\oplus$, which naturally lend themselves as the two operations above. For a nondeterministic automaton over a semiring, the value of a run is the semiring product $\otimes$ of the weights along the edges, and the value of a word is the semiring sum $\oplus$ of the values of its runs.

Common examples of such semirings are the \emph{$(\min,+)$} semiring $\tup{\bbZ\cup\{\infty\},\min,+}$ (also known as the \emph{tropical} semiring, and sometimes considered as the dual $(\max,+)$ semiring), and the rational field $\tup{\bbQ,+,\times}$ (and its \emph{probabilistic} sub-semiring).
Applications of such weighted automata fall on a wide spectrum including verification, rewriting systems, tropical algebra, and speech and image processing (in the pre-LLM era) (we refer the reader to \cite{daviaud2020register,droste2009handbook,Almagor2020Whatsdecidableweighted} and references therein). 
Famously, tropical weighted automata have also been key to proving the star-height conjecture~\cite{kirsten2005distance,Has82,Has00,LP04}.

For example, consider the tropical WFA in \cref{fig:min num a num b}. For every word $w\in \{a,b\}^*$ there are two runs: one cycles in $q_a$ and counts the number of $a$'s, and the other cycles in $q_b$ and counts the number of $b$'s. Thus, the weight of $w$ is the minimum between the number of $a$'s and the number of $b$'s.

As with most computational models, reasoning about WFAs becomes easier on the deterministic fragment (e.g., equivalence of tropical WFAs is undecidable for nondeterminstic automata, but decidable for deterministic ones~\cite{Kro94,Almagor2020Whatsdecidableweighted}). 
Unfortunately, unlike Boolean automata, nondeterministic WFAs are strictly more expressive than their deterministic fragment (for most semirings). 
Accordingly, a natural problem for WFAs is the \emph{determinization problem}: given a WFA $\cA$, is there a deterministic WFA $\cD$ such that $L_{\cA}\equiv L_{\cD}$. We focus on tropical WFAs.

The determinization problem was first raised for tropical WFAs in the 1990's by Mohri~\cite{mohri1994compact,mohri1997finite} (and was alluded to already by Choffrut in 1977~\cite{choffrut1977caracterisation}, but not explicitly). 
Despite much attention, the decidability of this problem remained open. Moreover, our ability to analyze the run structure of WFAs has made notoriously slow progress, despite seeing several exciting ideas over the years (e.g., the celebrated Forest-Factorization Theorem by Simon~\cite{simon1990factorization}).
Specifically, Mohri's original work proposed a determinization algorithm for unambiguous WFAs~\cite{mohri1997finite}. This was later improved to finitely-ambiguous WFAs~\cite{klimann2004deciding}, and much later to polynomially-ambiguous WFAs~\cite{kirsten2009deciding}. Crucially, all these works place strong restrictions on the nondeterminism, making their analysis possible.
Since then, the problem remained ``wide open''~\cite{daviaud2020register,filiot2017delay,kirsten2012decidability,lombardy2006sequential}.

We remark that many recent works shed light on various intricacies of the run structure of tropical WFAs, e.g.,~\cite{daviaud2016generalised, colcombet2014size,daviaud2017degree,daviaud2023big,aminof2013rigorous,chattopadhyay2021pumping, aminof2010reasoning}. 

\paragraph*{Contribution and Paper organization}
The centerpiece of our contribution is a proof that the determinization problem for tropical WFAs is decidable, thus resolving this longstanding open problem.
In obtaining this result, however, we develop several powerful new tools for reasoning about the run structure of WFAs, which are of independent interest. 

The full detailed proof starts at \cref{sec:prelim}, and is dauntingly long. This is in part due to our attempt to make it detailed and accessible.
In \cref{sec:abs:det and gaps} we define the model, and show a characterization of determinizability that is our starting point. In \cref{sec:abs:constructions} we briefly describe the tools and concepts we develop to approach the problem. In \cref{sec:abs:finale} we present the outline of the main proof, but leave two significant holes in it. These holes are patched in \cref{sec:abs:inc inf,sec:abs:discharging}, where we explain our novel technique for analyzing runs of a WFA.
Throughout, we stick to intuition and leave the technicalities to the full proof.

It is worth noting our algorithm has no complexity bounds. Indeed, it is composed of two semi-algorithms, and relies on purely existential arguments arising from Ramsey-style reasoning. We believe, however, that our approach will pave the way to also tightening the complexity (we present a naive PSPACE-hard lower bound).

\paragraph*{Related Work}
Aside from the works mentioned above, we remark that WFAs are strongly tied with algebraic/logic tools such as Rational Series~\cite{droste2009handbook,droste2005weighted} and in particular a standard approach to reasoning about them is algebraic, e.g., with Simon's factorization~\cite{simon1990factorization}. In contrast, our work stays entirely within the combinatorial approach. We posit that this combinatorial view is what allows us to ``untangle'' the various runs, and is what fails in previous approaches.

We should also mention that over the $(+,\times)$ semiring, the technical setting changes dramatically -- tools from real algebra arise, making the analysis entirely different. For $(+,\times)$ semirings, determinization was raised as open in 2006 by~\cite{lombardy2006sequential} and resolved very recently~\cite{bell2023computing,jecker2024determinisation} with clever usage of heavy linear-algebra machinery.





\section{WFAs, Determinizability and Gaps}
\label{sec:abs:det and gaps}
\paragraph{WFAs}
A WFA $\cA=\tup{Q,\Sigma,Q_0,\Delta}$ is a finite automaton where $Q$ is a set of states, $\Sigma$ is an alphabet, $Q_0\subseteq Q$ are initial states,
and $\Delta\subseteq Q\times \Sigma\times \bbZinf\times Q$ is a \emph{weighted} transition relation, assigning every $p,q,\sigma$ a weight in $\bbZinf=\bbZ\cup \{\infty\}$. Intuitively, weight $\infty$ means that there is no transition. 

A \emph{run} of $\cA$ a word $w$ is a sequence of transitions reading $w$, and its {\em weight} is the sum of the weights along its transitions. We denote a run $\rho$ from state $q$ to state $p$ reading $w$ by $\rho:q\runsto{w}p$, and its weight by $\weight(\rho)$.
The {\em weight of a word} $w$, is the minimal weight of a run on $w$ (where weight $\infty$ signifies having no run on $w$). More generally, we denote the minimal weight of a run on $w$ from state $q$ to $p$ by $\minweight(w,q\runsto{w}p)$, and the minimal weight of a run between sets of states $Q'$ and $P'$ by $\minweight(w,Q'\runsto{w}P')$. An example of a WFA is depicted in \cref{fig:min num a num b}

A WFA is \emph{deterministic} if it has a single initial state, and from every state there exists at most one finite-weight transition on each letter. It is \emph{determinizable} if it has an equivalent deterministic WFA (i.e., they describe the same function). 
It is well-known that not every WFA is determinizable~\cite{chatterjee2010quantitative}. In this paper, we show that determining whether a WFA is determinizable is decidable.

\begin{figure}[ht]
    \centering
    \begin{minipage}{0.35\textwidth} 
        \centering
       \includegraphics[width=1\linewidth]{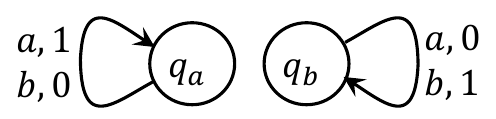}
        \caption{A tropical WFA that assigns to $w\in \{a,b\}^*$ the minimum between the number of $a$'s and $b$'s in $w$.}
        \label{fig:min num a num b}
    \end{minipage}
    \hfill
    \begin{minipage}{0.6\textwidth}
         \centering
    \includegraphics[width=0.85\linewidth]{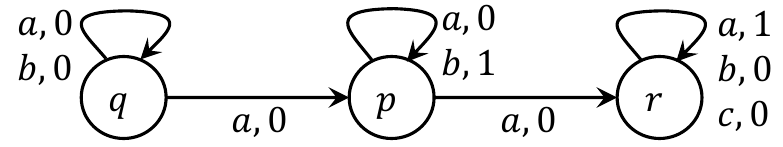}
    \caption{A tropical WFA $\cA$. All the states are initial. \cref{fig:abs:running dag example} depicts a run DAG of $\cA$.}
    \label{fig:running example}
    \end{minipage}
\end{figure}

\begin{figure}
  \begin{center}
    \includegraphics[width=0.7\textwidth]{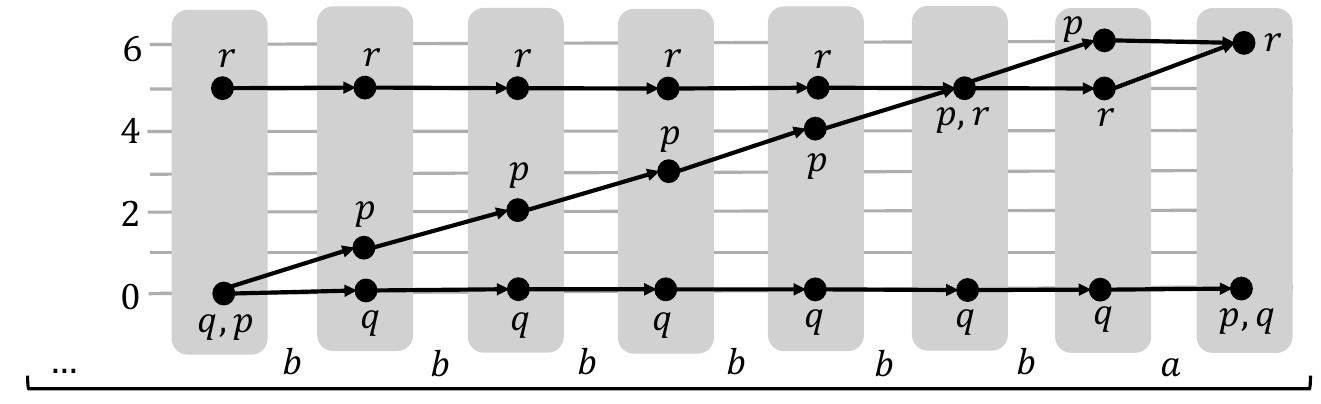}
  \end{center}
  \caption{After reading $bab^2ab^3ab^4ab^5a$, we arrive at the configuration $(0,0,5)$, as depicted. We then show how reading $b^6a$ increases the weight of $r$ by $1$. Note that this cannot be done with any shorter word.}
  \label{fig:abs:running dag example}
\end{figure}

We remark that there are several variants of WFAs: with or without initial and final weights, and with or without distinguished accepting states. 
We show that the determinizability problems for these variations are are inter-reducible (see \cref{rmk: initial and final weights}). Thus, 
our definition is of the simplest variant: no initial and final weights, and all states are accepting. 

\paragraph{Characterizing Determinizability by Gaps}
Intuitively, a WFA is determinizable if we can deterministically track its current \emph{configuration} -- a vector assigning to each state the minimal value it is reachable with (or $\infty$ if it is not reachable). That is, we need to update the current configuration upon reading a new letter (while incurring some weight). Then, the question is whether the number of states needed for this is finite.

A naive approach would be to track the minimal weight to each state. For example, in the WFA of \cref{fig:min num a num b}, given the word $aab$ we can track the vectors $(0,0)\runsto{a}(1,0)\runsto{a}(2,0)\runsto{b}(2,1)$ (where the two components represent the minimal weight to $q_a$ and to $q_b$, respectively).
Clearly, however, this requires an infinite state space. Also, it does not use the weights for each letter.
A simple optimization is therefore to always normalize the vectors such that the minimal state is $0$, and incur the normalization as weight. This makes sense, since we want to track the minimal runs.
Thus, we would track $aab$ as 
$(0,0)\runsto{a,0}(1,0)\runsto{a,0}(2,0)\runsto{b,1}(1,0)$, where in the last step we intuitively reach $(2,1)$, but ``shift'' all weights by $-1$, and therefore gain $1$ weight along the run.

However, this still clearly yields an infinite state space. The question is now whether we can make this finite. Intuitively, the state space can be made finite if there is some bound $B$, such that states in the configuration that go above $B$ in this normalized construction, i.e., go more than $B$ above the minimal run, cannot later become a minimal run, and can therefore be treated as ``unreachable''.

As it turns out, this is a precise characterization of determinization (it appears in \cite{filiot2017delay}, and is folklore). 
It is the starting point of our investigation, and we therefore make it explicit: for a bound $B\in \bbN$, a \emph{$B$-gap witness} is a pair of words $x,y\in \Sigma^*$ and a state $q\in Q$ such that upon reading $x$, we have $\minweight(x,Q_0\runsto{x}q)>\minweight(x,Q_0\runsto{x}Q)+B$, but $\minweight(xy,Q_0\runsto{xy}Q)=\minweight(xy,Q_0\runsto{x}q\runsto{y}Q)$ (see \cref{fig:abs:gap witness}). 
We then have the following.
\begin{theorem}
\label{thm:abs:det iff gaps}
    $\cA$ is \textbf{non}determinizable if and only if for every $B\in \bbN$, there is a $B$-gap witness.
\end{theorem}

Determinizability is in RE (recognizable): for a WFA $\cA$ and a deterministic WFA $\cD$, it is decidable whether they are equivalent~\cite{Almagor2020Whatsdecidableweighted}, so we can simply try all deterministic automata. Therefore, all that remains is to show that the existence of infinitely many $B$-gap witnesses, for increasing $B$, is also recognizable. This is the goal of our work.

At first glance, this problem seems easy: if there is a $B$-gap witness for very large $B$, surely we can apply some pumping argument to further increase this gap. This certainly holds for \cref{fig:min num a num b}. 
Unfortunately, things are far more involved, as runs can ``undercut'' one another.
\begin{example}
    \label{xmp:running pumping}
    Consider the WFA in \cref{fig:running example}. Intuitively, upon reading $b^k$, state $p$ gains weight $k$. Upon reading $a$, state $r$ gains $1$, but the weight in $p$ is ``reset'' to $0$ due to the run from $q$, and the run in $r$ is reset to the minimal run reaching $p$.  In particular, reading $aa$ resets all states to $0$.
    The letter $c$ ``kills'' all runs except that of $r$.

    Note that while reading $b^k$ for large $k$ brings $p$ very high, it is not a gap-witness, since we cannot make its run minimal. 
    In \cref{fig:abs:running dag example} we illustrate that more complicated gap-witnesses, namely words of the form $(b^na)_{n=1}^\infty$, can create large gaps, and form gap-witnesses with the suffix $c$, as they separate $r$ from the other states. Therefore, this WFA is nondeterminizable.

    We remark that \cref{fig:running example} can be used as a component in more elaborate constructions to require even more complicated gap-witnesses.
\end{example}

\section{Constructions and Definitions}
\label{sec:abs:constructions}
Our proof can be split into two significant elements: the definition of a well-behaved structure that allows a form of pumping, and a powerful set of techniques to find such well-behaved structures in a given WFA. 
In this section we present the former well-behaved structures, and begin to present some of the challenges that arise from them.

\subsection{The Baseline Augmented Construction (\cref{sec:augmented construction})}
\label{sec:abs:baseline aug}
Recall from \cref{sec:abs:det and gaps} the idea of tracking configurations and shifting so that the minimal run is always $0$. Our first construction is technically simple, but turns out to be a very versatile tool: instead of normalizing to the minimal run, we allow normalizing to \emph{any} run, referred to as the \emph{baseline} run, and the choice of run is given as part of the word. We also keep track of the reachable set of states, as in the Boolean subset construction. We refer to this as the \emph{Baseline-Augmented Subset Construction}, denoted $\augA$. 

More precisely, the states of $\augA$ (denoted $S$) are of the form $(q,p,T)\in Q\times Q\times 2^Q=S$, where $q$ is the ``usual'' current state, $p$ is the state in the baseline run, and $T$ is the reachable set of states. The alphabet of $\augA$, denoted $\Gamma$, is the \emph{transitions} of $\cA$. Then, from state $s=(q,p,T)$, reading the transition (i.e., letter) $\gamma=(p,\sigma,c,p')$, we have a transition to every $s'=(q',p',T')$ where $(q,\sigma,d,q')\in \Delta$, and $T'$ is updated deterministically. 
Note that the letter $\sigma$ (within $\gamma$) determines the options for the first component $q$, so that the first component simply tracks the ``standard'' runs of $\cA$, while the second component $p$ is determined by $\gamma$. 
Together, we have a transition $(s,\gamma, d-c,s')$ in $\augA$, where the
weight $d-c$ describes the weight $d$ of the corresponding transition in $\cA$, but ``shifted'' by the baseline weight $c$. 

Intuitively, for a word $x\in \Sigma^*$, we can choose any run $\rho$ of $\cA$ on $x$ as baseline, and then $\rho$ sets the ``perspective'' with which we consider the runs of $\cA$ on $x$. Crucially, these different perspectives do not affect the \emph{gaps} between the runs. 
Indeed, we prove (\cref{lem:A det iff augA det}) that $\cA$ is determinizable if and only if $\augA$ is, and from now on we proceed to work with $\augA$.

\subsection{Cactus Letters (\cref{sec:stable cycles and bounded behaviours,sec:cactus extension})}
\label{sec:abs:cactus letters}
A key difficulty in reasoning about weighted runs is their behavior in complex paths (c.f., \cref{xmp:running pumping}). Our most fundamental tools to handle this are \emph{stable cycles} and \emph{cactus letters}. Intuitively, cactus letters extend the alphabet by capturing nested cycles that can be ``easily pumped''. Unfortunately, their introduction leads to an \emph{infinite alphabet}, which becomes a central challenge later on.

A \emph{stable cycle} (\cref{sec:stable cycles}) is a pair $(S',w)$ where $S'\subseteq S$ is a subset of the states of $\augA$ and $w\in \Gamma^*$ is a word, such that reading $w$ from the states in $S'$ leads back to $S'$ (and therefore $w$ can be repeatedly read, i.e., pumped), and where the baseline run on $w$ is a minimal-weight cycle (with weight $0$). In particular, a stable cycle contains no negative cycles\footnote{The actual definition is somewhat more restrictive (\cref{def:stable cycle})}.

For example, in \cref{fig:abs:stable cycle} we illustrate a stable cycle $(S',u)$. The baseline run that cycles on $r$ gains weight $0$. The run from $q$ to $p$ gains negative weight, which is fine since it is not a cycle. The cycle from $t$ to $t$, must be non-negative.

\begin{figure}[ht]
    \centering
    \begin{minipage}{0.3\textwidth} 
    \centering
        \includegraphics[width=\textwidth]{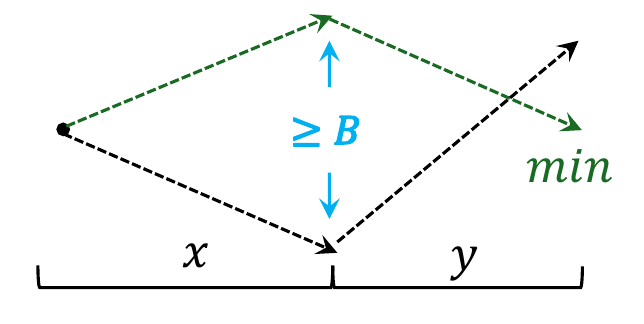}
        \caption{A $B$-gap witness.}
        \label{fig:abs:gap witness}
    \end{minipage}
    \hfill
    \begin{minipage}{0.3\textwidth} 
    \centering
        \includegraphics[trim=0 0pt 0 0, clip,scale=0.45]{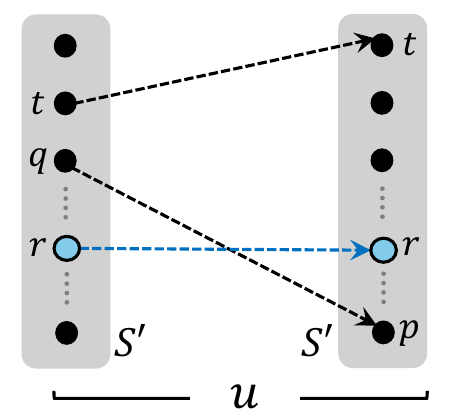}
        \caption{A stable cycle. The baseline $r$ is highlighted.}
        \label{fig:abs:stable cycle}
    \end{minipage}
    \hfill
    \begin{minipage}{0.3\textwidth} 
    \centering
        \includegraphics[trim=0 0 0 0, clip,scale=0.45]{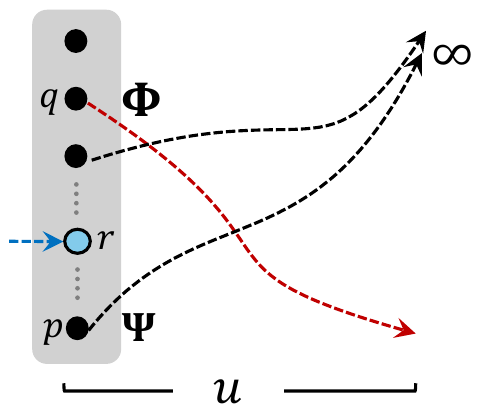}
        \caption{Potential and Charge, with baseline $r$.}
        \label{fig:abs:charge and pot}
    \end{minipage}
\end{figure}

For every stable cycle $(S',w)$, we identify pairs of states $p,q\in S$ called \emph{grounded} (\cref{def:grounded pairs}) such that reading $w$ a large number of times from $p$ to $q$ passes through a minimal cycle (with weight $0$), and so gains only bounded weight. 
We introduce a new \emph{cactus letter} $\alpha_{S',w}$ that has a transition only between grounded states of $(S',w)$, and represents a long repetition of $w$'s. The choice of weight for these transitions is delicate, and is omitted here (see~\cref{def:stabilization}).
Intuitively, reading $\alpha_{S',w}$ between grounded states represents pumping $w$ while remaining bounded (due to the $0$-weight cycle), while for pairs that are not grounded, pumping $w$ grows unboundedly (since all other cycles in $S'$ have positive weights).

Notice that we now have an infinite alphabet (since the words $w$ can be arbitrarily long). We now want to allow nesting of cactus letters, describing increasingly complex runs, and so we repeat the process using the new alphabet, and do this inductively ad-infinitum. Thus, a word $w$ in a cactus letter $\alpha_{S',w}$ may itself contain cactus letters. We end up with a WFA $\augA_\infty$ over an infinite alphabet $\Gamma_\infty$, whose letters represent unboundedly ``deep'' nested-cactus letters. 
We call a sequence of nested cactus letters a \emph{chain}.
\cref{fig:abs:cactus chain} depicts the formation of cactus letters and chains. 

\begin{figure}[ht]
  \begin{center}
    \includegraphics[trim=0 20pt 0 0, clip,width=0.8\textwidth]{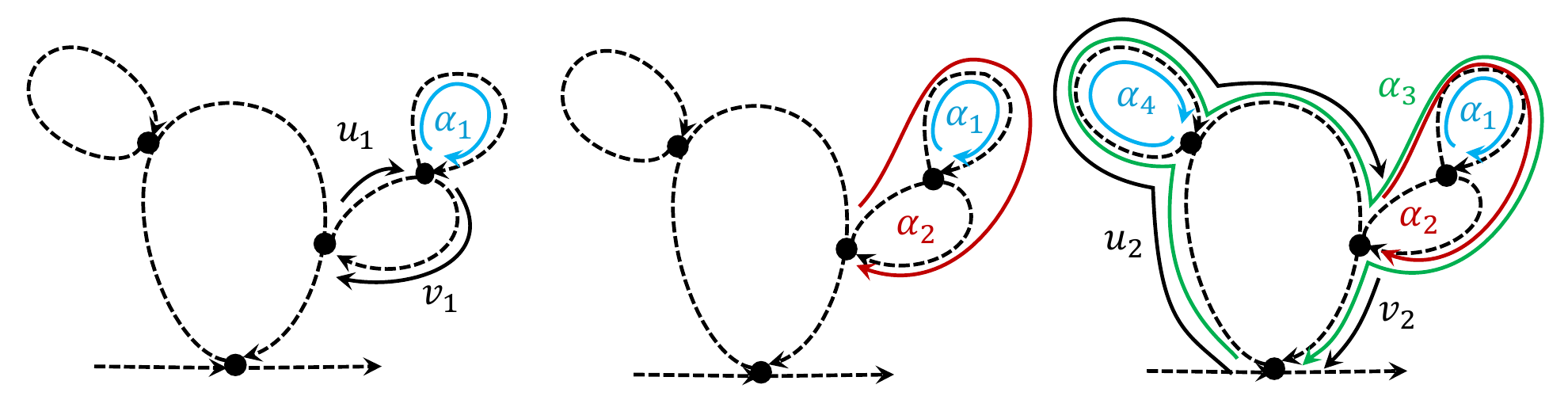}
  \end{center}
  \caption{Cactus letters and cactus chains: $\alpha_1$ is of depth $1$. The word $u_1\alpha_1v_1$ (left) forms a cactus letter $\alpha_2$ of depth $2$ (center). Then, $u_2\alpha_2v_2$ forms a cactus letter $\alpha_3$ of depth $3$ (right). Notice that $u_2$ contains the cactus letter $\alpha_4$. The sequence $\alpha_3,\alpha_2 ,\alpha_1$ is a cactus chain, and so is $\alpha_3,\alpha_4$. This also illustrates why they are called ``cactus letters''.}
  \label{fig:abs:cactus chain}
\end{figure}

The combination of the baseline construction with cactus letters yields a powerful and robust tool for reasoning about the behavior of $\augA_\infty$. Intuitively, it allows abstracting away parts of words as ``pumpable'', as well as shifting the perspective by replacing baselines. Concretely, in \cref{sec:cactus toolbox} we present three important tools:

\paragraph{A bound on chain depth}
Cactus chains can be arbitrarily deep. In our effort to make the alphabet finite, our first important insight is that if cactus chains become too deep, then they must contain \emph{degenerate} cactus letters. These  are letters where, intuitively, all pairs are grounded. Degenerate letters can be pumped ``for free'', in a sense that is made precise later. 
We remark that the proof of this result is nontrivial and involves interesting analysis of the strongly connected components of graphs that arise from the grounded pairs of the cactus letters on each level of the chain (\cref{sec:existence of degenerate cycle in chain}).

\paragraph{Baseline shifts} 
As mentioned in \cref{sec:abs:baseline aug}, by changing the baseline run we can shift the ``perspective'' on runs. In \cref{sec: baseline shift} we present a formal framework to do so, which also applies to cactus letters. However, cactus letters do not have a natural perspective shift. Thus, we introduce a new type of letters there, called \emph{rebase} letters, and another infinite hierarchy on top of the cactus hierarchy, which complicates things. We ignore such problems here, and leave them to the full proof.

\paragraph{Cactus unfolding}
A crucial usage of cactus letters is going back and forth between the different levels of abstraction: we can \emph{unfold} a cactus letter $\alpha_{S',w}$ to a word $w^n$ in a way that captures all the ``important'' features of runs. This, however, needs to be done carefully, and depends not only on the cactus letter, but on its surrounding prefix and suffix, e.g., $x\alpha_{S',w}y$. 
The framework in \cref{sec: cactus unfolding} captures this precisely. We remark that unfolding a cactus may change the reachable set of states, since abstracting away runs inevitably eliminates some of the original runs. This nuance, referred to as \emph{ghost states} ($\mathghost$) is considerably troublesome throughout the paper (see \cref{sec:reachable ghost states}).

One result of this toolbox is that while $\augA$ and ${\augA}_\infty$ are not equivalent (indeed, their alphabet is different), we can show that they are in fact equi-determinizable (\cref{cor:aug inf inf is det iff A is det}), so in particular it suffices to reason about the nondeterminizability of ${\augA}_\infty$.


\subsection{Potential and Charge (\cref{sec:dominance and potential})}
\label{sec:abs:pot and charge}
Recall that our goal is to reason about gap-witnesses. We now present two notions that take us a step closer. In a nutshell, the \emph{potential ($\pot$)} measures how far above the baseline run we can find ``relevant'' states, and the \emph{charge ($\charge$)} captures the weight of the current minimal run. 

We say that a state $q\in S$ is {\em dominant} in configuration $\vec{c}$, if there exists a word $u$ that has a finite-weight run from $q$, while all states below $q$ in $\vec{c}$ only have infinite-weight runs on $u$ (i.e., do not have runs).
Intuitively, $q$ is dominant in $\vec{c}$ if there is a suffix that ``kills'' all the states below it.
We say that $q$ is {\em maximal dominant} if $q$ has the maximal weight in $\vec c$ among all dominant states. 
The {\em potential} $\pot(w)$ of a word $w$ is the weight of a maximal dominant state in the configuration that $w$ reaches. 
Note that this weight represents the gap between the dominant state and the ($0$-weight) \emph{baseline} run (as opposed to the gap from the minimal state).

We define the {\em charge} $\charge(w)$ of $w$ to be the minimal weight in the configuration that $w$ reaches (negated, for technical reasons). Thus, $\charge(w)$ reflects the distance between the minimal run and the baseline run.
In \cref{fig:abs:charge and pot} we illustrate these concepts. The run from $r$ is baseline (notice that the baseline run might be above the potential). The weight of $q$ determines the potential, since $u$ sends all runs below it to $\infty$. The weight of $p$ is minimal, and determines the charge. 

A crucial property of the potential is that $\pot(w\gamma)$ cannot be much higher than $\pot(w)$ for a single letter $\gamma$. We refer to this property as \emph{bounded growth} (\cref{lem:bounded growth potential}). However, this property only holds when the alphabet is finite, which is no longer the case. 
We return to this later on.

Intuitively, the ``interesting'' value we wish to reason about is not the gap of a dominant state from the baseline (namely $\pot$), but the gap from the minimal run (to find a gap-witness), hence the definition of $\charge$.
The charge also has a ``bounded growth'' property, but it is of less interest. Importantly, the charge may \emph{decrease} abruptly, if the minimal run cannot continue and a much higher run becomes minimal.
Such behavior is very problematic for our proof (\cref{def:charge bounded decrease}). This is dealt with later.

For the proof, we need both the potential and the charge to ``play nice'' with the tools we develop earlier (namely baseline shifts and cactus unfolding). Developing such tools is a significant part of the work (see \cref{sec:growth of potential and charge}).

The technical reason for defining potential and charge is to break down the overall problem into two distinct settings. This becomes clearer as we explain the proof structure below. We mention that this is an important separation, without which the proof becomes unmanageable.

\subsection{A Witness for Nondeterminizability (\cref{sec:witness})}
\label{sec:abs:witness}
We are now ready to present our witness for nondeterminizability. Intuitively, the witness is 
a single word over $\Gamma_\infty$
that (finitely) represents an infinite set of $B$-gap witnesses for unbounded values of $B$, and hence implies nondeterminizability (i.e., it is \emph{sound}, \cref{lem:witness implies nondet}). 
Our main technical contribution in this work is to show that this characterization is \emph{complete}, i.e., that if $\augA_\infty$ is nondeterminizable, then it has such a witness.

A \emph{witness} (\cref{def:witness}) is a word $w_1w_2w_3\in \Gamma_\infty^*$ such that $w_1$ and $w_3$ are some prefix and suffix, and $w_2$ is a ``pumpable infix''. More precisely, $w_2$ can be replaced by a cactus letter $\alpha_{S',w_2}$ such that $w_1w_2w_3$ has finite weight in $\augA_\infty$, whereas $w_1\alpha_{S',w_2}w_3$ does not.

The existence of a witness suggests that replacing $\alpha_{S',w_2}$ with $w_2^k$ for some large $k$ will cause the non-grounded pairs to induce runs with unboundedly large weight. In addition, there cannot be any grounded pairs on $w_2$, as otherwise $\alpha_{S',w_2}$ has finite weight-transitions on them so $w_1\alpha_{S',w_2}w_3$ would attain finite weight (this follows the fact that the states reachable by $\alpha_{S',w_2}$ are a subset of those reachable by $w_2$, which is not strictly true, but is made precise and correct in the detailed proof).

Thus, $w_1w_2^kw_3$ is accepted only via non-grounded pairs. But the baseline run on $w_2$ remains low (at $0$), which means that after $w_1w_2^k$, there is a huge gap, after which $w_3$ causes the higher runs to continue with finite weight, and kills the baseline and lower runs. More precisely, these ``higher states'' are exactly the ghost-reachable states mentioned above: states that are not reachable via $\alpha_{S',w_2}$, but are reachable via $w_2^k$ with weight increasing as a function of $k$.
Therefore, a witness induces a family of $B$-gap witnesses for unbounded $B$. 
In \cref{lem:witness implies nondet} we give the precise details.

Once the definition of witness is set, we show an important result (\cref{lem:unfolding maintains potential}): for a word $u\alpha_{S',w}v$, if we replace $\alpha_{S',w}$ with $w^k$ for any large-enough $k$, then either $\pot(u\alpha_{S',w}v)=\pot(u w^k v)$, or there is a witness in $\cA_\infty$. 
The importance of this lemma is twofold. First, it (almost) shows that the potential ``plays nice'' with cactus unfolding, which is missing from the toolbox in \cref{sec:growth of potential and charge}, and second, it opens a line of  results of the form ``either something good happens, or there is a witness''. 
The latter is of importance because many failed attempts at this problem were due to too many cases where things ``go wrong''. Our techniques handle multiple such cases, but instead of things going wrong -- we end up witnesses (implying nondeterminizability).

\section{Nondeterminizable WFAs Have a Witness (\cref{sec:final nondet implies witness})}
\label{sec:abs:nondet implies witness}
We can now present the main proof structure. 
In the following, we present the proof outline, emphasizing the novel ideas. There are, however, several major results that we use without proof now, and discuss them after the proof in \cref{sec:abs:inc inf,sec:abs:discharging}. Note that the detailed proof has the opposite order. 
We remark that most details are left out of this sketch, but are referenced.

Consider a nondeterminizable WFA $\cA$. We reason about $\augA_\infty$ (as per their equi-determinizability) and show that it has a witness. This is done in two parts: first we prove in \cref{sec:abs:proof nondet to unbounded pot} that $\sup\{\pot(w)\mid w\in \Gamma^*\}=\infty$. That is, the potential over words in $\Gamma$ (without cactus letters) is unbounded in $\augA_\infty$. Next, we prove in \cref{sec:abs:proof unbounded pot to witness} that if the potential is unbounded, then $\augA_\infty$ has a witness.

\subsection{If $\cA$ is Nondeterminizable, Then the Potential is Unbounded}
\label{sec:abs:proof nondet to unbounded pot}
Since $\cA$ is nondeterminizable, then so is $\augA$, and it therefore has $B$-gap witnesses for every $B\in \bbN$. Thus, we have a sequence of words $\{x_ny_n\}_{n=1}^\infty$ such that the minimal run on $x_ny_n$ is at least $n$ above the minimal run after reading $x_n$ (where $n$ serves as the increasing bound $B$). 
Our first step is to change the baseline to a minimal run $\rho$ on $x_ny_n$. This means that $\rho$ now has weight $0$, so the minimal run on $x_n$ becomes very negative (\cref{fig:final baseline shift}). 
Observe that upon reading $y_n$, the charge $\charge$ decreases significantly (since the minimal run gets closer to the baseline). This is a useful property which we call \emph{discharging}, and is explained in \cref{sec:abs:discharging}. The implication of discharging is almost what we want: either the potential is unbounded, or there are words with infixes in which there are almost no decreasing runs (dubbed \emph{$D$-dip}, \cref{def:dip words}). If the potential is unbounded, we are done. 

A major step in the proof (\cref{sec:existence of separated increasing infixes}) is to show that the $D$-dip property implies the existence of a special structure called an \emph{increasing infix}, which we discuss in \cref{sec:abs:inc inf}. When such an infix is found, it can be turned into a cactus letter, in a process we refer to as \emph{budding} (\cref{sec:cactus budding}). Moreover, the budding process lowers the \emph{cost} -- a concept we now discuss.

The \emph{cost} of a word extends the notion of weight, by assigning a cactus letter $\alpha_{S',w}$ a cost much higher than that of $w$. In general, the cost upper-bounds the weight of any run on the word (in absolute value), not just the minimal one. 
The property of increasing infixes is that they can be turned into cactus letters while reducing the cost.

We now go back to the $y_n$ suffixes, which satisfy that $\charge(x_n)\gg \charge(x_ny_n)$. A crux of the proof is to replace each $y_n$ with a word $y'_n\in \Gamma_\infty$ (i.e., with cacti) that satisfies these charge constraints, but also has \emph{minimal cost}. Then, if we get a $D$-dip, we can find an increasing infix, and generate a new cactus. This reduces the cost of $y'_n$, leading to a contradiction to its minimality. 
Unfortunately, changing to $y'_n$ means that the alphabet is now infinite. This induces more complications, which we omit here.

We therefore conclude that the discharging sequence $y_n$ cannot have $D$-dips, so we have unbounded potential. We remark that the last part relies heavily on our results on bounded cactus chains, and in particular we show that we only need to unfold a bounded number of ``layers'' in cactus letters before reaching a finite alphabet simultaneously on all the $y'_n$ suffixes.

\subsection{If the Potential is Unbounded, then $\augA_\infty$ Has a Witness}
\label{sec:abs:proof unbounded pot to witness}
The second part resembles the structure of the first, and we therefore describe it even more briefly. Note, however, that the technical details of this part are actually more complicated than the first part above.

Since the potential is unbounded, we have a set of words $\{u_n\}_{n=1}^\infty$ over $\Gamma$ satisfying $\pot(u_n)\ge n$ for all $n\in \bbN$. 
We again replace them by words $\{u'_n\}_{n=1}^\infty$ over $\Gamma_\infty$ with minimal cost. This type of sequence is referred to as \emph{$\pot$-discharging} (due to the analogy with the first case). 
Similar techniques to the first case (again, the devil is in the details), allow us to conclude that there is either a witness or a $D$-dip. In the former case we are done, and in the latter we reach a contradiction to the minimal cost.

\subsection{Determinizability of WFA is Decidable}
\label{sec:abs:finale}
The decidability of determinizability now follows: we already mention above that determinizability is in RE. Thus, given a WFA $\cA$, we can recognize whether it is \textbf{non}determinizable by enumerating all tuples $(w_1,w_2,w_3)$ over $\Gamma_\infty$, and checking whether any of them form a witness (\cref{thm:decidability of determinization}). By our proof above, if $\cA$ (and therefore $\augA_\infty$) is indeed nondeterminizable, then there is a witness.

\section{Increasing Infix and Cactus Budding (\cref{sec:separated increasing infix,sec:cactus budding})}
\label{sec:abs:inc inf}
A separated increasing infix is a word $uxyv$ with a highly-organized structure (\cref{fig:abs:increasing infix}). Intuitively, after reading $u$, the states fall into ``sub-configurations'' $V_1,\ldots,V_k$, and this partition is kept after reading $x$ and reading $y$. Moreover, each sub-configuration increases by a fixed amount, or stays $0$. The three main properties of this structure are: 
\begin{enumerate}[itemsep=-3pt]
\item $y$ has much higher cost than $x$. In fact, higher than the cost of the cactus letter $\alpha_{S',x}$.
\item The gaps between the sub-configurations are huge. We can pump $x$ numerous times without the sub-configurations interfering with one another.
\item The behavior on $x$ and on $y$ is similar enough so we can replace $y$ with $x^k$ for some $k$.
\end{enumerate}

Intuitively, Conditions 2 and 3 allow us to pump $x$ enough times so that it effectively becomes $\alpha_{S',x}$ and also ``subsumes'' the behavior of $y$ (this intuition is made precise in the detailed proof). 
Then, we can replace $uxyv$ with $u\alpha_{S',x}v$, while maintaining similar behavior. 
Due to the huge cost of $y$ (Condition 1), this reduces the overall cost, and ``buds'' a cactus, leading to the contradiction in the main proof. 
This ``similar behavior'' mainly concerns the preservation of potential and charge, and is the focus of \cref{sec:cactus budding}.

\begin{figure}[ht]
  \begin{center}  
    \includegraphics[trim=0 10pt 0 0, clip,width=0.55\textwidth]{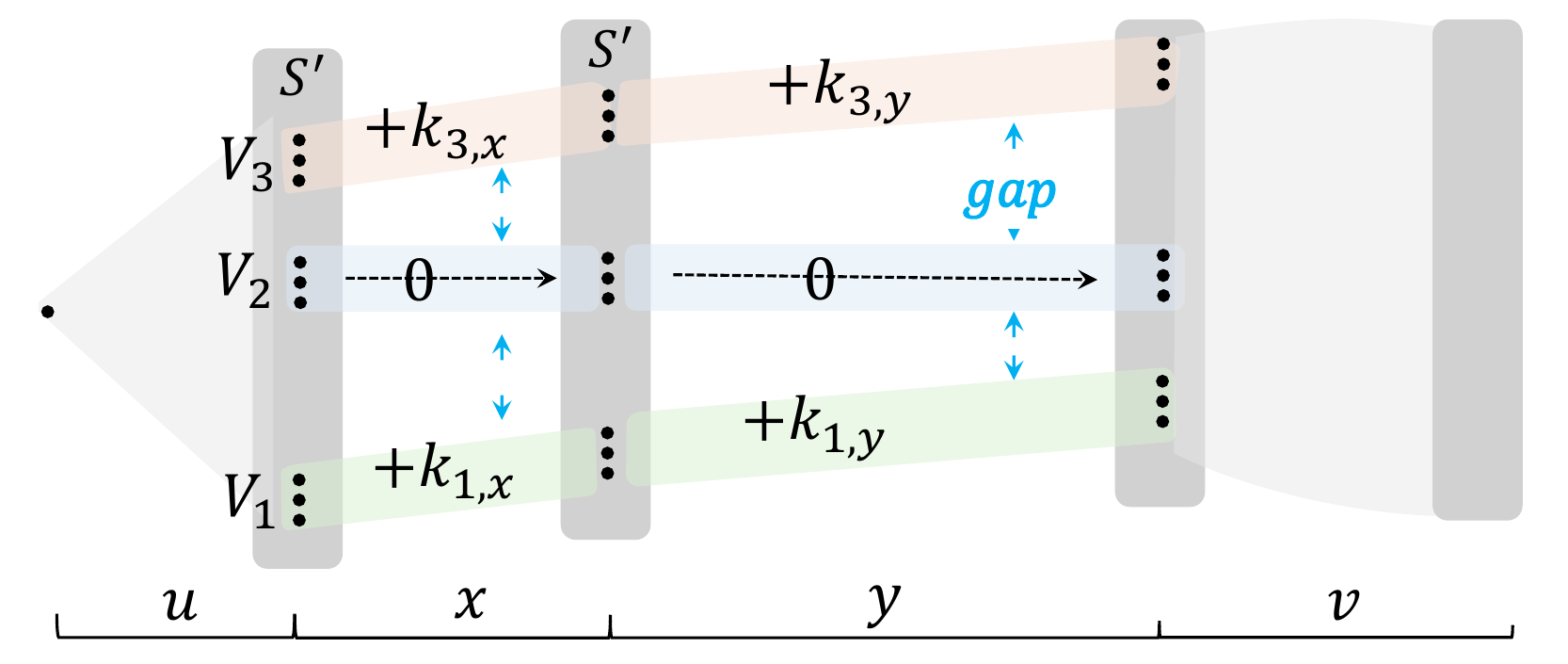}
  \end{center}
  \caption{A separated increasing infix. The runs in $V_1$ are all increased by $k_{1,x}$ upon reading $x$, and are all increased by $k_{1,y}$ upon reading $y$. Similarly for $V_3$ with $k_{3,x}$ and $k_{3,y}$. All the states in $V_2$ (the baseline run among them) remain in the same height upon reading both $x$ and $y$. The gaps between the sets is so large that we can pump $x$ enough times without e.g., any state in $V_1$ getting close to weight $0$.}
  \label{fig:abs:increasing infix}
\end{figure}

The most involved part is showing that such infixes exist (under the $D$-dip assumption, which we only name here without a proper definition). 
The first tool we develop to show their existence is a \emph{Ramsey-style theorem} for infinite colored sequences (\cref{sec:ramsey}). Notably, this is one part where our proof is purely existential. Therefore, in order to obtain complexity bounds (via our methods), this needs to be turned to a quantitative argument.
The second tool we develop (\cref{sec:existence of inc inf assuming covered,sec:existence of separated increasing infixes}) is a technique for ``untangling'' runs, by considering infinite sequences of words that get longer and longer, and then ``zooming-in'' on their infixes, so that the infixes also get longer, but also get simpler. We elaborate on this technique in \cref{sec:abs:discharging}.

The bottom line of this part is that in the presence of $D$-dip we can find an increasing infix $uxyv$ and safely replace it with the lower-cost $u\alpha_{S',x}v$.

\section{Discharging and Leveled Sequences (\cref{sec:leveled and discharging words,sec:potleveled and potdischarging words})}
\label{sec:abs:discharging}
One of our main conceptual contributions is a technique for ``untangling'' multiple runs, to find organized structures. We use this technique in several places (\cref{sec:leveled words sequences,sec:discharging word sequences,sec:potleveled words sequences,sec:potdischarging word sequences}), and we turn to briefly outline it, and then discuss its implications.

\subsection{The Zooming Technique (Fair Decompositions)}
\label{sec:abs:zooming}
Our goal is to find, under certain assumptions, an arbitrarily long infix $y$ and a word $xyz$, whose runs have the following properties: upon reading $y$, there is a small set of runs that do not ``interfere'' with each other (dubbed \emph{independent runs}, \cref{def:configuration independent runs}), and all other runs are within some small interval from one of these runs (dubbed \emph{$G$-cover}, \cref{def: config G cover words}). Moreover, the independent runs are extremely far apart, therefore allowing some pumping without risk of tangling runs. 
These properties are not dissimilar to those of increasing infixes, but the actual definitions have important differences, so we separate them also here.

On top of these requirements, we have an additional property we wish $y$ to satisfy, and this changes by instance. Specifically, the requirements are either that $\charge$ or $\pot$ remain within a small ``band'' (\emph{Leveled} or \emph{$\pot$-Leveled}) or that $\charge$ decreases / $\pot$ increases (\emph{Discharging} / \emph{$\pot$-Discharging}). 

The assumption under which we find such a structure is that we start with an infinite sequence of (unboundedly long) words satisfying only the desired additional property, but not necessarily inducing independent runs. 
Then, the key idea is to split an infix of each word into a number of parts, such that each part grows in length, the number of parts also grows, and the same general property can be found recursively in some inner part. 
This can be thought of as splitting $n$ into roughly $\sqrt{n}$ parts of size $\sqrt{n}$ (but is more involved, since we have other requirements from these parts). We refer to this as a \emph{fair decomposition} (e.g., \cref{def:decreasing gap fair decomposition}). Due to lack of space, illustrations of this are in the main proof (\cref{fig:discharging,fig:leveled}).

The main idea now is that if the independent runs do not form a $G$-cover, then there is a run that ``drifts'' far from other runs. We can then add it as an independent run and proceed inductively. 
Naturally, all phases in this construction have careful details to consider (in particular the ghost runs mentioned above are a challenge). 

\subsection{From Decompositions to Unbounded Potential and Witness}
\label{sec:abs:from decomposition to result horror}
The last key component of our proof (and the one most blatantly missing from the sketch in \cref{sec:abs:finale}) is how to use this leveled / discharging structure to either obtain unbounded potential over $\Gamma$, or a witness. This is where all our definitions come together (\cref{lem: decomposed seq with cover sparse no ghosts implies unbounded potential,lem: decomposed inc pot with cover sparse no ghosts implies type 1 witness}). 
We illustrate the idea by showing how we obtain unbounded potential given a fair-decomposition where an infix has decreasing charge. Again, we gloss over most of the details.

Consider a word $xy\in \Gamma_\infty^*$ where $y$ is very long, and upon reading $y$ there is a baseline run (of weight $0$) and the minimal run increases a lot, from below the baseline (i.e., $\charge$ decreases). Our goal is to show that $\sup\{\pot(w)\mid w\in \Gamma^*\}=\infty$. This is achieved in five steps, as illustrated in \cref{fig:proofHorror}.

\myparagraph{Step 1: Flatten $x$.} Our first step is to \emph{flatten} the prefix $x$, i.e., replace all its (nested) cactus letters by letters in $\Gamma$. This relies on our tools in \cref{sec: cactus unfolding}. Crucially, after flattening we maintain most of the properties we start with. Note that this does not explicitly appear in \cref{fig:abs:proofStep1}, since the alphabet is not depicted.

\myparagraph{Step 2: Baseline shift.} Next, we take the run on $y$ that decreases the most (the upper green run in \cref{fig:abs:proofStep1}), and we change our perspective to make it the baseline (\cref{sec: baseline shift}). Now, all the runs on the $y$ suffix are non-decreasing (intuitively, they are non-decreasing with respect to the run we took), as in \cref{fig:abs:proofStep2}.

\myparagraph{Step 3: Replace $y$ with $\alpha_{S',y}$.} Since all the runs on $y$ are non decreasing, then $y$ is already close to being a stable cycle (\cref{sec:stable cycles}). We show that it indeed is, and replace $y$ with the cactus $\alpha_{S',y}$ for some $S'$. This causes all the runs between non-grounded pairs to ``jump to $\infty$''. In particular, since the minimal run was increasing, we show that it is between non-grounded pairs, and hence jumps to $\infty$. Moreover, we start with a $G$-cover, so all the runs near this minimal run also jump to $\infty$, and the next-lowest run (that could possibly not jump to $\infty$) has a huge gap above this sub-configuration.

\myparagraph{Step 4: Baseline shift to minimal run.} We now almost have unbounded potential: some run above the minimal run stays bounded (namely the baseline run, or a lower run), but the minimal run increases unboundedly (as it is non-grounded). To conform with the definition of $\pot$, we need the increase to be above the baseline run. We therefore shift the baseline again, this time starting from the minimal run on $x$. We remark that this step requires the introduction of some additional constructs (\emph{jump letters} \cref{sec:jump letters}).

\begin{figure}[ht]
    \centering
    \begin{subfigure}{0.23\textwidth}
        \centering
        \includegraphics[width=\textwidth]{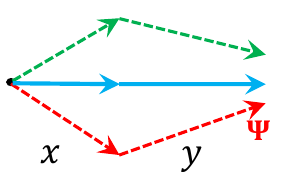}
        \caption{Step 1}
        \label{fig:abs:proofStep1}
    \end{subfigure}
    \begin{subfigure}{0.23\textwidth}
        \centering
        \includegraphics[width=\textwidth]{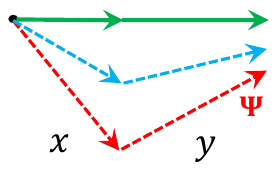}
        \caption{Step 2}
        \label{fig:abs:proofStep2}
    \end{subfigure}    
    \begin{subfigure}{0.23\textwidth}
        \centering
        \includegraphics[scale=0.8]{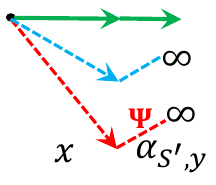}
        \caption{Step 3}
        \label{fig:abs:proofStep3}
    \end{subfigure}
    \begin{subfigure}{0.23\textwidth}
        \centering
        \includegraphics[scale=0.8]{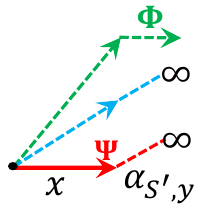}
        \caption{Step 4}
        \label{fig:abs:proofStep4}
    \end{subfigure}
    \caption{The steps of \cref{sec:abs:from decomposition to result horror} (\cref{lem: decomposed seq with cover sparse no ghosts implies unbounded potential}). The dashed lines are runs, the solid line is baseline.}
    \label{fig:proofHorror}
\end{figure}

By applying this technique to arbitrarily long $y$, and since $x$ was flattened to be over $\Gamma$, we obtain our unbounded potential (\cref{fig:abs:proofStep4}), with $\alpha_{S',y}$ (roughly) being the suffix that kills lower runs. 
This completes the holes left in the main proof in \cref{sec:abs:finale}, in very high level.

\section{Discussion}
\label{sec:abs:discussion}
The methods developed in this work represent a new type of reasoning for weighted automata: on top of standard (and even nested) ``cycle analysis'', we introduce the concept of baseline shifts, allowing versatile manipulations on runs (e.g., making all runs non-negative by shifting on the most decreasing one). We combine this with the novel cactus letters (and their budding framework), which allow both abstracting away behaviors as well as unfolding to get back those behaviors. This results in a powerful toolbox for reasoning about the runs of a WFA. 

Our Zooming technique is extensively used in this work, and we believe that it is generally useful for quantitative models. Indeed, it allows us to ``group'' runs together in a way that enables pumping. The latter is a central challenge in many proofs using quantitative models. 

Overall, we expect the methods and techniques introduced in this work to have far-reaching applications, beyond resolving the long-standing open problem of WFA determinizability. For example, the seminal work proving that boundedness of WFAs is decidable~\cite{Has82,LP04} is completely subsumed by our framework (however, note that (1) we do not provide complexity bounds and (2) it can also be obtained by Simon's Forest Factorization~\cite{simon1990factorization}).

Finally, now that the decidability of this problem is resolved, the race to establish its complexity can begin (with the current lower bound being PSPACE-hard as per \cref{apx:PSPACE hard}). 
An upper bound using our methods would require a quantitative Ramsey argument, as a first fundamental step (but would not suffice, as there are other existential steps).

In the next sections we present the complete work, with rigorous and detailed proofs (which account for the length of this paper). We start with a short reading guide.


\newpage
\pagestyle{fancy} 

\section{Preliminaries}
\label{sec:prelim}
\paragraph{A Reading Guide}
Due to the daunting length and complexity of our proof, we take several measures to assist the reader:
\begin{itemize}
    \item The current section is denoted in the header.
    \item Proof environments are colored in {\color{violet!75!black} purple hue, to distinguish them from other text}.
    \item We mark some results and definitions with \keyicon if they are key ideas, and with \lightbulbicon if their proof has special insights. We recommend focusing on \keyicon statements, and on \lightbulbicon proofs, for a first read.
    \item Most of our definitions use english terms rather than symbols, to assist in remembering their meaning. Very few (and widely used) ones are left  as symbols.
    \item Lastly, we note that on most PDF readers, the keyboard shortcut ``Alt + Left Arrow'' or ``Cmd + Left Arrow'' goes back after following a hyperlink (e.g., a reference to a definition). 
\end{itemize}

For an alphabet $\Sigma$, we denote by $\Sigma^*$ (resp. $\Sigma^+$) the set of finite words (resp. non-empty finite words) over $\Sigma$. 
For a word $w\in \Sigma^*$, we denote its length by $|w|$ and the set of its prefixes by $\pref(w)$. We write $w[i,j]=\sigma_i\cdots \sigma_j$ for the infix of $w$ corresponding to $1<i<j\le |w|$.

We denote by  $\bbNinf$ and $\bbZinf$ the sets $\bbN\cup \{\infty\}$ and $\bbZ\cup \{\infty\}$, respectively. We extend the addition and $\min$ operations to $\infty$ in the natural way: $a+\infty=\infty$ and $\min\{a,\infty\}=a$ for all $a\in \bbZinf$. By $\arg\min\{f(x)\mid x\in A\}$ we mean the set of elements in $A$ for which $f(x)$ is minimal for some function $f$ and set $A$.

\paragraph{Weighted Automata}
A \emph{$(\min,+)$ Weighted Finite Automaton} (WFA for short) is a tuple $\cA= \tup{Q,\Sigma, q_0, \Delta}$ with the following components:
\begin{itemize}
    \item $Q$ is a finite set of \emph{states}.
    \item $q_0\in Q$ is the \emph{initial} state.\footnote{Having a set of initial states does not add expressiveness, as it can be replaced by a single initial state that simulates the first transition from the entire set.}
    \item $\Sigma$ is an \emph{alphabet}. We remark that we often use an \emph{infinite alphabet} $\Sigma$, hence we do not require it to be finite. 
    \item $\Delta\subseteq Q\times \Sigma\times \bbZinf\times Q$ is a transition relation such that for every $p,q\in Q$ and $\sigma\in \Sigma$ there exists exactly\footnote{This is without loss of generality: if there are two transitions with different weights, the higher weight can always be ignored in the $(\min,+)$ semantics. Missing transitions can be introduced with weight $\infty$.} one weight $c\in \bbZinf$ such that $(p,\sigma,c,q)\in \Delta$.
\end{itemize}
If for every $p\in Q$ and $\sigma\in \Sigma$ there exists at most one transition $(p,\sigma,c,q)$ with $c\neq \infty$, then $\cA$ is called \emph{deterministic}.

\begin{remark}[On initial and final weights, and accepting states]
    \label{rmk: initial and final weights}
    Weighted automata are often defined with initial and final weights, i.e., 
    $q_0$ is replaced with an initial vector $\init\in \bbZinf^Q$ (and in particular may have several initial states with finite weight), and there are designated accepting states or a final weight vector $\fin\in \bbZinf^Q$.
    Then, the weight of a run also includes the initial weight and final weight (which may be $\infty$).

    In \cref{sec:no need for init and fin and acc} we show that the determinization problem for this general model can be reduced to that of our setting. Therefore, it is sufficient to consider our model, without initial and final weights, and with a single initial state. 
    We remark that while replacing $\init$ and $\fin$ with initial and accepting states is straightforward, making \emph{all} states accepting is nontrivial.
\end{remark}

\paragraph{Runs}
A \emph{run} of $\cA$ is a sequence of transitions $\rho=t_1,t_2,\ldots,t_m$ where $t_i=(p_i,\sigma_i,c_i,q_i)$ such that $q_i=p_{i+1}$ for all $1\le i<m$ and $c_i <\infty$ for all $1\le i \le m$.
We say that $\rho$ is a run \emph{on the word $w=\sigma_1\cdots\sigma_m$ from $p_1$ to $q_m$}, and we denote $\rho:p_1\runsto{w}q_m$. 
For an infix $x=w[i,j]$ we denote the corresponding infix of $\rho$ by $\rho[i,j]=t_i,\ldots,t_j$ (and sometimes by $\rho(x)$, if this clarifies the indices).
The \emph{weight} of the run $\rho$ is $\weight(\rho)=\sum_{i=1}^m c_i$. 

For a word $w$, the weight assigned by $\cA$ to $w$, denoted $\weight_{\cA}(w)$ is the minimal weight of a run of $\cA$ on $w$. For convenience, we introduce some auxiliary notations.

For a word $w\in \Sigma^*$ and sets of states $Q_1,Q_2\subseteq Q$, denote
\[\minweight_\cA(w,Q_1\to Q_2)=\min\{\weight(\rho)\mid \exists q_1\in Q_1,q_2\in Q_2,\ \rho:q_1\runsto{w}q_2\}\]
If $Q_1$ or $Q_2$ are singletons, we denote them by a single state (e.g., $\minweight_\cA(w,P\to q)$ for some set $P\subseteq Q$ and state $q$).
We can now define
\[
\weight_\cA(w)=\minweight_\cA(w,q_0\to Q)
\]
If there are no runs on $w$, then $\weight_\cA(w)=\infty$.
The function $\weight_\cA:\Sigma^*\to \bbZinf$ can be seen as the weighted analogue of the \emph{language} of an automaton.
We omit the subscript $\cA$ when clear from context.

For a word $w\in \Sigma^*$ and a state $q\in Q$, of particular interest are runs from $q_0$ to $q$ on $w$ that remain minimal throughout. Indeed, other runs are essentially ``cut short'' by some lower run. Formally, let $w=\sigma_1\cdots \sigma_n$ and consider a run $\rho:q_0\runsto{w}q$ with $\rho=t_1,t_2,\ldots,t_n$ and $t_i=(q_i,\sigma_i,e_i,q_{i+1})$. We say that $\rho$ is \emph{seamless} if for every $1\le j\le n$ it holds that $\weight(t_1,\ldots,t_j)=\minweight(\sigma_1\cdots \sigma_j,q_0\to q_{j+1})$.

We define $\wmax{w}$ to be the maximal weight (in absolute value) occurring in any possible transition on the letters of $w$, i.e., 
\[\wmax{w}=\max\{|c|<\infty \mid \exists p,q\in Q, 1\le i\le n,\  (p,\sigma_i,c,q)\in \Delta\}\] 
Additionally, the \emph{maximal effect} of $w$ is an upper bound on the change in weight that can happen upon reading $w$, defined as $\maxeff{w}=\sum_{i=1}^n\wmax{\sigma_i}$.
Thus, any finite-weight run $\rho$ on $w$ satisfies $|\weight(\rho)|\le \maxeff{w}\le \wmax{w}|w|$.

We write $p\runsto{w}q$ when there exists some run $\rho$ 
such that $\rho:p\runsto{w}q$.
We lift this notation to concatenations of runs, e.g., $\rho: p\runsto{x}q\runsto{y}r$ means that $\rho$ is a run on $xy$ from $p$ to $r$ that reaches $q$ after the prefix $x$. We also incorporate this to $\minweight$ by writing e.g., $\minweight(xy,q\runsto{x}p\runsto{y}r)$ to mean the minimal weight of a run $\rho:q\runsto{x}p\runsto{y}r$.

For a set of states $Q'\subseteq Q$ and a word $w\in \Sigma^*$, we define the {\em reachable set of states from $Q'$ upon reading $w$} as $\delta_{\bbB}(Q',w)=\{q\in Q\mid \exists p\in Q',\ p\runsto{w}q\}$ (where $\bbB$ stands for ``$\bbB$oolean'').

A WFA is \emph{trim} if every state is reachable from $q_0$ by some run. Note that states that do not satisfy this can be found in polynomial time (by simple graph search), and can be removed from the WFA without changing the weight of any accepted word. Throughout this paper, we assume that all WFAs are trim.

\paragraph{Determinizability}
We say that WFAs $\cA$ and $\cB$ are \emph{equivalent} if $\weight_\cA\equiv \weight_\cB$. We say that $\cA$ is \emph{determinizable} if it is equivalent to some deterministic WFA. Our central object of study is the following problem.
\begin{problem}[WFA Determinizability]
\label{prob:determinizability}
    Given a WFA $\cA$ over a finite alphabet $\Sigma$, decide whether $\cA$ is determinizable.
\end{problem}
The determinization problem is PSPACE-hard. This does not seem to be published anywhere, and we therefore provide a proof in \cref{apx:PSPACE hard} (we recommend reading this proof only after \cref{thm:det iff bounded gap}, as it relies on it).

\paragraph{Configurations}
A \emph{configuration} of $\cA$ is a vector $\vec{c}\in \bbZinf^Q$ which, intuitively, describes for each $q\in Q$ the weight $\vec{c}(q)$ of a minimal run to $q$ thus far (assuming some partial word has already been read). 
For a state $q$ we define the configuration $\vec{c_q}$ that assigns $0$ to $q$ and $\infty$ to $Q\setminus\{q_0\}$.
Intuitively, before reading a word, $\cA$ is in the configuration $\vec{c_\init}=\vec{c_{q_0}}$. 

We adapt our notations to include a given starting configuration $\vec{c}$, as follows.
Given a configuration $\vec{c}$ and a word $w$, they induce a new configuration $\vec{c'}$ by assigning each state the minimal weight with which it is reachable via $w$ from $\vec{c}$. We denote this by 
$\xconf(w,\vec{c})(q)=\minweight_{\vec{c}}(w,Q\to q)$ for every $q\in Q$.
In particular, $\xconf(w,\vec{c}_{\init})$ is the configuration that $\cA$ reaches by reading $w$ along a minimal run. 

We use the natural component-wise partial order on configurations. That is, for two configuration $\vec{c},\vec{d}$, we say that $\vec{d}$ is \emph{superior} to $\vec{c}$, denoted $\vec{c}\le \vec{d}$, if $\vec{c}(q)\le \vec{d}(q)$ for every $q\in Q$. We also denote the \emph{support} of a configuration by $\supp(\vec{c})=\{q\in Q\mid \vec{c}(q)<\infty\}$.

For a run $\rho:p\runsto{w}q$ and configuration $\vec{c}$, we define the \emph{weight of $\rho$ from $\vec{c}$} to be $\weight_{\vec{c}}(\rho)=\vec{c}(p)+\weight(\rho)$. We similarly extend the notation $\minweight$ to include configurations, by writing
\[
\minweight_{\vec{c}}(w,Q_1\to Q_2)=\min\{\weight_{\vec{c}}(\rho) \mid \exists q_1\in Q_1,q_2\in Q_2,\ \rho:q_1\runsto{w}q_2\}
\]
(we remark that we never need both the subscript $\cA$ and $\vec{c}$ together).
We say that $\rho$ is \emph{seamless from $\vec{c}$} if for every $1\le j\le n$ it holds that $\weight_{\vec{c}}(t_1,\ldots,t_j)=\minweight_{\vec{c}}(\sigma_1\cdots \sigma_j,T\to q_{j+1})$ where $T=\supp(\vec{c})$.

\subsection{A Characterization of Determinizability}
\label{sec:determinizability equiv characterization}
The following theorem gives an equivalent characterization of determinizability by means of runs, and of how far two potentially-minimal runs can get away from one another, which we commonly refer to as a \emph{gap}. 
Intuitively, we show that $\cA$ is determinizable if and only if there is some bound $B\in \bbN$ such that if two runs on a word $x$ obtain weights that are far from each other, but the ``upper'' run can still become minimal over some suffix, 
then the runs are at most $B$ apart. This is depicted in \cref{fig:abs:gap witness}.

We remark that this theorem is folklore, but to our knowledge has not been stated explicitly anywhere. Moreover, since we often treat infinite alphabets, we need to tweak the statement accordingly.
A similar characterization (with finite alphabet) was given for \emph{discounted-sum automata}~\cite{almagor2024determinization}.
\begin{definition}[$B$-Gap Witness]
\label{def: B gap witness}
    For $B\in \bbN$, a \emph{$B$-gap witness over alphabet $\Sigma'$} consists of a pair of words $x,y\in \Sigma'^*$ and a state $q\in Q$ such that
    \begin{itemize}
        \item $\minweight_\cA(x\cdot y,q_0\to Q)=\minweight_\cA(x\cdot y,q_0\runsto{x}q\runsto{y}Q)<\infty$, but 
        \item $\minweight_\cA(x,q_0\to q)-\minweight_\cA(x,q_0\to Q')>B$.
    \end{itemize}
\end{definition}
\noindent Then, the characterization is as follows.
\begin{theorem}
    \label{thm:det iff bounded gap}
    Consider a WFA $\cA$, then the following hold.
    \begin{enumerate}
        \item If there exists $B\in \bbN$ such that there is no $B$-gap witness over $\Sigma$, then $\cA$ is determinizable.
        \item If there is a finite alphabet $\Sigma'\subseteq \Sigma$ such that for every $B\in \bbN$ there is a $B$-gap witness over $\Sigma'$, the $\cA$ is not determinizable.
    \end{enumerate}
 \end{theorem}
We prove the theorem in \cref{sec:apx det iff bounded gap}. We remark that the proof there considers also initial and final weights, as it is later used in \cref{sec:no need for init and fin and acc} to establish \cref{rmk: initial and final weights}.

Observe that for WFAs over a finite alphabet, \cref{thm:det iff bounded gap} provides an exact characterization of determinizability.

\section{The Baseline-Augmented Subset Construction}
\label{sec:augmented construction}
We now construct a WFA $\augA$ that augments $\cA$ with additional information regarding its runs. This information is kept in both the states and in the alphabet.
Technically, the alphabet of $\augA$ consists of the \emph{transitions} of $\cA$, i.e. $\Delta$. Each state of $\augA$ contains information about the current state, the current reachable subset of states, as well as a special \emph{baseline} state. 
Intuitively, the baseline states aim to describe the minimal-weight run on the current word. 

We proceed with the precise construction.
Let $\cA= \tup{Q,\Sigma, q_0, \Delta}$. We define $\augA=\tup{\augStates,\Delta,\augInitState,\augTrans}$ with the following components:
\begin{itemize}
    \item The states are 
    $\augStates=\{(q,p,T) \in Q \times Q \times 2^Q \mid q,p \in T\}$.
    We refer to the components of a state $(q,p,T)\in S$ as the \emph{inner state} $q$, the \emph{baseline component} $p$, and the \emph{reachable subset} $T$.
    \item The alphabet of $\augA$ is the set of transitions $\Delta$ of $\cA$. For clarity, we denote elements of $\Delta$ using either explicit tuples, or by $\tau$. In practice, we consider reading only transitions of finite weight, but this is enforced in the transition relation.
    \item The initial state is $(q_0,q_0,\{q_0\})$.
    \item The transitions $\augTrans$ are defined as follows. 
    Consider states $(p_1,q_1,T_1)$ and $(p_2,q_2,T_2)$ in $\augStates$ and a letter $(q_1,\sigma,c,q_2)\in \Delta$ with $c\neq \infty$
    (i.e., the transition in $\cA$ from $q_1$ to $q_2$ on some $\sigma\in \Sigma$). 
    If $T_2=\booltrans(T_1,\sigma)$, then consider the unique transition $(p_1,\sigma,c',p_2)\in \Delta$. We add to $\augTrans$ the transition $((p_1,q_1,T_1),(q_1,\sigma,c,q_2),c'-c,(p_2,q_2,T_2))$.


    Finally, we complete any undefined transitions with $\infty$.
    
\end{itemize}
To gain some intuition on the construction (specifically on the transition function), observe that for a state $(p_1,q_1,T_1)$ and letter $(q,\sigma,c,q')$, the transition on the second and third component is entirely deterministic (assuming it has finite weight): any state $(p_2,q_2,T_2)$ reached must satisfy $T_2=\booltrans(T_1,\sigma)$, and it must hold that $q_1=q$ and $q_2=q'$.
The first component simulates $\cA$: from $p_1$ we can reach any state $p_2$ that satisfies $(p_1,\sigma,c',p_2)$ for some $c'$. 
Finally, the cost of the transition is $c'-c$: this reflects the fact that we ``normalize'' all weights so that if we had started at state $p$, the cost of the transition is $0$.

The remainder of our analysis is carried out on $\augA$. We now show that this is sound and complete.

We now elaborate further on the behavior of $\augA$. Recall that $\augA$ reads a sequence $x$ of transitions of $\cA$ (and this sequence must form a run in order to have finite weight). Each state in $\augA$ keeps track of two runs of $\cA$, and of the set of reachable states. This is done as follows.
Each transition $(q_1,\sigma,c,q_2)$ that is read (as a letter for $\augA$) induces the letter in $\sigma\in \Sigma$ that $\cA$ reads, but in addition ``suggests'' a transition on $\sigma$ from $q_1$ to $q_2$. 
Then, this suggestion is only ``accepted'' in states whose baseline component is $q_1$, i.e., $(p_1,q_1,T_1)$ for some $p_1,T_1$ with $p_1,q_1\in T_1$. 
This suggestion means that $\augA$ can now move to any state $p_2$ reachable from $p_1$ via $\sigma$, while the baseline component is now $q_2$. 
Further observe that upon reading a word $x$, all the states reachable by $\augA$ are of the form $(\cdot,p,T)$, i.e., share the same second and third components. Moreover, the possible reachable inner states (the first component) are exactly those in $T$, which are the states reachable by $\cA$ upon reading the letters in the word in $\Sigma^*$ induced by $x$. 
Thus, $\augA$ essentially simulates the run tree of $\cA$, with some additional information.

Of particular importance are states of the form $(q_1,q_1,T_1)$, which are reachable and can accept the suggestion $(q_1,\sigma,c,q_2)$. Then, $\augA$ may move to $(q_2,q_2,T_2)$ (as well as possibly to other states). Due to the weight normalization, the weight of this transition in $\augA$ is $c-c=0$. Such states are central to our analysis.
\begin{definition}[Baseline States and Baseline Runs]
    \label{def:baseline states}
    A state in $S$ is a \emph{baseline} if it is of the form $(p,p,T)$ for some $p\in T$. 
    A run $\rho=t_1\cdots t_m$ of $\augA$ with $t_i=(s_i,\tau_i,d,s_{i+1})$ is called a \emph{baseline} if $s_i$ is a baseline for all $i$.
\end{definition}

\begin{observation}
    \label{obs:baseline runs have weight 0}
    Recall that the weight of a transition in $\augA$ on letter $(p,\sigma,c,q)$ is normalized by reducing $c$ from any transition. In particular, a transition between baseline states $(p,p,T)\runsto{\sigma}(q,q,T)$ has weight $0$. Consequently, if $\rho$ is a baseline run, then $\weight(\rho)=0$. Note that there may still be other runs of both positive and negative weights.
\end{observation}
The intuition behind the baseline run is that it forms a ``point of view'' (POV) of the various runs of $\cA$ on a given word, in the sense that the weights of the transitions in $\cA$ are shifted so that the baseline run has weight $0$, and all other runs are shifted accordingly. We formalize this intuition in \cref{sec:shift POV}.
In particular, this enables us to show that $\cA$ is determinizable if and only if $\augA$ is determinizable.

\subsection{Shifts Between $\cA$ and $\augA$}
\label{sec:shift POV}
Intuitively, the runs of $\cA$ and of $\augA$ have the same relative structure, up to the choice of baseline run dictated in $\augA$. 
Indeed, for a word $w\in \Sigma^*$ and a run $\rhobase$ of $\cA$ on $w$, when $\augA$ reads $\rhobase$, each run corresponds to a run of $\cA$ on $w$, with the weights shifted according to $\rhobase$. In particular, the difference in weights between any two runs of $\augA$ on $\rhobase$ is the same as the difference in weights between the corresponding runs of $\cA$ on $w$.
Conversely, for a word $\aug{w}\in \Delta^*$, we can associate a word $w\in \Sigma^*$ such that $\aug{w}$ is a run of $\cA$ on $w$, and such that each run $\aug{\rho}$ of $\augA$ on $\aug{w}$ corresponds to a run $\rho$ of $\cA$ on $w$, and such that the difference in weight is similarly maintained as above.

We now formalize this intuition.
Consider a word\footnote{Observe that the definition does not actually depend on $w$, as $w$ is already induced by $\rhobase$.} $w\in \Sigma^*$ and a run $\rhobase=t_1,t_2,\ldots,t_n$ of $\cA$ on $w$ with $t_i=(q_i,\sigma_i,c_i,q_{i+1})$. We think of the run $\rhobase$ as a word read by $\cA$. 
Consider a run $\augrho=\aug{r_1},\aug{r_2},\ldots, \aug{r_n}$ of $\augA$ on $\rhobase$ where for all $1\le i< n$ we have
\[\aug{r_i}=((p_i,q_i,T_i),t_i,d_i,(p_{i+1},q_{i+1},T_{i+1}))\in \augTrans\]
The \emph{$\rhobase$-shifted run of $\cA$ corresponding to $\augrho$},
denoted $\shifttoorig(\augrho)$, 
is then the run $\rho=t'_1,t'_2,\ldots, t'_n$ where $t'_i=(p_i,\sigma_i,d_i+c_i,p_{i+1})$.
Note that $\rho$ follows the inner states of the run $\augrho$, but has weights shifted by $\rhobase$. 

We first observe that $\rho$ is indeed a run of $\cA$ on $w$: the initial state of $\augrho$ is $\augInitState=(q_0,q_0,\{q_0\})$, and therefore the initial state of $\rho$ is $q_0$. For the transitions, for every $1\le i< n$, since $\aug{r_i}\in \augTrans$ is a transition on the letter $t_i$, then (by the definition of $\augTrans$) we have that $(p_i,\sigma_i,c'_i,p_{i+1})$ for $c'_i$ such that $c'_i-c_i=d_i$. But then $c'_i=c_i+d_i$, so the transition $t'_i$ is indeed a valid transition in $\Delta$.

We now show that $\rhobase$-shifted runs maintain the weight differences of their origin runs. Intuitively, this holds because both runs are shifted by the same corresponding weights of $\rhobase$.
\begin{proposition}
\label{prop: pov shift from aug to A maintains distance}
    Consider a run $\rhobase$ of $\cA$ and two runs $\augrho=\aug{r_1},\aug{r_2},\ldots,\aug{r_n}$ and $\aug{\mu}=\aug{x_1},\aug{x_2},\ldots,\aug{x_n}$ of $\augA$ on $\rhobase$. 
    Let $\rho=r'_1,r'_2,\ldots,r'_n$ and $\mu=x'_1,x'_2,\ldots,x'_n$ be 
    $\shifttoorig(\augrho)$ and $\shifttoorig(\aug{\mu})$, respectively,
    then for every $1\le i\le n$ we have
    \[\weight_{\augA}(\aug{r_1},\aug{r_2},\ldots,\aug{r_i})-\weight_{\augA}(\aug{\mu}=\aug{x_1},\aug{x_2},\ldots,\aug{x_i})=\weight_\cA(r'_1,r'_2,\ldots,r'_i)-\weight_\cA(x'_1,x'_2,\ldots,x'_i)\]
\end{proposition}
\begin{proof}\ifproofs   
    The proof follows easily by induction on $i$. The base case is trivial: for $i=0$ all runs accumulate $0$ weight, so the differences are both $0$. 
    Assume correctness for $i-1$, we prove for $i$.
    By the induction hypothesis, up to prefix $i$ the differences are maintained, and thus it suffices to prove that
    \[\weight_{\augA}(\aug{r_i})-\weight_{\augA}(\aug{x_i})=\weight_\cA(r'_i)-\weight_\cA(x'_i)\]
    Denote $\aug{r_i}=((p_i,q_i,T_i),t_i,d_i,(p_{i+1},q_{i+1},T_{i+1}))$ and $\aug{x_i}=((p'_i,q_i,T'_i),t_i,d'_i,(p'_{i+1},q_{i+1},T'_{i+1}))$ (the second and third components are the same since both runs are runs on $\rhobase$)
    with $t_i=(q_i,\sigma_i,c_i,q_{i+1})$. 
    Then by the definition above we have
    $r'_i=(p_i,\sigma_i,d_i+c_i,p_{i+1})$ and $x'_i=(p'_i,\sigma_i,d'_i+c_i,p'_{i+1})$.
    Thus, we need to prove that
    \[d_i-d'_i=d_i+c_i-(d'_i+c_i)\]
    which indeed holds.
\else \textbf{PROOFS REMOVED} \fi\end{proof}
Recall that a run is seamless if at all prefixes, it is the minimal run to its most recent state. By \cref{prop: pov shift from aug to A maintains distance} we immediately have the following.
\begin{corollary}
\label{cor:pov shift from Aug to A maintains seamless}
    Consider a run $\rhobase$ of $\cA$ and a run $\augrho=\aug{r_1},\aug{r_2},\ldots,\aug{r_n}$ of $\augA$ on $\rhobase$. Let $\rho=r'_1,r'_2,\ldots,r'_n$ be $\shifttoorig(\augrho)$, then $\augrho$ is seamless if and only if $\rho$ is seamless.
\end{corollary}

We now turn to the converse shift: from runs of $\cA$ to runs of $\augA$. 
To this end, consider a word $w\in \Sigma^*$ and a run $\rhobase$ of $\cA$ on $w$. With every run $\rho$ of $\cA$ on $w$ we associate the \emph{$\rhobase$-shifted run of $\augA$ corresponding to $\rho$} denoted $\augrho=\shifttoaug(\rho)$ as follows.

Let $\rhobase=t_1,t_2,\ldots,t_n$ with $t_i=(q_i,\sigma_i,c_i,q_{i+1})$ and $\rho=r_1,r_2,\ldots,r_n$ with $r_i=(p_i,\sigma_i,d_i,p_{i+1})$, then $\augrho$ is the run of $\cA$ on $\rhobase$ defined by $\augrho=\aug{r_1},\aug{r_2},\ldots,\aug{r_n}$ with 
\[\aug{r_i}=((p_i,q_i,T_i),t_i,d_i-c_i,(p_{i+1},q_{i+1},T_{i+1}))\]
and $T_i$ is determined deterministically by $\augA$ on $\rhobase$.

As above, we observe that $\aug{r_i}$ is a run of $\augA$ on $\rhobase$. The initial state of $\rhobase$ and of $\rho$ is $q_0$, so the initial state of $\augrho$ is $(q_0,q_0,\{q_0\})$, i.e., the initial state of $\augA$.
For the transitions, for every $1\le i< n$, the transition $r_i$ is valid in $\augA$ on the internal state by the transition of $\rho$, and on the baseline and reachable set by the transitions (i.e., letters) of $\rhobase$. It remains to check that the weight of $r_i$ matches the transitions in $\augA$, but this follows directly by the definition of $\augTrans$.

The shift from $\cA$ to $\augA$ and from $\augA$ to $\cA$ can be seen as inverse operators, in that when shifting from $\cA$ to $\augA$, apart from the change in the state space, the weight of transitions in $\rhobase$ is subtracted, and conversely, when going back from $\augA$ to $\cA$ the weight of transitions in $\rhobase$ is added (again, apart from the change in state space). We thus have the following
\begin{observation}
    \label{obs:pov shift are inverses}
Consider a word $w\in \Sigma^*$ and a run $\rhobase$ of $\cA$ on $w$.
\begin{itemize}
    \item For runs $\rho,\rhobase$ of $\cA$ on $w$, let $\augrho=\shifttoaug(\rho)$ and let $\mu=\shifttoorig(\augrho)$, then $\mu=\rho$.
    \item For a run $\augrho$ of $\cA$ on $\rhobase$, let $\rho=\shifttoorig(\augrho)$ and $\aug{\mu}=\shifttoaug(\rho)$, then $\aug{\mu}=\augrho$.
\end{itemize}
\end{observation}
\begin{remark}
    \label{rmk: rhobase in shifts}
    Note that while $\rhobase$ is a parameter in $\shifttoaug(\rho)$, it is actually fixed in $\shifttoorig(\augrho)$, as it is determined by $\augrho$ (it is the word upon which $\augrho$ runs).
\end{remark}
We can now readily prove an analogue of \cref{prop: pov shift from aug to A maintains distance}.
\begin{proposition}
\label{prop: pov shift from A to aug maintains distance}
    Consider a run $\rhobase$ of $\cA$ on $w$, and two runs $\rho=r'_1,r'_2,\ldots,r'_n$ and $\mu=x'_1,x'_2,\ldots,x'_n$ of $\cA$ on $w$. 
    Let $\augrho=\aug{r_1},\aug{r_2},\ldots,\aug{r_n}$ and $\aug{\mu}=\aug{x_1},\aug{x_2},\ldots,\aug{x_n}$ of $\augA$ on $\rhobase$ be the respective $\rhobase$-shifted runs, then for every $1\le i\le n$ we have
    \[\weight_{\augA}(\aug{r_1},\aug{r_2},\ldots,\aug{r_i})-\weight_{\augA}(\aug{\mu}=\aug{x_1},\aug{x_2},\ldots,\aug{x_i})=\weight_\cA(r'_1,r'_2,\ldots,r'_i)-\weight_\cA(x'_1,x'_2,\ldots,x'_i)\]
\end{proposition}
\begin{proof}\ifproofs   
    By \cref{prop: pov shift from aug to A maintains distance} the $\rhobase$-shifted runs of $\cA$ that correspond to $\augrho$ and $\aug{\mu}$ maintain the differences, but by \cref{obs:pov shift are inverses} these runs are exactly $\rho$ and $\mu$, respectively.
\else \textbf{PROOFS REMOVED} \fi\end{proof}

We can now conclude that $\cA$ and $\augA$ are equi-determinizable, allowing us to proceed with an analysis of $\augA$ instead of $\cA$.

\begin{lemma}
    \label{lem:A det iff augA det}
    $\cA$ is determinizable if and only if $\augA$ is determinizable.
\end{lemma}
\begin{proof}\ifproofs    
For the first direction, assume $\cA$ is not determinizable. By \cref{thm:det iff bounded gap,def: B gap witness}, for every $B\in \bbN$ there exist words $x,y\in \Sigma^*$ and a state $q\in Q$ such that $\minweight_\cA(x\cdot y,q_0\to Q)=\minweight_\cA(x\cdot y,q_0\runsto{x}q\runsto{y}Q)<\infty$, and $\minweight_\cA(x,q_0\to q)-\minweight_\cA(x,q_0\to Q)> B$. 

By \cref{prop: pov shift from A to aug maintains distance}, the same gap between runs exists between the corresponding runs after a $\rhobase$-shift for any $\rhobase$. Moreover, a minimal run in $\cA$ is shifted to a minimal run in $\augA$. Thus, again by \cref{thm:det iff bounded gap}, we conclude that $\augA$ is not determinizable.

The converse direction is similar: assume $\augA$ is not determinizable, then for every $B\in \bbN$ there exist words $\aug{x},\aug{y}\in \Sigma^*$ and a state $s\in S$ such that $\minweight_{\augA}(\aug{x}\cdot \aug{y},s_0\to S)=\minweight_{\augA}(\aug{x}\cdot \aug{y},s_0\runsto{\aug{x}}s\runsto{\aug{y}}S)<\infty$, and $\minweight_{\augA}(\aug{x},s_0\to s)-\minweight_{\augA}(\aug{x},s_0\to S)> B$. 

By \cref{prop: pov shift from aug to A maintains distance}, the gap is maintained by the corresponding shifted runs. Note that here, the shift is dictated by the word being read, namely $\aug{x}$ (comparing the shifted runs of $s_0\runsto{\aug{x}}s$ and $s_0\runsto{\aug{x}}Q$). 
Thus, again by \cref{thm:det iff bounded gap}, we conclude that $\cA$ is not determinizable.
\else \textbf{PROOFS REMOVED} \fi\end{proof}

\section{Stable Cycles and Bounded Behaviors}
\label{sec:stable cycles and bounded behaviours}
A fundamental approach used in reasoning about weighted automata is that of \emph{stabilization}. Intuitively, the idea is to symbolically represent a repeated application of a certain transition matrix, until entries either become bounded, unreachable, or tend to $\infty$. In essence, this is a form of ``pumping''.
While useful in many setting, this form of stabilization seems to be too crude for showing that determinizability is decidable (evidently, all attempts for using it thus far have failed).

In this work, we use a much more delicate notion of stabilization, whereby we restrict attention to very specific subsets of states and transition matrices (i.e., words). 
The basic structure that allows some form of pumping is that of \emph{stable cycles}. Intuitively a stable cycle is a pair $(S',w)$ where $S'$ is a set of states and $w$ is a word such that reading $w$ any number of times from the states in $S'$ leads back to $S'$, and the weighted behaviors of runs are ``stable'', in the sense that we can identify which states generate minimal runs, and which tend to unbounded weights, regardless of the starting weights. 
We make this notion precise in the following. 

\subsection{Stable Cycles}
\label{sec:stable cycles}
We start with a weaker definition than stable cycles.
\begin{definition}
    \label{def:reflexive cycle and state}
    Consider a set of states $S'\subseteq S$ and $w\in \Delta^*$. We say that $(S',w)$ is a \emph{reflexive cycle} if there exists $q\in Q$ and $T\subseteq Q$ such that $S'=\{(p,q,T)\mid p\in T\}$ and $\booltrans(S',w)\subseteq S'$.

    A state $s\in S'$ is a \emph{reflexive state} if $s\runsto{w}s$. We denote the set of reflexive states of $(S',w)$ as $\RefStates(S',w)=\{s\in S'\mid s\runsto{w}s\}$.
\end{definition}

That is, the states in $S'$ share the last two components, and $T$ is ``saturated'' in the sense that every state $p\in T$ represented in $S'$ as $(p,q,T)$. Reflexivity requires that reading $w$ from $S'$ leads back to $S'$, or a subset thereof.
Note that if $(S',w)$ is a reflexive cycle, then so is $(S',w^n)$ for every $n\in \bbN$ (this is immediate by the requirement $\booltrans(S',w)\subseteq S'$).

A-priori, $(S',w)$ might not have reflexive states at all. Clearly, however, for $w^n$ with $n$ large enough, some states become reflexive. Our first two observation is that we can stabilize the reachability set and the reflexive states by taking a large-enough repetition of $w$. 
To this end, let $\bigN= |S|!$ and $\bigM=|S| \cdot \bigN$.

A basic observation from Boolean automata is that the transition relation becomes idempotent after $\bigM$ iterations. More precisely, we have the following.
\begin{proposition}
    \label{prop: transitions stabilize at M}
    For every word $w$, states $s,r\in S$ and $n\in \bbN$ we have $s\runsto{w^\bigM}r$ if and only if $s\runsto{w^{\bigM n}}r$.
\end{proposition}
\begin{proof}
    For readers versed with the algebraic view of automata, this can be shown using the order of the transition monoid of the automaton (see e.g.,~\cite{bojanczyk2020languages}). 
    We bring a state-based proof for completeness. 

    If $s\runsto{w^\bigM}r$, then since $\bigM>|S|$ there is a simple cycle in such a run, i.e., $s\runsto{w^a}r'\runsto{w^b}r'\runsto{w^{\bigM -a-b}}r$ with $b<|S|$. But then $b$ divides $\bigM$, so we can repeat this cycle for additional $(n-1) \bigM/b$, thus yielding a run $s\runsto{w^{\bigM n}}r$.

    Conversely, if $s\runsto{w^{\bigM n}}r$, since $\bigM=|S|\bigN$, we can decompose $s$ into $\frac{\bigM n}{\bigN}=|S|n$ segments of $w^{\bigN}$.
    \[
    \rho:s\runsto{w^{\bigN}}s_1\runsto{w^{\bigN}}s_2\cdots \runsto{w^{\bigN}}s_{|S|n}=r
    \]
    Since there are more than $|S|$ such segments, we can remove cycles and obtain a run 
    \[
    \mu:s\runsto{w^{\bigN}}t_1\runsto{w^{\bigN}}t_2\cdots \runsto{w^{\bigN}}t_{k}=r
    \]
    with $k\le |S|$. 
    If $k=|S|$, then $k\bigN=\bigM$, and therefore $s\runsto{w^{\bigM}}r$ and we are done.
    Otherwise, $1\le k<|S|$, and in particular since $\bigN>|S|$ the prefix $s\runsto{w^{\bigN}}t_1$ contains a cycle of the form $t'\runsto{w^m}t'$ for some $1\le m<|S|$. Since $\bigM=|S|\cdot |S|!$, in particular $m$ divides $|S|!=\bigN$, and therefore $m$ divides $\bigM-k\bigN$.
    Thus, by repeating this cycle an additional $\frac{\bigM-k\bigN}{m}$ times, we obtain a run $s\runsto{w^{\bigM}}r$, and we are done.
\end{proof}

\begin{proposition}
\label{prop:reflexive cycles stabilize at N}
Consider a reflexive cycle $(S',w)$. For every $n,m\in \bbN$ we have that
$\RefStates(S',w^{m})\subseteq \RefStates(S',w^\bigN)=\RefStates(S',w^{\bigN \cdot n})$.
Additionally, $\booltrans(S',w^{\bigN})=\booltrans(S',w^{\bigN \cdot n})$.
\end{proposition}
\begin{proof}\ifproofs 
Let $s\in \RefStates(S',w^m)$ for some $m\in \bbN$, and take this $m$ to be minimal. 
Thus, $\rho:s\runsto{w^m}s$ for some run $\rho$. We start by showing that $s\in \RefStates(S',w^\bigN)$. First, note that if $m\le |S|$, then $m| \bigN = |S|!$, so $\rho^{\bigN/m}:s\runsto{w^\bigN}s$. Thus, $s\in \RefStates(S',w^\bigN)$. 
Next, if $m>|S|$, then by the pigeonhole principle $\rho$ has a cycle between some $w^i$ and $w^j$. 
That is, there are $1\le i<j\le m$  and some state $s'\in S$ such that $\rho:s\runsto{w^i}s'\runsto{w^{j-i}}s'\runsto{w^{m-j}}s$. We can now remove this cycle and obtain the run
$\rho':s\runsto{w^{m-(j-i)}}s$, contradicting the minimality of $m$.

Next, we show that $\RefStates(S',w^\bigN)=\RefStates(S',w^{\bigN \cdot n})$. The direction $\RefStates(S',w^\bigN)\subseteq\RefStates(S',w^{\bigN \cdot n})$ trivially holds, by repeating cycles (since $\bigN$ divides $\bigN\cdot n$). The converse holds by the argument above, plugging $m=\bigN\cdot n$.

For the reachability part, by reflexivity we inductively have that $\{\booltrans(S',w^k)\}_{k=0}^\infty$ forms a decreasing chain, and in particular $\booltrans(S',w^{\bigN \cdot n})\subseteq \booltrans(S',w^{\bigN})$.
Conversely, consider $s\in \booltrans(S',w^{\bigN})$, then by a similar pigeonhole argument as above, there is $s',s''\in S'$ such that $s'\runsto{w^i}s''\runsto{w^{j}}s'' \runsto{w^{\bigN-i-j}}s$ for some $j<|S|$, and by repeating the $w^j$ cycle we can obtain a run $s'\runsto{w^{\bigN\cdot n}}s$, so $s\in \booltrans(S',w^{\bigN\cdot n})$.
\else \textbf{PROOFS REMOVED} \fi\end{proof}
\cref{prop:reflexive cycles stabilize at N} shows that all the reflexive cycles on any number of repetitions of $w$ are already present for $w^\bigN$.

In the $(\min,+)$ semantics, of particular importance are reflexive cycles whose cycles have minimal weight. We capture this as follows. 
\begin{definition}
    \label{def:minimal reflexive states}
    Consider a reflexive cycle $(S',w)$. We say that a state $s\in \RefStates(S',w)$ is \emph{minimal reflexive} if $\minweight(w,s\to s)=\min\{\minweight(w,r\to r)\mid r\in \RefStates\}$.    
    We then define 
    \[\MinRefStates(S',w)=\{s\in S'\mid s\text{ is minimal reflexive.}\}\]
\end{definition}
That is, the cycle on $w$ from $s$ to $s$ has minimal weight among all reflexive cycles on $w$.
Recall that $S'=\{(p,q,T)\mid p\in T\}$  for some $T\subseteq Q$. In particular, $(q,q,T)\in S'$ is the unique baseline state in $S'$. 
We now ask whether $(q,q,T)$ is a minimal reflexive state, and whether its minimal reflexive run is a baseline (and therefore by \cref{obs:baseline runs have weight 0} has weight $0$). Moreover, we ask whether this holds for $w^n$ for every $n\in \bbN$. In case this holds, we say that $(S',w)$ is a \emph{Stable Cycle}. We now formalize this.

\begin{definition}[\keyicon Stable Cycle]
\label{def:stable cycle}
A reflexive cycle $(S',w)$ with baseline state $s=(q,q,T)$ is a \emph{stable cycle} if for every $n\in \bbN$ it holds that $s\in \MinRefStates(S',w^n)$ and $\minweight(w^n,s\to s)=0$.
\end{definition}

Note that if $(S',w)$ is a stable cycle, then the baseline run $\rho:s\runsto{w^n} s$ is a cycle of minimal weight on any state $s'\in S'$, and its weight is $0$.
There could be other runs from $S'$ on some $w^n$ that are negative, but not cycles.

We now show an analogous result to \cref{prop:reflexive cycles stabilize at N}, showing that the minimal reflexive states also become stable quickly, under the assumption that we have a stable cycle.
\begin{proposition}
\label{prop:minimal reflexive cycles stabilize at N}
Consider a stable cycle $(S',w)$. For every $n,m\in \bbN$ we have that 
\[\MinRefStates(S',w^{m})\subseteq \MinRefStates(S',w^\bigN)=\MinRefStates(S',w^{\bigN \cdot n})\]   
\end{proposition}
\begin{proof}\ifproofs 
Let $s\in \MinRefStates(S',w^m)$ for some $m\in \bbN$ (not necessarily the baseline state), and take this $m$ to be minimal. 
Thus, $\rho:s\runsto{w^m}s$ for some minimal-weight run $\rho$. We start by showing that $s\in \MinRefStates(S',w^\bigN)$. If $m\le |S|$, then $m| \bigN=|S|!$, so $\rho^{\bigN/m}:s\runsto{w^\bigN}s$. 
By \cref{def:stable cycle} we have $\weight(\rho)=0$, so $\weight(\rho^{\bigN/m})=0$ as well. Since $(S',w)$ is stable, then in particular $\minweight(w^{\bigN},s'\to s')\ge 0$ for every $s'\in S'$, since the baseline state is minimal reflexive for $w^{\bigN}$ (indeed, for all repetitions of $w$) and has cycle weight $0$.
It follows that $s\in \MinRefStates(S',w^{\bigN})$

Next, if $m>|S|$, then by the pigeonhole principle $\rho$ has a cycle between some $w^i$ and $w^j$. 
That is, there are $1\le i<j\le m$  and some state $s'\in S$ such that $\rho^i:s\runsto{w^i}s'$ and  $\rho^j:s\runsto{w^j}s'$. Therefore, $\rho^{j-i}:s'\runsto{w^{j-i}}s'$. Moreover, $\weight(\rho^{j-i})\ge 0$. Indeed, otherwise enough repetitions of $\rho^{j-i}$ would yield a negative cycle, contradicting the fact that the baseline state is minimal reflexive with weight 0 for \emph{every} repetition of $w$. 

We can now remove this cycle and obtain the run
$\rho':s\runsto{w^{m-(j-i)}}s$ with weight at most $0$ (and by the above -- exactly $0$), contradicting the minimality of $m$.

Next, we show that $\MinRefStates(S',w^\bigN)=\MinRefStates(S',w^{\bigN \cdot n})$. 
The direction $\MinRefStates(S',w^\bigN)$ $\subseteq\MinRefStates(S',w^{\bigN \cdot n})$ holds by repeating cycles (since $\bigN$ divides $\bigN\cdot n$), and by observing that negative cycles cannot occur in a stable cycle. 
The converse holds by the argument above, plugging $m=\bigN\cdot n$.
\else \textbf{PROOFS REMOVED} \fi\end{proof}

\subsection{States with Bounded Behaviors}
\label{sec:bounded states}
Recall that our motivation in this section is to establish a ``pumpable'' structure. Let $(S',w)$ be a stable cycle, and consider the long-term behavior of runs on $w^m$ for large $m$, starting from states in $S'$.

Recall that for a state $s\in \MinRefStates(S',w)$ we have $\minweight(w,s\to s)=0$, and that this is the minimal weight of a cycle on $w$ from $S'$. Then, for every state $r\in S'$, if either $s\runsto{w}r$ or $r\runsto{w}s$, more  repetitions of $w$ intuitively do not increase (nor decrease) the weight further, since they incur weight $0$ by cycling on $s$. 
The converse is almost correct: we would want to say that runs that avoid $\MinRefStates(S',w)$
must increase unboundedly. This is almost true, except that a state might not be minimal-reflexive for $w$, but becomes minimal reflexive for $w^\bigM$, so does not lead to unbounded weights.
In the following we show that this is indeed the case: runs that avoid $\MinRefStates(S',w^{\bigM})$ inevitably yield increasing weights.
Moreover, we identify after how many repetitions of $w$ these eventual behaviors stabilize.

We start with the definitions for the three types of ``bounded behavior'': states that are reachable via $w$ from a minimal-reflexive state are called \emph{tethered} (as they are ``tethered'' to a $0$-weight cycle), states from which a minimal-reflexive state is reachable are called \emph{plateau}, as they reach a $0$-weight cycle, and their combination is referred to as \emph{grounded}, as follows.
\begin{definition}[Tethered, Plateau and Grounded Pairs]
    \label{def:tethered}
    \label{def:plateau}
    \label{def:grounded pairs}
    For a stable cycle $(S',w)$ we define:
    \begin{itemize}
        \item The \emph{tethered pairs} as
        \[\TethPairs(S',w)=\{(s,r)\in S'\times S'\mid s\in \MinRefStates(S',w)\wedge s\runsto{w}r\}\]
        We refer to each such $r$ as a \emph{tethered state}.
        \item The \emph{plateau pairs} as
        \[\PlatPairs(S',w)=\{(s,r)\in S'\times S'\mid r\in \MinRefStates(S',w)\wedge s\runsto{w}r\}\]
        We refer to each such $s$ as a \emph{plateau state}.
        \item The \emph{grounded pairs} as
        \[\GroundPairs(S',w)=\{(s,r)\in S'\times S'\mid \exists g\in \MinRefStates(S',w^{\bigM})\ST s\runsto{w^{\bigM}}g\runsto{w^{\bigM}}r\}\]
        For $(s,r)\in \GroundPairs(S',w)$, the state $g\in \MinRefStates(S',w^{\bigM})$ for which $\weight(s\runsto{w^{\bigM}}g\runsto{w^{\bigM}}r)$ is minimal (among all states in $\MinRefStates(S',w^\bigM)$) is the \emph{grounding state of $(s,r)$.}
    \end{itemize}
\end{definition}
\cref{fig:pair_types} depicts the three types of pairs in \cref{def:tethered}.
\begin{figure}[ht]
    \centering
    \includegraphics[width=0.9\linewidth]{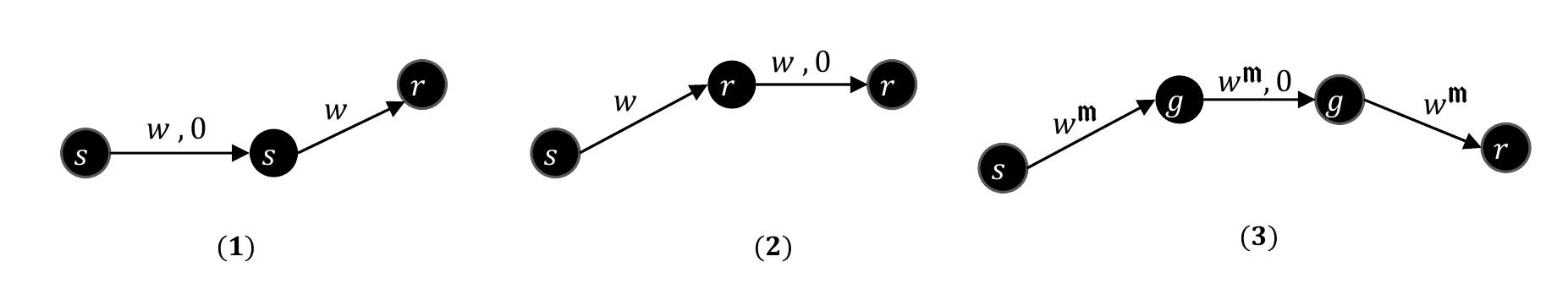}
    \caption{$(s,r)$ is a tethered pair (1), a plateau pair (2), and grounded pair with a grounding state $g$ (3)}
    \label{fig:pair_types}
\end{figure}
As a sanity-check, note that since $\booltrans(S',w)\subseteq S'$, then for every tethered, plateau, or grounded pair $(s,r)$, the fact that $r\in S'$ already follows from $s\in S'$.

Similarly to \cref{prop:reflexive cycles stabilize at N,prop:minimal reflexive cycles stabilize at N}, we show that the tethered and plateau pairs stabilize. This, however, no longer holds for any repetition of $w$, but for repetitions of $w^\bigN$, and the stabilization occurs at $\bigM$ (recall $\bigM=|S| \cdot \bigN$), as follows.

\begin{proposition}
\label{prop:tethered states stabilize at N}
Consider a stable cycle $(S',w)$. For every $n,m\in \bbN$ we have that 
\[\TethPairs(S',w^{m \cdot \bigN})\subseteq \TethPairs(S',w^\bigM)=\TethPairs(S',w^{\bigM \cdot n}).\]
Moreover, for every $(s,r) \in \TethPairs(S',w^{m \cdot \bigN})$ it holds that
\[\minweight(w^{m\cdot \bigN}, s \to r)\ge \minweight(w^{\bigM}, s \to r) = \minweight(w^{n\cdot \bigM}, s \to r) 
\]
\end{proposition}
\begin{proof}\ifproofs The proof proceeds in three steps.
\paragraph{Step 1:}
Consider a tethered pair $(s,r)\in \TethPairs(S',w^{m\bigN})$ for some $m$, and let $m$ be minimal for this pair. We claim that $m\le |S|$. Indeed, similarly to the proofs of \cref{prop:reflexive cycles stabilize at N,prop:minimal reflexive cycles stabilize at N}, if $m>|S|$, by the pigeonhole principle there exists a state $s'$ that repeats in the minimal weight run $\rho:s\runsto{w^{m\bigN}}r$. We can thus write
$\rho:s\runsto{w^{m_1\bigN}}s'\runsto{w^{m_2\bigN}}s'\runsto{w^{m_3\bigN}}r$ with $m_1+m_2+m_3=m$. We then have that $s\runsto{w^{(m_1+m_3)\bigN}}r$. In order to reach a contradiction to the minimality of $m$, we need to show that $s\in \MinRefStates(S',w^{(m_1+m_3)\bigN})$, so that $(s,r)$ is a tethered pair in $(S',w^{(m_1+m_3)}\bigN)$. This, however, follows directly from \cref{prop:minimal reflexive cycles stabilize at N}, since $\MinRefStates(S',w^{m\bigN})=\MinRefStates(S',w^{(m_1+m_3)\bigN})$. We can  therefore assume $m\le |S|$.

In addition, since stable cycles do not contain negative cycles, we have that $\minweight(w^{m_2\bigN},s'\to s')\ge 0$, so $\minweight(w^{(m_1+m_3)\bigN},s\to r)\le \minweight(w^{m\bigN},s\to r)$.

\paragraph{Step 2:} Next, we prove that $\TethPairs(S',w^{m\bigN})\subseteq \TethPairs(S',w^{k\bigN})$ for all $k\ge m$. 
Indeed, consider a tethered pair $(s,r)\in \TethPairs(S',w^{m\bigN})$, then $s\runsto{w^{m\bigN}}r$ and $s\in \MinRefStates(S',w^{m\bigN})$. 
By \cref{prop:minimal reflexive cycles stabilize at N} and since $k\ge m$ we also have 
\[s\in \MinRefStates(S',w^{\bigN})=\MinRefStates(S',w^{(k-m)\bigN})=\MinRefStates(S',w^{k\bigN}).\]
In particular, $s\runsto{w^{(k-m)\bigN}}s$, and therefore $s\runsto{w^{k-m}\bigN}s\runsto{w^{m\bigN}}r$, so $s\runsto{w^{(k\bigN)}}r$ and thus $(s,r)\in \TethPairs(S',w^{k\bigN})$.

Moreover, since $s\in \MinRefStates(w^{(k-m)\bigN})$, it follows that $\minweight(w^{(k-m)\bigN},s\to s)=0$, and from the run description above we get that
$\minweight(w^{m\bigN},s\to r)\ge \minweight(w^{k\bigN},s\to r)$ (where we have an inequality instead of equality, since conceptually there could be other runs with lower weight).

\paragraph{Step 3:} We can now complete the proof. for a pair $(s,r)\in \TethPairs(S',w^{m\bigN})$, we can assume by Step 1 that $m\le |S|$, and this can only decrease the weight. By Step 2 we also have $(s,r)\in\TethPairs(S',w^{\bigM})$, since $\bigM=|S|\bigN>m\bigN$, and that this again decreases the weight, i.e., 
\[\minweight(w^{m\bigN},s\to r)\ge \minweight(w^{\bigM},s\to r)\ge \minweight(w^{n\bigM},s\to r).\]
Conversely, if $(s,r)\in\TethPairs(S',w^{\bigM})$ then since $\bigM$ is a multiple of $\bigN$, by Step 2 we have $(s,r)\in\TethPairs(S',w^{n\bigM})$ for all $n\ge 1$. 
To conclude the equality of weights, by Step 1 and Step 2 combined, we have that 
$\minweight(w^{n\bigM},s\to r)\ge \minweight(w^{\bigM},s\to r)$. Indeed, decreasing from $n\bigM$ to at most $|S|\bigN$ decreases the weight (Step 1), and it further decreases to $|S|\bigN=\bigM$ (Step 2).
\else \textbf{PROOFS REMOVED} \fi\end{proof}

A similar claim holds for plateau states, with an analogous proof.
\begin{proposition}
\label{prop:plateau states stabilize at N}
    Consider a stable cycle $(S',w)$. For every $n,m\in \bbN$ we have that 
    \[\PlatPairs(S',w^{m \cdot \bigN})\subseteq \PlatPairs(S',w^\bigM)=\PlatPairs(S',w^{\bigM \cdot n}).\]
    Moreover, for every $(s,r) \in \PlatPairs(S',w^{m \cdot \bigN})$ it holds that
    \[\minweight(w^{m\cdot \bigN}, s \to r)\ge \minweight(w^{\bigM}, s \to r) = \minweight(w^{n\cdot \bigM}, s \to r) 
    \]
\end{proposition}
\begin{proof}\ifproofs The proof proceeds in three steps.
    \paragraph{Step 1:}
        Consider a plateau pair $(s,r)\in \PlatPairs(S',w^{m\bigN})$ for some $m$, and let $m$ be minimal for this pair. We claim that $m\le |S|$. Indeed, similarly to the proofs of \cref{prop:reflexive cycles stabilize at N,prop:minimal reflexive cycles stabilize at N}, if $m>|S|$, by the pigeonhole principle there exists a state $s'$ that repeats in the minimal weight run $\rho:s\runsto{w^{m\bigN}}r$. We can thus write
        $\rho:s\runsto{w^{m_1\bigN}}s'\runsto{w^{m_2\bigN}}s'\runsto{w^{m_3\bigN}}r$ with $m_1+m_2+m_3=m$. We then have that $s\runsto{w^{(m_1+m_3)\bigN}}r$. In order to reach a contradiction to the minimality of $m$, we need to show that $r\in \MinRefStates(S',w^{(m_1+m_3)\bigN})$, so that $(s,r)$ is a plateau pair in $(S',w^{(m_1+m_3)}\bigN)$. This, however, follows directly from \cref{prop:minimal reflexive cycles stabilize at N}, since $\MinRefStates(S',w^{m\bigN})=\MinRefStates(S',w^{(m_1+m_3)\bigN})$. We can  therefore assume $m\le |S|$.
        
        In addition, since stable cycles do not contain negative cycles, we have that $\minweight(w^{m_2\bigN},s'\to s')\ge 0$, so $\minweight(w^{(m_1+m_3)\bigN},s\to r)\le \minweight(w^{m\bigN},s\to r)$.
    
    \paragraph{Step 2:} 
        Next, we prove that $\PlatPairs(S',w^{m\bigN})\subseteq \PlatPairs(S',w^{k\bigN})$ for all $k\ge m$. 
        Indeed, consider a plateau pair $(s,r)\in \PlatPairs(S',w^{m\bigN})$, then $s\runsto{w^{m\bigN}}r$ and $r\in \MinRefStates(S',w^{m\bigN})$. 
        By \cref{prop:minimal reflexive cycles stabilize at N} and since $k\ge m$ we also have 
        \[r\in \MinRefStates(S',w^{\bigN})=\MinRefStates(S',w^{(k-m)\bigN})=\MinRefStates(S',w^{k\bigN}).\]
        In particular, $r\runsto{w^{(k-m)\bigN}}r$, and therefore $s\runsto{w^{m\bigN}}r\runsto{w^{k-m}\bigN}r$, so $s\runsto{w^{(k\bigN)}}r$ and thus $(s,r)\in \PlatPairs(S',w^{k\bigN})$.
    
    Moreover, since $r\in \MinRefStates(w^{(k-m)\bigN})$, it follows that $\minweight(w^{(k-m)\bigN},r\to r)=0$, and from the run description above we get that
    $\minweight(w^{m\bigN},s\to r)\ge \minweight(w^{k\bigN},s\to r)$ (where we have an inequality instead of equality, since conceptually there could be other runs with lower weight).
    
    \paragraph{Step 3:} 
        We can now complete the proof. for a pair $(s,r)\in \PlatPairs(S',w^{m\bigN})$, we can assume by Step 1 that $m\le |S|$, and this can only decrease the weight. By Step 2 we also have $(s,r)\in\PlatPairs(S',w^{\bigM})$, since $\bigM=|S|\bigN>m\bigN$, and that this again decreases the weight, i.e., 
        \[\minweight(w^{m\bigN},s\to r)\ge \minweight(w^{\bigM},s\to r)\ge \minweight(w^{n\bigM},s\to r).\]
        Conversely, if $(s,r)\in\PlatPairs(S',w^{\bigM})$ then since $\bigM$ is a multiple of $\bigN$, by Step 2 we have $(s,r)\in\PlatPairs(S',w^{n\bigM})$ for all $n\ge 1$. 
        To conclude the equality of weights, by Step 1 and Step 2 combined, we have that 
        $\minweight(w^{n\bigM},s\to r)\ge \minweight(w^{\bigM},s\to r)$. Indeed, decreasing from $n\bigM$ to at most $|S|\bigN$ decreases the weight (Step 1), and it further decreases to $|S|\bigN=\bigM$ (Step 2).
\else \textbf{PROOFS REMOVED} \fi\end{proof}

The stabilization of $\TethPairs$ and $\PlatPairs$ implies that the grounding state in grounded pairs does not have to be reachable and co-reachable with $w^\bigM$, as follows.
\begin{proposition}
    \label{prop:grounding states reachable with any bigM k}
    Consider a stable cycle $(S',w)$. For every $s,r,g\in S'$ with $g\in \MinRefStates(S',w^{\bigM})$, if $s\runsto{w^{\bigM i}}g\runsto{w^{\bigM j}}r$ for some $i,j\ge 1$, then $(s,r)\in \GroundPairs(S',w)$.
\end{proposition}
\begin{proof}
    Since $g\in \MinRefStates(S',w^{\bigM})$, 
    then $(s,g)\in \PlatPairs(S',w^{\bigM i})=\PlatPairs(S',w^{\bigM})$ (by \cref{prop:plateau states stabilize at N}) and that $(g,r)\in \TethPairs(S',w^{\bigM j})=\TethPairs(S',w^{\bigM})$. 
    We thus have $s\runsto{w^{\bigM}}g\runsto{w^{\bigM}}r$, so by \cref{def:grounded pairs} we have $(s,r)\in \GroundPairs(S',w)$.
\end{proof}

The main result of this section aggregated the asymptotic behavior of grounded pairs, and shows that they stabilize. We give a brief overview of the argument.

Consider a pair of states $(s,r)$ in a stable cycle $(S',w)$. If $(s,r)\notin \GroundPairs(S',w)$, then every run $\rho:s\runsto{w^m}r$ has only strictly positive cycles (focusing on the states after each $w$), and is therefore intuitively increasing. There may, however, be negative simple paths upon reading $w$, and the negative weight can depend on the length of $w$. Thus, it is guaranteed that every long-enough run has very large weight, but achieving this weight may take a while.

Conversely, if $(s,r)\in \GroundPairs(S',w)$ and $g$ is their grounding state (c.f. \cref{def:grounded pairs}), then it is possible to spend any number $m\bigM>2\bigM$ of iterations of $w$ between $r$ and $s$ by taking $0$-weight cycles on $g$. While this is optimal among runs that go through a grounding state, there may be other runs with lower weight that initially take a very negative transition. Nonetheless, after enough repetitions of $w$, the path through the grounding state becomes optimal altogether. This is captured as follows.

\begin{lemma}[\keyicon Pumping Grounded Pairs]
\label{lem:pumping grounded pairs}
Consider a stable cycle $(S',w)$. For every $n\in \bbN$ there exists $M_0\in \bbN$ (efficiently computable) such that for all $m\ge M_0$ the following hold.\footnote{Note that Item 2 does not depend on $n$.}  
\begin{enumerate}
    \item If $(s,r)\notin \GroundPairs(S',w)$ then $\minweight(w^{m\cdot 2\bigM},s\to r)>n$.
    \item If $(s,r)\in \GroundPairs(S',w)$ with grounding state $g$, then 
    \[
    \minweight(w^{m\cdot 2\bigM},s\to r)=\minweight(w^{M_0\cdot 2\bigM},s\to r)=\minweight(w^{2\bigM},s\runsto{w^{M}} g\runsto{w^{M}}  r)
    \]
    Moreover, we have $\minweight(w^{m'\cdot 2\bigM},s\to r)\le \minweight(w^{2\bigM},s\runsto{w^{M}} g\runsto{w^{M}}  r)$ for every $m'\in \bbN$.
\end{enumerate}
\end{lemma}
\begin{proof} \ifproofs
We first prove Item 1, and then use it to prove Item 2.
\paragraph{Proof of Item 1.}
Let $(s,r)\notin \GroundPairs(S',w)$. We construct a weighted directed graph $G=\tup{V,E}$ whose vertices are $V=S'$ and for $p,q\in S'$ we add the edge $(p,q)$ with weight $\minweight(w^\bigM,p\to q)$ (which might be $\infty$). Observe that there is a bijection between paths from $s$ to $r$ in $G$ and runs of the form  $\rho:s=p_1\runsto{w^\bigM}p_2\runsto{w^{\bigM}}\cdots \runsto{w^\bigM}p_{k}=r$ provided each section $p_i\runsto{w^\bigM}p_{i+1}$ in the run has minimal weight (indeed, otherwise it can be replaced by a lower-weight run).
In particular, for all $k$, the minimal weight of a run from $s$ to $r$ on $w^{k\bigM}$ is attained by a path in the graph.

Since $(S',w)$ is stable, this graph has no strictly negative cycles.
Since $(s,r)\notin \GroundPairs(S',w)$, then every cycle in the graph that is reachable from $s$ and can reach $r$, has strictly positive weight. Indeed, otherwise there would be a run $s\runsto{w^{k_1\bigM}}q\runsto{w^{k_2\bigM}}q\runsto{w^{k_3\bigM}}r$ with $\minweight(q\runsto{w^{k_2\bigM}}q)=0$. However, by \cref{prop:minimal reflexive cycles stabilize at N} this implies $q\in \MinRefStates(S',w^{\bigM})$, which in turn implies by \cref{prop:tethered states stabilize at N,prop:plateau states stabilize at N,prop:minimal reflexive cycles stabilize at N} that $(s,r)\in \GroundPairs(S',w)$ (as $(s,q)\in \PlatPairs(S',w^{\bigM})$ and $(q,r)\in \TethPairs(S',w^{\bigM})$), but this contradicts the assumption that $(s,r)\notin \GroundPairs(S',w)$.

We conclude that every cycle reachable between $s$ and $r$ has strictly positive weight. The claim readily follows from this, by noting that increasingly long runs have increasingly large weight (which corresponds to all the runs having large weight). We spell out the argument for completeness.

By Johnson's algorithm~\cite{johnson1977efficient} we can (efficiently) compute a new weight function on the edges $\mu:E\to \bbN$ and two numbers $v_s,v_r\in \bbZ$ such that the weight of a path $s=p_1,p_2,\ldots,p_k=r$ in the original graph $G$ is $v_s+\sum_{i=1}^{k-1}E(p_i,p_{i+1})-v_r$. Note that in the new weights, there are no negative edges at all, and the property that there are no reachable $0$-cycles between $s$ and $r$ is maintained by Johnson's construction. We can now readily see that in every segment of length $|S|$ in the path there is at least one cycle, whose weight is at least $1$. Therefore, for every $n\in \bbN$, every path of length $|S|(|v_s|+|v_r|+n)$ has weight at least $n$ in $G$. Thus, for every $n\in \bbN$ we can set $M_0=|S|(|v_s|+|v_r|+n)$, and the claim follows (note that we prove the claim already for $w^{m\cdot \bigM}$, but we take $w^{m\cdot 2\bigM}$ for uniformity with the other constants).

Note that in fact, we prove a slightly stronger result: for any pair $(s,r)$, if $\rho:s\runsto{w^{m\cdot 2\bigM}}r$ does not traverse a $0$-weight cycle while reading some $w^i$ infix, then $\weight(\rho)>n$. In particular, if $\weight(\rho)\le n$ then it must make a $0$-weight cycle on some state $q\in \MinRefStates(S',w^\bigM)$.
This observation is used in the next part of the proof.

\paragraph{Proof of Item 2.}
Consider $(s,r)\in \GroundPairs(S',w)$ and let $g$ be the grounding state for $(s,r)$ (see \cref{def:grounded pairs}). 
We first note that for every $m\in \bbN$ it holds that there is always a run on $w^{m\cdot 2\bigM}$ that uses the grounding state as a $0$-cycle. Indeed, $s\runsto{w^\bigM}g\runsto{w^{(m-1)2\bigM}}g\runsto{w^{M}}r$ is a run from $s$ to $r$ on $w^{m\cdot 2\bigM}$. Since $\minweight(w^{\bigM},g\to g)=0$ it follows that the weight of this run is exactly $\minweight(w^{2\bigM},s\runsto{w^{\bigM}} g\runsto{w^{\bigM}}r)$. We thus have the inequality:
\[
\minweight(w^{2\bigM},s\runsto{w^{\bigM}} g\runsto{w^{\bigM}}r)\ge \minweight(w^{m\cdot 2\bigM},s\to r).
\]
This concludes the ``Moreover'' part.
It remains to show that there is some $M_0\in \bbN$ such that for all $m\ge M_0$, the reverse inequality holds.

Fix $n=\minweight(w^{2\bigM},s\runsto{w^{\bigM}} g\runsto{w^{\bigM}}r)$ and let $M_0$ be the constant constructed in the proof of Item 1 for this $n$. The observation at the end of the proof of Item 1 gives us that for every $m\ge M_0$, every run $\rho:s\runsto{w^{m\cdot 2\bigM}}r$ that does not visit a state in $\MinRefStates(S',w^\bigM)$ must have weight larger than $n$. By the inequality above, such runs cannot be minimal, and we can therefore ignore them.

Consider then a minimal-weight run $\rho:s\runsto{w^{m\cdot 2\bigM}}r$ that visits a state $q\in \MinRefStates(S',w^\bigM)$, then we can write
$\rho:s\runsto{w^{k_1\cdot 2\bigM}}q\runsto{w^{k_2\cdot 2\bigM}}r$ for some $k_1,k_2\ge 1$. 
 Then, by \cref{prop:tethered states stabilize at N,prop:plateau states stabilize at N} we have that $(s,q)\in \PlatPairs(w^{\bigM},s\to q)$ and $(q,r)\in \TethPairs(w^{\bigM},q\to r)$, and moreover: $\minweight(w^\bigM,s\to q)=\minweight(w^{k_1\bigM},s\to q)$ and similarly $\minweight(w^\bigM,q\to r)=\minweight(w^{k_2\bigM},q\to r)$.
 Thus, we have $\weight(\rho)=\minweight(w^{2\bigM},s\runsto{w^\bigM}q\runsto{w^{\bigM}}r)$. However, recall that $g$ is the grounding state, which by definition means that $g$ minimizes this expression over all $q\in \MinRefStates(S',w^{\bigM})$, and therefore $\weight(\rho)\ge \minweight(w^{2\bigM},s\runsto{w^\bigM}g\runsto{w^{\bigM}}r)$. We thus have 
 \[
\minweight(w^{m\cdot 2\bigM},s\to r)\ge \minweight(w^{2\bigM},s\runsto{w^{\bigM}} g\runsto{w^{\bigM}}r)
\]
and we conclude the equality.
We can now either plug in $m=M_0$ or leave $m\ge M_0$ to get the two equalities of Item 2.
\else \textbf{PROOFS REMOVED} \fi \end{proof}

\section{The Cactus Extension}
\label{sec:cactus extension}
Consider a baseline-augmented subset construction $\augA=\tup{S,\Gamma,\augInitState,\augTrans}$. Note that we now denote the alphabet by $\Gamma$; we no longer assume $\Gamma$ is finite. 
More precisely, we do not assume the precise transitions defined in \cref{sec:augmented construction}, but the following:
\begin{itemize}
    \item The states are of the form $(q,p,T)$ where $q,p\in T$.
    \item The transitions are deterministic with respect to the second and third components (referred to as \emph{baseline component} and \emph{reachable set}, respectively).
\end{itemize}
In the following, we repeatedly extend the alphabet and transitions of $\augA$ to obtain a sequence of automata over the same state space with the properties above. We briefly sketch our approach.
Our extension consists of three phases. First, we introduce a \emph{stabilization} operator, which adds letters referred to as \emph{cactus letters}. Intuitively, a cactus letter is built by nesting stable cycles, where each stable cycle becomes a letter. 
The next phase is a \emph{rebase} operator, which further adds the \emph{rebase letters}. Intuitively, these letters represent a stable cycle, but also shift the baseline to a new state (which is not possible in cactus letters). 
We then repeatedly alternate adding cactus letters and rebase letters to an infinite limit. 
Finally, we add another set of letters, referred to as \emph{jump} letters. These allow us to change the baseline. The need for jump letters is technical, and we explain their usage in due course.

\subsection{Cactus Letters}
\label{sec:cactus letters}
Cactus letters are defined as the infinite union of an inductively defined set. The induction steps are defined by adding new letters and transitions with the following \emph{stabilization} operator.
\begin{definition}[\keyicon Stabilization]
\label{def:stabilization}
    For each stable cycle $(S',w)$ we introduce a new letter $\alpha_{S',w}$ with the following transitions: for every $(s,r)\in \GroundPairs(S',w)$ with grounding state $g$ we have the transition $(s,\alpha_{S',w},\minweight(w^{2\bigM},s\runsto{w^\bigM}g\runsto{w^\bigM} r),r)$. We refer to this operation as \emph{stabilization}, and denote the new letters by $\stab_{\augA}(\Gamma)$ and the new transitions by $\stab_{\augA}(\augTrans)$.
\end{definition}
\begin{remark}
\label{rmk:stabilization other works}
The term \emph{stabilization} is often used in the algebraic view of automata to denote a closure operation on the transition monoid of a WFA~\cite{lombardy2006sequential,daviaud2023big} whereby runs that increase too much are ``sent to $\infty$''.
Our usage is technically different, but has a similar flavor, and hence we reuse the name. 
Indeed, by \cref{lem:pumping grounded pairs}, the effect of reading arbitrarily many repetitions of $w$ from a state in a stable cycle is either bounded (for a grounded pair) or tends to infinity. 
\end{remark}

We define the stabilization by essentially considering $w^{2\bigM}$. This, however, can be replaced by $w^{m2\bigM}$ for any $m\in \bbN$, as we now show. Intuitively, the reason is that at $2\bigM$ most of the behaviors stabilize.
\begin{proposition}
    \label{prop:cactus letters stabilizes at 2M}
    Consider a stable cycle $(S',w)$, then for every $m\in \bbN$ we have that $(S',w^{m\cdot 2\bigM})$ is also a stable cycle, and the transitions defined on $\alpha_{S',w}$ are identical to those of $\alpha_{S',w^{m\cdot 2\bigM}}$. 
    I.e, for every $s,r\in S$ we have $\minweight(\alpha_{S',w},s\to r)=\minweight(\alpha_{S',w^{m\cdot 2\bigM}},s\to r)$.
\end{proposition}
\begin{proof}
    Let $m\in \bbN$. We first claim that $(S',w^{m\cdot 2\bigM})$ is a stable cycle. 
    Since $(S',w)$ is reflexive, then $\booltrans(S',w)\subseteq S'$, so $\booltrans(S',w^{k})\subseteq S'$ for every $k\in \bbN$, and in particular for $k=m2\bigM$. 
    Similarly, by \cref{def:stable cycle} we have that for every $n\in \bbN$, the baseline state $s\in S$ satisfies $s\in \MinRefStates(S',w^{n})$ and has a $0$-cycle, and in particular this holds for multiples of $m 2\bigM$, so $(S',w^{m\cdot 2\bigM})$ is a stable cycle.

    Next, consider $(s,r)\in \GroundPairs(S',w)$, then by definition we have $s\runsto{w^\bigM}g\runsto{w^\bigM}r$  for $g\in \MinRefStates(S',w^\bigM)$, but then we also have $s\runsto{(w^{m2\bigM})^\bigM} g\runsto{(w^{m2\bigM})^\bigM} r$ via the run
    $
    s\runsto{w^\bigM}g \runsto{w^{m4\bigM^2-2}} g   \runsto{w^\bigM} r 
    $,
    so $(s,r)\in \GroundPairs(S',w^{m\cdot 2\bigM})$. Moreover, by the run above we also have
    $\minweight(w^{2\bigM},s\to r)\le \minweight((w^{m2\bigM})^{2\bigM},s\to r)$.

     Conversely, consider $(s,r)\in \GroundPairs(S',w^{m\cdot 2\bigM})$, then by \cref{prop:grounding states reachable with any bigM k} we have $(s,r)\in \GroundPairs(S',w)$, and by \cref{lem:pumping grounded pairs} we have $\minweight((w^{m2\bigM})^{2\bigM},s\to r)\le \minweight(w^{2\bigM},s\to r)$.

     Since the transitions on cactus letters are only defined for grounded pairs, this concludes the proof.
\end{proof}

We now apply stabilization inductively, and take the infinite union, as follows.
\begin{definition}[\keyicon Stabilization Closure]
\label{def:stab closure}
For every $k\in \bbNinf$ we define  
$\stab_k(\augA) = \augA_k =\tup{S,\Gamma_k,\augInitState,\augTrans_k}$ inductively as follows.
\begin{itemize}
    \item For $k=0$ we define $\Gamma_0=\Gamma$, $\augTrans_0=\augTrans$.
    \item For $k>0$ we define $\Gamma_k=\Gamma_{k-1}\cup \stab_{\augA_{k-1}}(\Gamma_{k-1})$ and $\augTrans_k=\augTrans_{k-1}\cup \stab_{\augA_{k-1}}(\augTrans_{k-1})$ 
    \item For $k=\infty$ we define $\Gamma_\infty=\bigcup_{k\in \bbN}\Gamma_k$ and $\augTrans_\infty=\bigcup_{k\in \bbN}\augTrans_k$.
\end{itemize}
\end{definition}
We refer to the letters in $\Gamma_\infty$ as \emph{cactus letters of $\Gamma$} and we denote 
$\augA_\infty = \stab_\infty(\augA) = \bigcup_{k\in \bbN}\augA_k$. Notice that while the sets of letters and transitions of $\augA_\infty$ are infinite, its set of states is $S$.  
\begin{example}[Cactus Letters]
    \label{xmp:cactus letters}
    Suppose $\Gamma=\{a,b\}$, we construct some cactus letters (See \cref{fig:cactus}). 
    For the sake of the example we do not have a concrete WFA, as we just illustrate the syntax of cactus letters. 
    Let ${\color{Cerulean}w_1=aaa}$ and ${\color{Cerulean}w_4=bbb}$, we define ${\color{Cerulean}\alpha_1=\alpha_{S_1,w_1}}$ and ${\color{Cerulean}\alpha_4=\alpha_{S_4,w_4}}$. 
    We then recursively define ${\color{red} w_2=ab{\color{Cerulean}\alpha_1 }ab}$ with ${\color{red}\alpha_2=\alpha_{S_2, w_2}}$. 
    We go up another level: ${\color{Green} w_3=baa{\color{Cerulean}\alpha_4}abba{\color{red}\alpha_2}ba}$ with ${\color{Green}\alpha_3=\alpha_{S_3, w_3}}$. We can consider another word outside, e.g., $ab{\color{Green}\alpha_3} aa$.
    Note that we could write the entire word as a single expression, but it is very cumbersome:
    \[ 
    ab \cdot {\color{Green}\alpha_{S_3,baa\cdot{\color{Cerulean}\alpha_{S_4,bbb}}\cdot abba\cdot{\color{red}\alpha_{S_2,ab\cdot{\color{Cerulean}\alpha_{S_1,aaa}}\cdot ab}}\cdot ba}}\cdot aa
    \]
\end{example}

\begin{figure}[ht]
        \centering
        \includegraphics[width=0.8\linewidth]{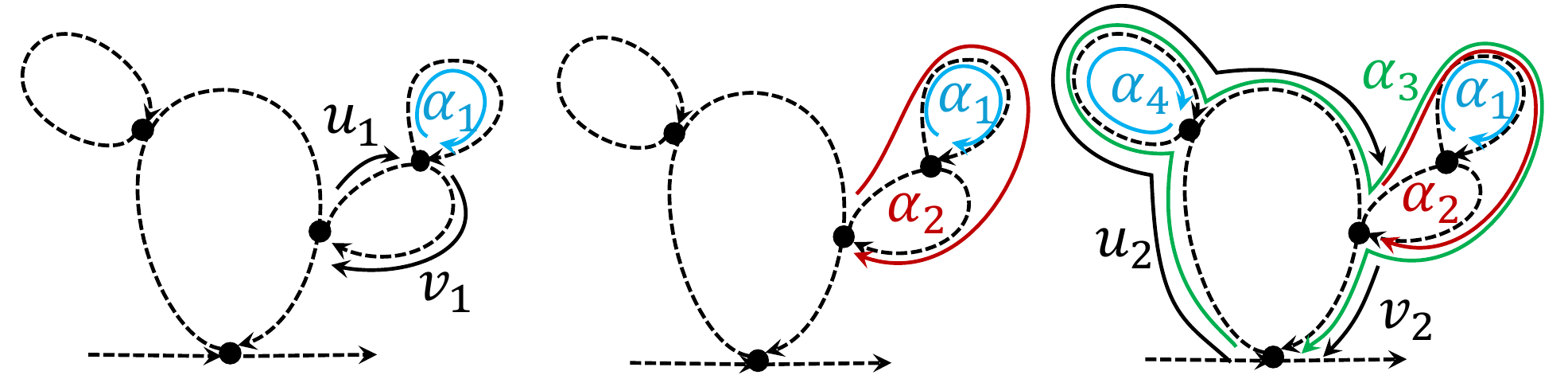}
        \caption{Cactus letters and cactus chains: $\alpha_1$ is of depth $1$. The word $u_1\alpha_1v_1$ (left) forms a cactus letter $\alpha_2$ of depth $2$ (center). Then, $u_2\alpha_2v_2$ forms a cactus letter $\alpha_3$ of depth $3$ (right). Notice that $u_2$ contains the cactus letter $\alpha_4$. The sequence $\alpha_3,\alpha_2 ,\alpha_1$ is a cactus chain, and so is $\alpha_3,\alpha_4$.}
        \label{fig:cactus}
    \end{figure}

\begin{remark}
\label{rmk:baseline transitions with cactus are zero}
Recall that in $\augA$, transitions between baseline states have weight 0. We remark that this is preserved in $\augA_\infty$. Indeed, a transition on $\alpha_{S',w}$ between baseline states $(p,p,T)$ and $(q,q,T')$ implies that $p=q$, and that $w$ prescribes a $0$-weight self loop on $p$, since $S'$ is a stable cycle and $p\in \MinRefStates(S',w^{\bigM})$ (see \cref{def:stable cycle}). 
\end{remark}

Consider a word $w\in \Gamma_\infty^*$, then $w$ is a concatenation of letters from various $\Gamma_k$'s. In the following we define the \emph{depth} of $w$ as the maximal $k$ used in the letters of $w$. This is illustrated in \cref{fig:cactus}.

In addition, we define a \emph{cost} of a word. Intuitively, the cost represents a measure of ``length'' of words, where cactus letters are defined to be extremely long, and dramatically increasing in length with each added depth. The notion of cost is central to our proof, and becomes clear later on.
\begin{definition}
\label{def:cost depth and sub cactus}
The functions $\cost:\Gamma_\infty\to \bbN$ and $\depth:\Gamma_\infty\to \bbN$ are defined inductively on the structure of $\Gamma_\infty^*$ in \cref{tab:cost depth and sub}.
\begin{table}[h!]
\centering
\begin{tabular}{|c|c|c|}
\hline
 & $\cost(w)=$ & $\depth(w)=$ \\
\hline
$w = a \in \Gamma_0=\Gamma$ & $ \wmax{a}+1$ & $1$  \\
\hline
$w = \alpha_{S',x} \in \Gamma_\infty$ & $ 2^{16(\bigM |S| \cdot \cost(x))^2}$ & $ 1 + \depth(x)$  \\
\hline
$w = \sigma_1 \cdots \sigma_m \in \Gamma_\infty^*,\  m\ge 2$ & $ \sum_{i=1}^m \cost(\sigma_i)$ & $ \max_{1\le i\le m}(\depth(\sigma_i))$ \\
\hline
\end{tabular}
\caption{The cost and depth functions.}
\label{tab:cost depth and sub}
\end{table}
\end{definition}

Note that \cref{def:cost depth and sub cactus} is indeed inductive, since if $\alpha_{S',w}\in \Gamma_k$ for some $k\in \bbN$, then $w\in \Gamma_{k-1}^*$. That is, considering the ``inner'' letters of a cactus letter must decrease the depth, so the definition reaches a terminating condition. 

\begin{remark}[$\cost$ is strictly increasing on prefixes]
    \label{rmk:cost is strictly increasing}
    The reason for $+1$ in the cost of a single letter is just so that the cost function is always strictly positive. Then, since the cost of a word is defined as the sum of the costs of its letters, we have that if $x$ is a prefix of $y$, then $\cost(x)<\cost(y)$.
\end{remark}

The precise reason for the cost of $\alpha_{S',x}$ being defined as it is becomes apparent in later sections. For now, we show that the cost can be used to upper-bound the maximal weight.
\begin{proposition}
\label{prop:cost higher than max weight}
    For every $w \in \Gamma_\infty^*$ it holds that $\maxeff{w}\le \cost(w)$.
\end{proposition}
\begin{proof}
The proof is immediate by induction according to \cref{tab:cost depth and sub}, as follows.
    \begin{itemize}
        \item If $w=a\in \Gamma_0$, then by definition $\cost(a)=\wmax{a}+1=\maxeff{a}+1>\maxeff{a}$.
        \item If $w=\alpha_{S',x}$, let $p,q$ such that $(p,\alpha_{S',x},c,q)\in \augTrans_\infty$ with maximal (finite) $c$ in absolute value, i.e., $|c|=\wmax{\alpha_{S',x}}$. 
        Since there is a finite-weight transition on $\alpha_{S',x}$ from $p$ to $q$, it follows from \cref{def:stabilization} that $c=\minweight(x^{2\bigM},p\runsto{x^\bigM} g\runsto{x^\bigM} q)$ for some state $g$. In particular, $c\le \wmax{x^{2\bigM}}2\bigM |x|$. By induction we have $\wmax{x}|x|\le \cost(x)$, so 
        \[
        \begin{split}
            c&\le \wmax{x^{2\bigM}}2\bigM |x| =\wmax{x}2\bigM |x|\\
            &\le \cost(x)2\bigM \le 2^{16\bigM |S|^2 \cost(x)}=\cost(\alpha_{S',x})=\cost(\alpha_{S',x})
        \end{split}
        \]
        where we use the fact that $\wmax{x^{2\bigM}}=\wmax{x}$, since $\wmax{}$ depends only of the set of letters occurring in the word, and the fact that $b 2\bigM \le 2^{16 \bigM |S|^2 b}$, which holds for all $b\in \bbN$.

        We conclude that $\wmax{\alpha_{S',x}}|\alpha_{S',x}|=c\le \cost(\alpha_{S',x})$, as required.
        
        \item If $w=\sigma_1\cdots \sigma_m$, we have by the induction hypothesis that
        \[ \cost(w)=\sum_{i=1}^n\cost(\sigma_i)\ge \sum_{i=1}^n\maxeff{\sigma_i}=\sum_{i=1}^n\wmax{\sigma_i}=\maxeff{w}\] 
        where we use the fact that for single letters, $\maxeff{\sigma_i}=\wmax{\sigma_i}$.
    \end{itemize}
\end{proof}

\subsection{Rebase Letters}
\label{sec:rebase letters}
In light of \cref{rmk:baseline transitions with cactus are zero}, we note that cactus letters maintain the baseline. However, later on in the proof (in \cref{sec:condition for type 1 witness}), we need to change the baseline while reading a stable cycle. 
To this end, we introduce an additional set of letters and transitions that allow to (intuitively) traverse a grounded pair by reading a cactus letter $\alpha_{S',w}$, while changing the baseline to a new state. The weight of these transitions is then normalized by the weight of the transition on $\alpha_{S',w}$.
\begin{definition}[Rebase]
\label{def:rebase}
    Let $\augA_\infty$ be the stabilization closure of $\augA$.
    Consider a cactus letter $\alpha_{S',w}\in \Gamma_\infty$ and write $S'=\{(q,p,T)\mid q\in T\}$. 
    For each grounded pair $(s,r)\in \GroundPairs(S',w)$ where $s=(q_1,p,T)$ and $r=(q_2,p,T)$ with the transition $(s,\alpha_{S',w},c,r)\in \augTrans_\infty$, we add a letter $\beta_{S',w,s\to r}$ with the following transitions:
 \[
\{((q',q_1,T),\beta_{S',w,s\to r},d-c,(q'',q_2,T))\mid ((q',p,T),\alpha_{S',w},d,(q'',p,T))\in \augTrans_\infty\}
 \]
We refer to the resulting WFA as $\rebase(\augA_\infty)$. 
\end{definition}
\begin{figure}[ht]
    \centering
    \includegraphics[width=0.75\linewidth]{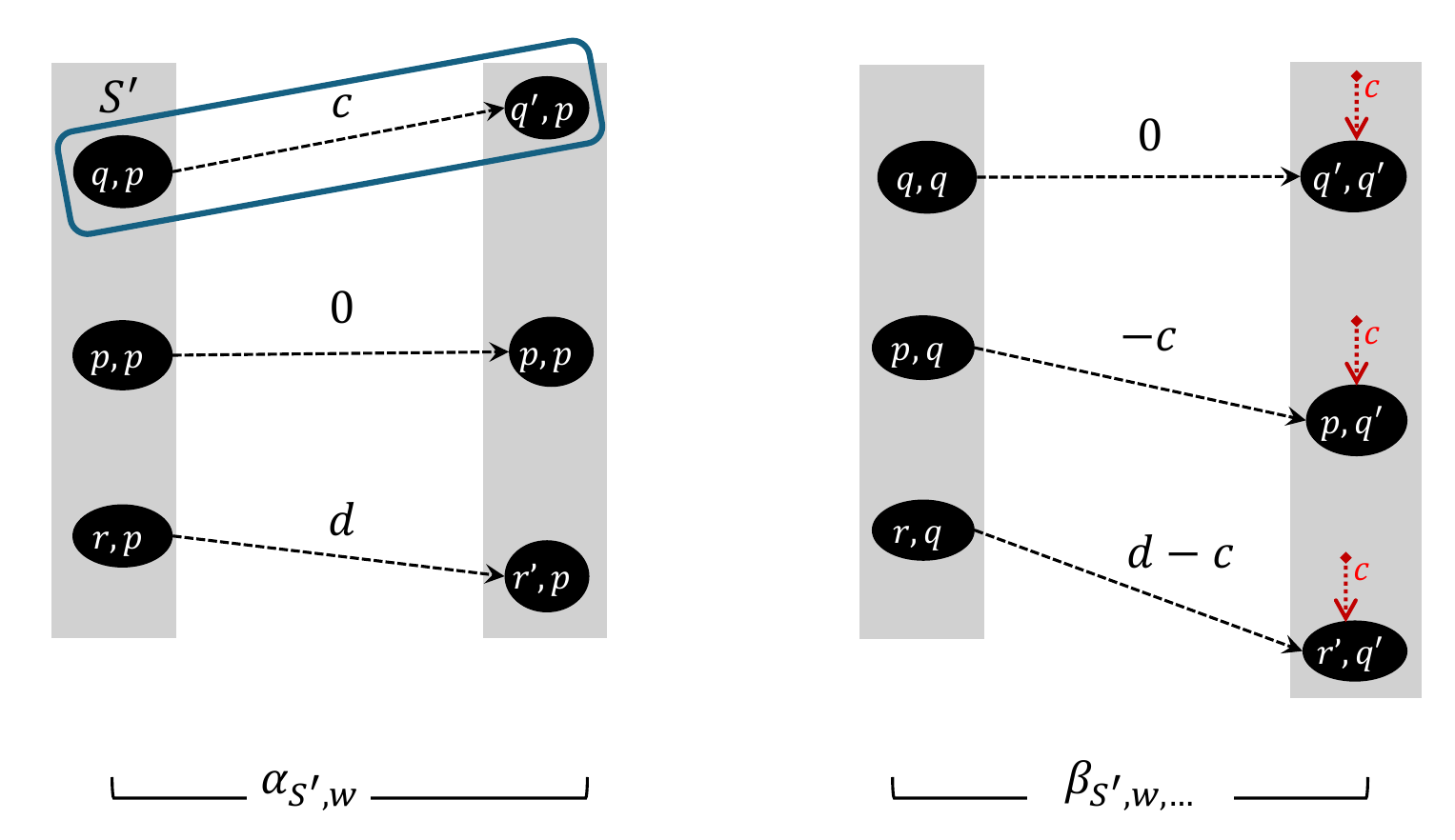}
    \caption{Rebase letter $\beta_{S',w,s\to r}$ with $s=(q,p,T)$ and $r=(q',p,T)$ induced by the $\alpha_{S',w}$ transition from $(q,p)$ to $(q',p)$. The weight of the resulting transitions, as well as the baseline run, shifts down by $c$, in accordance with the new base transition from $q$ to $q'$. }
    \label{fig:rebase}
\end{figure}
We illustrate the behavior of a rebase letter in \cref{fig:rebase}.

\begin{remark}[Cactus as Rebase]
\label{rmk:cactus as rebase}
Consider a cactus letter $\alpha_{S',w}$ where $S'=\{(q,p,T)\mid q\in T\}$. We observe that for the baseline state $s_p=(p,p,T)$, the transitions defined on the letter $\beta_{S',w,s_p\to s_p}$ are identical to those on $\alpha_{S',w}$.

Indeed, $(s_p,\alpha_{S',w},0,s_p)\in \augTrans_\infty$, so plugging in $0$ for $c$ in \cref{def:rebase} gives identical weights to those of $\alpha_{S',w}$.
\end{remark}

As a sanity check, note that the weight of the transition from $s$ to $r$ upon reading $\beta_{S',w,s\to r}$ is $0$ (since $(s,\alpha_{S',w},c,r)\in \augTrans_\infty$, and by plugging in $q'=q_1$ and $q''=q_2$ in \cref{def:rebase}). Therefore, the property that baseline transitions have weight $0$ is maintained.

\begin{remark}
\label{rmk:rebase is idempotent}
    Observe that rebase introduces new letters only for letters obtained by stabilization. Therefore, it is an idempotent operator. That is, if $\augA'$ is obtained from $\augA$ by applying rebase, then applying rebase on $\augA'$ does not introduce new letters.
\end{remark}

We now alternate the stabilization closure and rebase ad infinitum, as follows.
\begin{definition}[Cactus Extension]
\label{def:cactus extension}
Define $\augA^0_0=\augA$. 
For every $k,j\in \bbNinf$ we define 
$\augA^j_k=\tup{S,\Gamma^j_k,\augInitState,\augTrans^j_k}$ inductively as follows.
\begin{itemize}
    \item We define $\augA^0_k=\stab_k(\augA)$ as per \cref{def:stab closure}.
    \item For $j>0$ we define $\augA^j_0=\rebase(\augA^{j-1}_\infty)$.
    \item For $j>0$ and $k>0$ (including $k=\infty$) we define $\augA^j_k = \stab_k(\augA^j_0)$.
    \item For $j=\infty$ and every $k$ we define $\augA^\infty_k$ as the union (of alphabets and transitions) of $\augA^j_\infty$ for $j\in\bbN$. 
\end{itemize}
\end{definition}
\begin{remark}
\label{rmk:rebase not relevant in infinite level}
    Observe that $\augA^\infty_k=\augA^\infty_\infty$ for all $k$, since applying stabilization on the infinite union does not add new letters (since stabilization is applied to words of some finite depth anyway, and is therefore already in the union).
\end{remark}
\begin{remark}
    \label{rmk:cost functions on cactus extension levels}
    Recall that the $\cost$ and $\depth$ functions (\cref{def:cost depth and sub cactus}) are defined with respect to the stabilization closure. Thus, when considered on a word in the scope of $\augA^j_\infty$ for some $j>0$, the underlying alphabet $\Gamma$ is actually $\Gamma^j_0$. 
    That is, rebase letters are considered as having depth $1$, like ``regular'' letters. 
    In the context of our results, we never use $\cost$ and $\depth$ over rebase letters. However, since we also establish a general toolbox, this is worth noticing.
\end{remark}

\subsection{Jump Letters}
\label{sec:jump letters}
Our final type of letters is a technical addition that simplifies (some) proofs.
These letters, referred to as \emph{jump letters}, allow us to change the baseline state without other effect. 
\begin{definition}
    \label{def:jump letters}
    Consider a cactus extension $\augA^j_k$. For states $s=(q,p_1,T), r=(q,p_2,T)\in S$ and $T\subseteq Q$ (recall that $Q$ is the state space of the original WFA $\cA$) we introduce a \emph{jump letter} $\jl_{s\to r}$ with the following transitions. For states $t=(q',p_1,T)$ and $g=(q',p_2,T)$ (note the same \emph{first} and third component), we add the transition $(t,\jl_{s\to r},0,g)$.
    For states not of the form $(\cdot,p_1,T)$, the transitions have weight $\infty$.
    Note that the transitions on $\jl_{s\to r}$ are deterministic.
\end{definition}
Thus, jump letters can change the baseline without further effect. 
Note that since the baseline component is deterministic, then upon reading $\jl_{s\to r,T}$, all reachable states change their baseline. 
Also observe that jump can only change the baseline component to a state in $T$. However, as we discuss in the following (\cref{sec:reachable ghost states}), some states in $T$ may be unreachable. And indeed, $\jl$ letters may allow the baseline to be at an unreachable state, thus making the baseline run not seamless.

\begin{remark}
    \label{rmk:jump letters first component}
    Note that in \cref{def:jump letters}, the component $q$ in $s$ and $r$ is ignored (indeed, the same transitions are induced regardless of $q$). Thus, we sometimes denote the states $s,r$ of $\jl_{s\to r}$ as e.g., $s=(\cdot,p_1,T)$, where $\cdot$ stands for some arbitrary state.
\end{remark}

Finally, due to the ``violent'' nature of jump letters, we emphasize whenever they are present in a run, and in particular in a seamless run. We say that a run is \emph{jump free} if its transitions do not contain a jump letter.
We do not annotate whether an alphabet contains jump letters or not, but we state so explicitly when relevant.

\subsection{Ghost Reachable States}
\label{sec:reachable ghost states}
Consider a word $w$. Recall that $\booltrans(s_0,w)$ is the set of states reachable from $s_0$ by reading $w$. 
In $\augA$, before the introduction of cactus letters, it holds that 
$\booltrans(s_0,w)=\{(p,q,T)\mid p\in T\}$ for some set $T$. That is, the reachable first components are exactly captured by the third component $T$.
After introducing cactus letters, however, this no longer holds. Indeed, \cref{def:cactus extension,def:rebase} limit the transitions to certain states (namely those obtained from grounded pairs). In particular, we have that
$\booltrans(s_0,w)=\{(p,q,T)\mid p\in T'\}$ for some $T'\subseteq T$, but equality does not necessarily hold.
Intuitively, the states in $T\setminus T'$ are 
those for which repetitions of cactus infixes lead to unbounded increase in weight, and are therefore abstracted to $\infty$ in the stabilization.

Nonetheless, our analysis (in particular in \cref{sec:witness,sec:existence of separated increasing infixes}) requires us to reason about these states, which we term \emph{ghost states}. 

\begin{definition}[Ghost-Reachable States $\mathghost$]
    \label{def:ghost states}
    Consider a state $s=(p_1,q_1,T_1)$ and word $w$ with $\booltrans(s,w)=\{(p_2,q_2,T'_1)\mid p_2\in T_2\}$ for some $T_2\subseteq T'_1$. 
    We define the set \emph{ghost-reachable states of $w$ from $s$} as
$\ghostTrans(s,w)=\{(p_2,q_2,T'_1)\mid p_2\in T'_1\}$.
\end{definition}
Note that $\ghostTrans(s,w)$ includes the ghost states as well as the standard reachable states. That is, $\booltrans(s,w)\subseteq \ghostTrans(s,w)$.
Moreover, observe that $\ghostTrans(s,w)$ is determined by $\booltrans(s,w)$, and therefore if $\booltrans(s,w)=\booltrans(s',w')$ then $\ghostTrans(s,w)=\ghostTrans(s',w')$.

\begin{remark}[Stable Cycles Maintain Ghost States]
\label{rmk:stable cycles maintain ghost states}
    For every stable cycle $(S',w)$ we have $\ghostTrans(S',w)=S'$. 
    Indeed, recall that for a cycle to be stable, and in particular reflexive (\cref{def:reflexive cycle and state}), it holds that $S'=\{(p,q,T)\mid p\in T\}$ for some $T$ (i.e., the set of states is ``saturated'' with respect to $T$), and that $\booltrans(S',w)\subseteq S'$. In particular, all states in $\booltrans(S',w)$ are of the form $(\cdot,q,T)$ for some $q,T\in Q$.
    Thus, we have $\ghostTrans(S',w)=\{(p,q,T)\mid p\in T\}=S'$.
\end{remark}

We also lift the definition to runs, as follows.
\begin{definition}[Ghost Run]
\label{def:ghost run}
    Consider words $w_1,w_2$ and states $p\in \ghostTrans(s_0,w_1)\setminus\booltrans(s_0,w_1)$ and $q\in \ghostTrans(s_0,w_1w_2)\setminus\booltrans(s_0,w_1w_2)$.
    We then say that a run $q \runsto{w_2} p$ on the infix $w_2$ is a \emph{ghost run from $w_1$ to $w_1w_2$}. 
\end{definition}
We note that a ghost run cannot visit the same state as a ``real'' run on any prefix of $w_2$, as that would entail $q\in \booltrans(s_0,w_1w_2)$. Thus, ghost runs are ``uninterrupted'' by real ones.

\section{A Toolbox for Cacti}
\label{sec:cactus toolbox}
Cactus, rebase and jump letters intuitively do two things: (1) abstract away parts of words as ``pumpable'' cycles, in a \emph{nested} manner, and (2) shift the baseline run between states. 
A central theme in our reasoning is to utilize these features by ``folding'' parts of words into a cactus, ``unfolding'' other parts, possibly in a nested fashion, and shifting the baseline to a certain run as a ``point-of-view change''. These operations demonstrate the versatility of the cactus framework. 
In this section, we present the toolbox that allows us to perform these operations. This is divided to three parts: a framework to bound the depth of ``interesting'' nested cacti (\cref{sec:degeneracy and cactus depth bound}), a framework for shifting the baseline run (\cref{sec: baseline shift}), and a framework for getting rid of all cacti and rebase letters by \emph{flattening} (\cref{sec: cactus unfolding}).

\subsection{Degeneracy and A Bound on Cactus Depth}
\label{sec:degeneracy and cactus depth bound}
In this section, we consider words built up by repeatedly nesting cactus letters, as the following definition captures (see \cref{fig:cactus}).
\begin{definition}[Cactus Chain]
\label{def:cactus chain}
    A \emph{cactus chain} is a finite sequence of letters $\alpha_{S'_1,w_1},\ldots,\alpha_{S'_n,w_n}$ such that for every $1\le i<n$ we have $w_i=u_i \alpha_{S'_{i+1},w_{i+1}}v_i$.
\end{definition}

Cactus chains can be arbitrarily deep. Our goal in this section is to show that while this is the case, only a bounded number of ``layers'' can include interesting behaviors. In order to formalize this, we introduce the notions of \emph{degenerate} and \emph{non-degenerate} stable cycles, as follows.

\begin{definition}[(Non-)Degenerate Stable Cycle]
\label{def:degenerate stable cycle}
    A stable cycle $(S',w)$ is \emph{degenerate} if for every $s\in S'$, $t\in \RefStates(S',w^{2\bigM})$, if there is a run $s\runsto{w^{2\bigM k}}t$ for some $k\ge 1$, then $(s,t)\in \GroundPairs(S',w)$.

    Otherwise, $(S',w)$ is \emph{non-degenerate}, and a state $s\in S'$ such that there exists $t\in \RefStates(S',w^{2\bigM})$ and $k\ge 1$ for which $s\runsto{w^{2\bigM k}}t$ and $(s,t)\notin \GroundPairs(S',w)$ is called a \emph{non-degenerate state}.
\end{definition}
Intuitively, the ``canonical'' degenerate cycle is one where all the reflexive states are also minimal-reflexive (i.e., have a $0$-cycle on $w^{2\bigM}$). Indeed, then any reflexive state can serve also as a grounding state. The definition of degeneracy generalizes this notion, by allowing some reflexive states that are not minimal, but they cannot be reached without going through a grounding state, and therefore their non-minimality does not really add any unboundedness. We formalize this in \cref{sec:boundedness of degenerate cycles}. 
\subsubsection{Boundedness of Degenerate Cycles}
\label{sec:boundedness of degenerate cycles}
\begin{proposition}
\label{prop:degenerate cycle has bounded weights}
    Consider a degenerate stable cycle $(S',w)$ and $s,r \in S'$. For every $k\in \bbN$, if $s\runsto{w^{2\bigM |S| k}}r$, then 
    $|\minweight(w^{2\bigM |S| k},s \to r)|\leq 2\bigM \cost(w)$.
\end{proposition}
\begin{proof}
    We split the proof to the upper and lower bound, starting with the former.
    Consider $s,r\in S'$ and $k\in \bbN$ such that $s\runsto{w^{2\bigM |S| k}}r$. We show that $\minweight(w^{2\bigM |S| k},s \to r)\le  2\bigM \cost(w)$. Consider the minimal-weight run $\rho:s\runsto{w^{2\bigM |S| k}}r$, i.e., $\weight(\rho)=\minweight(w^{2\bigM |S| k},s \to r)$.
    
    By the pigeonhole principle there exists $t\in S'$ such that we can write $\rho:s \runsto{w^{2\bigM k_1}} t\runsto{w^{2\bigM  k_2}} t \runsto{w^{2\bigM  k_3}}r$ with $k_1+k_2+k_3=|S|k$ and $k_2>0$. In particular, we have $t\in \RefStates(S',w^{2\bigM k_2})= \RefStates(S',w^{2\bigM})$ (by \cref{prop:reflexive cycles stabilize at N}). 

    Since $(S',w)$ is degenerate, it follows from \cref{def:degenerate stable cycle} that $(s,t)\in \GroundPairs(S',w)$. We claim that this implies $(s,r)\in \GroundPairs(S',w)$. Indeed, take the grounding state of $(s,t)$, i.e., some $t'\in \MinRefStates(S',w^{\bigM})$ such that $s\runsto{w^{\bigM}}t'\runsto{w^{\bigM}}t$, then we can write
    \[s \runsto{w^{\bigM}} t'\runsto{w^{\bigM}}t \runsto{w^{2\bigM  k_3}}r\]
    and therefore $s\runsto{w^{\bigM}}t'\runsto{w^{(2k_3+1)\bigM}}r$. By \cref{prop:grounding states reachable with any bigM k} this implies $(s,r)\in \GroundPairs(S',w)$.
    Now, by \cref{lem:pumping grounded pairs,prop:cost higher than max weight}, we have that 
    \[\minweight(w^{2\bigM |S| k},s\to r)\le \minweight(w^{2\bigM},s\runsto{w^\bigM} t'\runsto{w^\bigM}r)\le 2\bigM\cost(w)\]
    Concluding the upper bound.

    We proceed with the lower bound, which is significantly simpler. Intuitively, in order for the weight to decrease so much, there must be a negative cycle, which cannot happen in a stable cycle. We now formalize this argument.

    Recall that $\bigM\gg |S|$. We prove a stronger claim than the lower bound: for every $m\in \bbN$ we have $\minweight(w^m,s\runsto{m}r)\ge -|S|\cost(w)>-2\bigM \cost(w)$.
    To this end, assume otherwise by way of contradiction, and let $m\in \bbN$ be the minimal number such that $\minweight(w^m,s\to r)<-|S|\cost(w)$. By \cref{prop:cost higher than max weight} it must hold that $m\ge |S|$, since $|\maxeff{w^m}|\le m\cost(w)$.
    
    Let $\rho$ be a minimal-weight run such that $\rho:s\runsto{w}s_1\runsto{w}s_2\runsto{w}\cdots \runsto{w}s_{m-1}\runsto{w}r$. Denote $s=s_0$ and $r=s_m$. Since $m\ge |S|$, then by the pigeonhole principle there exist $0\le i<j\le m$ such that $s_i=s_j$. Denote $t=s_i=s_j$, then $t\in \RefStates(S',w)$. 
    Now, if $\minweight(w^{j-i},t\to t)\ge 0$, then we can shorten the run $\rho$ while still having value lower than $-|S|\cost(w)$, contradicting the minimality of $m$. However, if $\minweight(w^{j-i},t\to t)< 0$, this contradicts the fact that $(S',w)$ is a  stable cycle (\cref{def:stable cycle}). Thus, we are done.
\end{proof}
A simple corollary of \cref{prop:degenerate cycle has bounded weights} is that enough repetitions of $w$ cause a repeated configuration, as follows.
\begin{proposition}
\label{prop:degenerate cycle has repeating configuration}
    Consider a degenerate stable cycle $(S',w)$. There exist
    $m<n\le (2+4\bigM\cost(w))^{|S|^2}$
    such that for every $t\in S'$ and corresponding configuration $\vec{c_t}$ we have \[\xconf(\vec{c_t},w^{2\bigM |S| m})=\xconf(\vec{c_t},w^{2\bigM |S| n})\]
\end{proposition}
\begin{proof}
    Fix $t\in S'$. 
    By \cref{prop:degenerate cycle has bounded weights}, for every $k\in \bbN$ it holds that 
    \[ 
    \xconf(\vec{c_t},w^{2\bigM |S| k})\in (\{-2\bigM\cost(w),\ldots, 2\bigM\cost(w)\}\cup\{\infty\})^{S}
    \]
    Denote $D=\{-2\bigM\cost(w),\ldots, 2\bigM\cost(w)\}\cup\{\infty\}$, then $|D|=4\bigM\cost(w)+2$. 
    Taking this view on all states simultaneously, we have that for every $k\in \bbN$, the vector 
    $\mathcal{T}_k=(\xconf(\vec{c_t},w^{2\bigM |S| k}))_{t\in S}$ is in $(D^{S})^{S}$.

    Thus, by the pigeonhole principle, there exists $0\le m<n\le (4\bigM\cost(w)+2)^{|S|^2}$ such $\mathcal{T}_m=\mathcal{T}_n$, which proves the claim.
\end{proof}
Naturally, once a configuration repeats in all states, this cycle can be arbitrarily pumped without any effect on the weight, as follows.
\begin{proposition}
\label{prop:stable cycle config repeat}
    Consider a stable cycle $(S',w)$ and $m<n \in \bbN$ such that for every $t \in S'$ it holds that $\xconf(\vec{c_t},w^{2\bigM m})=\xconf(\vec{c_t},w^{2\bigM n})$. Then for every $t,r \in S'$ and every $i \in \bbN$, we have that $\minweight(w^{2\bigM (m+i (n-m))},t \to r)=\minweight(\alpha_{S',w}, t \to r)$.
\end{proposition}
\begin{proof}
    Denote $k=n-m$ and $n_i=ik+m$ for every $i\in \bbN$. In particular, $n_0=m$ and $n_1=n$. By induction it is immediate that  $\xconf(\vec{c_t},w^{2\bigM m})=\xconf(\vec{c_t},w^{2\bigM n_i})$ for every $i\in \bbN$.

    The slightly more involved part of the proof is concluding that for every $t,r\in S'$ it holds that $\minweight(w^{2\bigM (m+i (n-m))},t \to r)=\minweight(\alpha_{S',w}, t \to r)$. 
    Intuitively, this holds because in order for the configurations to repeat, either $r$ is not reachable from $t$ at all, or the weights remain bounded, meaning that $(t,r)$ is a grounded pair. We now formalize this intuition.
    
    First, if $\minweight(w^{2\bigM (m+i (n-m))},t \to r)=\infty$ for some $i\in \bbN$, then in particular $\minweight(w^{2\bigM},t\to r)=\infty$. Indeed, if there is a finite-weight run from $t$ to $r$ on 
    $w^{2\bigM}$, then there is a reflexive state within this run, which can be pumped to a finite-weight run on $w^{2\bigM (m+i (n-m))}$ (indeed, on any multiple of $w^{2\bigM}$).
    
    Next, assume $\minweight(w^{2\bigM (m+i (n-m))},t \to r)=c<\infty$, we claim that $(t,r)\in \GroundPairs(S',w)$. Indeed, assume by way of contradiction that $(t,r)\notin \GroundPairs(S',w)$, then by \cref{lem:pumping grounded pairs} there exists $M_0\in \bbN$ such that for every $m'\ge M_0$ it holds that $\minweight(w^{2\bigM m'},t \to r)>c$. Take $m'=m+M_0(n-m)$, then $m'\ge M_0$, so this inequality holds. However, by first part of the proof (the repetition of configurations), we have that 
    \[\xconf(\vec{c_t},w^{2\bigM m})=\xconf(\vec{c_t},w^{2\bigM (m+i (n-m))})=\xconf(\vec{c_t},w^{2\bigM (m+M_0(n-m))})\]
    and therefore
    \[
    c<\minweight(w^{2\bigM m'},t \to r)=\minweight(w^{2\bigM (m+i (n-m))},t \to r)=c
    \]
    which is a contradiction.

    We conclude that $(t,r)\in \GroundPairs(S',w)$, and therefore by \cref{lem:pumping grounded pairs,def:stabilization} we have $\minweight(w^{2\bigM (m+i (n-m))},t \to r)=\minweight(\alpha_{S',w},t\to r)$.
\end{proof}

We conclude this section with its main result, stating that for a degenerate cycle $(S',w)$, we can find $m$ that is not too large, such that the behavior of $w^{2\bigM |S| m}$ is identical to $\alpha_{S',w}$, but $w^{2\bigM |S| m}$ has lower cost than $\alpha_{S',w}$. Intuitively, this means that it is worthwhile  replacing $\alpha_{S',w}$ with $w^{2\bigM |S| m}$. This, in turn, implies that in order to minimize the cost of a word, degenerate cycles can be eliminated.

\begin{lemma}
    \label{lem:degenerate cycles high cost}
    Consider a degenerate stable cycle $(S',w)$. There exists $m<(2+4\bigM\cost(w))^{|S|^2}$ such that $\cost(w^{2\bigM |S| m})<\cost(\alpha_{S',w})$ and for every $t,r\in S'$ it holds that $\minweight(w^{2\bigM |S| m},t \to r)=\minweight(\alpha_{S',w},t \to r)$.
\end{lemma}
\begin{proof}
    Combining \cref{prop:degenerate cycle has repeating configuration,prop:stable cycle config repeat} (note that the constant $|S|$ needs to be introduced in \cref{prop:stable cycle config repeat}), there exist $m<(2+4\bigM\cost(w))^{|S|^2}$ such that for every $t,r\in S'$ it holds that $\minweight(w^{2\bigM |S| m},t\to r)=\minweight(\alpha_{S',w},t\to r)$.

    It remains to work out the cost constraint, using \cref{def:cost depth and sub cactus} (\cref{tab:cost depth and sub}).
    \[
    \begin{split}
        &\cost(w^{2\bigM |S| m})=2\bigM |S| m \cost(w) < 2\bigM |S| (2+4\bigM\cost(w))^{|S|^2} \cost(w)\\
        &<_{(\dagger)}2^{16(\bigM |S| \cost(w))^2}=\cost(\alpha_{S',w})
    \end{split}
    \]
    where the inequality $(\dagger)$ is a crude upper bound that can be seen as follows: denote $|S|=s$ and $\cost(w)=c$, then $\bigM s<\bigM^{s^2}$, and we have
    \[
    \begin{split}
        &2\bigM s (2+4\bigM c)^{s^2} c <  2\bigM s c(8\bigM c)^{s^2} < 2^{s^2}\bigM^{s^2} c^{s^2}(8\bigM c)^{s^2} =\\
        &(16 \bigM^2c^2)^{s^2}< 2^{(16 \bigM^2c^2)^{s^2}}=2^{16 (\bigM s c)^2}
    \end{split}
    \]    
\end{proof}

\subsubsection{Existence of Degenerate Cycles in Cactus Chains}
\label{sec:existence of degenerate cycle in chain}
Our goal in this section is to show that a deep-enough cactus chain must contain a degenerate stable cycle. Specifically, ``deep-enough'' is more than $|S|$. We need some definitions and result along the way.

We now define the \emph{Minimal-Reflexive Graph} of a stable cycle $(S',w)$. This directed graph captures the reachability relation between minimal-reflexive states, upon reading $w^{2\bigM}$.

\begin{definition}[$\MinGraph$]
    Consider a stable cycle $(S',w)$. Its \emph{minimal-reflexive graph} is the directed graph
    $\MinGraph(S',w)=\tup{V,E}$ with vertices
    $V=\MinRefStates(S',w^{\bigM})$ and 
    $E=\GroundPairs(S',w) \cap (\MinRefStates(S',w^{\bigM}) \times \MinRefStates(S',w^{\bigM}))$
\end{definition}
That is, we have $(s,r)\in E$ if $s,r\in \MinRefStates(S',w^{\bigM})$ and there exists $g\in \MinRefStates(S',w^{\bigM})$ such that $s\runsto{w^M}g\runsto{w^M}r$ (as per \cref{def:grounded pairs}).
We present a simple but useful property of $\MinGraph(S',w)$, namely that it is transitively closed.
\begin{proposition}
    \label{prop:mingraph is transitively closed}
    For every stable cycle $(S',w)$, the graph $\MinGraph(S',w)$ is transitively closed. That is, if $(s,r)\in E$ and $(r,t)\in E$, then $(s,t)\in E$.
\end{proposition}
\begin{proof}
    Assume $(s,r)\in E$ and $(r,t)\in E$. Thus, there exist $g_1,g_2\in \MinRefStates(S',w^{\bigM})$ such that $s\runsto{w^\bigM}g_1\runsto{w^\bigM}r$ and $r\runsto{w^\bigM}g_2\runsto{w^\bigM}t$. In particular, $s\runsto{w^{2\bigM}}r$ and $r\runsto{w^{2\bigM}}t$. 
    Since $V=\MinRefStates(S',w^{\bigM})$, then $s,r\in \MinRefStates(S',w^{\bigM})$. Since the minimal states stabilize at $w^{\bigM}$ (\cref{prop:minimal reflexive cycles stabilize at N}) we also have $s,r\in \MinRefStates(S',w^{2\bigM})$. 
    Thus, we have $(s,r),(r,t)\in \TethPairs(S',w^{2\bigM})$ (see \cref{def:tethered}). 
    Since tethered states also stabilize at $\bigM$ (\cref{prop:tethered states stabilize at N}), we have in particular that $(s,r),(r,t)\in \TethPairs(S',w^{\bigM})$, so $s\runsto{w^{\bigM}}r\runsto{w^\bigM}t$, and since $r\in \MinRefStates(S',w^{\bigM})$, we conclude that $(s,t)\in \GroundPairs(S',w)$, so $(s,t)\in E$.
\end{proof}
We now decompose $\MinGraph$ into (maximal) \emph{Strongly-Connected Components (SCC)}. Recall that an SCC is $C\subseteq V$ such that for every $s,r\in C$ there is a path from $s$ to $r$ and from $r$ to $s$, and that $C$ is maximal with this property (with respect to containment). By \cref{prop:mingraph is transitively closed}, in the case of transitive graphs we can just define an SCC to be a maximal clique (with respect to containment). 

We now turn to the main technical result of this section, relating the number of SCCs of consecutive levels in a cactus chain, where the inner cactus is non-degenerate.
\begin{lemma}[\lightbulbicon SCCs of Nested $\MinGraph$]
\label{lem:nondegenerate mingraph more SCC}
    Consider a stable cycle $(S',w)$ with $w=u\alpha_{B,x}v$. Then the number of SCCs of $\MinGraph(B,x)$ is at least the number of SCCs of $\MinGraph(S',w)$. Moreover, if $(S',w)$ is a non-degenerate stable cycle, then this number is strictly greater. 
\end{lemma}
\begin{proof}
    Denote the SCCs of $\MinGraph(S',w)$ by $C_1,...,C_m$. 
    For each SCC $C_i$ choose a representative $s_i\in C_i$. For every $1\le i\le m$, since $s_i\in \MinRefStates(S',w^{\bigM})$, then in particular there is a minimal reflexive run $s_i\runsto{\alpha_{S',w}}s_i$. By \cref{def:stabilization}, this means that $s_i\runsto{w^{2\bigM}}s_i$. That is, $s_i\runsto{(u\alpha_{B,x}v)^{2\bigM}}s_i$. We split the run as follows:
    \[ s_i\runsto{u}t_i\runsto{\alpha_{B,x}}r_i\runsto{v(u\alpha_{B,x}v)^{2\bigM-1}}s_i
    \]
    Again by \cref{def:stabilization}, we have $(t_i,r_i)\in \GroundPairs(B,x)$, so there exists $g_i\in \MinRefStates(B,x^\bigM)$ such that $t_i\runsto{x^\bigM}g_i\runsto{x^\bigM}r_i$. We associate with each $s_i$ the corresponding state $g_i$. Note that $g_i$ is a vertex in $\MinGraph(B,x)$. 

    We claim that for every $j\neq i$, if there is an edge $(g_j,g_i)$ in $\MinGraph(B,x)$, then there is also an edge $(s_j,s_i)$ in $\MinGraph(S',w)$. Indeed, 
    By the assumption we have that $(g_j,g_i)\in \GroundPairs(B,x)$.
    We therefore have the following:
    \[\begin{split}        &s_j\runsto{u}t_j\runsto{x^\bigM}g_j\runsto{x^\bigM}r_j\runsto{v(u\alpha_{B,x}v)^{2\bigM-1}}s_j\\        &s_i\runsto{u}t_i\runsto{x^\bigM}g_i\runsto{x^\bigM}r_i\runsto{v(u\alpha_{B,x}v)^{2\bigM-1}}s_i\\        
        &g_j\runsto{x^\bigM}h\runsto{x^\bigM}g_i
        \end{split}
    \]
    where $h$ is the grounding state of $(g_j,g_i)$. In particular, we have $g_j\runsto{x^{2\bigM}}g_i$, and therefore $t_j\runsto{x^{\bigM}}g_j\runsto{x^{2\bigM}}g_i$, so $t_j\runsto{x^{3\bigM}}g_i$. 
    Since $g_i\in \MinRefStates(B,x^\bigM)= \MinRefStates(B,x^{3\bigM})$ (by \cref{prop:minimal reflexive cycles stabilize at N}), it follows that $t_j\in \PlatPairs(B,x^{3\bigM})$ (\cref{def:plateau}). By the stabilization of Plateau  pairs (\cref{prop:plateau states stabilize at N}) we have that $t_j\runsto{x^{\bigM}}g_i$.
    But then we can compose the run 
    \[ s_j\runsto{u}t_j\runsto{x^\bigM}g_i\runsto{x^\bigM}r_i\runsto{v(u\alpha_{B,x}v)^{2\bigM-1}}s_i
    \]
    So $s_j\runsto{w^{2\bigM}}s_i$. 
    Again by the stabilization of plateau pairs (\cref{prop:plateau states stabilize at N}) we have that $s_i\in \PlatPairs(S',w^{2\bigM})$ and therefore $s_j\runsto{w^\bigM}s_i$. But $s_i\in  \MinRefStates(S',w^\bigM)$, and it can serve as a grounding state. We therefore have $s_j\runsto{w^\bigM}s_i\runsto{w^\bigM}s_i$, so $(s_j,s_i)\in \GroundPairs(S',w)$, and is therefore an edge in $\MinGraph(S',w)$.
    We can now prove the desired properties.

    \paragraph{All the $g_i$ are in different SCCs of $\MinGraph(B,x)$}
    Assume by way of contradiction that there are $i\neq j$ such that $g_i$ and $g_j$ are in the same SCC $D$ of $\MinGraph(B,x)$.
    By \cref{prop:mingraph is transitively closed} we have that $(g_j,g_i)$ and $(g_i,g_j)$ are edges in $\MinGraph(B,x)$, and by the above we get that $(s_j,s_i)$ and $(s_i,s_j)$ are edges in $\MinGraph(S',w)$, which contradicts the fact that $s_j,s_i$ are in difference SCCs.
    
    We conclude that each $g_i$ lies in a separate SCC, so the number of SCCs in $\MinGraph(B,x)$ is at least that of $\MinGraph(S',w)$.

    \paragraph{If $(S',w)$ is non-degenerate, there is some $g'\in \MinRefStates(B,x^\bigM)$ that is not in any of the SCCs of the $g_i$.}
    We now proceed to the case where $(S',w)$ is non-degenerate, and we exhibit at least one additional SCC in $\MinGraph(B,x)$. 

    Since $(S',w)$ is non-degenerate (\cref{def:degenerate stable cycle}), it has a non-degenerate state $s'\in S'$. That is, there is some (non-minimal) reflexive state $r\in \RefStates(S',w^{2\bigM})$ with $s'\runsto{w^{2\bigM k}}r$ for some $k\in \bbN$, such that $(s',r)\notin \GroundPairs(S',w)$. Similarly to the previous case, we expand the run above as
    \[
    s'\runsto{u}t'\runsto{\alpha_{B,x}}r'\runsto{vw^{2\bigM k-1}}r
    \]
    By \cref{def:stabilization} we have $(t',r')\in \GroundPairs(B,x)$, so there exists a state $g'\in \MinRefStates(B,x^\bigM)$ such that
    $t'\runsto{x^{\bigM}}g'\runsto{x^{\bigM}}r'$.
    We therefore have
    \[ s'\runsto{u}t'\runsto{x^{\bigM}}g'\runsto{x^{\bigM}}r'\runsto{vw^{2\bigM k-1}}r
    \]
    We claim that $g'$ is not in any SCC of $g_i$ for $1\le i\le m$. Assume otherwise by way of contradiction, then by \cref{prop:mingraph is transitively closed} there exists some $g_i$ such that $(g',g_i)$ and $(g_i,g')$ are edges in $\MinGraph(B,x)$. Thus, $(g_i,g'), (g',g_i)\in \GroundPairs(B,x)$, and in particular $g_i\runsto{x^{2\bigM}}g'$ and $g'\runsto{x^{2\bigM}}g_i$. We now have the following runs:
    \[
    \begin{split} &s'\runsto{u}t'\runsto{x^{\bigM}}g'\runsto{x^{\bigM}}r'\runsto{vw^{2\bigM k-1}}r\\ &s_i\runsto{u}t_i\runsto{x^{\bigM}}g_i\runsto{x^{\bigM}}r_i\runsto{vw^{2\bigM k-1}}s_i\\
    &g_i\runsto{x^{2\bigM}}g'\\
    &g'\runsto{x^{2\bigM}}g_i
    \end{split}
    \]
    We can therefore obtain the runs:
    \[
      \begin{split} &s'\runsto{u}t'\runsto{x^{\bigM}}g'\runsto{x^{2\bigM}}g_i\runsto{x^{\bigM}}r_i\runsto{vw^{2\bigM k-1}}s_i\\ &s_i\runsto{u}t_i\runsto{x^{\bigM}}g_i\runsto{x^{2\bigM}}g'\runsto{x^{\bigM}}r'\runsto{vw^{2\bigM k-1}}r\\
    \end{split}
    \]
    That is, $s'\runsto{w^{2\bigM}}s_i\runsto{w^{2\bigM}}r$. Since $s_i\in \MinRefStates(S',w^\bigM)$, then by \cref{prop:grounding states reachable with any bigM k} we have $(s',r)\in \GroundPairs(S',w)$, in contradiction to the non-degeneracy assumption.

    We conclude that $g'$ is not in any of the SCCs of the $g_i$, and therefore $\MinGraph(B,x)$ has strictly more SCCs than $\MinGraph(S',w)$.
\end{proof}

Since $\MinGraph(S',w)$ can have at most $|S|$ SCCs, an immediate corollary of \cref{lem:nondegenerate mingraph more SCC} is that a cactus chain of depth more than $|S|$ must have a degenerate cycle.
\begin{corollary}
    \label{cor:deep cactus chain has degenerate cycle}
    Consider a cactus chain $\alpha_{S'_1,w_1},\ldots,\alpha_{S'_n,w_n}$. If $n>|S|$ then there exists $1\le i\le n$ such that $\alpha_{S'_i,w_i}$ is a degenerate stable cycle.
\end{corollary}
Finally, by \cref{lem:degenerate cycles high cost}, if a cactus chain $\alpha_{S'_1,w_1},\ldots,\alpha_{S'_n,w_n}$ contains a degenerate stable cycle $\alpha_{S'_i,w_i}$, then we can replace this degenerate stable cycle by $w^{2\bigM |S| m}$ for some $m$ while maintaining equivalent transitions, but lowering the cost. Thus, minimal-cost cactus chains do not contain degenerate cycles. In particular, by \cref{cor:deep cactus chain has degenerate cycle}, their depth is bounded by $|S|$. We conclude this section with this important result.
\begin{corollary}[\keyicon Minimal-Cost Cactus Chains are Not Deep]
    \label{cor:minimal cost cactus chain has bounded depth}
    Consider a cactus chain $\alpha_{S'_1,w_1},\ldots,\alpha_{S'_n,w_n}$ of minimal cost, i.e., \[\cost(\alpha_{S'_1,w_1})=\min\{\cost(\gamma)\in \Gamma_\infty^0\mid \forall s,r\in S, \minweight(\alpha_{S'_1,w_1},s\to r)=\minweight(\gamma,s\to r)\}\]
    Then $n\le |S|$. In particular, $\depth(\alpha_{S'_1,w_1})\le |S|$.
\end{corollary}

\subsection{Baseline Shift}
\label{sec: baseline shift}
In \cref{sec:shift POV} we establish the correspondence between the runs of $\augA$ and of $\cA$. In particular, this shows that regardless of the choice of baseline run, the run structure of $\augA$ (i.e., the differences between runs) is fixed.

In this section, we show a similar result in the presence of cactus and rebase letters. Since these letters are not present in $\cA$, we show this directly on $\augA_\infty^\infty$. More precisely, we show that for every word $w\in \augTrans_\infty^\infty$ (without jump letters), and for every run $\rho_0$ on $w$, we can construct a new word $w'$ with the same ``inner letters'' as $w$, and establish a bijection between runs on $w$ and runs on $w'$, such that differences are maintained and such that the run corresponding to $\rho_0$ is the baseline run (and therefore has weight $0$ throughout).

The intuition behind the construction is simple: given the run $\rho_0$, we only need to change the letters of $w$ so that their  baseline component follows the state component dictated by $\rho_0$. Then, $\rho_0$ (or rather, the corresponding run that uses the same state components) becomes a baseline run, and all other runs shift their weights accordingly. In order to change the letters in this manner, cactus letters may become rebase letters (and some rebase letters become cactus letters).

The formal construction mainly tests the limits of the English and Greek alphabets, as well as the reader's patience, but is otherwise simple, as follows.

Consider a word $w=\tau_1\cdots\tau_k$ where for every $1\le i\le k$ we have that 
\[\tau_i\in \{(p_{i-1},\sigma_i,c_i,p_i),\alpha_{S'_i,w_i},\beta_{S'_i,w_i,s_i\to r_i}\}\]
i.e., each letter is either in $\augTrans$, a cactus, or a rebase.
Let $\rho_0=t_1,\ldots,t_k$ be a run of $\augA_{\infty}^\infty$ on $w$, and denote 
$t_i=((q_{i-1},p_{i-1},T_{i-1}),\tau_i,d_i,(q_i,p_i,T_i))$ (where the $d_i$ are the weights).

We define the word $w'=\gamma_1\cdots \gamma_k$, where we split the definition of $\gamma_i$ to cases according to whether $\tau_i$ is a standard letter or a cactus/rebase.
\begin{itemize}
    \item If $\tau_i=(p_{i-1},\sigma_i,c_i,p_i)$, then $d_i=c'_i-c_i$ where $c'_i$ is the weight of the transition $(q_{i-1},\sigma_i,c'_i,q_i)\in \augTrans$ (by \cref{sec:augmented construction}).
    In this case, we define $\gamma_i=(q_{i-1},\sigma_i,c'_i,q_i)$.
    \item If $\tau_i=\alpha_{S'_i,w_i}$ or $\tau_i=\beta_{S'_i,w_i,r_i\to s_i}$, denote $S'_i=\{(q,p,T)\mid q\in T\}$ for some $p,T$. We define $\gamma_i=\beta_{S'_i,w'_i,(q_{i-1},p,T)\to (q_i,p,T)}$. In case $q_{i-1}=q_i=p$, then we identify $\gamma_i=\alpha_{S'_i,w_i}$ as per \cref{rmk:cactus as rebase}.

    In order to show that $\gamma_i$ is a well-defined rebase letter, we need to show that $((q_{i-1},p,T),(q_i,p,T))\in \GroundPairs(S'_i,w'_i)$. 
    Again by \cref{rmk:cactus as rebase} we can uniformly consider $\tau_i$ as $\tau_i=\beta_{S'_i,w_i,r_i\to s_i}$ (identifying $\alpha_{S'_i,w_i}$ with a rebase letter on the baseline state). Denote $r_i=(q',p,T)$ and $s_i=(q'',p,T)$. 
    
    Since $t_i=((q_{i-1},p_{i-1},T_{i-1}),\tau_i,d_i,(q_i,p_i,T_i))$ then by \cref{def:rebase} we have $p_{i-1}=q'$ and $p_i=q''$, as well as $T_{i-1}=T_i=T$. So we can write 
    $t_i=((q_{i-1},q',T),\tau_i,d_i,(q_i,q'',T))$. Still by \cref{def:rebase} (and since $d_i<\infty$) we have $((q_{i-1},p,T),\alpha_{S'_i,w_i},e_i,(q_i,p,T))$ for some $e_i$, and in particular $((q_{i-1},p,T),(q_i,p,T))\in \GroundPairs(S'_i,w'_i)$, otherwise there would not be a transition as per \cref{def:stabilization}.
\end{itemize}
Having defined $w'$, it now remains to show a bijection between runs on $w$ and on $w'$.
\begin{proposition}
    \label{prop:baseline shift run bijection}
    In the notations above, for every sequence of states $f_1,\ldots,f_k\in Q$ (states of $\cA$), the following are equivalent.
    \begin{itemize}
        \item $\eta=x_1,\ldots,x_k$ with $x_i=((f_{i-1},p_{i-1},T_{i-1}),\tau_i,a_i,(f_i,p_i,T_i))$ is a run of $\augA_\infty^\infty$ on $w$.
        \item $\mu=y_1,\ldots,y_k$ with $y_i=((f_{i-1},q_{i-1},T_{i-1}),\gamma_i,a_i-d_i,(f_i,q_i,T_i))$ is a run of $\augA_\infty^\infty$ on $w'$ (recall that $d_i$ are the weights in $\rho_0$).
    \end{itemize}
\end{proposition}
\begin{proof}
    Consider a sequence of states $f_1,\ldots,f_k\in Q$, and $\eta$ and $\mu$ as in the statement. We show that the transitions $x_i$ have finite weight if and only if the transition $y_i$ have finite weight, thus concluding the claim.
    For every $1\le i\le k$ the following hold:
    \begin{itemize}
        \item If $\tau_i=(p_{i-1},\sigma_i,c_i,p_i)$, then the transition $x_i$ is in $\augTrans_\infty^\infty$ if and only if 
        $a_i=c''_i-c_i$ where $c''_i$ is the weight of the transition $(f_{i-1},\sigma_i,c''_i,f_i)\in \augTrans$.  Since $\gamma_i=(q_{i-1},\sigma_i,c'_i,q_i)$ (and in particular has the same letter $\sigma_i$), this happens if and only if  
        \[((f_{i-1},q_{i-1},T_{i-1}),\gamma_i,c''_i-c'_i,(f_i,q_i,T_i))\in \augTrans_{\infty}^\infty\] 
        Recall that by construction we have $d_i=c'_i-c_i$, and therefore 
        \[c''_i-c'_i=c''_i-c_i+c_i-c'_i=c''_i-c_i-(c'_i-c_i)=a_i-d_i\]
        and so the transition above is $y_i$.
        
        \item If $\tau_i=\beta_{S'_i,w_i,r_i\to s_i}$ (we again use \cref{rmk:cactus as rebase} to consider cactus and rebase jointly), denote $S'_i=\{(q,p,T)\mid q\in T\}$ and $r_i=(q',p,T)$ and $s_i=(q'',p,T)$. Since $(r_i,s_i)\in \GroundPairs(S'_i,w_i)$, we have $(r_i,\alpha_{S'_i,w_i},b'_i,s_i)\in \Delta_\infty^\infty$ for some $b'_i<\infty$. 
        As in the construction of $w'$, by \cref{def:rebase} 
        we have that $x_i\in \augTrans_\infty^\infty$ if and only if 
        $p_{i-1}=q'$, $p_i=q''$ and $T_{i-1}=T_i=T$, and $a_i=b_i-b'_i$ where $b_i$ is such that $((f_{i-1},p,T),\alpha_{S'_i,w_i},b_i,(f_i,p,T))\in \augTrans_\infty^\infty$ (and in particular these states are a grounded pair).

        As we show in the construction of $w'$, we have $((q_{i-1},p,T),\alpha_{S'_i,w_i},e_i,(q_i,p,T))$ for some $e_i$.
        By \cref{def:rebase}, the above holds if and only if 
        \[((f_{i-1},q_{i-1},T_{i-1}),\beta_{S'_i,w'_i,(q_{i-1},p,T)\to (q_i,p,T)},b_i-e_i,(f_i,q_i,T_i))\in \augTrans_\infty^\infty\]
        Note that $d_i=e_i-b'_i$, and therefore $a_i-d_i=b_i-b'_i-(e_i-b'_i)=b_i-e_i$.
        Therefore, this transition is exactly $y_i$.
    \end{itemize}
\end{proof}
We wrap up this section by introducing the baseline-shift operator, and providing a concise formalism for using it.
For a word $w$ and $\rho_0$ as above, we denote the constructed word $w'$ by 
$\baseshift{w}{\rho_0}$.
We overload this notation to runs: for a run $\eta=x_1,\ldots,x_k$ with $x_i=((f_{i-1},p_{i-1},T_{i-1}),\gamma_i,a_i,(f_i,p_i,T_i))$ we denote by $\baseshift{\eta}{\rho_0}$ the run $\mu=y_1,\ldots,y_k$ with $y_i=((f_{i-1},q_{i-1},T_{i-1}),\gamma_i,a_i-d_i,(f_i,q_i,T_i))$. Note that a run already encapsulates the information about the word, and therefore we do not need to specify $w$ in the latter notation.

\cref{prop:baseline shift run bijection} readily implies our two main results of this section, namely that the baseline shift of $\rho_0$ by itself is a seamless run, and that baseline shifts maintain the gaps between runs.
\begin{corollary}
    \label{cor:baseline shift to seamless run}
    Consider a word $w$ and a run $\rho_0$ on $w$, then $\baseshift{\rho_0}{\rho_0}$ is a seamless run (and in particular gains $0$ weight in every transition).
\end{corollary}
\begin{corollary}
    \label{cor:baseline shift maintains gaps}
    Consider a word $w$ and a run $\rho_0$ on $w$. Let $\rho_1,\rho_2$ be runs on $w$, and consider the shifted runs $\mu_1=\baseshift{\rho_1}{\rho_0}$ and $\mu_2=\baseshift{\rho_2}{\rho_0}$ then for every $0\le i\le |w|$ it holds that $\weight(\rho_1[1,i])-\weight(\rho_2[1,i])=\weight(\mu_1[1,i])-\weight(\mu_2[1,i])$.
\end{corollary}

\subsection{Cactus Flattening}
\label{sec: cactus unfolding}
Intuitively, cactus (and rebase) letters encapsulate many repetitions of the same word, where rebase also allows a change in the baseline. Naturally, we wish to consider the corresponding words that are obtained by repeatedly ``unfolding'' cactus and rebase letters, in a process we refer to as \emph{flattening}. 
We make extensive use of flattening throughout the paper.

Recall that for a word $u$ the value $\maxeff{u}$ is an upper-bound on the absolute value of the change in weight that a run can incur upon reading $u$.

\begin{definition}[\keyicon Cactus-Unfolding]
    \label{def:unfolding function}
    Consider a cactus letter $\alpha_{S',w}$, a word $u=x \cdot \alpha_{S',w} \cdot y$ and $F\in \bbN$.

    By \cref{lem:pumping grounded pairs} there exists $M_0\in \bbN$ such that for every $k\ge M_0$ and every $s,r\in S'$ with $(s,r)\notin\GroundPairs(S',w)$ it holds that $\minweight(w^{k\cdot 2\bigM},s\to r)>F$.
    
    The \emph{Unfolding of $u$ with constant $F$} is then $\unfold(x,\alpha_{S',w},y \wr F)=xw^{2\bigM M_0}y$, and any word $xw^{2\bigM m}y$ with $m\ge M_0$ is \emph{an unfolding of $u$ with constant $F$.}
\end{definition}
\begin{remark}[``The'' v.s. ``an'' unfolding]
\label{rmk:increasing repetitions in unfolding}
We remark that the distinction between ``the'' unfolding and ``an'' unfolding merely means that we can add repetitions of $w$ (in multiples of $2\bigM$), and still be referred to as an unfolding.

In the following, we mostly refer to $M_0$ as the number of repetitions, but we sometimes need to increase this number, in which case we remark that all the results still hold.
\end{remark}


Intuitively, while $\alpha_{S',w}$ makes transitions $s\runsto{\alpha_{S',w}}r$ attain value $\infty$ if $(s,r)\notin \GroundPairs(S',w)$, the unfolding repeats $w$ enough times such that these runs may exist, but attain a value as high as we like (namely $F$). 
We always choose $F$ large enough such that any continuation of such runs is irrelevant for the value of the word. Therefore, the values of runs after unfolding are the same as those before unfolding. We capture this as follows.

\begin{proposition}
    \label{prop:unfolding cactus maintains seamlesss gaps}
    Consider a cactus letter $\alpha_{S',w}$, a word $x \cdot \alpha_{S',w} \cdot y$ and $F> 2\maxeff{x \cdot \alpha_{S',w} \cdot y}$.
    
    Let $\rho$ be a seamless run on $x \cdot \alpha_{S',w} \cdot y$, then there exists a \emph{seamless} run $\mu$ on $\unfold(x,\alpha_{S',w},y \wr F)=xw^{2\bigM M_0}y$ such that 
    $\mu[1,|x|]=\rho[1,|x|]$, and for every $i\ge 0$, if $z=x\cdot\alpha_{S',w}\cdot y[1,i]$ is a prefix of $x \cdot \alpha_{S',w} \cdot y$ and the corresponding prefix of $\unfold(x,\alpha_{S',w},y \wr F)$ is $z'=x\cdot w^{2\bigM M_0}\cdot  y[1,i]$, then it holds that 
    $\weight(\rho(z))=\weight(\mu(z'))$.
\end{proposition}
\begin{proof}
    Let $u=x \cdot \alpha_{S',w} \cdot y$ and $v=\unfold(x,\alpha_{S',w},y \wr F)=xw^{2\bigM M_0}y$, and let $\rho$ be a seamless run on $u$.
    Denote $\rho:s_0\runsto{x}s\runsto{\alpha_{S',w}}r\runsto{y}q$

    Since $s\runsto{\alpha_{S',w}}r$ then by the definition of transitions on cactus letters (\cref{def:stabilization}) we have that $(s,r)\in \GroundPairs(S',w)$. Let $g$ be the corresponding grounding state, i.e., $g\in \MinRefStates(S',w^\bigM)$
    and $s\runsto{w^\bigM}g\runsto{w^\bigM}r$. 
    In particular, $\minweight(g\runsto{w^\bigM}g)=0$. Take $\eta:s\runsto{w^\bigM}g$ and $\xi:g\runsto{w^\bigM}r$ to be runs of minimal weights (from/to their respective states), and $\zeta: g\runsto{w^\bigM}g$ with weight $0$.
    We construct $\mu$ as the following concatenation of runs:
    \[
    \mu:\rho(x)\eta\zeta^{(2M_0-2)}\xi\rho(y)
    \]
    and therefore
    \[
    \mu:s_0\runsto{x}s\runsto{w^\bigM}g\runsto{w^{(2M_0-2)\bigM}}g\runsto{w^\bigM}r\runsto{y}q
    \]
    First, note that $\mu$ and $\rho$ are identical on their $x$ prefix. That is,  
    $\mu[1,|x|]=\rho[1,|x|]$, as required.

    Next, let $i\ge 0$ and denote $z=x\cdot\alpha_{S',w}\cdot y[1,i]$ and $z'=x\cdot w^{2\bigM M_0}\cdot  y[1,i]$.
    Recall (from \cref{def:stabilization}) that $\weight(s\runsto{\alpha_{S',w}}r)=\minweight(s\runsto{w^\bigM}g\runsto{w^\bigM} r)$. Since $\zeta: g\runsto{w^\bigM}g$ has weight $0$, and $\xi$ and $\eta$ are of minimal weight, it follows that $\weight(s\runsto{\alpha_{S',w}}r)=\weight(\eta\zeta^{(2M_0-2)}\xi)$. Thus, upon reaching $r$, both $\rho$ and $\mu$ have the same weight.

    Since the suffix read from $r$ in both $z$ and $z'$ is $y[1,i]$ and the runs are identical at this suffix (namely $\rho(y)$), it follows that     
    $\weight(\rho(z))=\weight(\mu(z'))$.

    Finally, it is left to prove that $\mu$ is seamless (and recall that $\rho$ is seamless). 
    Assume by way of contradiction that $\mu$ is not seamless, and let $\mu'$ be a run on some prefix $u$ of $\unfold(x,\alpha_{S',w},y \wr F)$ such that both $\mu(u):s_0\runsto{u} s_1$ and $\mu'(u):s_0\runsto{u}s_1$ but 
    $\weight(\mu'(u))<\weight(\mu(u))$.
    That is, $\mu'$ arrives at $s_1$ with smaller weight than $\mu$, thus ``breaking'' $\mu$.
    We can assume without loss of generality that $\mu$ and $\mu'$ are identical from $s_1$ (i.e., after reading $u$). Indeed, by concatenating the suffix of $\mu$ to the prefix $\mu'(u)$, the obtained run remains strictly below $\mu$ from $s_1$ and on, and also ``breaks'' $\mu$ at $s_1$. 
    We can further assume that $s_1$ is the last state in both runs, i.e., $q$. In particular $\weight(\mu')<\weight(\mu)$.

    Denote by $p_1,p_2$ the states visited by $\mu'$ after reading $x$  and $xw^{2M_0\bigM}$ respectively (we implicitly think of $p_1$ and $p_2$ as states at specific indices, to avoid cumbersome notations). See~\cref{fig:unfolding}.
    
    
      \begin{figure}[ht]
        \centering
        \begin{subfigure}{0.9\textwidth}
        \includegraphics[width=0.9\linewidth]{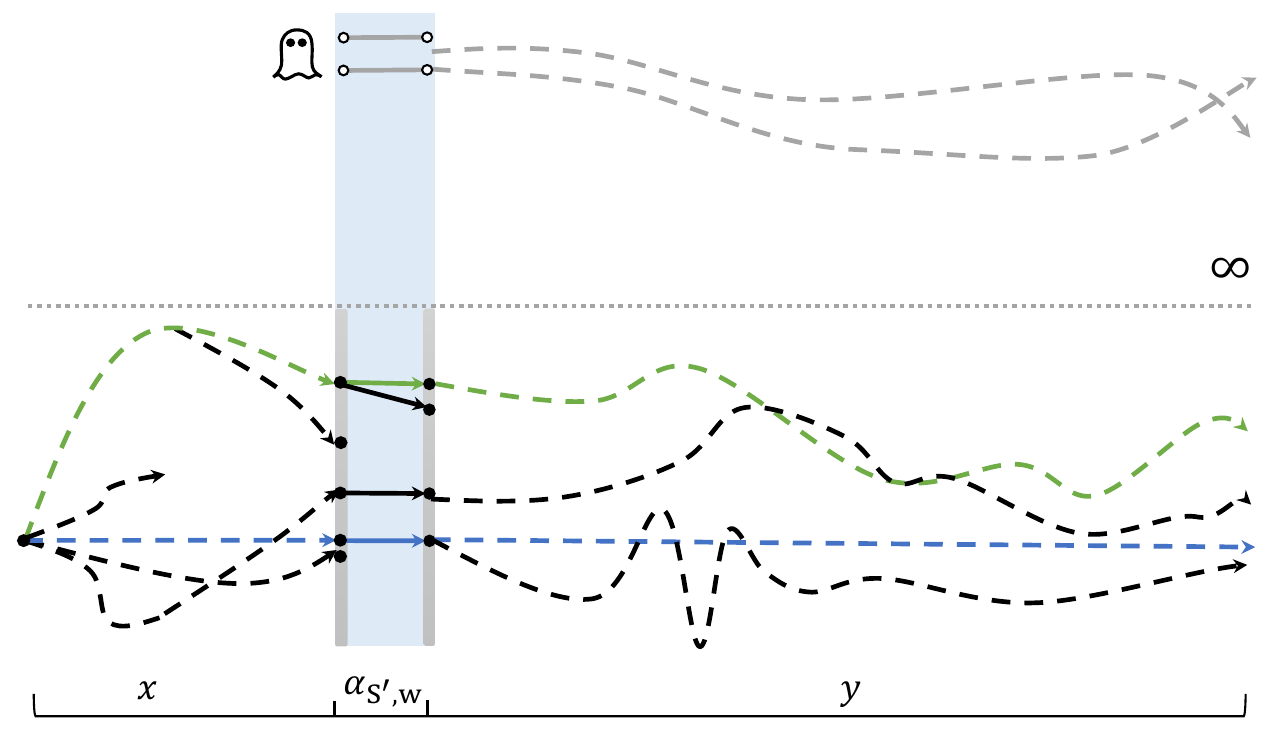}
        \caption{The cactus letter $\alpha_{S',w}$ induces ghost states and runs. The baseline run is in blue, and the green run is seamless.}
        \label{fig:unfolding before}
        \end{subfigure}

        \begin{subfigure}{0.9\textwidth}
        \includegraphics[width=0.9\linewidth]{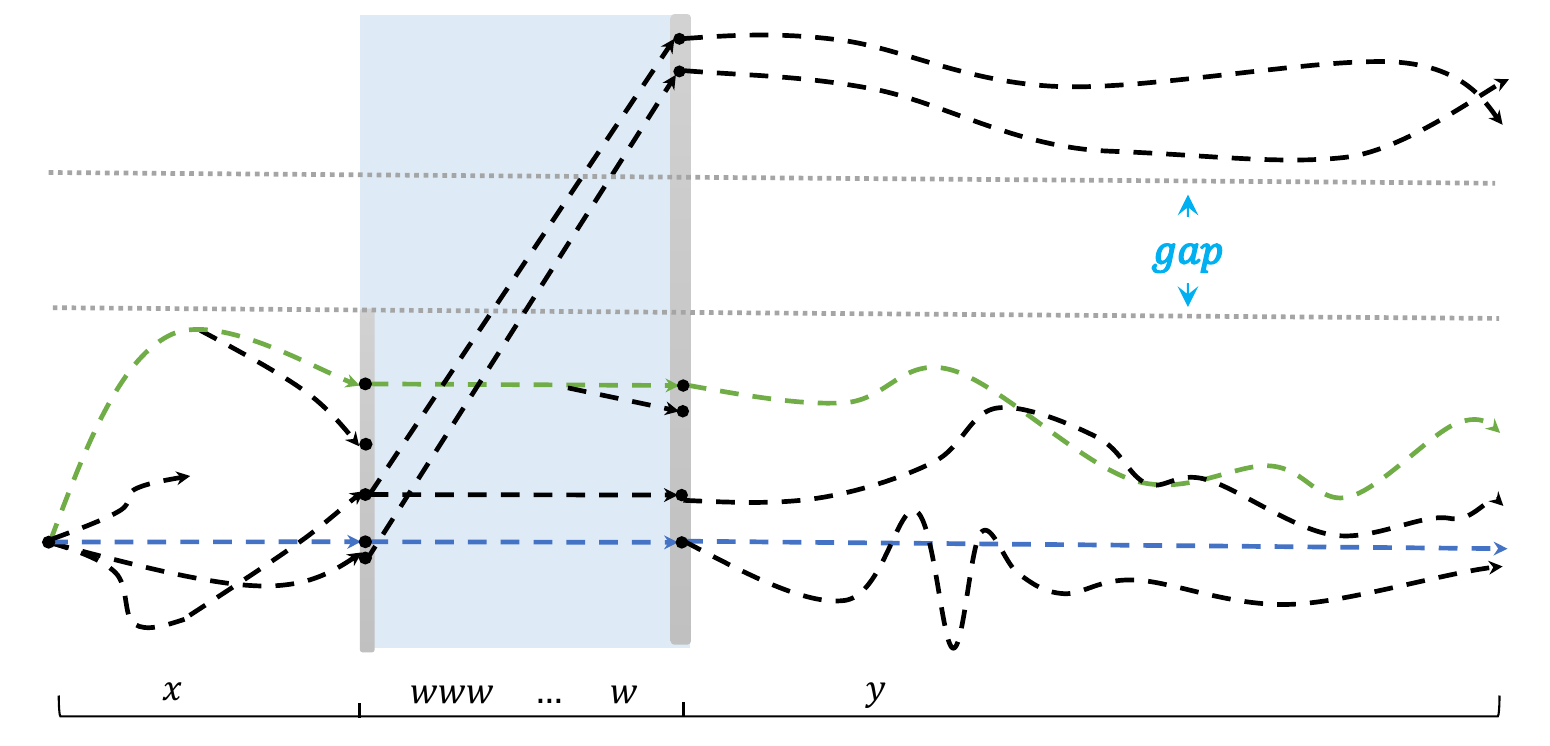}
        \caption{After unfolding, $\alpha_{S',w}$ is replaced by repetitions of $w$. The ghost runs now become concrete, and are far above the original runs. The updated green run is still seamless.}
        \label{fig:unfolding after}
        \end{subfigure}
        \caption{Cactus letter unfolding, before (\cref{fig:unfolding before}) and after (\cref{fig:unfolding after}).}
        \label{fig:unfolding}
    \end{figure}
    
    There are now two possible cases: either $(p_1,p_2)\in \GroundPairs(S',w)$, or not. Note that $p_1,p_2\in S'$ (since $p_1\in \booltrans(s_0,x)\subseteq S'$ and $p_2\in \booltrans(S',w)\subseteq S'$ by \cref{def:stable cycle}).

    If $(p_1,p_2)\in \GroundPairs(S',w)$, then $p_1\runsto{\alpha_{S',w}}p_2$. Moreover, the weight of this transition on $\alpha_{S',w}$ is $\minweight(p_1\runsto{w^{2M_0\bigM}}p_2)$ (by \cref{def:stabilization,lem:pumping grounded pairs}).
    Thus, the weight gained by $\mu'$  over the infix is at least $\minweight(p_1\runsto{w^{2M_0\bigM}}p_2)$.
    We can now consider the run $\rho':s_0\runsto{x}p_1\runsto{\alpha_{S',w}}p_2\runsto{y}q$ obtained by replacing the infix of $\mu'$ on $w^{2M_0\bigM}$ with the transition on $\alpha_{S',w}$. It is immediate that 
    $\weight(\rho')=\weight(\mu')<\weight(\mu)=\weight(\rho)$, but this is a contradiction to the fact that $\rho$ is seamless.

    The most involved case is when $(p_1,p_2)\notin \GroundPairs(S',w)$, where we heavily rely on the choice of $M_0$ (the above does not actually use $M_0$, if examined carefully).

    Note that in this case, there is no transition on $\alpha_{S',w}$ from $p_1$ to $p_2$, and we cannot readily obtain a run $\rho'$ as above.
    However, by \cref{lem:pumping grounded pairs} and the choice of $M_0$ in \cref{def:unfolding function}  we have that $\minweight(w^{2M_0\bigM},p_1\runsto{}p_2)>F$.
    In particular, the weight gained by $\mu'$ in the infix between $p_1$ and $p_2$ is more than $2\maxeff{x\alpha_{S',w}y}$, and therefore we have
    \[
    \weight(\mu')>2\maxeff{x\alpha_{S',w}y}-\maxeff{x}-\maxeff{y}\ge \maxeff{x\alpha_{S',w}y}
    \]
    However, $\weight(\mu)=\weight(\rho)\le \maxeff{x\alpha_{S',w}y}$ and therefore $\weight(\mu')>\weight(\mu)$, in contradiction to our assumption.

    Thus, we conclude that $\mu$ is seamless.
\end{proof}

Naively, we would want a converse of \cref{prop:unfolding cactus maintains seamlesss gaps}, i.e., that a seamless run on an unfolded word induces an equivalent seamless run on the folded word. Unfortunately, this is generally not true. Indeed there may be seamless runs that do not use grounded pairs, and therefore attain a huge weight on the unfolded word, but are simply not reachable in the folded word. Therefore, the converse needs to be stated carefully, as follows.

\begin{proposition}
    \label{prop:un-unfolding maintains cheap seamless runs}
    Consider a cactus letter $\alpha_{S',w}$, a word $x\cdot \alpha_{S',w}\cdot y$ and $F>2\maxeff{x\cdot \alpha_{S',w}\cdot y}$. Let $\mu$ be a seamless run on $\unfold(x,\alpha_{S',w},y \wr F)=xw^{2\bigM M_0}y$, then the following hold.
    \begin{enumerate}
        \item If $\weight(\mu)\le \maxeff{x\alpha_{S',w}y}$ then there exists a seamless run $\rho$ on $x\alpha_{S',w}y$ such that 
    $\mu[1,|x|]=\rho[1,|x|]$, and for every $i\ge 0$, if $z=x\cdot\alpha_{S',w}\cdot y[1,i]$ is a prefix of $x \cdot \alpha_{S',w} \cdot y$ and the corresponding prefix of $\unfold(x,\alpha_{S',w},y \wr F)$ is $z'=x\cdot w^{2\bigM M_0}\cdot  y[1,i]$, then it holds that 
    $\weight(\rho(z))=\weight(\mu(z'))$.
    \item If $\weight(\mu)> \maxeff{x\alpha_{S',w}y}$, then $\weight(\mu)>F-\maxeff{x\cdot \alpha_{S',w}\cdot y}$.
    \end{enumerate}
    
\end{proposition}
\begin{proof}
    The proof can be viewed as a ``reverse'' proof of \cref{prop:unfolding cactus maintains seamlesss gaps}.

    Consider a seamless run $\mu$ on $\unfold(x,\alpha_{S',w},y \wr F)=xw^{2\bigM M_0}y$. 
    Write $\mu:s_0\runsto{x}s\runsto{w^{2\bigM M_0}}r\runsto{y}q$. By \cref{lem:pumping grounded pairs}, if such that $\weight(\mu)\le \maxeff{x\alpha_{S',w}y}$ then it must hold that $(s,r)\in \GroundPairs(S',w)$. Indeed, otherwise the weight gained by $\mu$ in the infix $s\runsto{2^{2\bigM M_0}}r$ would be more than $F>2\maxeff{x\alpha_{S',w}y}$ (by \cref{def:cactus extension}), which would in turn imply that $\weight(\mu)>\maxeff{x\alpha_{S',w}y}$, which is a contradiction.

    Since $(s,r)\in \GroundPairs(S',w)$ and $\mu$ is seamless, we can assume without loss of generality that 
    \[\mu:s_0\runsto{x}s\runsto{w^\bigM}g\runsto{w^{(2M_0-2)\bigM}}g\runsto{w^{\bigM}}r\runsto{y}q\]
    where $g$ is the grounding state of $(s,r)$, and the infix $s\runsto{w^\bigM}g\runsto{w^{(2M_0-2)\bigM}}g\runsto{w^{\bigM}}r$ of $\mu$ gains weight $\minweight(s\runsto{w^{2\bigM}}r)$ (as per \cref{lem:pumping grounded pairs}).
    We can now readily construct $\rho$ by replacing this infix with a transition on $\alpha_{S',w}$ with the same weight.
    
    The fact that $\rho$ is seamless now comes ``for free'' from \cref{prop:unfolding cactus maintains seamlesss gaps}. Indeed, if $\rho$ is not seamless, we can take a lower-weight seamless run, and construct from it a seamless run $\mu'$ that would ``break'' $\mu$, contradicting the assumption that $\mu$ is seamless.

    For the second case, if $\weight(\mu)> \maxeff{x\alpha_{S',w}y}$ then again by \cref{lem:pumping grounded pairs} it now follows that $(s,r)\notin \GroundPairs(S',w)$ (by the same analysis as above). By \cref{def:cactus extension} we then have $\minweight(w^{2\bigM M_0},s\to r)>F$, so $\weight(\mu)> F-\maxeff{x}-\maxeff{y}\ge F-\maxeff{x\alpha_{S',w}y}$.
\end{proof}
The combination of \cref{prop:unfolding cactus maintains seamlesss gaps,prop:un-unfolding maintains cheap seamless runs} gives us the following useful tool.
\begin{lemma}[\keyicon The Effect of Unfolding]
    \label{lem:unfolding configuration characterization}
    Consider a cactus letter $\alpha_{S',w}$, a word $x\cdot \alpha_{S',w}\cdot y$ and $F>2\maxeff{x\cdot \alpha_{S',w}\cdot y}$, and let $\unfold(x,\alpha_{S',w},y \wr F)=xw^{2\bigM M_0}y$.
    Let $y'$ be a prefix of $y$, and consider the configurations $\vec{c_1}=\xconf(\vec{c_\init},x\cdot \alpha_{S',w}\cdot y')$ and $\vec{c_2}=\xconf(\vec{c_\init},xw^{2\bigM M_0}y')$, then $\supp(\vec{c_1})\subseteq \supp(\vec{c_2})$ and for every state $q\in S$ we have the following.
    \begin{enumerate}
        \item If $q\in \supp(\vec{c_1})$ then $\vec{c_2}(q)=\vec{c_1}(q)$.
        \item If $q\notin \supp(\vec{c_1})$ then $\vec{c_2}(q)>F-\maxeff{x\alpha_{S',w}y}$ (and could be $\infty$).
    \end{enumerate}
\end{lemma}
\begin{proof}
    By \cref{prop:unfolding cactus maintains seamlesss gaps} we immediately get that $\supp(\vec{c_1})\subseteq \supp(\vec{c_2})$ and moreover, for every $q\in \supp(\vec{c_1})$ we have $\vec{c_2}(q)=\vec{c_1}(q)$, since a seamless run to $q$ on $x\cdot \alpha_{S',w}\cdot y'$ induces a seamless run of the same weight on $xw^{2\bigM M_0}y'$.

    Next, consider $q\notin \supp(\vec{c_1})$. If $q\notin \supp(\vec{c_2})$ then we are done. Otherwise, it must hold that $\vec{c_2}(q)> \maxeff{x\alpha_{S',w}y}$. Indeed, if (by way of contradiction) $\vec{c_2}(q)\le \maxeff{x\alpha_{S',w}y}$ then by \cref{prop:un-unfolding maintains cheap seamless runs} we have $\vec{c_1}(q)=\vec{c_2}(q)$ and in particular $q\in \supp(\vec{c_1})$, which is a contradiction.

    It then follows (still by \cref{prop:un-unfolding maintains cheap seamless runs}) that $\vec{c_2}(q)>F-\maxeff{x\alpha_{S',x},y}$, and we are done.
\end{proof}

Our next task is to eliminate rebase letters. Intuitively, a rebase letter is similar to a cactus letter that allows a change of baseline runs (with restrictions on what changes are allowed). 
Thus, we replace a rebase letter with the corresponding cactus letter, surrounded by jumps to change the baseline.

\begin{definition}[Rebase Removal]
\label{def:rebase removal}
Consider a rebase letter $\beta_{S',w,s\to r}$ and a word $x\cdot \beta_{S',w,s\to r}\cdot y$ where $s=(q_1,p,T)$ and $r=(q_2,p,T)$. 
Let $t_1=(\cdot,q_1,T)$, $t_p=(\cdot,p,T)$ and $t_2=(\cdot,q_2,T)$ where $\cdot$ stands for some arbitrary state (c.f., \cref{rmk:jump letters first component})

The \emph{Rebase Removal} of $x\cdot \beta_{S',w,s\to r}\cdot y$ is $\rebaserm(x, \beta_{S',w,s\to r}, y)=x\cdot \jl_{t_1\to t_p}\alpha_{S',w}\jl_{t_p\to t_2}y$.
\end{definition}

Rebase removal is depicted in~\cref{fig:rebase removal}.
We claim that applying rebase removal doesn't change the structure of the run tree, in the sense that the gaps between all runs remain the same. The following proposition captures this. 

\begin{figure}[ht]
    \centering
    \includegraphics[width=1\linewidth]{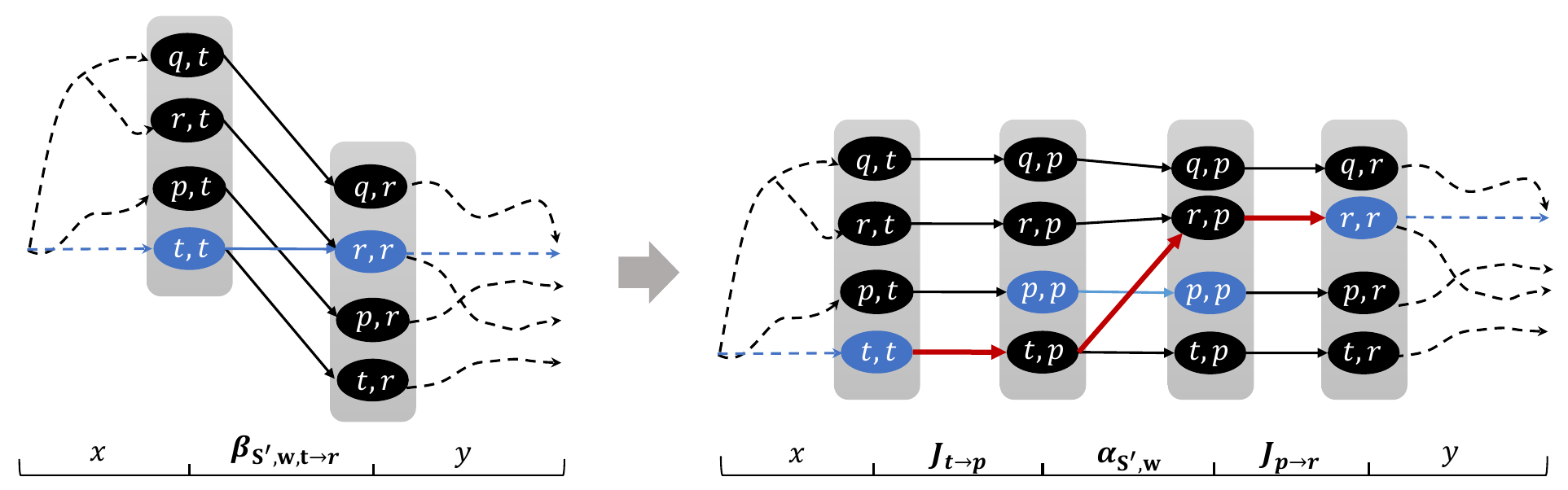}
    \caption{The rebase letter $\beta_{S',w,t\to r}$ is replaced by the sequence $\jl_{t\to p}\alpha_{S',w}\jl_{p\to r}$. The baseline changes}
    \label{fig:rebase removal}
\end{figure}

\begin{proposition}
    \label{prop:rebase removal shifts runs}
    Consider a rebase letter $\beta_{S',w,s\to r}$ with $s=(q_1,p,T)$ and $r=(q_2,p,T)$. Denote $t_1=(\cdot,q_1,T)$, $t_p=(\cdot,p,T)$ and $t_2=(\cdot,q_2,T)$ as per \cref{def:rebase removal}. Let $c\in \bbN$ be the weight of the transition $(s,\alpha_{S',w},c,r)$.

    For every $q',q''\in Q$ (states of the original WFA $\cA$), the transition
    \[
    ((q',q_1,T),\beta_{S',w,s\to r},d-c,(q'',q_2,T))
    \]
    exists if and only if the\footnote{Observe that this is indeed a single run, since each part of it is a single transition.} run
    \[
    (q',q_1,T)\runsto{\jl_{t_1\to t_p}}(q',p,T)\runsto{\alpha_{S',w}}(q'',p,T)\runsto{\jl_{t_p\to t_2}}(q'',q_2,T)
    \] 
    exists and has weight $d$.

    Moreover, there are no other runs from $(q',q_1,T)$ to $(q'',q_2,T)$ on $\jl_{t_1\to t_p}\alpha_{S',w}\jl_{t_p\to t_2}$.
\end{proposition}
\begin{proof}
    The proof is actually immediate by the respective definitions, but does require untangling each definition.

    First, observe that every transition of the form $(s',\beta_{S',w,s\to r},e,r')$ implies by \cref{def:rebase} that $s'=(q',q_1,T)$ and $r'=(q'',q_2,T)$ for some $q',q''\in Q$. Thus, these are the only types of transition that can be considered. Therefore, consider some $q',q''\in Q$.

    Recall that $c$ is the weight of $(s,\alpha_{S',w},c,r)$. Still by \cref{def:rebase}, we have that 
    \[((q',q_1,T),\beta_{S',w,s\to r},d-c,(q'',q_2,T))\] 
    if and only if $((q',p,T),\alpha_{S',w},d,(q'',p,T))$.
    By \cref{def:jump letters}, this holds if and only if we have the following run
   \[
    (q',q_1,T)\runsto{\jl_{t_1\to t_p}}(q',p,T)\runsto{\alpha_{S',w}}(q'',p,T)\runsto{\jl_{t_p\to t_2}}(q'',q_2,T)
    \] 
    Where the transitions on jump letters have weight $0$, and the transition on $\alpha_{S',w}$ has weight $d$. This concludes the claim.
\end{proof}
By applying \cref{prop:rebase removal shifts runs} to entire runs, we obtain a bijection between runs on $u$ and on $\rebaserm(u)$ which, intuitively, preserves the gaps between runs. More precisely, by carefully tracking indices and state names, we have the following.
\begin{corollary}
    \label{cor:rebase removal preserves gaps}
    Consider a word $x\cdot \beta_{S',w,s\to r}\cdot y$ with $s=(q,p,T)$ and $r=(q',p,T)$.
    There exists a constant $c\in \bbN$ such that the following are equivalent for every sequence of transitions $\rho=t_1,t_2,\ldots,t_k$ where $t_i=(s_i,\sigma_i,d_i,s_{i+1})$
 and $s_i=(q_i,p_i,T_i)$ for all $i$.
 \begin{itemize}
        \item $\rho:s_0\runsto{x}s_1\runsto{\beta_{S',w,s\to r}}s_2\runsto{y}s_3$ is a run on $x\cdot \beta_{S',w,s\to r}\cdot y$.
        \item The sequence of transition 
        \[        \rho'=t_1,t_2,\ldots,t_{|x|},t'_{1},t'_2,t'_3,t_{|x|+1},\ldots,t_k
        \] 
        with 
        \[
        \begin{split}
        &t'_1=(s_{|x|+1},\jl_{(\cdot,q,T)\to (\cdot,p,T)},0,(q_{|x|+1},p,T)) \\
        &t'_2=((q_{|x|+1},p,T),\alpha_{S',w},c,(q_{|x|+2},p,T)) \\
        &t'_3=(q_{|x|+2},\jl_{(\cdot,p,T)\to (\cdot,q',T)},0,s_{|x|+2})    
        \end{split}
        \]
        is a run on $\rebaserm(x,\beta_{S',w,s\to r},y)$, i.e.,  
        \[
        \rho':s_0\runsto{x}s_1\runsto{\jl_{(\cdot,q,T)\to (\cdot,p,T)}}s'_1\runsto{\alpha_{S',w}}s'_2\runsto{\jl_{(\cdot,p,T)\to (\cdot,q',T)}}s
        _2\runsto{y}s_3
        \]
    \end{itemize}
    And observe that in case both runs exist, the weights of the runs are identical up to $s_1$, and from $s_2$ the run $\rho'$ has weight higher by $c$ than $\rho$.
\end{corollary}

Having established how to remove cactus letters and rebase letters, we can now recursively apply these definition to \emph{flatten} any word $u\in (\Gamma_\infty^\infty)^*$ back to a word in $\Gamma^*$ (with jump letters). Before defining the flattening precisely, note that if $u\in (\Gamma_\infty^\infty)^*$ then in fact $u\in (\Gamma_k^j)^*$ for some finite $k,j$ (indeed, the infinite indices are reached by infinite unions, but cannot be attained in a single finite word).

We remark that the order with which one removes cactus letters (e.g., leftmost first, or highest-depth first) may result in different flattenings. We fix as a canonical order the highest-depth first approach (and left to right among those), but this is arbitrary. It is captures below with the requirement that the prefix $x$ is of ``lower'' rank.

\begin{definition}[Flattening]
    \label{def:cactus rebase flattening}
    Consider a word $u\in (\Gamma_\infty^\infty)^*$ and $F\in \bbN$. 
    We define the \emph{flattening of $u$ with constant $F$}, denoted $\flatten(u\wr F)$,  inductively as follows.
    \begin{itemize}
        \item \textbf{Base case:} If $u\in (\Gamma_0^0)^*=\Gamma^*$, then $\flatten(u \wr F)=u$.
        \item \textbf{Cactus Case:} If $u\in (\Gamma^j_k)^*\setminus (\Gamma^{j}_{k-1})^*$ for $j\ge 0$ and $k\ge 1$, write $u=x\alpha_{S',w}y$ such that  $x\in(\Gamma^{j}_{k-1})^* $ and $\alpha_{S',w}\in \Gamma^j_k\setminus \Gamma^{j}_{k-1}$, then we define 
        \[\flatten(u\wr F)=\flatten(\unfold(x,\alpha_{S',w},y \wr F) \wr F+2\maxeff{u})\]
        \item \textbf{Rebase case:} If $u\in (\Gamma^j_0)^*\setminus (\Gamma^{j-1}_\infty)^*$ for $j\ge 1$, write $u=x\beta_{S',w,s\to r}y$ such that $x\in (\Gamma^{j-1}_\infty)^*$ and  $\beta_{S',w,s\to r}\in \Gamma^j_0\setminus \Gamma^{j-1}_\infty$, then we define $\flatten(u \wr F)=\flatten(\rebaserm(x,\beta_{S',w,s\to r},y) \wr F)$.
    \end{itemize}
\end{definition}
Observe that \cref{def:cactus rebase flattening} always terminates. Indeed, by induction after every application of the inductive cases, the number of maximal-depth cactus letters, or number of rebase letters of highest rank $j$, decreases by one. Note that each application of $\unfold$ or $\rebaserm$ may create additional cacti and rebase letters that were nested, but these are of lower ranks. 

We illustrate the flattening process of the cactus in \cref{fig:cactus} in \cref{fig:flattening example}.
\begin{figure}[ht]
    \centering
    \includegraphics[width=0.8\linewidth]{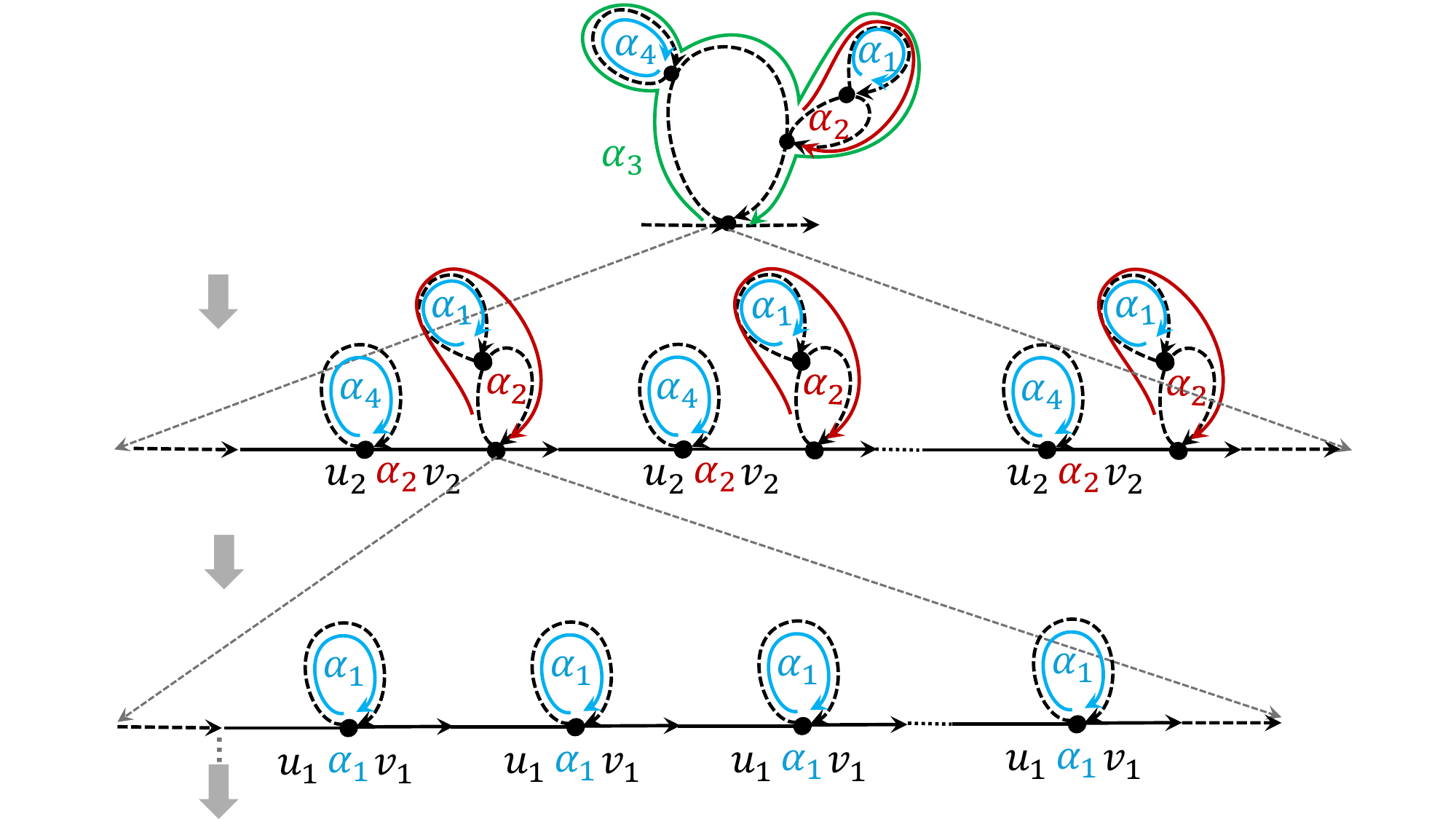}
    \caption{Illustration of flattening the cactus in \cref{fig:cactus}. The horizontal dotted lines represent many more iterations.}
    \label{fig:flattening example}
\end{figure}

Also note that in the recursive cactus case, we increase $F$ by $2\maxeff{u}$. This ensures that if there are ``inner letters'' of higher cost that are nested in the cactus, their weight is also accounted for.

The main characterization of flattening is that it maintains the weights to the reachable states, makes ghost states reachable, but gives the latter very high weight, as we now formulate.
\begin{lemma}
    \label{lem:flattening configuration characterization}
    Consider a word $u\in (\Gamma_\infty^\infty)^*$ and $F> 2\maxeff{u}$. 
    Let $\vec{c}=\xconf(\vec{c_\init},u)$ and $\vec{d}=\xconf(\vec{c_\init},\flatten(u \wr F))$.
    Then \[\booltrans(s_0,u)=\supp(\vec{c})\subseteq \supp(\vec{d})=\booltrans(s_0,\flatten(u \wr F))=\ghostTrans(s_0,\flatten(u \wr F))=\ghostTrans(s_0,u)\]
    and for every $q\in \supp(\vec{d})$, if $q\in \supp(\vec{c})$ then $\vec{c}(q)=\vec{d}(q)$, and otherwise $\vec{d}(q)\ge \max\{\vec{c}(p)\mid p\in \supp(\vec{c})\}+F$.
\end{lemma}
\begin{proof}
    The containment/equality chain is immediate by the construction of unfolding and rebase removal. The interesting part of the proof is the weight accounting.
    The claim is proved by induction on the word $u$, as per \cref{def:cactus rebase flattening}.

    \subparagraph*{Base case:}
    For the base case, if $u\in (\Gamma_0^0)^*$, then $\flatten(u \wr F)=u$ and the claim is trivial, since $\vec{c}=\vec{d}$.

    \subparagraph*{Cactus case:}
    For the cactus case, write $u=x\alpha_{S',w}y$ and denote $u'=\unfold(x,\alpha_{S',w},y \wr F+2\maxeff{u})$. 
    Define $\vec{c},\vec{d}$ as in the premise, and additionally $\vec{c'}=\xconf(\vec{c_\init},u')$.
    Let $q\in \supp(\vec{d})=\booltrans(s_0,\flatten(u \wr F))$.
    
    If $q\in \supp(\vec{c})$, then in particular there is a seamless run $\rho:s_0\runsto{u}q$, so by \cref{prop:unfolding cactus maintains seamlesss gaps} there is a seamless run with the same weight $\mu:s_0\runsto{u'}q$. Assume $\rho$ has minimal weight, then since $\weight(\rho)\le \maxeff{u}$ (which holds regardless of minimality) it follows by \cref{prop:un-unfolding maintains cheap seamless runs} that $\mu$ is also minimal weight (otherwise there would be a lower run to $q$).
    Thus, we have $\vec{c}(q)=\vec{c'}(q)$ and in particular $q\in \supp(\vec{c'})$. 
    We have that $\flatten(u \wr F)=\flatten(u' \wr F+2\maxeff{u})$, so by $q\in \supp(\vec{c'})$ we can apply the induction hypothesis to obtain $\vec{c'}(q)=\vec{d}(q)$, but then $\vec{c}(q)=\vec{d}(q)$, so this case is proved. Intuitively, all we show here is that if $q$ is already reachable, then this property is maintained through the inductive step (which is fairly obvious by the definition of unfolding).

    Now assume $q\notin \supp(\vec{c})$, we split to two cases. If $q\in \supp(\vec{c'})$ then $q\in\ghostTrans(s_0,u)\setminus \booltrans(s_0,u)$ but $q\in \booltrans(s_0,u')$. 
    In particular, $\vec{c'}(q)\le \maxeff{u'}$.
    Also, by the induction hypothesis, we have $\vec{c'}(q)=\vec{d}(q)$.
    By \cref{prop:un-unfolding maintains cheap seamless runs}, there is no seamless run $\mu:s_0\runsto{u'}q$ with $\weight(\mu)\le \maxeff{u}$ (otherwise we would have $q\in \booltrans(s_0,u)$), and therefore (by the second item of \cref{prop:un-unfolding maintains cheap seamless runs})  we have 
    \[
    \begin{split}
    \vec{d}(q)=&\vec{c'}(q)=\minweight(u',s_0\to q)>F+2\maxeff{u}-\maxeff{u}=\\
    &F+\maxeff{u}\ge F+\max\{\vec{c}(p)\mid p\in \supp(\vec{c})\}    
    \end{split}
    \]
    as required.
    Intuitively, this case amounts to showing that at the first level of the unfolding where $q$ becomes reachable, its weight is very high.

    The remaining case is if $q\notin \supp(\vec{c'})$. Then, by the induction hypothesis we have 
    \[
    \begin{split}
         \vec{d}(q)\ge &\max\{\vec{c'}(p)\mid p\in \supp(\vec{c})\}+F+2\maxeff{u}>\\
         &F+\maxeff{u}\ge F+\max\{\vec{c}(p)\mid p\in \supp(\vec{c})\}
    \end{split}
    \]
    Intuitively, this case just shows that if $q$ is not yet reachable even in the current unfolding level, then by the end of the unfolding its weight is extremely high.

    \subparagraph*{Rebase case:}
    The rebase case is immediate by \cref{cor:rebase removal preserves gaps}. Indeed, the reachable set of configurations does not change, nor do the weights.   
\end{proof}
Conceptually, flattening is simply recursive unfolding (as depicted in \cref{fig:unfolding,fig:flattening example}), with some adjustments to account for rebase letters.

\cref{prop:unfolding cactus maintains seamlesss gaps,prop:un-unfolding maintains cheap seamless runs,cor:rebase removal preserves gaps,lem:flattening configuration characterization} give us the tools to reason about the flattening of words, and show that they maintain the gaps in runs. In \cref{apx:aug inf inf is det iff aug A is det} we introduce some additional reasoning to eliminate jump letters, and show the following.
\begin{theorem}
\label{thm:aug inf inf is det iff aug A is det}
    $\augA_\infty^\infty$ is determinizable if and only if $\augA$ is determinizable.
\end{theorem}
By combining this with \cref{lem:A det iff augA det} we also have the following.
\begin{corollary}
    \label{cor:aug inf inf is det iff A is det}
     $\augA_\infty^\infty$ is determinizable if and only if $\cA$ is determinizable.
\end{corollary}

\section{Dominance, Potential and Charge}
\label{sec:dominance and potential}
The construction in \cref{sec:cactus extension} can be thought of as a pumping argument, in that we make sure that states that can be (relatively) easily pumped unboundedly are out of the way. This construction is now equipped with a variety of tools from \cref{sec:cactus toolbox}.
Our next task is to define certain measures of configurations, that essentially refine the notion of $B$-gap witness (\cref{def: B gap witness}), and refer specifically to the cactus extension. We present two such measures, the \emph{potential} and \emph{charge}, and then provide a toolbox to work with them.

In the following section we refer to $\augA^\infty_\infty$ (sometimes with jump letters). Specifically, recall that in $\augA^\infty_\infty$, the weight of baseline transitions is $0$, and that there is at most one baseline run on each word (see \cref{rmk:baseline transitions with cactus are zero}). 

\subsection{Dominant States, Potential and Charge -- Definitions}
\label{sec:dominant states}
Consider a configuration $\vec{c}$ of $\augA^\infty_\infty$. 
Our first definition describes when a state $q$ is still ``relevant'' in $\vec{c}$, in the sense that some suffix read from $\vec{c}$ may make the run starting from $q$ minimal, or at least better (i.e., lower) than states that are currently lower. We make this precise as follows.
\begin{definition}[\keyicon Dominant State]
    \label{def:dominant state}
    Consider a configuration $\vec{c}\in \bbZinf^S$. We say that a state $q\in S$ is \emph{dominant} if there exists $w\in\Sigma^*$ such that 
    $\minweight(w,q\to S)<\infty$ 
    and for every $p\in S$, if $\vec{c}(p)<\vec{c}(q)$ then 
    $\minweight(w,p\to S)=\infty$.

    We say that $q$ is \emph{maximal dominant} if $\vec{c}(q)=\max\{\vec{c}(p)\mid p\text{ is dominant}\}$. We then denote this maximal value by $\domval(\vec{c})=\vec{c}(q)$ and the set of maximal dominant states by $\maxdom(\vec{c})$.
\end{definition}
Intuitively, $q$ is dominant if it can yield a run with finite weight, whereas all states below $q$ only yield infinite weight runs (i.e, they cannot complete a valid run).

Note that states $q'$ with $\vec{c}(q')\ge \vec{c}(q)$ are not considered in the definition, and may yield even lower runs than those starting from $q$. Thus, we distinguish the maximal dominant state as well.
Also note that dominance seems strongly tied to the notion of gap witnesses (\cref{def: B gap witness}). However, there is an important difference: the baseline run can be chosen to be some arbitrary run, which may cause the potential to be large, but such that way below the baseline there is another accepting run, and therefore there are no gap-witnesses. 
The precise connection between dominance and gap witnesses/determinizability is made in \cref{sec:witness}.

Next, recall that some states in $\augA^\infty_\infty$ are baseline, and that a run is a baseline if all the states along it are baselines. We wish to extend the definition of a baseline run to words. Naively, one could require that a word has a baseline run. This, however, is trivial -- every word has a baseline run (this follows from the definition of the baseline-augmented subset construction in \cref{sec:augmented construction}). 
Instead, we require that the baseline run is \emph{seamless} (c.f. \cref{sec:prelim}), i.e., that for every state $q$ visited by the baseline run after a prefix, the baseline run is also minimal to that state (although there may be other, lower runs, that go through other states).

Thus, we say that $w$ \emph{has a seamless baseline run} if its baseline run is seamless.
%
Recall that in the baseline-augmented subset construction, the weight of a baseline run is $0$, as baseline transitions are normalized to weight 0. This special status of the baseline run, provided it is seamless, allows us to capture the growth of other runs. We measure this in two ways: the difference between the highest dominant state to the baseline run is the \emph{potential}, and the difference between the baseline run and the current minimal state is the \emph{charge}\footnote{The term ``potential'' makes sense as it refers to states that can potentially become minimal. The name ``charge'' was chosen to keep in line with the physics-based naming.}. Recall that $\xconf(w,\vec{c}_{\init})$ is the configuration reached by $\augA^\infty_\infty$ after reading $w$. We then have the following.
\begin{definition}[\keyicon Potential and Charge]
    \label{def:potential}
    \label{def:charge}
    Consider a word $w$ that has a seamless jump-free baseline run and let $\vec{c}_w=\xconf(w,\vec{c}_{\init})$.
    \begin{enumerate}
        \item The \emph{potential} of $w$ is  $\pot(w)=\domval(\vec{c}_w)$.
        \item The \emph{charge} of $w$ is  $\charge(w)=-\min\{\vec{c}(q)\mid q\in Q\}$.
    \end{enumerate}
\end{definition}
Note that $\pot(w)$ and $\charge(w)$ are always nonnegative (when they are defined), since the minimal run is of weight at most $0$, due to the seamless baseline run.


\subsection{A Toolbox for Potential and Charge}
\label{sec:growth of potential and charge}
We turn to give some fundamental results about potential and charge, as well as link them to the tools of \cref{sec:cactus toolbox} -- specifically to baseline shift and unfolding.

Observe that if $w\cdot \sigma$ has a seamless baseline run, then so does $w$. In the following lemmas, we show that the difference in potential and in charge between $w\cdot \sigma$ and $w$ is bounded from above. This, however, requires that the alphabet is finite, and therefore the bound is not uniform for all words. Nonetheless, this suffices for our needs.
\begin{lemma}[\keyicon Potential Bounded Growth]
    \label{lem:bounded growth potential}
    For every finite alphabet $\Gamma\subseteq \Gamma^\infty_\infty$ there exists $k\in \bbN$ such that for every $w\in \Gamma^*$ and $\sigma\in \Gamma$, if $w\cdot \sigma$ has a seamless baseline run then $\pot(w\cdot \sigma)-\pot(w)\le k$.
\end{lemma}
\begin{proof}\ifproofs 
    Intuitively, the proof is simple: since $\Gamma$ is finite, the maximal weight that can be accumulated in a single step is bounded. 
    Thus, if $\maxdom(w\cdot \sigma)$ is much higher than $\maxdom(w)$, then the predecessor of $\maxdom(w\cdot \sigma)$ must have already started higher than $\maxdom(w)$, and is either a dominant state itself, or some state slightly lower but still higher than $\maxdom(w)$ is dominant. This contradicts the maximality of $\maxdom(w)$. 
    We turn to formalize this argument.
        
    Consider a finite alphabet $\Gamma$, and let $W$ be the maximal weight (in absolute value) that appears in a transition on a letter $\gamma\in \Gamma$.

    Consider a word $w\cdot \sigma\in \Gamma^*$ that has a seamless baseline run (whose prefix on $w$ is therefore also a seamless baseline run), and assume by way of contradiction that $\pot(w\cdot \sigma)-\pot(w)> W$.
    Let $p=\maxdom(w)$ and $p'=\maxdom(w\cdot \sigma)$, and let $\vec{c}_w=\xconf(w,\vec{c_\init})$ and $\vec{c}_{w\sigma}=\xconf(\sigma,\vec{c}_{w})$. 
    Consider the predecessor $q$ of $p'$. That is, $\minweight(w\sigma,q_0\runsto{w}q\runsto{\sigma}p')=\minweight(w\sigma,q_0\to p')$. We claim that $\vec{c}_w(q)>\vec{c}_w(p)$. 
    Indeed, more generally, if $r$ is a state for which $\vec{c}_w(p)\ge \vec{c}_w(r)$, then $r$ cannot yield any finite-weight run to $p'$ on $\sigma$, otherwise
    \begin{equation}
    \label{eq:pot bounded growth eq1}
    \begin{split}
    &\minweight(w\sigma,q_0\runsto{w}r\runsto{\sigma}p')\le \minweight(w\sigma,q_0\to r)+W=
    \vec{c}_w(r)+W\le \vec{c}_w(p)+W=\\    &\pot(w)+W<\pot(w\sigma)=\minweight(w\sigma,q_0\runsto{w}q\runsto{\sigma}p')=\minweight(w\sigma,q_0\to p')
    \end{split}
    \end{equation}
    which would be a contradiction.
        
    Since $p'=\maxdom(w\cdot \sigma)$ then by \cref{def:dominant state} there exists a word $z\in \Sigma^*$ (note that $z$ is not assumed to be over $\Gamma$) such that 
    $\minweight(z,p'\to Q)<\infty$ 
    and for every $q'$ with $\vec{c}_{w\sigma}(q')<\vec{c}_{w\sigma}(p')$ it holds that 
    $\minweight(z,q'\to Q)=\infty$.
    In particular, we have that 
    $\minweight(\sigma z,q\to Q)\le \minweight(\sigma z,q\runsto{\sigma}p'\runsto{z}Q)<\infty$.
    
    Let $r$ be a state attaining $\min\{\vec{c}_w(r)\mid \minweight(\sigma z, r\to Q)<\infty\}$. By similar argument to \cref{eq:pot bounded growth eq1} it must hold that $\vec{c}_w(r)> \vec{c}_w(p)$ (since any transition from $r$ on $\sigma$ can accumulate at most $W$, and therefore cannot surpass $\vec{c}_w(p')$ if starting below $\vec{c}_w(p)$).
    It follows that $r$ is a dominant state in $\vec{c}$, contradicting the assumption that $p=\maxdom(w)$. This concludes the proof.
\else \textbf{PROOFS REMOVED} \fi \end{proof}

\begin{lemma}[Charge Bounded Growth]
    \label{lem:bounded growth charge}
    For every finite alphabet $\Gamma\subseteq \Gamma^\infty_\infty$ there exists $k\in \bbN$ such that for every $w\in \Gamma^*$ and $\sigma\in \Gamma$, if $w\cdot \sigma$ has a seamless baseline run then $\charge(w\cdot \sigma)-\charge(w)<k$.
\end{lemma}
\begin{proof}\ifproofs 
    The proof idea is dual to that of \cref{lem:bounded growth potential} (but simpler): the minimal (negative) weight that can be accumulated in a single step is bounded from below, and therefore the minimal run cannot drop too much in a single step.
    
    Consider a finite alphabet $\Gamma$, and let $W$ be the maximal weight (in absolute value) that appears in a transition on a letter $\gamma\in \Gamma$.
    Consider a word $w\cdot \sigma\in \Gamma^*$ that has a seamless baseline run (whose prefix on $w$ is therefore also a seamless baseline run).

    Let $p$ and $p'$ be the states attaining $-\charge(w)$ and $-\charge(w\sigma )$, respectively. In particular, it holds that $\minweight(w,q_0\to Q)=\minweight(w,q_0\to p)=-\charge(w)$ and $\minweight(w\sigma ,q_0\to Q)=\minweight(w\sigma ,q_0\to p')=-\charge(w\sigma)$.
    We then have 
    \[\minweight(w\sigma,q_0\to p')\ge \minweight(w,q_0\to Q)+\minweight(\sigma,Q\to Q)\ge \minweight(w,q_0\to p)-W\]
    We thus have $-\charge(w\sigma)\ge -\charge(w)-W$, so $\charge(w\sigma)-\charge(w)\le W$, concluding the claim.
\else \textbf{PROOFS REMOVED} \fi \end{proof}
In contrast to \cref{lem:bounded growth charge}, $\charge$ can \emph{decrease} abruptly, if a minimal run cannot continue, so a much higher run becomes minimal. 
The following definition, when applicable, guarantees this doesn't happen.
\begin{definition}[Bounded-Decrease Charge]
\label{def:charge bounded decrease}
Consider a word $xyz$ where $\charge$ is defined, and let $B\in \bbN$. 
We say that the infix $y$ has \emph{$B$ bounded-decrease charge in $xyz$} if for every prefix $u\sigma$ of $y$ it holds that $\charge(xu)-\charge(xu\sigma)<B$.
\end{definition}
A useful property linking charge with potential is that when the charge does not have bounded decrease, it implies that the potential is unbounded. For this result, we restrict the alphabet to $\Gamma_\infty^0$, which is our main focus for most of the paper.
More precisely, we have the following.
\begin{lemma}[\lightbulbicon Unbounded Charge Decrease Implies Unbounded Potential]
\label{lem:charge no bounded decrease then potential unbounded}
    Let $\Gamma'\subseteq \Gamma_\infty^0$ be a finite alphabet, and assume that for every $B\in \bbN$ there exist a word $u\in \Gamma'^*$ and a letter $\sigma\in \Gamma'$ such that $\charge(u)-\charge(u \sigma)\ge B$.
    %
    Then $\sup\{\pot(w)\mid w\in (\Gamma_0^0)^*\}=\infty$.
\end{lemma}
\begin{proof}
    Intuitively, since $\sigma$ causes a large jump in the charge, this means that $\sigma$ cannot be read by the minimal run after $u$, and that the only run that can read it is much higher. This already suggests that the potential is unbounded. 
    Technically, all that needs to be done is to perform a baseline shift (\cref{sec: baseline shift}) on a minimal run on $u$. This, in turn, also requires us to first flatten (\cref{sec: cactus unfolding}) this prefix, so that we remain in $\Gamma_0^0$. 
    We recommend following this argument, as gentle preparation to the upcoming horrors of 
    \cref{lem: decomposed seq with cover sparse no ghosts implies unbounded potential}.
    We illustrate the proof in \cref{fig:charge no bounded decrease potential unbounded}.
    
\begin{figure}[ht]
    \centering
    \begin{subfigure}{0.2\textwidth}
        \centering
        \includegraphics[width=\textwidth]{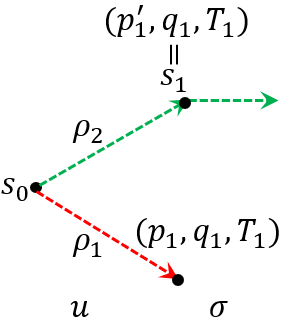}
        \caption{Step 1}
        \label{fig:charge to pot Step1}
    \end{subfigure} 
    ~
    \begin{subfigure}{0.25\textwidth}
        \centering
        \includegraphics[width=\textwidth]{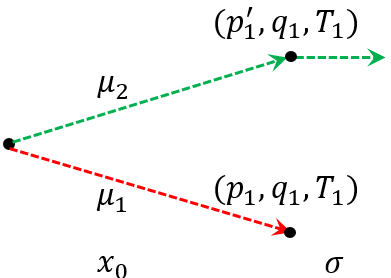}
        \caption{Step 2}
        \label{fig:charge to pot Step2}
    \end{subfigure}    
    ~
    \begin{subfigure}{0.45\textwidth}
        \centering
        \includegraphics[width=\textwidth]{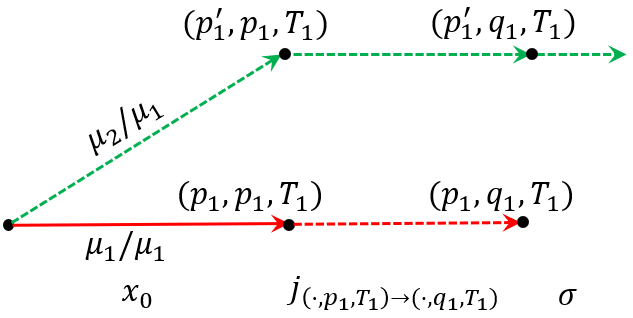}
        \caption{Step 3}
        \label{fig:charge to pot Step3}
    \end{subfigure}
    \caption{The proof flow of \cref{lem:charge no bounded decrease then potential unbounded}}
    \label{fig:charge no bounded decrease potential unbounded}
\end{figure}

    Denote by $W_{\max}$ the maximal weight appearing in a transition over any letter in $\Gamma'$.
    Let $P\in \bbN$, we show that the potential can be greater than $P$. By the assumption, there exist $u,\sigma$ such that $\charge(u)\ge \charge(u\sigma)+ P+2W_{\max}$.

    Since $\charge$ is the negation of the minimal weight, there exists a seamless run $\rho_1:s_0\runsto{u} S$ such that $\weight(\rho_1)=-\charge(u)=\minweight(u,s_0\to S)$. 
    Let $\vec{c_u}=\xconf(\vec{c_\init},u)$, and consider the minimal state $s_1$ that can read $\sigma$. That is, $s_1\in \arg\min\{s\in \supp(\vec{c_u})\mid \minweight(\sigma,s_1\to S)<\infty\}$. It must hold that $\vec{c_u}(s_1)>\weight(\rho_1)+P$ as otherwise (since $W_{\max}$ bounds the weight incurred on $\sigma$) the minimal run on $u\sigma$ has weight at most 
    \[-\charge(u)+P+W_{\max}< -\charge(u)+P+2W_{\max}\le  -\charge(u\sigma)\]
    where the last inequality follows by our requirement on $\charge(u\sigma)$. But this is a contradiction, since $\minweight(s_0,u\sigma)=-\charge(u\sigma)$.
    Put simply -- if the charge has a large jump on $\sigma$, then the minimal state that can read $\sigma$ is much higher than the minimal state up to $\sigma$. Note that $s_1$ is therefore dominant (\cref{def:dominant state}).
    Let $\rho_2:s_0\runsto{u}s_1$ be a seamless run to $s_1$ (see \cref{fig:charge to pot Step1}).
    
    We now turn to flattening the prefix (\cref{def:cactus rebase flattening}).
    Let $x_0=\flatten(u \wr 2W_{\max}|u\sigma|)$. Since $\Gamma'$ does not contain rebase letters, we have $x_0\in (\Gamma_0^0)^*$ (since flattening may only introduce jump letters in case of rebase letters).
    Let $\mu_1:s_0\runsto{x_0}S$ be a seamless run such that $\weight(\mu_1)=-\charge(x_0)$ and similarly $\mu_2:s_0\runsto{x_0}s_1$ a seamless run such that $\weight(\mu_2)=\weight(\rho_2)$. These exist by 
    \cref{lem:flattening configuration characterization}.
    By the same lemma, for every state $s'$ with $\minweight(x_0,s_0\to s')<\weight(\mu_2)$ we have $\minweight(\sigma,s'\to S)=\infty$ i.e., $s_1$ is still the minimal state from which $\sigma$ can be read, and $\weight(\mu_2)-\weight(\mu_1)\ge P$ (see \cref{fig:charge to pot Step2}).
    
    We now perform a baseline shift on $\mu_1$. 
    By \cref{cor:baseline shift maintains gaps}, we have 
    \[P\le \weight(\baseshift{\mu_2}{\mu_1})-\weight(\baseshift{\mu_1}{\mu_1})=\weight(\baseshift{\mu_2}{\mu_1})\]
    (where $\weight(\baseshift{\mu_1}{\mu_1})=0$ as it is the baseline run, as per \cref{cor:baseline shift to seamless run}).
    Essentially, all that remains to do is adapt $\sigma$ to the baseline shift, so that it can be read from the ``shifted $s_1$'' but not from any state below it. This is achieved using jump letters, as follows.

    Denote the state reached by $\rho_1$ as $(p_1,q_1,T_1)$ and $s_1=(p'_1,q_1,T_1)$ (since the last two components are deterministic, they are the same for both states). Keep in mind that there is a finite weight transition from $s_1$ on $\sigma$ to some state. Note that $\mu_1$ also ends in $(p_1,q_1,T_1)$, and $\mu_2$ ends in $(p'_1,q_1,T_1)$ (since flattening does not change the reachable states and the baseline component).

    Therefore, after the baseline shift we have that $\baseshift{\mu_1}{\mu_1}:s_0\to (p_1,p_1,T_1)$ and $\baseshift{\mu_2}{\mu_1}:s_0\to (p'_1,p_1,T_1)$. 
    We now consider the suffix $\jl_{(\cdot,p_1,T_1)\to (\cdot,q_1,T_1)}\cdot \sigma$. There exists some $r'\in S$ such that
    \[
    s_0\runsto{\baseshift{x_0}{\mu_1}}(p'_1,p_1,T_1) \runsto{\jl_{(\cdot,p_1,T_1)\to (\cdot,q_1,T_1)}} (p'_1,q_1,T_1)\runsto{\sigma} r'
    \]
    However, for any state below $(p'_1,p_1,T_1)$ in the configuration $\vec{c}=\xconf(\vec{c_\init},\baseshift{x_0}{\mu_1})$, there is no finite-weight run on $\jl_{(\cdot,p_1,T_1)\to (\cdot,q_1,T_1)}\cdot \sigma$. Indeed, such a run implies a run from a state lower than $s_1$ on $\sigma$, which contradicts the definition of $s_1$ (see \cref{fig:charge to pot Step3}).

    We conclude that $\pot(\baseshift{x_0}{\mu_1})\ge P$, so we are done (note that we do not have equality, since we do not assume that $s_1$ is maximal dominant, only dominant).   
\end{proof}

Our next simple result is that both potential and charge are monotone with respect to the initial configuration, as follows.
\begin{lemma}
\label{lem:potential and charge are monotone}
        Consider configurations $\vec{c},\vec{d}$ where $\vec{c}\le \vec{d}$ such that $\supp(\vec{c})=\supp(\vec{d})$, there is a single baseline state $q\in \supp(\vec{c})$, and $\vec{c}(q)=\vec{d}(q)=0$.
        
        Then for every $w\in \Sigma^*$ with a seamless baseline run from $\vec{c}$ it holds that:
        \begin{enumerate}
            \item The base run of $\xconf(w,\vec{d})$ is seamless.
            \item $\pot(\xconf(w,\vec{c}))\le \pot(\xconf(w,\vec{d}))$.
            \item $\charge(\xconf(w,\vec{c}))\ge \charge(\xconf(w,\vec{d}))$
        \end{enumerate}
    \end{lemma}
\begin{proof}
    Since $\vec{c}\le \vec{d}$, then for every run $\rho:p\runsto{y}q$ we have $\weight_{\vec{c}}(\rho)\le \weight_{\vec{d}}(\rho)$, and in particular 
    \begin{equation}
    \label{eq:potential and charge are monotone main}
        \minweight_{\vec{c}}(y,P\to R)\le \minweight_{\vec{d}}(y,P\to R) \quad  \text{ for every } P,R\subseteq Q.
    \end{equation}
        Using this, we prove the three claims.
    \begin{enumerate}
        \item Consider the baseline run $\rhobase=(t_1,\ldots,t_n)$ on $w$ from $\vec{c}$. Since transitions between baseline states have weight $0$, and since $\vec{c}(q)=\vec{d}(q)=0$ for their single baseline state $q$, it follows that for every prefix $x$ of $w$ we have  
        $\weight_{\vec{c}}(t_1,\ldots,t_{|x|})=\weight_{\vec{d}}(t_1,\ldots,t_{|x|})=0$.
        
        Since $\rhobase$ is seamless from $\vec{c}$, then for every prefix $x$ of $w$ and baseline state $q_{|x|}$ reached by $\rhobase$ after reading $x$ we have 
        $\weight_{\vec{c}}(t_1,\ldots,t_{|x|})=\minweight_{\vec{c}}(x,Q\to q_{|x|})$.
        By \cref{eq:potential and charge are monotone main} we have
        \[\weight_{\vec{d}}(t_1,\ldots,t_{|x|})=\weight_{\vec{c}}(t_1,\ldots,t_{|x|})=\minweight_{\vec{c}}(x,Q\to q_{|x|})\le \minweight_{\vec{d}}(x,Q\to q_{|x|})\]
        so $\rhobase$ is seamless from $\vec{d}$ as well.

        \item Denote $\vec{c'}=\xconf(w,\vec{c})$ and $\vec{d'}=\xconf(w,\vec{d})$
        By \cref{eq:potential and charge are monotone main} we have that $\vec{c'}\le \vec{d'}$. 
        Consider the state $s=\maxdom(\vec{c'})$ and let $z\in \Sigma^*$ be a suffix such that 
        $\minweight(z,s\to Q)<\infty$
        and 
        $\minweight(z,s'\to Q)=\infty$ 
        for every $s'$ such that $\vec{c'}(s')< \vec{c'}(s)$, as per \cref{def:potential}.
        Since $\supp(\vec{c'})=\supp(\vec{d'})$ (following $\supp(\vec{c})=\supp(\vec{d})$), it holds that 
        $\minweight(z,s'\to Q)<\infty$. 
        Let $s''$ be a state that attains the minimum $\min\{\minweight_{\vec{d'}}(z,s''\to Q)<\infty\mid s''\in Q\}$. That is, $s''$ is dominant with respect to $z$. 
        Notice that $\vec{c'}(s)\le \vec{c'}(s'')$ since $s=\maxdom(\vec{c'})$, but $\vec{c'}\le \vec{d'}$, so $\vec{c'}(s)\le \vec{d'}(s'')$. 
        In particular, 
        \[\domval(\vec{c'})=\vec{c'}(s)\le \vec{d'}(s'')\le \domval(\vec{d'})\]

        \item Again denote $\vec{c'}=\xconf(w,\vec{c})$ and $\vec{d'}=\xconf(w,\vec{d})$. Since $\vec{c'}\le \vec{d'}$ and there are seamless baseline runs to $\vec{c'}$ and $\vec{d'}$, we have $\min\{\vec{c'}(q)\mid q\in Q\}\le \min\{\vec{d'}(q)\mid q\in Q\}$, so by \cref{def:potential} we have $\charge(\vec{c'})\ge \charge(\vec{d'})$.
    \end{enumerate}
\end{proof}

We now turn to link the potential with baseline shift (\cref{sec: baseline shift}).
\begin{proposition}
    \label{prop:potential over baseline shift}
    Consider a word $u$ such that $\pot(u)$ is defined, and let $\rho$ be a seamless run on $u$. For every prefix $v$ of $u$, we have that $\pot(\baseshift{v}{\rho})=\pot(v)-\weight(\rho(v))$.
\end{proposition}
\begin{proof}
    The idea is quite intuitive: we shift all so that $\rho$ becomes baseline, so while the potential itself may change, it is still attained by the same dominating state, and the difference in weight is exactly $\rho(v)$. 
    The only nontrivial part of the proof is translating a ``separating suffix'' $z$ that works for $v$ to one that works for $\pot(\baseshift{v}{\rho})$. This is achieved using jump letters.
    
    Let $\rho_0$ be the baseline run over $v$, so that $\weight(\rho_0)=0$.
    Let $\vec{c}=\xconf(\vec{c_{\init}},v)$ and $\vec{c'}=\xconf(\vec{c_{\init}},\baseshift{v}{\rho})$, and let $q=\maxdom(\vec{c})$, i.e., $\vec{c}(q)=\pot(v)$. Thus, there exists $z\in \Gamma_\infty^\infty$ (possibly with jump letters) such that $\minweight(z,q\to S)<\infty$ and for every $p$ with $\vec{c}(p)<\vec{c}(q)$ we have $\minweight(z,p\to S)=\infty$. Let $\mu$ be a seamless run $\mu:s_0\runsto{v}q$. We now consider the corresponding ``shifted'' state $q'\in S$ such that $\baseshift{\mu}{\rho}:s_0\runsto{\baseshift{v}{\rho}}q'$. 
    By \cref{cor:baseline shift maintains gaps} we have that $\baseshift{\mu}{\rho}$ is also a seamless run (since its gaps from other runs are the same as those of the seamless $\mu$), and in addition
    \[
    \pot(v)=\weight(\mu)-\weight(\rho_0)=\weight(\baseshift{\mu}{\rho})-\weight(\baseshift{\rho_0}{\rho})=\vec{c'}(q')+\weight(\rho)
    \]
    where the latter equality, namely $\weight(\rho)=-\weight(\baseshift{\rho_0}{\rho})$ follows from \cref{prop:baseline shift run bijection}.

    It thus remains to prove that $\vec{c'}(q')=\pot(\baseshift{v}{\rho})$. To this end, consider the suffix $\jl_{q'\to q}$ (note that by \cref{prop:baseline shift run bijection}) this is indeed a legal jump letter (and in particular has cost $0$ from $q'$ to $q$), then we have 
    \[
    \minweight(\jl_{q'\to q}\cdot z,q'\to S)\le 
    \minweight(\jl_{q'\to q}\cdot z,q'\runsto{\jl_{q'\to q}}q\runsto{z} S)<\infty
    \]
    On the other hand, still by \cref{prop:baseline shift run bijection}, for every state $s'\in \supp(\vec{c'})$, if $\vec{c'}(s')<\vec{c'}(q')$, then the ``corresponding state'' $s\in \supp(\vec{c})$, i.e., the state such that a seamless run $\eta:s_0\runsto{v} s$ has $\baseshift{\eta}{\rho}:s_0\runsto{\baseshift{v}{\rho}} s'$, satisfies $\vec{c}(s)<\vec{c}(q)$. Moreover, we have that $s'\runsto{\jl_{q'\to q}} s$, and this is the only transition from $s'$ on $\jl_{q'\to q}$ (this is immediate by the shape of the states, as described in \cref{prop:baseline shift run bijection}).
    We therefore have
    \[
    \minweight(\jl_{q'\to q}\cdot z,s'\to S)=
    \minweight(\jl_{q'\to q}\cdot z,s'\runsto{\jl_{q'\to q}}s\runsto{z} S)=\infty
    \]
    So indeed $q'$ is dominant. We therefore have $\vec{c'}(q')\le \pot(\baseshift{v}{\rho})$, so we conclude
    $\pot(\baseshift{v}{\rho})\ge \pot(v)-\weight(\rho)$.

    In order to obtain equality, we could use the same technique as above. However, a simpler argument is just to shift again to the original baseline. That is, we apply the proof thus far to the word $\baseshift{v}{\rho}$ and shift to the run $\baseshift{\rho_0}{\rho}$. We thus have
    \[\pot(\baseshift{\baseshift{v}{\rho}}{\rho_0})\ge\pot(\baseshift{v}{\rho})-\weight(\baseshift{\rho_0}{\rho})\]
    However, recall from above that $-\weight(\baseshift{\rho_0}{\rho})=\weight(\rho)$, and observe that nesting the baseline shift operator is equivalent to applying the outermost one, so since $\rho_0$ is the baseline run on $v$, we have
    \[\baseshift{\baseshift{v}{\rho}}{\rho_0}=\baseshift{v}{\rho_0}=v\]
    We thus get $\pot(v)\ge \pot(\baseshift{v}{\rho})+\weight(\rho)$
    which concludes the second inequality.
\end{proof}

Our last result of this section is that the charge is invariant under cactus unfolding (\cref{def:unfolding function}). Intuitively, this is because the charge is attained via a grounded pair, whereas unfolding mainly changes non-grounded pairs.
\begin{proposition}
\label{prop:unfolding maintains charge}
    Consider a word $u=x \cdot \alpha_{B,w} \cdot y$ and let $F>2\maxeff{x \cdot \alpha_{B,w} \cdot y}$. Let $M_0\in \bbN$ such that $xw^{2\bigM M_0}y=\unfold(x , \alpha_{B,w} , y \wr F)$.
    Then for every prefix $u$ of $y$ it holds that
    $\charge(x\cdot w^{2\bigM M_0}\cdot u)=\charge(x\cdot \alpha_{B,w}\cdot u)$.
\end{proposition} 
\begin{proof}
    Consider the configurations $\vec{c_1}=\xconf(x\cdot \alpha_{B,w}\cdot u)$ and $\vec{c_2}=\xconf(x\cdot \alpha_{B,w}\cdot u)$. Let $p_1,p_2\in S$ such that $\vec{c_1}(p_1)=-\charge(x\cdot \alpha_{B,w}\cdot u)$ and $\vec{c_2}(p_2)=-\charge(x\cdot w^{2\bigM M_0}\cdot u)$.
    In particular, we have $p_1\in \supp(\vec{c_1})$ and $p_2\in \supp(\vec{c_2})$. 
    By \cref{lem:unfolding configuration characterization} we have $\vec{c_2}(p_1)=\vec{c_1}(p_1)$, so 
    $-\charge(x\cdot w^{2\bigM M_0}\cdot u)\le -\charge(x\cdot \alpha_{B,w}\cdot u)$.

    Conversely, if $\vec{c_1}(p_2)=\vec{c_2}(p_2)$, then we have $-\charge(x\cdot w^{2\bigM M_0}\cdot u)\ge -\charge(x\cdot \alpha_{B,w}\cdot u)$ and we are done. 

    Otherwise (by way of contradiction), we have $\vec{c_1}(p_2)\neq \vec{c_2}(p_2)$, so again by \cref{lem:unfolding configuration characterization} it follows that $p_2\notin \supp(c_1)$, and therefore 
    \[\vec{c_2}(p_2)>F-\maxeff{x \cdot \alpha_{B,w} \cdot y}\ge \maxeff{x \cdot \alpha_{B,w} \cdot y}>0\] 
    and in particular (since the baseline run exists and has weight $0$) it follows that $\vec{c_2}(p_2)$ is not minimal in $\vec{c_2}$, and therefore is not the charge, which contradicts the assumption.
\end{proof}
Ideally, we would now want to prove a similar result regarding potential. Unfortunately, such a result does not generally hold. Nonetheless, in \cref{sec:witness}, we introduce concepts that allow us to obtain a strong-enough similar result.

\section{A Witness for Nondeterminizability}
\label{sec:witness}
We are now ready to describe the type of witnesses we use for nondeterminizability. Intuitively, a witness consists of a prefix $w_1$, a ``pumpable infix'' $w_2$, and a suffix $w_3$, where the idea is that after enough iterations of $w_2$ the potential grows unboundedly, and $w_3$ serves to separate the dominant state from the baseline, thus demonstrating unbounded gaps, which imply nondeterminizability by \cref{thm:det iff bounded gap}.

In this section, after defining witnesses, we prove that they provide a sound characterization for nondeterminizability: if a witness exists, then $\augA_\infty^\infty$ is nondeterminizable. Moreover, we show that checking whether a given input is indeed a witness, is decidable.
The converse, namely that witnesses are a complete characterization, is in fact the main result of this entire paper.

There are two cruxes to the approach. First, $w_1,w_2$ and $w_3$ are over our cactus extension, and thus capture nested cycles. Indeed, it is well known that single cycles (even non-simple) in the original WFA $\cA$ may not suffice to demonstrate unbounded gaps (e.g., \cref{xmp:running pumping}). 
Second, we require $w_3$ to essentially cause a jump from finite to infinite weight. This may seem overly strict. 
Indeed, in some cases all runs have finite weight, and such a separation is impossible in $\cA$, despite it being nondeterminizable (e.g., if all states are accepting and the transition function is total).
However, since the witness is over a cactus alphabet, it actually captures an ``asymptotic'' behavior, and in particular may induce $\infty$ weight transitions, which we rely on. This becomes clearer after the definition and the proof of \cref{lem:witness implies nondet}.

Consider the WFA $\augA_\infty^\infty=\tup{S,\Gamma_{\infty}^\infty,s_0,\augTrans_\infty^\infty}$, and recall that $\Gamma_{\infty}^\infty$ is obtained by alternating the stabilization closure (\cref{def:stab closure}) and rebase (\cref{def:rebase}) (and it also includes jump letters, as per \cref{def:jump letters}). 

For generality and brevity, we define witnesses parameterized by a type $k\in \bbN$. For this paper, however, we only use types $0$ and $1$. 
We explain the intuition behind the definition in the proof of \cref{lem:witness implies nondet} below. 
\begin{definition}[\keyicon Witness]
    \label{def:witness}
    For $k\in \bbN$, a \emph{witness of type $k$} consists of three words $w_1,w_2,w_3$ such that the following hold.
    \begin{enumerate}
        \item $w_1,w_2\in (\Gamma_\infty^k)^*$,  $w_3\in(\Gamma_\infty^\infty)^*$ and there is a seamless baseline run on $w_1w_2$.
        \item Let $S_1= \ghostTrans(s_0,w_1)$, then $(S_1,w_2)$ is a stable cycle with corresponding cactus letter $\alpha_{S_1,w_2}$ (recall that $\ghostTrans$ are the ghost-reachable states, as defined in \cref{def:ghost states}).
        \item $\booltrans(s_0,w_1)=\booltrans(s_0,w_1w_2)$ (i.e., $w_2$ cycles on the reachable states).
        \item $\minweight(w_1w_2w_3,s_0\to S)<\infty$ whereas $\minweight(w_1\alpha_{S_1,w_2}w_3,s_0\to S)=\infty$.
    \end{enumerate}
\end{definition}

We can now give the formal justification for witnesses, namely that they imply nondeterminizability. In the proof we also give further intuition about the details of the definition.

\begin{lemma}[\keyicon]
    \label{lem:witness implies nondet}
    If there is a type-$k$ witness for some $k\in \bbN$, then $\augA^\infty_\infty$ is not determinizable.
\end{lemma}
\begin{proof}
    We start with some intuition regarding the proof and the definition of witnesses.
    The ``core'' of \cref{def:witness} is Requirement $4$: the word $w_1w_2w_3$ can be read with finite weight, but $w_1\alpha_{S_1,w_2}w_3$ cannot. Recall that $\alpha_{S_1,w_2}$ allows finite runs only on grounded pairs, whose weight is bounded when reading $w_2^*$. In addition, $\alpha_{S_1,w_2}$ has a seamless baseline run. Since we also require a seamless baseline run on $w_1w_2$, it follows that there is a run of weight $0$ on $w_1(w_2^{2\bigM})^*$. However, this run, and indeed -- all runs between grounded pairs -- are ``killed'' by $w_3$ (otherwise we would have finite weight runs on $w_1\alpha_{S_1,w_2}w_3$). 
    Therefore,  all the finite weight runs on $w_1w_2w_3$ must traverse, along the $w_2$ infix, pairs of non-grounded states (this is not strictly true, because the grounded pairs require $w_2^{2\bigM}$, which is why we need to pump $w_2$, hence Requirement 3). 
    Since runs on grounded pairs increase arbitrarily, we can construct arbitrarily-high gap witnesses, implying nondeterminizability. 
    This is illustrated in \cref{fig:witness structure}.
    We now turn to the formal details.

    We show that the existence of a type-$k$ witness implies the gap criterion of nondeterminizability as per \cref{thm:det iff bounded gap}. That is, assume there exists a type-$k$ witness $(w_1,w_2,w_3)$, then we claim that for every gap $G\in \bbN$ there exist words $x,y\in (\Gamma_\infty^\infty)^*$ and a state $s\in S$ such that 
    $\minweight(xy,s_0\to S)=\minweight(xy,s_0\runsto{x}s\runsto{y}S)<\infty$, and $\minweight(x,s_0\to s)-\minweight(x,s_0\to S)>G$.

    Let $G\in \bbN$. 
    As per \cref{def:witness}, denote $S_1=\ghostTrans(s_0,w_1)$, then $(S_1,w_2)$ is a stable cycle corresponding to $\alpha_{S_1,w_2}$. Denote by $W=\maxeff{w_1\alpha_{S_1,w_2}w_3}$, and recall that $W$ is an upper bound on the weight (in absolute value) of any run of any prefix of $w_1\alpha_{S_1,w_2}w_3$.
    By \cref{lem:pumping grounded pairs} there exists $M_0\in \bbN$ such that for every $m\ge M_0$ and every pair of states $(r,s)\notin \GroundPairs(S_1,w_2)$ it holds that $\minweight(w_2^{2\bigM m},r\to s)>G+W$. 

    \begin{figure}[ht]
        \centering
        \includegraphics[width=0.7\linewidth]{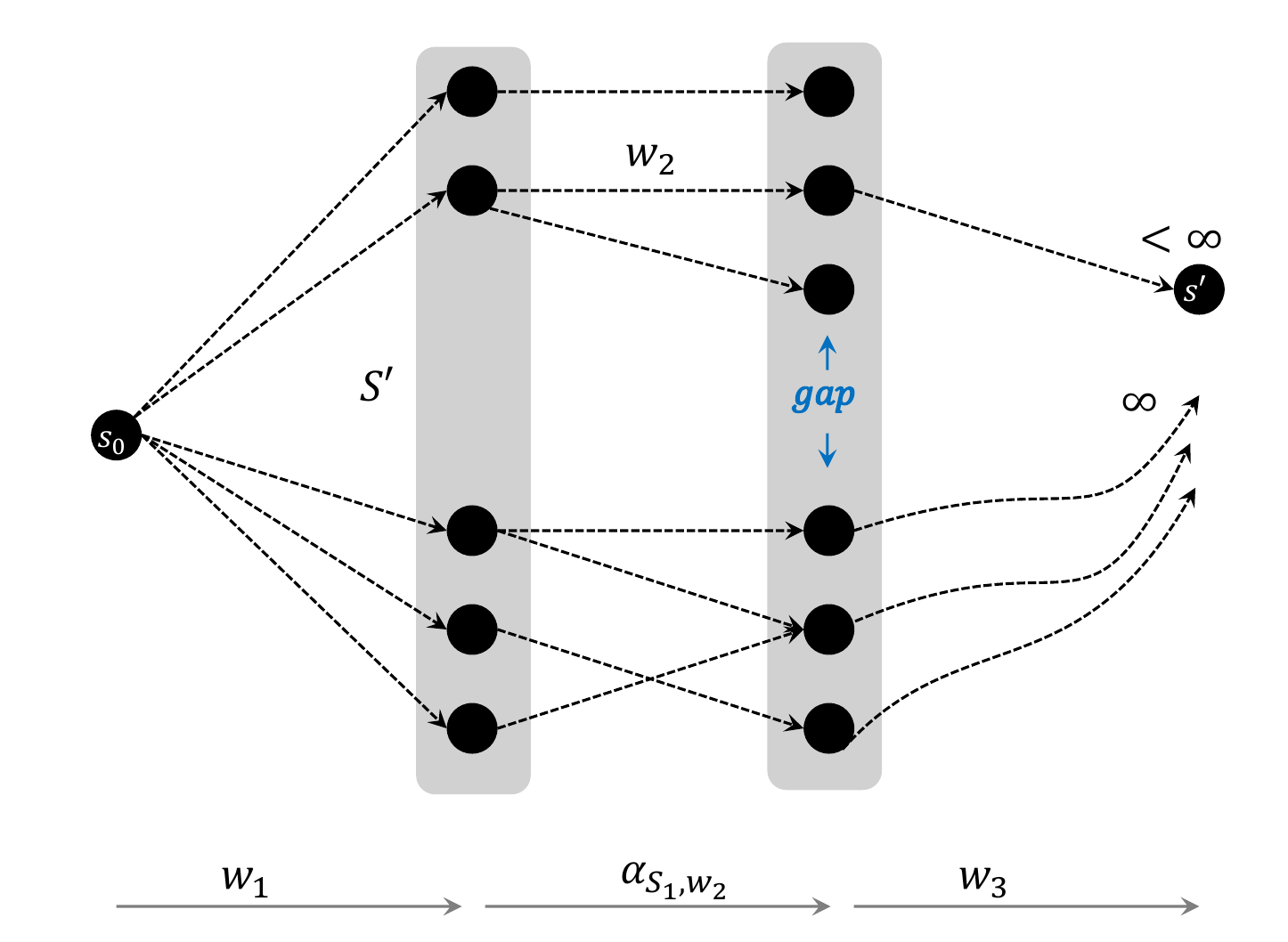}
        \caption{The structure of a witness, where  
        $S_1 = \protect\ghostTrans(s_0,\widetilde{w_1})$. Notice that the grounded-pairs of $\alpha_{S_1,w_2}$ get killed by $w_3$, but there are ghost runs that survive $w_3$.}
        \label{fig:witness structure}
    \end{figure}
    
    Note that by Requirement 3 we have $\booltrans(s_0,w_1)=\booltrans(s_1,w_1w_2^n)$ for every $n\in \bbN$.
    Denote 
    \[S_2=\booltrans(s_0,w_1\alpha_{S_1,w_2}) \quad \text{ and } \quad U_2=\booltrans(s_0,w_1w_2)\setminus S_2=\booltrans(s_0,w_1w_2^{2\bigM M_0})\setminus S_2\]
    By Requirement 4, we have $\minweight(w_3,S_2\to S)=\infty$ (otherwise there is a finite-weight run on $w_1\alpha_{S_1,w_2}w_3$). Since we also have $\minweight(w_1w_2w_3)<\infty$, it follows that there exists $s\in U_2$ such that 
    $\minweight(w_3,s\to S)<\infty$ and moreover, we can assume $s$ is minimal, i.e., 
    \[\minweight(w_1w_2^{2\bigM M_0}w_3,s_0\to S)=\minweight(w_1w_2^{2\bigM M_0}w_3,s_0\runsto{w_1w_2^{2\bigM M_0}}s\runsto{w_3} S)\]
    In particular, by our choice of $s$ we have $\minweight(w_3,s\to S)<\infty$.

    Next, for every $r\in S_1$ such that $r\runsto{w_2^{2\bigM M_0}}s$ we have $(r,s)\notin \GroundPairs(S_1,w_2)$ (otherwise we would have $s\in S_2$, but $s\in U_2$). Thus, we have $\minweight(w_1w_2^{2\bigM M_0},s_0\to s)> -\maxeff{w_1}+ G+W\ge G$. 
    Since $w_1w_2$ has a seamless baseline run, we have
    $\minweight(w_1w_2^{2\bigM M_0}, s_0\to S)\le 0$.
    We therefore have that 
    \[\minweight(w_1w_2^{2\bigM M_0},s_0\to s)-\minweight(w_1w_2^{2\bigM M_0}, s_0\to S)> G-0=G\]
    Thus, for $x=w_1w_2^{2\bigM M_0}$ and $y=w_3$, the gap requirements hold. Moreover, the family of witnesses is over the alphabet of $w_1,w_2,w_3$, and in particular finite.
    Since this is true for all $G$, we conclude from \cref{thm:det iff bounded gap} that $\augA_\infty^\infty$ is nondeterminizable. 
\end{proof}

Algorithmically, \cref{lem:witness implies nondet} is not of much use unless we can detect when words form a witness. Fortunately, this is easy (intuitively -- all the conditions are easy to check, the only mildly-nontrivial part is computing the transition weights on cactus letters).
\begin{theorem}
    \label{thm:verifying witness is decidable}
    The following problem is decidable: given a WFA $\cA$ and words $(w_1,w_2,w_3)\in \Gamma_\infty^
\infty$, decide whether $(w_1,w_2,w_3)$ is a witness.
\end{theorem}
\begin{proof}
    We describe an algorithm to check whether the words form a witness. First, for every letter $\sigma$ occurring in $w_1,w_2,w_3$, if $\sigma\in \Gamma_0^0$, then we know its associated transitions by constructing $\augA$ from $\cA$.

    If $\sigma=\alpha_{S',x}\in \Gamma_k^j$ for some $k,j$, we need to compute its transitions between all pairs of states. This is done inductively as follows: first compute the transitions on the letters of $x$ (which are of a lower rank). Next, compute the grounded pairs of $\alpha_{S',x}$ by first computing the set $\MinRefStates(S',x^\bigM)$, which may form the grounding states, and then checking for each $g\in \MinRefStates(S',x^\bigM)$ which states $(s,r)\in S'$ satisfy $s\runsto{x^\bigM}g\runsto{x^\bigM}r$, as per \cref{def:grounded pairs} (which is a simple reachability query).
    The weights of the transitions between grounded states $(s,r)$ are then computable as per \cref{def:stabilization} (by iterating over the possible grounding states).
    In addition, we similarly compute the transitions on $\alpha_{S_1,w_2}$.

    If $\sigma=\rebase_{S',x,s\to r}$, we first verify that $(s,r)$ are indeed a grounded pair, as above, and the readily compute the transitions and their weights (inductively) by \cref{def:rebase}.

    Thus, we have all the transitions weights on the letters of $w_1,w_2,w_3$. We can now easily check the conditions of \cref{def:witness}: condition 1 is syntactic, there is nothing to check. Condition 2 can be checked as follows: $S_1$ can be computed as per \cref{def:ghost states} from the transitions, and checking that the cycle is stable is a simple reachability query as per \cref{def:stable cycle}.
    Condition 3 is checked by composing the transitions on $w_1$ and on $w_1w_2$ and checking reachability equivalence. Finally, since we already computed the transitions over  $\alpha_{S_1,w_2}$, we can evaluate both expressions in Condition 4, and we are done.
\end{proof}

Our next result is that either a type-0 witness exists, or stable cycles can be safely unfolded without affecting the potential. Recall that type-0 means that we allow nesting of cacti, but no rebase or jump letters on the $w_1,w_2$ parts of the witness. 
Intuitively, this result is as close as we can get to an analogue of \cref{prop:unfolding maintains charge} for Potential.

\begin{lemma}[\keyicon \lightbulbicon Unfolding Maintains Potential]
    \label{lem:unfolding maintains potential}
        Either $\augA_\infty^\infty$ has a type-$0$ witness, or the following holds.
        Consider an initial state\footnote{When applying this lemma, we sometimes use a different initial state than that of $\augA$. We therefore allow any initial state, but still call it $s_0$ for clarity.} $s_0$ and $u,\alpha_{B,x},v\in (\Gamma_\infty^0)^*$ (where $(B,x)$ is a stable cycle).
        Let  $F\ge 2\maxeff{u\alpha_{B,x}v}$ and $ux^{2\bigM M_0}v=\unfold(u,\alpha_{B,x},v \wr F)$ as per \cref{def:unfolding function}, then for every prefix $v'$ of $v$ we have:
        \begin{enumerate}
            \item $\xconf(\vec{c_{s_0}},ux^{2\bigM M_0}v')\le \xconf(\vec{c_{s_0}},u\alpha_{B,x}v')$.
            \item $\pot(ux^{2\bigM M_0}v')=\pot(u\alpha_{B,x}v')$ (if the potential is defined, with respect to the given initial state $s_0$).
        \end{enumerate}
\end{lemma}
\begin{proof}
    Item 1 follows immediately from \cref{lem:unfolding configuration characterization}. Indeed, let $\vec{d}=\xconf(\vec{c_{s_0}},u\alpha_{B,x}v')$ and $\vec{d'}=\xconf(\vec{c_{s_0}},ux^{2\bigM M_0}v')$, then  we have 
    $\supp(\vec{d})\subseteq  \supp(\vec{d'})$ and for every $q\in \supp(\vec{d})$ we have $\vec{d'}(q)=\vec{d}(q)$ (and anyway $\vec{d}(q)=\infty$ otherwise). 

    It is now tempting to use \cref{lem:potential and charge are monotone} and use the superiority of the configurations to conclude the proof. Unfortunately, the condition there also requires that the configurations have equal supports, which fails in our case. We therefore need to work much harder to show the lemma.

    Intuitively, the proof proceeds as follows. We consider the unfolding $ux^{2\bigM M_0}v=\unfold(u,\alpha_{B,x},v \wr F)$, and specifically its prefix $ux^{2\bigM M_0}$. 
    By \cref{prop:unfolding cactus maintains seamlesss gaps}, the unfolding preserves the reachable configuration, i.e. the weights with which states in $\booltrans(s_0,u\alpha_{B,x})$ are reached. However, we might also get finite weights on other states in $\ghostTrans(s_0,u\alpha_{B,x})$, due to the unfolding. 
    The question now is whether these newly reachable states interfere with the potential. 

    In order to reason about the potential, we consider a maximal dominant state $q$, and split to two cases. The first case is when $q$ is in $\booltrans(s_0,ux^{2\bigM M_0})\setminus \booltrans(s_0,u\alpha_{B,x})$ (namely -- very high, reached by pumping $x$ on non-grounded pairs). 
    Then, we show that we can construct a type-0 witness. Specifically, the witness is $(u x^{2\bigM M_0},x^{2\bigM M_0},z)$, where $z$ is the suffix that shows $q$ is dominant. Proving that this is indeed a witness is technical, and uses the intricacies of our definitions of stable cycles and grounded states. 
    The intuition, however, is not too complicated: if the dominant state is high, then replacing $x^{2\bigM M_0}$ by its corresponding cactus letter sends all states above $q$ to $\infty$, and the remaining states are sent to $\infty$ by $z$.

    The second case is when $q$ is in $\booltrans(s_0,u\alpha_{B,x})$, namely low, reachable by grounded states. Then we use the unfolding toolbox of \cref{sec: cactus unfolding} to show that the potential is indeed maintained.

    We now turn to formalize this argument.
    Consider $ux^{2\bigM M_0}v=\unfold(u,\alpha_{B,x},v \wr F)$, and assume that the potential is defined (otherwise we are done).
    Consider the following configurations and sets of states:
    \[
\begin{array}{ l l l }
\vec{c_1} = \xconf(\vec{c_{s_0}}, u) & \vec{c_2} = \xconf(\vec{c_{s_0}}, u\alpha_{B,x}) & \vec{c'_2} = \xconf(\vec{c_{s_0}}, ux^{2\bigM M_0}) \\ 
S_1 = \ghostTrans(s_0, u) & S_2 = \booltrans(s_0, u\alpha_{B,x}) & S'_2 = \booltrans(s_0, ux^{2\bigM M_0}) \\ 
\end{array}
\]
    Note that $S_2=\supp(\vec{c_2})$ and $S'_2=\supp(\vec{c'_2})$. 
    By \cref{lem:unfolding configuration characterization} we have that 
    $S_2\subseteq S'_2$, for every $s\in S_2$ it holds that $\vec{c'_2}(s)=\vec{c_2}(s)$ and 
    for every $s'\notin S_2$ it holds that $\vec{c'_2}(s')> \vec{c_2}(s)$ (where $s'\notin S_2$ means either $s'\in  S'_2\setminus S_2$, or $s'\notin S'_2$).

    Consider a maximal-dominant state $q\in S$ in $\vec{c'_2}$. 
    By \cref{def:dominant state}, this means that there exists $z\in \Gamma^\infty_\infty$ such that 
    $\minweight(z,q\to S)<\infty$ 
    whereas for every $p\in S$ with $\vec{c'_2}(p)<\vec{c'_2}(q)$ we have 
    $\minweight(z,p\to S)=\infty$. We now split to cases according to whether $q\in S'_2\setminus S_2$ (i.e., $q$ is ``very high'') or $q\in S_2$ (i.e., $q$ is ``low'').

    \paragraph*{Dominant State is Very High}
    If $q\in S'_2\setminus S_2$, we claim that $(u x^{2\bigM M_0},x^{2\bigM M_0},z)$ is a type-0 witness. 
    Indeed, following the requirements in \cref{def:witness}, we have:
    \begin{enumerate}
        \item $ux^{2\bigM M_0},x^{2\bigM M_0}\in \Gamma_\infty^0$ and $z\in \Gamma^\infty_\infty$.
        \item Since $\alpha_{B,x}$ is a cactus letter then by \cref{prop:cactus letters stabilizes at 2M} we have that $\alpha_{B,x^{2 \bigM M_0}}$ is a cactus letter with the same transitions. Also, $B=\ghostTrans(s_0,u)=\ghostTrans(s_0,ux)=\ghostTrans(s_0,ux^{2\bigM M_0})$, since $(B,x)$ is a stable cycle.
        \item By \cref{prop: transitions stabilize at M} we have that  $\booltrans(s_0,u x^{2\bigM M_0})=\booltrans(s_0,u x^{4\bigM M_0})=\booltrans(s_0,u x^{2\bigM M_0}x^{2\bigM M_0})$.
        \item Since $q\in S'_2=\booltrans(s_0,u x^{2\bigM M_0})$ then by Item 3 above we have $q\in \booltrans(s_0,u x^{2\bigM M_0}x^{2\bigM M_0})$. In particular, $\minweight(u x^{2\bigM M_0}x^{2\bigM M_0},s_0\to q)<\infty$. Since $z$ satisfies that $\minweight(z,q\to S)<\infty$, we conclude that:
        \[
        \minweight(u x^{2\bigM M_0}\cdot x^{2\bigM M_0}\cdot z,s_0\to S)\le\minweight(u x^{2\bigM M_0}x^{2\bigM M_0},s_0\to q)+\minweight(z,q\to S)<\infty
        \]

        It therefore remains to prove that $\minweight(u x^{2\bigM M_0}\cdot \alpha_{B,x^{2\bigM M_0}}\cdot z,s_0\to S)=\infty$. 
        To this end, we claim that for every state $p$, if $\minweight(u x^{2\bigM M_0}\cdot \alpha_{B,x^{2\bigM M_0}},s_0\to p)<\infty$, then $\vec{c'_2}(p)<\vec{c'_2}(q)$ (from which we can obtain the desired result).
        Indeed, if $\minweight(u x^{2\bigM M_0}\cdot \alpha_{B,x^{2\bigM M_0}},s_0\to p)<\infty$ then 
        there exists a run 
        \[
        \rho:s_0\runsto{u}s_1\runsto{x^{2\bigM M_0}}p'\runsto{\alpha_{B,x^{2\bigM M_0}}}p
        \]
        By \cref{def:stabilization}, this implies that $(p',p)\in \GroundPairs(\alpha_{B,x^{2\bigM M_0}})$. Then, however, it follows by \cref{prop:grounding states reachable with any bigM k} that $(s_1,p)\in \GroundPairs(B,x^{2\bigM M_0})$ (by using the same grounding state of $(p',p)$). 
        Recall that by \cref{prop:cactus letters stabilizes at 2M}, the letters $\alpha_{B,x^{2\bigM M_0}}$ and $\alpha_{B,x}$ have the exact same transitions. That is, $(s_1,p)\in \GroundPairs(B,x)$. 
        Thus, we have the run
        $s_0\runsto{u}s_1\runsto{\alpha_{B,x}}p$
        so $p\in S_2$. 
        As noted above, all the states in $S_2$ the configurations $\vec{c_2}$ and $\vec{c'_2}$ coincide, so $\vec{c_2}(p)=\vec{c'_2}(p)$. On the other hand, $q\in S'_2\setminus S_2$, and again as noted above this means that $\vec{c'_2}(q)>\vec{c_2}(p)=\vec{c'_2}(p)$.

        We are therefore in the following setting:
        \begin{itemize}
            \item If $\minweight(u x^{2\bigM M_0}\cdot \alpha_{B,x^{2\bigM M_0}},s_0\to p)<\infty$, then $\vec{c'_2}(p)<\vec{c'_2}(q)$, but then $\minweight(z,p\to S)=\infty$, so 
            \[\minweight(u x^{2\bigM M_0}\cdot \alpha_{B,x^{2\bigM M_0}}\cdot z,s_0\runsto{u x^{2\bigM M_0}\cdot \alpha_{B,x^{2\bigM M_0}}}p\runsto{z}S)=\infty\]
            \item If $\minweight(u x^{2\bigM M_0}\cdot \alpha_{B,x^{2\bigM M_0}},s_0\to p)=\infty$ then trivially 
            \[\minweight(u x^{2\bigM M_0}\cdot \alpha_{B,x^{2\bigM M_0}}\cdot z,s_0\runsto{u x^{2\bigM M_0}\cdot \alpha_{B,x^{2\bigM M_0}}}p\runsto{z}S)=\infty\]
        \end{itemize}
        We conclude that anyway we have
        $\minweight(u x^{2\bigM M_0}\cdot \alpha_{B,x^{2\bigM M_0}}\cdot z,s_0\to S)=\infty$
    \end{enumerate}
    This shows that all the requirements of \cref{def:witness} hold, so there is a type-0 witness.

    \paragraph*{Dominant State is Low}
    If $q\notin S'_2\setminus S_2$, then $q\in S_2$ (since $q\in \supp(\vec{c'_2})=S'_2$). Then, by the observations above, we have $\vec{c'_2}(q)=\vec{c_2}(q)$. We claim that for every prefix $v'$ of $v$ it holds that $\pot(ux^{2\bigM M_0}v')=\pot(u\alpha_{B,c}v')$. 
    Indeed, fix such a prefix $v'$ and denote
    \[
    \begin{split}
        &\vec{c_3}=\xconf(s_0,u\alpha_{B,x} v')=\xconf(c_2,v') \quad S_3=\supp(\vec{c_3})\\
        &\vec{c'_3}=\xconf(s_0,u x^{2\bigM M_0} v')=\xconf(c'_2,v') \quad S'_3=\supp(\vec{c_3})
    \end{split}
    \]
    Similarly to the above, we have $S_3\subseteq S'_3$, and due to the choice of $F$ as the unfolding constant (\cref{def:unfolding function,lem:unfolding configuration characterization}) for every $p\in S_3$ it holds that $\vec{c_3}(p)=\vec{c'_3}(p)$, and for every $p\in S'_3\setminus S_3$ we have $\vec{c'_3}(p)>\max\{\vec{c_3}(r)\mid r\in S_3\}+\maxeff{u\alpha_{B,c}v}$.

    Let $p,p'$ be maximal dominant states in $\vec{c_3}$ and $\vec{c'_3}$, respectively. We claim that $p,p'\in S_3$. Indeed, $p\in S_3$ by definition. If $p'\in S'_3\setminus S_3$, then there exists a suffix $z$ such that $\minweight(z,p'\to S)<\infty$ and for every $r'$ with $\vec{c'_3}(r')<\vec{c'_3}(p')$ it holds that $\minweight(z,r'\to S)=\infty$. In particular, for every $r'\in S_3$ this holds. 
    Since $p'\in S'_3\setminus S_3$, then $p'\in \booltrans(S'_2,v')\setminus \booltrans(S_2,v')$. But then $v'z$ is a suffix that exhibits a dominant state in $S'_2\setminus S_2$, and in particular the maximal dominant state of $\vec{c'_2}$ is in $S'_2\setminus S_2$, so we are back in the previous case of a very high dominant state (which we assume does not hold now).

    Thus, $p'\in S_3$. It now readily follows that $\pot(ux^{2\bigM M_0}v')=\pot(u\alpha_{B,c}v')$: restricted to $S_3$, the configurations $\vec{c_3}$ and $\vec{c'_3}$ are equivalent. Moreover, for every $r'\in S'_3\setminus S_3$ we have $\vec{c'_3}(r')>\vec{c'_3}(p')$, so does not come into consideration for the dominance of $p'$. 

    It thus follows that $\vec{c'_3}(p')=\vec{c_3}(p)$ (in fact, we can assume without loss of generality that $p=p'$, since we can choose any maximal-dominant state, if there are several), and since the baseline run is reachable and its state is also in $S_3$, we conclude that
    \[
    \pot(ux^{2\bigM M_0}v')=\vec{c'_3}(p')=\vec{c_3}(p)=\pot(u\alpha_{B,c}v').
    \]
\end{proof}

\begin{remark}[The Prefix ``Either $\augA_\infty^\infty$ has a type-0 witness, or...'']
\label{rmk:the either witness prefix}
    Most of our reasoning about potential in the coming sections is based on \cref{lem:unfolding maintains potential}. 
    Therefore, the prefix ``Either $\augA_\infty^\infty$ has a type-0 witness, or the following holds'' appears in many other arguments. Since having such a witness immediately implies nondeterminziability (by \cref{lem:witness implies nondet}), this condition can be read as ``either we are done, or the following holds''.
\end{remark}

\section{A Ramsey Theorem for Colored Infixes}
\label{sec:ramsey}
In this section we present a general Ramsey-type theorem, concerning finite coloring of an infinite set of words. Specifically, we assume that we have an infinite set of words such that all the infixes of each word are colored. 
We show that we can find arbitrarily-long words with arbitrarily-many indices such that all infixes between these indices are identically colored. 

These results are then used to reason about the types of runs between configurations in \cref{sec:existence of inc inf assuming covered}.

For $n\in \bbNinf$ denote $\interval_n=\{(i,j)\mid 1\le i\le j\le n\wedge  i,j\in \bbN\}$ denote the set of finite \emph{intervals} of $\{1,\ldots, n\}$ (think of $(i,j)$ as denoting the set $\{i,i+1,\ldots, j\}$). Note that for $n=\infty$, the set $\interval_\infty$ has all finite intervals. 
Fix a finite set $\colset$ of \emph{colors}. 
An \emph{interval coloring} is a function $\lambda:\interval_n\to \colset$ that assigns a color to each interval.

Fix a finite alphabet $\Sigma$. An \emph{infix-colored set} is an infinite set $C=\{(w_1,\lambda_1),(w_2,\lambda_2),\ldots\}$ where $w_i\in \Sigma^*$ and $\lambda_i:\interval_{|w_i|}\to \colset$ is an interval coloring for all $i\in \bbN$. 

Intuitively, the following lemma shows that given an infix-colored set $C$, there is a single infinite word $w$ and a corresponding interval coloring such that $w$ has infinitely many prefixes in $C$, and their infix coloring coincides with that of $w$. Technically, we state the lemma slightly differently. 
For $\lambda:\interval_{k}\to \colset$ and $k'<k$, we denote $\lambda|_{k'}$ as $\lambda$ restricted to the domain $\interval_{k'}$.
We denote by $\Sigma^\omega$ is the set of infinite words over $\Sigma$.

\begin{lemma}
    \label{lem:ramsey infinite word}
    Let $C$ be an infix-colored set, then there exists an infinite word $w\in \Sigma^\omega$ and an interval coloring $\lambda_\omega:\interval_\infty\to \colset$ such that for every finite prefix $v$ of $w$ there exists $(x,\lambda')\in C$ such that $v$ is a prefix of $x$ and $\lambda_\omega|_{|v|}\equiv \lambda'|_{|v|}$
\end{lemma}
\begin{proof}
    Consider a pair $(w,\lambda)\in C$, then for every prefix $v$ of $w$, we can consider the pair $(v,\lambda|_{|v|})$, which is a coloring of $v$ that coincides with $\lambda$ over the relevant indices.

    We now construct a finitely-branching forest\footnote{The graph is a forest since $\epsilon$ may appear with different colorings.} (i.e., a union of trees) $F=\tup{V,E}$ \`a la K\"onig, as follows. 
    The vertices of the tree are $F=\{(v,\lambda)\mid v\in \Sigma^*, \lambda:\interval_{|v|}\to \colset\}$, i.e., each vertex is an infix-colored word. The edges are
    \[
    E=\{((v,\lambda),(v\sigma,\lambda')\mid \lambda'|_{|v|}\equiv \lambda,\ \sigma\in \Sigma)\}
    \]
    That is, we connect a colored word $v$ to any word obtained by adding a single letter $\sigma$ to $v$, such that the coloring of $v\sigma$ coincides with that of $v$ on the relevant indices.

    Note that all the outgoing edges from vertices $(v,\lambda)$ with $|v|=n$ are to vertices $(v',\lambda')$ with $|v'|=n+1$. 
    Since there are only finitely many vertices $(v,\lambda)$ with $|v|=n$ (since $\colset$ is finite), then clearly the branching degree of the tree is also finite.

    We now restrict $F$ to the sub-forest $F'$ obtained by keeping only the vertices in $F$ from which there is a path to some $(x,\lambda')\in C$ (observe that the elements of $C$ are also vertices in $F$ by definition). 
    Since $C$ is infinite and $\Sigma$ is finite, it follows that $C$ contains arbitrarily long words. Since we keep all vertices along paths from $\epsilon$ (with some coloring) to words in $C$, it follows that $F'$ is a finitely-branching forest with arbitrarily long paths.

    By K\"onig's lemma, it follows that there is an infinite path in $F'$. We claim that such a path induces an infinite word as required. 
    Indeed, we can assume without loss of generality that the infinite path begins in $(\epsilon,\lambda_\epsilon)$ (otherwise we add prefixes to it until reaching $\epsilon$). We can therefore write the path as $(\epsilon,\lambda_\epsilon),(\sigma_1,\lambda_1),(\sigma_1\sigma_2,\lambda_2),\ldots$, i.e., the $n$-th vertex in the path is $(\sigma_1\cdots \sigma_n,\lambda_n)$. 

    We now define the infinite word $w$ and $\lambda_\omega$ as follows. $w=\sigma_1\sigma_2\cdots$ (i.e., $w[n,n]=\sigma_n$ for all $n\in \bbN$). For every $(i,j)\in \interval_\omega$ we define $\lambda_\omega(i,j)=\lambda_j(i,j)$.

    It remains to show that $(w,\lambda_\omega)$ satisfies the required conditions. Consider a finite prefix $v$ of $w$, then in particular $(v,\lambda_{|v|})\in F'$ (since $w$ is described by a path in $F'$), and $\lambda_{|v|}\equiv \lambda_\omega|_{|v|}$.
    Moreover, by the definition of $F'$, it follows that there is a path in $F$ from $(v,\lambda_{|v|})$ to some $(x,\lambda')\in C$. By the way we construct $F$, we have that $\lambda'|_{|v|}\equiv \lambda_{|v|}$. Putting these together, we have $\lambda_{\omega}|_{|v|}=\lambda'|_{|v|}$, as required.   
\end{proof}

We can now show that every infix-colored set $C$ has words that can be decomposed to arbitrarily many infixes such that all infixes have the same coloring, and moreover -- this coloring is idempotent with respect to concatenation. In addition, we are able to place some constraints on the cost of these infixes, where ``cost'' here is a general extension of the length of a word.
Naturally, our usage is with the cost functions of \cref{def:cost depth and sub cactus}.

Formally, a \emph{prefix-cost function} is a function $\precost:\Sigma^+\to \bbN\setminus \{0\}$ without any restrictions.
 A letter-cost function is $\lc:\Sigma^+\to \bbN\setminus \{0\}$ such that $\lc$ is ``strictly increasing on prefixes'', i.e., if $x$ is a strict prefix of $y$ then $\lc(x)<\lc(y)$ (c.f., \cref{rmk:cost is strictly increasing}).
We also denote $w[i,j]=\sigma_i\cdots \sigma_{j}$ for $1\le i\le j\le n$. Before we proceed, we recall a specific theorem by Ramsey that we use in the proof.
For a set $A$, write $(A)_n=\{A'\subseteq A: |A'|=n\}$
\begin{theorem}[Ramsey, 1930]
	\label{thm: ramsey original}
Let $f:(\bbN)_n\to \colset$, then there exists an infinite set $A\subseteq \bbN$ on which $f$ is homogeneous, i.e., $f$ is constant on $(A)_n$.
\end{theorem}
\begin{theorem}[\keyicon \lightbulbicon A Ramsey Theorem for Infix Coloring]
    \label{thm: our ramsey}
    Consider an infix-colored set $C$, a prefix-cost function $\precost$ and a letter-cost function $\lc$. 
    For every $L\in \bbN$ there exists a pair $(v,\lambda)$ and an index set $1\le i_0<i_1<\ldots <i_L\le |v|+1$ such that there are infinitely many pairs $(w_i,\lambda_i)\in C$ (for $i\in \bbN$) where $v$ is a prefix of $w_i$, and the following hold:
    \begin{enumerate}
        \item 
        $\lambda\equiv \lambda_i|_{|v|}$ for all $i\in \bbN$. I.e., the coloring of $v$ matches that of the $w_i$.
        \item For every $0\le j<k\le L$ we have $\lambda(i_0,i_1-1)=\lambda(i_j,i_{k-1})$, i.e., all the infixes between the indices in the set have the same color.
        \item $\precost(v[i_0,i_1-1])< \lc(v[i_1,i_{L}])$, i.e., the first infix is much ``cheaper'' than the remaining infix.
    \end{enumerate}
\end{theorem}
\begin{proof}
Let $w\in \Sigma^\omega$ and $\lambda_\omega$ be the infix-colored infinite word obtained as per \cref{lem:ramsey infinite word}. Note that $\lambda_\omega:(\bbN)_2\to \colset$, and we can therefore apply \cref{thm: ramsey original} to $\lambda_\omega$.
Thus, there exists an infinite subset $A\subseteq \bbN$ and a color $r\in \colset$ such that $\lambda_{\omega}(i,j)=r$ for all $i<j\in A$. Sort $A$ as $i_0<i_1<\ldots$. 

Recall that the cost function $\lc$ strictly increasing on prefixes, and in particular strictly increasing with \emph{length} along any infix of $w$.
Since $A$ is infinite, then there exists $L\in \bbN$ large enough such that $\precost(w[i_0,i_1-1])<\lc(w[i_1,i_L])$ (note that $\precost(w[i_0,i_1-1])$ is a constant).

Take $v=w[1,i_L]$. By the above, Requirement \textit{3} holds (since $v$ is identical to $w$ up to index $i_L$). 
Requirement \textit{2} holds by the application of \cref{thm: ramsey original}, since all the infixes with indices in $A$ are colored $r$.
The existence of infinitely many pairs in $C$ where $v$ is a prefix follows from \cref{lem:ramsey infinite word}. Indeed, for every prefix $x$ of $w$ such that $v$ is a prefix of $x$, there exists $(w_i,\lambda_i)\in C$ such that $x$ is a prefix of $w_i$ and $\lambda_\omega|_{|x|}\equiv \lambda_i|_{|x|}$. In particular, since there are arbitrarily long such prefixes $x$, it follows that there are infinitely many such pairs $(w_i,\lambda_i)$. This also shows that Requirement \textit{1} holds, concluding the proof.
\end{proof}

\section{Separated Increasing Infixes}
\label{sec:separated increasing infix}
In this section we define word infixes whose runs exhibit a very ``organized'' behavior, which are later used in our proof. We then show that such infixes exist.

\subsection{Separated Increasing Infixes -- Definition}
\label{sec:separated increasing infix definition}
We start with an intuitive overview of the properties we require, follows by the formal definition.
\begin{figure}[ht]
    \centering
    \includegraphics[scale=0.5]{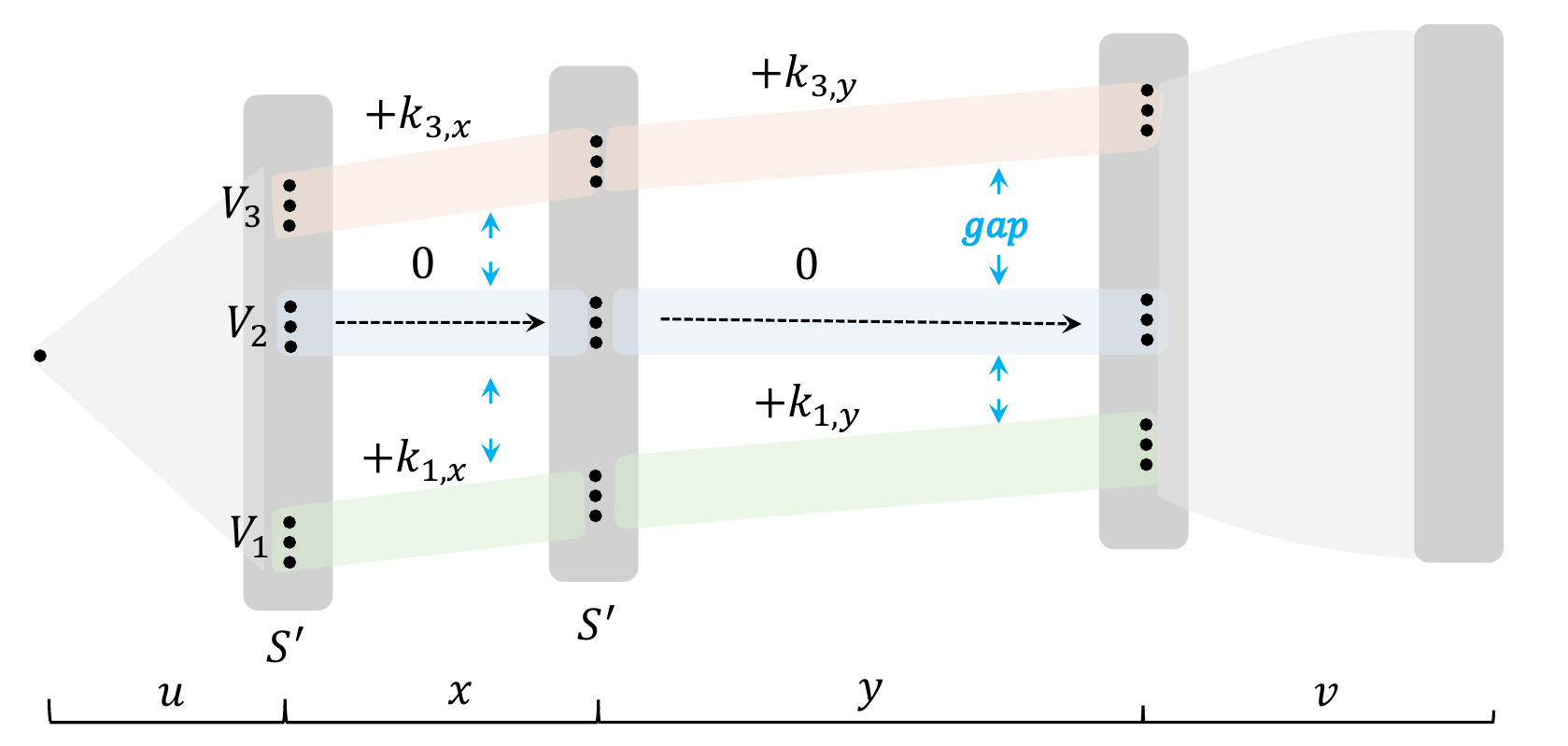}
    \caption{A Separated Increasing Infix from State $p$. After reading $u$, the configuration is separated to sub-configurations corresponding to the $V_i$. Reading $x$ maintains this partition and the exact distances within each set. The same holds for $y$, but with different shifts. 
    In addition, $y$ is much longer than $x$ (or rather, has much higher cost). Moreover, the gaps between the sub-configurations are larger still -- much larger than the maximal effect of $x$ and $y$.
    }
    \label{fig:increasing infix}
\end{figure}

Consider the WFA $\augA_\infty^0=\tup{S,\Gamma^0_\infty,s_0,\augTrans^0_\infty}$. That is, we focus on the construction with cactus letters, but without rebase or jump.
An separated increasing infix (from state $p$) consists of a concatenation $uxyv$ with the following structure (see \cref{fig:increasing infix}): 
\begin{itemize}
    \item Reading $u$ from $p$ reaches a configuration that is partitioned to ``sub-configurations'' that have a huge gap between them, we call the state sets in these sub configurations $V_1,\ldots, V_{\ell}$. 
    \item Next, reading $x$ maintains the structure of these $V_j$ sets, and increases all the states in each set $V_j$ by a constant $k_{j,x}\ge 0$. In particular, the entire configuration maintains its support after $x$.
    \item Reading $y$ has very similar behavior to $x$, in the sense that the $V_j$ are maintained, and each is increased by a constant $k_{j,y}$, with the property that $k_{j,y}=0$ if and only if $k_{j,x}=0$. 

    However, an additional property of $y$ is that it is \emph{very} long (or rather, very costly). Specifically, its cost (in the sense of \cref{def:cost depth and sub cactus}) is larger than the cost of a cactus letter on $x$. This means that replacing the infix $xy$ by the cactus $\alpha_{B,x}$ decreases the overall cost (which indeed we do in \cref{lem:increasing infix budding}).

    Despite the huge cost of $y$, the gaps between the $V_j$ are even higher, so much so that ``pumping'' $x$ should be safe. These notions are made concrete throughout the section.
    \item Finally, $v$ is just a harmless suffix.
\end{itemize}

We now turn to the formal definition, starting with some notations.
For a word $w$, recall that $\wmax{w}$ is the maximal absolute value of a weight over any transition on a letter of $w$.
Further recall that $\wmax{w}\cdot |w|$ upper-bounds the maximal finite change in weights that can be incurred upon reading $w$.
For a state $p\in S$, let $\vec{c_p}$ be the configuration that assigns $p$ weight $0$ and $\infty$ otherwise.
\begin{definition}[\keyicon Separated Increasing Infix from $p$]
    \label{def:separated increasing infix from state}
    Consider a state $p\in S$ and a word $w\in (\Gamma_\infty^0)^*$ decomposed as $w=u\cdot x\cdot y\cdot v$. 
    Denote $\vec{c_u}=\xconf(\vec{c_p},u)$ and similarly $\vec{c_{ux}}=\xconf(\vec{c_p},ux)$ and $\vec{c_{uxy}}=\xconf(\vec{c_p},uxy)$.
    We say that \emph{$w$ is a separated increasing infix from $p$} if the following conditions hold.
    \begin{enumerate}
        \item \label{itm:separated infix length diff} 
        $\cost(y)>2^{16(\bigM |S| \cost(x))^2}$
        (where $\cost$ is as per \cref{def:cost depth and sub cactus}).
        \item \label{itm:separated infix reachable states} $\supp(\vec{c_u})=\supp(\vec{c_{ux}})=\supp(\vec{c_{uxy}})=B$ for some subset $B\subseteq S$. 
        That is, the subsets of states reached after $u,ux,uxy$ are the same.
        \item \label{itm:separated infix partition}
        $B$ can be partitioned to $B=V_1\cup \ldots \cup V_\ell$ for some $\ell$ such that the following holds.
        \begin{enumerate}
            \item 
            \label{itm:separated infix gaps V}
            For every $1\le j<\ell$, every $s\in V_j,s'\in V_{j+1}$ and  every $\xi\in \{u,ux,uxy\}$ we have 
            $\vec{c_\xi}(s')-\vec{c_\xi}(s)> \incinfGapConst$. 
            \acctodo{If this constant changes, the accounting in the definition of $\sparse$ needs to be updated, along with every part influenced by it. These parts are next to comments labeled ACCOUNTING-1}
            
            That is, the values assigned to states in $V_{j+1}$ are significantly larger than those assigned to $V_{j}$ in the configurations $\vec{c_{u}}, \vec{c_{ux}}$ and $\vec{c_{uxy}}$.
            \item 
            \label{itm:separated infix linear k}
            For every $1\le j\le \ell$ there exist $k_{j,x},k_{j,y}\in \bbN$ such that for every $s\in V_j$ we have  $\vec{c_{ux}}(s)=\vec{c_u}(s)+k_{j,x}$ and $\vec{c_{uxy}}(s)=\vec{c_{ux}}(s)+k_{j,y}$.

            Moreover, $k_{j,x}=0$ if and only if $k_{j,y}=0$.
        \end{enumerate}
        \item \label{itm:separated infix seamless baseline}  If $p$ is a baseline state, then there is a seamless baseline run $\rho:p\runsto{uxyv}p$ from $\vec{c_p}$ (and so $\rho$ is constantly of weight $0$ in all prefixes).
    \end{enumerate}
\end{definition}
The partition and gaps described in \cref{itm:separated infix partition,itm:separated infix gaps V} give rise to a restriction on the runs on the infixes $x,y$. Specifically, there are only runs from $V_j$ to $V_{j'}$ for $j\ge j'$, and there is always a run from $V_{j'}$ to $V_{j'}$. In addition, the constants $k_{j,x}$ and $k_{j,y}$ are bounded.
\begin{proposition}
\label{prop:run characteristics of increasing infix from state}
    Consider a separated increasing infix $uxyv$ from state $p$, and let $B=V_1\cup\ldots \cup V_\ell$ and $k_{j,x},k_{j,y}$ as per \cref{def:separated increasing infix from state}. Then the following hold for every $1\le j\le \ell$,.
    \begin{itemize}
        \item $k_{j,x}\le \maxeff{x}$ and $k_{j,y}\le \maxeff{y}$.
        \item For every $z\in \{x,y\}$ and $s\in S$, if $s\in V_{j}$ then there exists $s'\in V_{j}$ such that $s'\runsto{z}s$.
        \item For every $z\in \{x,y\}$ and $s,s'\in S$ such that $s\runsto{z}s'$, if $s'\in V_{j'}$ then $s\in V_{j}$ for some $j\ge j'$.
    \end{itemize}
\end{proposition}
\begin{proof}
    We prove the claim for $z=x$, and the case of $y$ identical up to working with $\vec{c_{uxy}}$ instead of $\vec{c_{ux}}$.

    We start with the first item. Assume by way of contradiction that there exists $1\le j\le \ell$ such that $k_{j,x}>\maxeff{x}$, and let $j$ be maximal with this property (i.e., the ``highest'' part in the partition where this holds). We depict the proof in \cref{fig:increasing infix run characterization}.
    \begin{figure}[ht]
        \centering
        \includegraphics[width=0.4\linewidth]{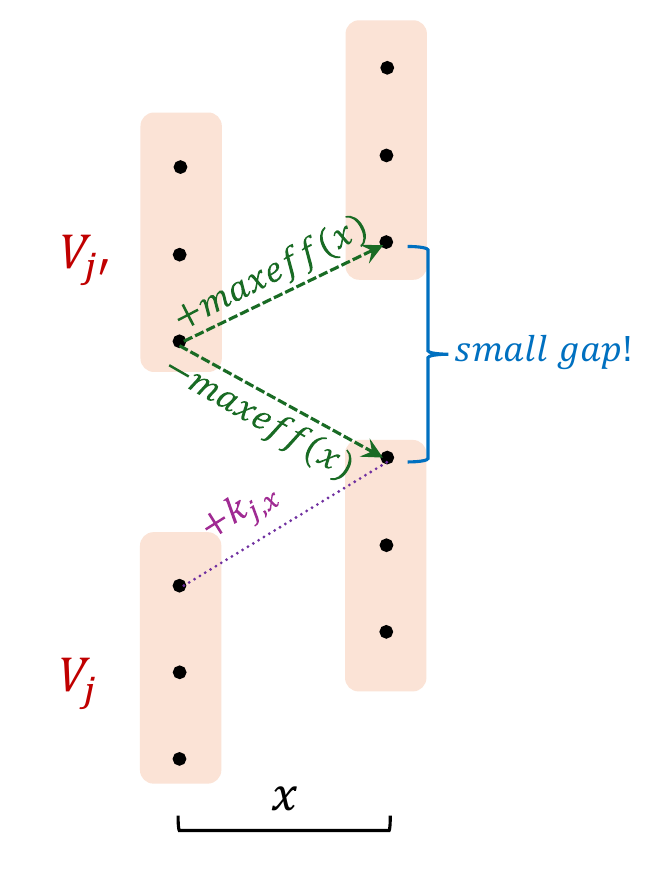}
        \caption{The part $V_j$ and the state $s$ whose minimal run stems from $V_{j'}$, implying a small gap.}
        \label{fig:increasing infix run characterization}
    \end{figure}
    
    Consider a state $s\in \arg\max\{\vec{c_u}(q)\mid q\in V_j\}$. We then have $\vec{c_{ux}}(s)>\vec{c_u}(s)+\maxeff{x}$. This means that the minimal seamless run $\rho:s'\runsto{x}s$ is from $s'\in V_{j'}$ with $j'>j$. Indeed, no run from any state in $V_j$ or lower can get to weight $\vec{c_u}(s)+\maxeff{x}$. In particular, we get $j<\ell$.
    We therefore have $\vec{c_{ux}(s)}\ge \vec{c_{u}}(s')-\maxeff{x}$.

    In the following, it is useful to think of $j'$ as $j+1$ (although we do not actually assume that). 
    Consider $k_{j',x}$. By the maximality assumption on $j$, we have that $k_{j',x}\le \maxeff{x}$. Therefore, $\vec{c_{ux}}(s')\le \vec{c_{u}}(s')+\maxeff{x}$. This, however, contradicts the gap in $\vec{c_{ux}}$: we have that $s'\in V_{j'}$ and $s\in V_j$, but 
    \[
    \vec{c_{ux}}(s')-\vec{c_{ux}}(s)\le 2\maxeff{x}\ll \incinfGapConst
    \]
    \acctodo{check}
    in contradiction to \cref{itm:separated infix gaps V} of \cref{def:separated increasing infix from state}.
    We conclude that $k_{j,x}\le \maxeff{x}$ for all $j$, and similarly for $y$.

    We now proceed to the latter two items.
    Let $1\le j\le \ell$ and $s\in V_{j}$ and let $s'\in V_{j'}$ such that $s'\runsto{x}s$. 
    Assume by way of contradiction that $j'<j$, then by the gap criterion \cref{itm:separated infix gaps V}, since the gap is much larger than $\maxeff{x}$, we have 
    $\vec{c_{ux}}(s)\le \vec{c_u}(s')+\maxeff{x}<\vec{c_{u}}(s)$, but this is a contradiction to \cref{itm:separated infix linear k} of \cref{def:separated increasing infix from state}.
    Thus, $j'\ge j$, as required.

    Next, assume by way of contradiction that for every $s'$ such that $s'\runsto{x}s$ we have $s'\in V_{j'}$ for $j'>j$. As we show above, this immediately implies $k_{k,x}>\maxeff{x}$, which cannot hold.
\end{proof}

Our next observation is that an increasing infix does not allow negative cycles on $x^k$ for any $k$. Note that negative runs on $x$ may occur. For example, a state in $V_7$ may have a negative run on $x$, but to a state in $V_3$, which is much lower and therefore this negative run is not seamless, and thus does not really ``matter''.

\begin{proposition}
\label{increasing infix no negative cycles on x}
    Consider a separated increasing infix $uxyv$ from $p\in S$. In the notations of \cref{def:separated increasing infix from state}, for every $r\in B$ and $k\in \bbN$ we have $\minweight(x^k,r\to r)\ge 0$.
\end{proposition}
\begin{proof}
    Assume by way of contradiction that the claim does not hold, i.e., $\minweight(x^k,r\to r)<0$ for some $r\in B$ and $k\in \bbN$. Let 
    $\rho:r=r_0\runsto{x}r_1\runsto{x}r_2\cdots r_{k-1}\runsto{x}r_k=r$ be a run such that $\weight(\rho)<0$. In particular, there exists some $1\le i_0\le k$ such that $\minweight(x,r_{i_0-1}\to r_{i_0})<0$.
    
    By \cref{itm:separated infix partition} of \cref{def:separated increasing infix from state}, for every $0\le i\le k$ we have $r_i\in V_{j_i}$ for some $1\le j_i\le \ell$. 
    By \cref{prop:run characteristics of increasing infix from state}, we have $j_{i-1}\ge j_i$ for all $1\le i\le k$, since there are only runs from $V_j$ to ``lower'' parts.
    Since $r_0=r_k$, it therefore follows that all the states must be in the same part $V_j$.

    Let $k_{j,x}$ be the corresponding constant to $V_j$. Observe that for every $r',r''\in V_j$ it holds that $\minweight(x,r'\to r'')\ge \vec{c_u}(r'')-\vec{c_u}(r')+k_{j,x}$ (otherwise we would have $\vec{c_{ux}}(r'')<\vec{c_u}(r'')+k$).
    Thus, as a telescopic sum, and since $r_k=r_0=r$, we have
    \[
    \begin{split}
        &\weight(\rho)\ge \sum_{i=1}^k \minweight(x,r_{i-1}\to r_i)\ge  \sum_{i=1}^k (\vec{c_u}(r_i)-\vec{c_u}(r_{i-1})+k_{j,x}) \\
        &= \vec{c_u}(r_k)-\vec{c_u}(r_0)+k\cdot k_{j,x}=k\cdot k_{j,x}
    \end{split}
    \]
    This contradicts the negativity of $\rho$, and we are done.
\end{proof}
We now have that for every $k\in \bbN$, the minimal weight that can be attained by a run over $x^k$ cannot be too small (as the run cannot contain negative cycles). Specifically, we have the following.
\begin{corollary}
\label{cor:increasing infix no very negative runs on x k}
For every $k\in \bbN$ and run $\rho:B\runsto{x^k} B$ we have $\weight(\rho)\ge -\maxeff{x}|S|$.
\end{corollary}


%
We now lift \cref{def:separated increasing infix from state} to a set of states, obtaining the main definition of this section.
\begin{definition}[Separated Increasing Infix]
    \label{def:separated increasing infix}
    Consider a set of states $S'\subseteq S$ and a word $w$ decomposed as $w=u\cdot x\cdot y\cdot v$. 
    We say that \emph{$w$ is a separated increasing infix from $S'$} if $w$ is a separated increasing infix from every $p\in S'$ and in addition $(\ghostTrans(S',u),x)$ is a stable cycle.
\end{definition}

The next central lemma is key to our proofs, and captures our usage of separated increasing infixes. Intuitively, it states that given a separated increasing infix $u\cdot x\cdot y\cdot v$, we can replace it with $u\cdot \alpha_{\ghostTrans(S',u),x}\cdot v$, and then many ``nice'' properties hold. Specifically, the cost of the word becomes lower, its potential becomes higher, and it induces higher runs. A caveat to the lemma is that these guarantees might not hold, but in such a case we have a type-$0$ witness (this is based on \cref{lem:unfolding maintains potential}).

\begin{lemma}[\keyicon \lightbulbicon Increasing Infix to Cactus]
\label{lem:increasing infix budding} 
Either $\augA_\infty^\infty$ has a type-$0$ witness, or the following holds. 
Consider a separated increasing infix $S',w$ where $w=u\cdot x\cdot y\cdot v$ and $S'\subseteq S$. Write $B=\ghostTrans(S',u)$ and let  $w'=u\alpha_{B,x}v$, then we have:
\begin{enumerate}
    \item $\cost(w)>\cost(w')$.
    \item For every $r\in S',t\in S$ it holds that $\minweight(w,r\to t)\le \minweight(w',r\to t)$.
    \item For every $w_1,w_2\in (\Gamma_\infty^0)^*$ with $\booltrans(s_0,w_1)\subseteq S'$ we have $\pot(w_1ww_2)\le \pot(w_1w'w_2)$ (if the potential is defined).
\end{enumerate}
\end{lemma}
\begin{proof}
    The first requirement follows immediately from the \cref{def:separated increasing infix from state}: since $\cost(y)>2^{16(\bigM |S| \cost(x))^2}=\cost(\alpha_{B,x})$  (by \cref{def:cost depth and sub cactus}) so 
\[
\begin{split}
&\cost(w)=\cost(u)+\cost(x)+\cost(y)+\cost(v)>\\
&\cost(u)+\cost(x)+\cost(\alpha_{B,x})+\cost(v)\ge \cost(u)+\cost(\alpha_{B,x})+\cost(v)=\cost(w')
\end{split}
\]
The remainder of the proof is broken into two steps.
\paragraph{Step 1: From $y$ to $x^*$}
The first step is to replace $y$ with many iterations of $x$. Intuitively, we imagine that $x$ is already replaced with $\alpha_{B,x}$, and we unfold it.

Let $w_1,w_2\in (\Gamma_\infty^0)^*$ with $\booltrans(s_0,w_1)\subseteq S'$, and consider the word $\unfold(w_1u,\alpha_{B,x},vw_2 \wr F)=w_1ux^{2\bigM M_0}vw_2$ for some $F>2\maxeff{w_1u\alpha_{B,x}yvw_2}$.
We assume without loss of generality that $M_0$ is large enough so that for every state $r\in S'$ it holds that $ux^{2\bigM M_0}v$ is also an unfolding $\unfold(u,\alpha_{B,x},v \wr F)$ when starting\footnote{Note that the initial state $r$ is implicit in \cref{def:unfolding function}, but the unfolding may depend on it nonetheless.} from $r$, and that 
$M_0> \maxeff{xy}$. Note that $y$ does not appear in the latter unfolding -- this is deliberate, and makes the conclusion of the proof simpler.

We can assume these requirements on $M_0$ by \cref{rmk:increasing repetitions in unfolding}, which states that increasing the number of repetitions of $x$ retains all the properties of unfolding.

 We show an analogue of the properties of the lemma. That is, we show:
\begin{itemize}
    \item For every $r\in S',t\in S$ it holds that $\minweight(uxyv,r\to t)\le \minweight(ux^{2\bigM M_0}v,r\to t)$.
    \item $\pot(w_1 uxyv  w_2)\le \pot(w_1ux^{2\bigM M_0}v w_2)$ (if the potential is defined).
\end{itemize}
We start with the first property.
The intuitive idea is as follows. We repeat the infix $x$ many times. By the increasing infix property, for each $V_j$ in the partition, this either causes all states to increase, or stay in place. However, this is only true as long as the $V_j$ sub-configurations stays far from one another. At some point, it could be that e.g., $V_2$ rises a lot, while $V_3$ stays in place, and when the weights ``cross'', the states in $V_3$ now become minimal. 
Now, while state in $V_3$ grow when reading $x$, it may still be the case that $V_3$ generates \emph{negative} runs to states in $V_2$. Before pumping $x$, this had no effect. Now, however, it generates lower runs.

Nonetheless, these negative runs are ``local'', since $V_3$ cannot generate any negative cycle in $V_2$ (otherwise $V_2$ would have a negative cycle itself). Since the gap we start with is huge, we are guaranteed that even after these local decreasing runs, the overall configuration increases in comparison to the first $x$.
Moreover, even when removing $y$, the negative effect this might create is not large enough to overcome the gaps we create by pumping $x$, so we end up with a superior configuration. This idea is depicted in~\cref{fig:inc infix cactus unfolding} 

An important point to stress is that the pumped runs of the $V_j$ sets might indeed cross each other, but if this happens, it must be so far above the original configuration, that it maintains the overall increase.

\begin{figure}[ht]
    \centering
    \includegraphics[width=0.8\linewidth]{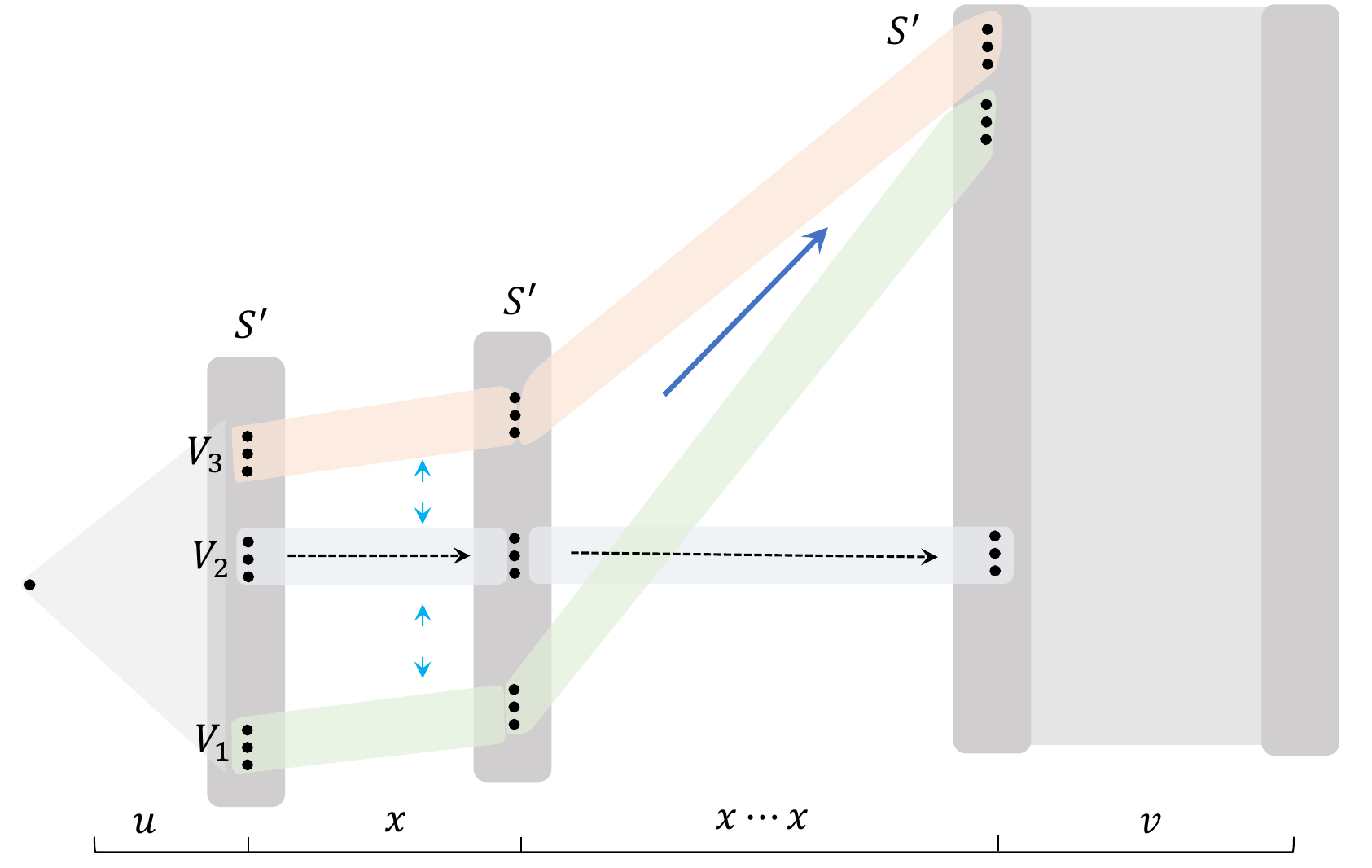}
    \caption{Replacing $y$ with $x^{2\bigM M_0}$. Visually, the increase seems ``steeper'', but it is the same slope as $x$, only much much longer. Also notice that runs from $V_1$ eventually cross those of $V_2$ (and so some runs might be overtaken by $V_2$), but this happens way above the original $V_1$ states.}
    \label{fig:inc infix cactus unfolding}
\end{figure}

We now turn this wild hand-waving to a precise argument. 
Let $r\in S'$. 
We start by analyzing runs on the words $uxy$ and $ux^{2\bigM M_0}$. Specifically, we show that 
$\xconf(\vec{c_r},uxy)\le \xconf(\vec{c_r},ux^{2\bigM M_0})$.
Using the notations of \cref{def:separated increasing infix from state} (with $r$ as the starting state), 
since $\supp(\vec{c_u})=\supp(\vec{c_{ux}})=\supp(\vec{c_{uxy}})=B$, then by induction we also have that $B= \supp(\vec{c_{ux^{k}}})$ for any $k\in \bbN$.

Consider some $r'\in B$, and assume $r'\in V_j$ in the partition given by \cref{itm:separated infix partition} of \cref{def:separated increasing infix from state}. 
By \cref{prop:run characteristics of increasing infix from state}, there exist $k_{j,x},k_{j,y}\in \bbN$ such that $k_{j,x}=0$ if and only if $k_{j,y}=0$ if and only if $k_{j,x}+k_{j,y}=0$. 

Let 
\[\rho:r\runsto{u}p_0\runsto{x}p_1\runsto{x}p_2\cdots p_{2\bigM M_0-1} \runsto{x} p_{2\bigM M_0}=r'\]
be a minimal-weight run. By \cref{prop:run characteristics of increasing infix from state}, for all $0\le i\le 2\bigM M_0$ we have that  $p_i\in V_{j'}$ for $j'\ge j$ 
(otherwise there is a transition from lower $V_{j''}$ to a higher one, which is a contradiction).
Denote $\vec{d_i}=\xconf(r,ux^i)$. We aim to show that restricted to $V_j$, the configuration $\vec{d_{2\bigM M_0}}$ is superior to $\vec{c_{uxy}}$.
Note that for every $0\le i<2\bigM M_0$ we have $\vec{d_{i-1}}\le \vec{d_i}$.
Indeed, by \cref{def:separated increasing infix from state} each $V_j$ can only (weakly) increase. This, however, is not enough to prove our claim, since once $i$ gets high enough, the gaps are no longer maintained and the $V_j$ sets may ``interfere'' with each other, and prevent certain states from strictly increasing. Then, removing $y$ may decrease things, which can be problematic.

We consider two cases.
\begin{itemize}
    \item If $p_i\in V_j$ for all $0\le i\le 2\bigM M_0$. In particular, the minimal run leading to each $p_i$ stems from $V_j$, when starting at $\vec{d_{i-1}}$. 
    We claim (by induction) that $\vec{d_i}(s)=\vec{c_u}(s)+ik_{j,x}$ for all $i$.
    
    The base case $i=0$ is trivial, since $\vec{d_0}=\vec{c_u}$. Assume correctness for $i$, we prove for $i+1$.
    By the induction hypothesis, we have that $\vec{d_i}(s)=\vec{c_u}(s)+ik_{j,x}$.
    By the observation above, we have that
    $\vec{d_{i+1}}(s)=\minweight_{\vec{d_i}}(x,V_j\to s)$. 
    Since this expression depends only on states in $V_j$, then by the above we can restrict attention to $V_j$, and therefore write
    \[\vec{d_{i+1}}(s)=\minweight_{\vec{c_u}}(x,V_j\to s)+ik_{j,x}=\vec{c_{ux}}(s)+ik_{j,x}=\vec{c_u}(s)+(i+1)k_{j,x}\]
    which completes the induction.

    Now, if $k_{i,j}=0$ then for every $s'\in V_j$ we have $\vec{d_{ 2\bigM M_0}}(s')=\vec{c_u}(s')=\vec{c_{ux}}(s')=\vec{c_{uxy}}(s')$ (since by \cref{def:separated increasing infix from state} we also have $k_{j,y}=0$). So the ``sub-configurations'' of $uxy$ and $ux^{2\bigM M_0}$ restricted to $V_j$ are equal. In particular, $\minweight(uxy,r\to r')=\minweight(ux^{2\bigM M_0},r\to r')$.
    
    If $k_{j,x}>0$, then by our choice that $M_0> \maxeff{xy}$ and the inductive claim, we have 
    \[
    \vec{d_{2\bigM M_0}}(r')> \vec{c_u}(r')+2\bigM \maxeff{xy} k_{j,x}\ge \vec{c_u}(r')+\maxeff{xy}
    \]
    However, by \cref{def:separated increasing infix from state} and \cref{prop:run characteristics of increasing infix from state} we have
    \[
    \vec{c_{uxy}}(r')\le \vec{c_u}(r')+k_{j,x}+k_{k,y}\le \vec{c_u}(r')+\maxeff{x}+\maxeff{y}=\vec{c_u}(r')+\maxeff{xy}
    \]
    and we conclude that restricted to $V_j$, the sub-configuration of $ux^{2\bigM M_0}$ is superior to that of $uxy$. In particular, $\minweight(uxy,r\to r')<\minweight(ux^{2\bigM M_0},r\to r')$.

    \item  The second case is when there exists $0\le i\le 2\bigM M_0$ such that $p_i\in V_{j'}$ for some $j'>j$. 
    Intuitively, in this case the run on $x^{2\bigM M_0}$ goes so high that it mingles with e.g., $V_{j+1}$. From there, it must remain much higher than any weight that $xy$ can attain starting from $V_j$ (and in particular the weight of reaching $r'$)
    
    By \cref{cor:increasing infix no very negative runs on x k}, for every $s\in V_{j''}$ with $j''>j$ and every $k\in \bbN$ we have $\minweight_{\vec{c_u}}(x^k,s\to B)\ge \vec{c_u}(s)-\maxeff{x}|S|$. 
    As observed above, for every $i'$ we have $\vec{d_{i'-1}}\le \vec{d_{i'}}$, and therefore we in particular have 
    \[\minweight_{\vec{d_i}}(x^{2\bigM M_0-i},p_i\to B)\ge \vec{c_u}(s)-\maxeff{x}|S|\]
    but by \cref{itm:separated infix gaps V} of \cref{def:separated increasing infix from state}, for the maximal state $s'\in V_j$ we have
    \[
        \begin{split}
        &\vec{c_u}(s)-\maxeff{x}|S| > \vec{c_u}(s')+\incinfGapConst -\maxeff{x}|S| \gg \\
        &\vec{c_u}(s')+\maxeff{xy}\ge \minweight_{\vec{c_u}}(xy,V_j\to B)
        \end{split}
    \]
    Since $p_0\runsto{xy}r'$ and $p_0\in V_j$, it follows that $\minweight(uxy,r\to r')< \minweight(ux^{2\bigM M_0},r\to r')$.

    We remark that this latter case cannot occur if $k_{j,x}=0$, but we do not use this fact.
\end{itemize}

Since the above works for every $j$ and every $r'$, we conclude that $\vec{c_{uxy}}\le \vec{d_{2\bigM M_0}}$. 
Therefore, we immediately have that
$\xconf(\vec{c_{uxy}},v)\le \xconf(\vec{d_{2\bigM M_0}},v)$.  Recall that this is all under $r$ as the initial state, and works for every $r\in S'$.
We can conclude the first property we want to prove, namely:
for every $r\in S',t\in S$ it holds that $\minweight(uxyv,r\to t)\le \minweight(ux^{2\bigM M_0}v,r\to t)$.

Obtaining the second item is now easy. Observe that since both $x$ and $y$ are cycles on $B$, i.e., $\booltrans(B,x)=\booltrans(B,y)=B$, then 
\[\supp(\xconf(\vec{c_{uxy}},v))= \supp(\xconf(\vec{d_{2\bigM M_0}},v))\]
Also, by the first property above, we can also prefix $w_1$ and still maintain configuration superiority. That is,
\[
    \xconf(\vec{c_{\init}},w_1uxyv)\le \xconf(\vec{s_{\init}},w_1ux^{2\bigM M_0}v)
\]
and again the configurations have the same support. 
This brings us under the conditions of \cref{lem:potential and charge are monotone} (two support-equivalent configurations with a superiority relation). Applying the lemma with the suffix $w_2$, we have
\[\pot(\xconf(\xconf(\vec{c_{\init}},w_1uxyv),w_2))\le \pot(\xconf(\xconf(\vec{c_{\init}}w_1ux^{2\bigM M_0}v),w_2))\]
And since we start from $\vec{c_{\init}}$, this simplifies to $\pot(w_1uxyvw_2)\le \pot(w_1ux^{2\bigM M_0}vw_2)$, as required.

\paragraph{Step 2: from $x^*$ to $\alpha_{B,x}$}
Using the previous step, we now need to prove the following.
\begin{itemize}
    \item For every $r\in S'$, $t\in S$ it holds that $\minweight(ux^{2\bigM M_0}v,r\to t)\le \minweight(u\alpha_{B,x}v,r\to t)$.
    \item $\pot(w_1ux^{2\bigM M_0}vw_2)\le \pot(w_1u\alpha_{B,x} vw_2)$.
\end{itemize}
Intuitively, we already repeat $x$ enough times so that all its behaviors are ``stabilized'', and can be replaced by $\alpha_{B,x}$.

Technically, all the work is already done by \cref{lem:unfolding maintains potential}, as follows.
Recall that $M_0$ is chosen so that for every state $r\in S'$ it holds that $ux^{2\bigM M_0}v$ is an unfolding $\unfold(u,\alpha_{B,x},v \wr F)$ with $F$ large enough. 
We are therefore within the conditions of \cref{lem:unfolding maintains potential}, with $r$ as the initial state. We then have that either $\augA_\infty^\infty$ has a type-0 witness (in which case we are done), or from Item 1 \cref{lem:unfolding maintains potential} in we get in particular that 
$\xconf(\vec{c_r},ux^{2\bigM M_0v})\le \xconf(\vec{c_r},u\alpha_{B,x}v)$ (by applying the lemma with the suffix $v$). This in means that for every $t\in S$ we have 
\[
\begin{split}
&\minweight(ux^{2\bigM M_0}v,r\to t)=\xconf(\vec{c_r},ux^{2\bigM M_0v})(t)\le\\ 
&\xconf(\vec{c_r},u\alpha_{B,x}v)=\minweight(u\alpha_{B,x}v,r\to t)
\end{split}
\]
concluding the first item, and combined with Step 1 concludes Item 2 of our proof.

Finally, to conclude the potential inequality we again apply \cref{lem:unfolding maintains potential}. This time, we recall that $M_0$ is also chosen such that $M_0=\unfold(w_1u,\alpha_{B,x},vw_2 \wr F)$ for $F$ large enough. Again we are within the scope of \cref{lem:unfolding maintains potential}, this time with the initial state $s_0$ (i.e., the ``real'' initial state). 
Thus, again either $\augA_\infty^\infty$ has a type-0 witness (in which case we are done), or $\pot(w_1ux^{2\bigM M_0}vw_2)=\pot(w_2u\alpha_{B,x}vw_2)$, and combined with the potential inequality of Step 1, we conclude Item 3 of our proof for the specific $w_1,w_2$ we start with.
Since the above works for every $w_1,w_2$, we are done. 

Note that while the proof requires us to first fix $w_1,w_2$ in order to choose $M_0$, eventually we end up with the word $u\alpha_{B,x}v$, which satisfies the claim regardless of $w_1,w_2$.
\end{proof}

\subsection{Existence of Separated Increasing Infixes}
\label{sec:existence of separated increasing infixes}
We now turn to show that separated increasing infixes exist. The overall approach is to analyze the run structure of words, combined with the Ramsey-style argument we show in \cref{thm: our ramsey}. 
In the remainder of the section, we fix a finite alphabet $\Gamma'\subseteq \Gamma_\infty^\infty$ and a corresponding constant $W_{\max}$ that is the maximal weight (in absolute value) appearing in any transition on the letters in $\Gamma'$. 

In the following we focus on words over $\Gamma'$ decomposed as $w_1w_2w_3w_4$. The intuition behind this decomposition is the following: $w_1$ is a prefix leading to some set of states. 
$w_2$ separates several outgoing runs from each state after $w_1$ to a large gap. 
$w_3$ is a very long infix that maintains a separation between runs (and induces a separated increasing infix in some form), and $w_4$ is a suffix. 
\begin{remark}[Implicit Decomposition to $(w_1,w_2,w_3,w_4)$]
    \label{rmk:w1w2w3w4 is a decomposition}
    When we write $w_1w_2w_3w_4$ we assume that the decomposition is explicit (instead of writing $(w_1,w_2,w_3,w_4)$). When we wish to consider the entire word $w_1w_2w_3w_4$ we explicitly mention that we consider it ``as a concatenation''.
\end{remark}
We start with several definitions regarding the runs structure of different parts of this decomposition, mainly focused around $w_3$.

\begin{definition}[Independent Runs]
    \label{def:independent run}
    Consider a word $w_1w_2w_3w_4$, a state $s\in \ghostTrans(s_0,w_1)$ and two seamless runs $\rho_1:s\runsto{w_2}p_1\runsto{w_3}q_2$ and $\rho_2:s\runsto{w_2}p_3\runsto{w_3}q_4$. We say that $\rho_1$ and $\rho_2$ are \emph{independent with respect to $s$ and $w_3$} if there exists $G\in \bbN$ such that for every prefix $u$ of $w_3$ we have $\weight(\rho^{w_2u}_2)-\weight(\rho^{w_2u}_1)\ge G$ where $\rho^{w_2u}_i$ is the prefix of $\rho_i$ on the word $w_2u$.

    In this case we also say that \emph{$\rho_2$ is above $\rho_1$ with gap $G$} (we also use ``not within gap $G$'' when the order is not important).

    Similarly, a set of runs $\{\rho_1,\ldots,\rho_\ell\}$ from $s$ on $w_2w_3$ is \emph{independent with gap $G$} if its runs are pairwise-independent with gap $G$ (a singleton set $\{\rho_1\}$ is independent by definition).
\end{definition}
Observe that two seamless runs with distinct weights in some prefix, must be state-disjoint from that prefix and on (otherwise at least one of them is not seamless). It follows that for a state $s$, every set of independent runs is of size at most $|S|$. 

Our overall approach in the next few sections is to ``untangle'' runs on words in order to obtain a very structured behavior. To do so, we consider a sequence of words $w_1w_2w_3w_4$ with $w_3$ increasing in length. This is captured by our first definition.

\begin{definition}[\keyicon Elongating Words Sequence]
    \label{def:elongated words sequence}
    A function from $\bbN$ to $\Gamma'^*$ is an \emph{elongating words sequence}, denoted $\words:\bbN\to \Gamma'^*$, if for every $m\in \bbN$ we have $\words(m)=w_1w_2w_3w_4$ such that $w_1w_2w_3w_4$ has a seamless baseline run and $|w_3|\ge m$. 

    For elongating words sequences $\words,\words'$, we say that $\words'$ is a \emph{faithful restriction} of $\words$ if for every $m\in \bbN$ there exists $t\in \bbN$ with $\words'(m)=w'_1w'_2w'_3w'_4$ and $\words(t)=w_1w_2w_3w_4$ such that $w'_1=w_1$ and $w'_1w'_2w'_3w'_4=w_1w_2w_3w_4$ as a concatenation.
\end{definition}
Intuitively, a faithful restriction $\words'$ of $\words$ only returns ``similar'' words to $\words$, but $\words$ may return things that are not returned by $\words'$. We typically use this when we modify $\words$ to obtain other function in a way that is not too ``disruptive''.

Since $|w_3|\ge m$ in an elongated words sequence, we can obtain a faithful restriction of it by truncating each $w_3$ to be of length exactly $m$, and assigning the suffix of $w_3$ to $w_4$. For convenience, we often make this assumption in this section (but drop it in following sections). We explicitly define it as follows.
\begin{definition}[Exact Words Sequence]
\label{def:exact word sequence}
An elongated words sequence is \emph{exact}, denoted $\xwords$, if for every $m\in \bbN$ and $\xwords(m)=w_1w_2w_3w_4$ we have $|w_3|=m$.
\end{definition}

We now specify two important possible properties of elongating word sequences. The first concerns the gaps between independent runs.

\begin{definition}[$\ell$-Sparse]
    \label{def:sparse words}
    For $\ell\in \bbN$, an exact word sequence $\xwords$ is \emph{$\ell$-sparse} if for every $m\in \bbN$ with $\xwords(m)=w_1w_2w_3w_4$ and     
    for every $p\in \ghostTrans(s_0,w_1)$ there is a set $I_p$ of independent runs with gap $\sparseGapBound$ such that $|\bigcup_{p}I_p|= \ell$.
    \acctodo{Changed -- Main thing}
\end{definition}
Intuitively, an $\ell$-sparse function returns words where the gaps between independent runs are large (and grow with $m$). We assume that when $\xwords$ is $\ell$-sparse, the sets $I_p$ are returned together with the word $\xwords(m)$.

Notice that the requirement on $\ell$ is that all words have the same number of independent runs. We show that in fact it is sufficient that infinitely many words have $\ell$ independent runs.

\begin{proposition}
    \label{prop:faithful restriction of words infinitely many ells}
    Let $\xwords$ be an exact word sequence such that there are infinitely many $m\in \bbN$  with $\xwords(m)=w_1w_2w_3w_4$ and     
    for every $p\in \ghostTrans(s_0,w_1)$ there is a set $I_p$ of independent runs with gap $\sparseGapBound$ such that $|\bigcup_{p}I_p|= \ell$.
    Then there is a faithful restriction $\xwords'$ of $\xwords$ that is $\ell$-sparse.
\end{proposition}
\begin{proof}
    Consider $m\in \bbN$ and let $\xwords(m)=w_1w_2w_3w_4$. 
    For every $n<m$, consider the word $w'_1w'_2w'_3w'_4$ where $w'_1=w_1$, $w'_2=w_2$, $w'_3=w_3[1,n]$ and $w'_4=w_3[n+1,m]w_4$ (i.e., we move letters from $w_3$ to $w_4$ so that $|w'_3|=n$). 
    Clearly we can replace $\xwords(n)$ by $w'_1w'_2w'_3w'_4$ and obtain a faithful restriction. Moreover, the sparsity requirement holds in $w'_1w'_2w'_3w'_4$, since the independent runs on $w_3$ can be shortened to independent runs in $w'_3$, and the gap requirement for $m$ is stricter than for $n<m$.

    Thus, if there are infinitely many $m$ for which the sparsity requirement holds, we can ``fill the gaps'' between these values of $m$ by shortening the respective $w_3$'s.
\end{proof}

The next property concerns words that have a long infix in which no run decreases too much.
\begin{definition}[\keyicon $D$-Dip]
    \label{def:dip words}
    For $D\in \bbN$, a word $w_1w_2w_3w_4$ is \emph{$D$-dip} if for every decomposition $w_3=xyz$, state $p\in \ghostTrans(s_0,w_1w_2x)$ and run $\rho:p\runsto{y}S$ it holds that $\weight(\rho)>-D$.
    
    A function $\words$ is \emph{$D$-dip} if $\words(m)$ is $D$-dip for every $m\in \bbN$.
\end{definition}
Intuitively, a word $w_1w_2w_3w_4$ is a dip word if every infix of $w_3$ induces ``dips'' (i.e. decreasing runs) losing at most $D$ weight ($D$ stands for ``Dip'', also in the following). 
Thus, the overall trend of runs on $w_3$ in $\words$ is non-decreasing, with some small exceptions.

Our final definition for this section concerns the states that are not part of the independent runs. Specifically, we measure whether such states are close to independent runs prescribed by the $\ell$-sparse function. 
If all states are within gap $G$ of an independent run, we say that the function is a \emph{$G$-cover}.
\begin{definition}[State $G$-Cover]
    \label{def:G cover inc infix}
    Consider $G\in \bbN$ and an $\ell$-sparse function $\words$. 
    We say that $\words$ is a \emph{$G$-cover} if for every $m\in \bbN$ the word $\words(m)=w_1w_2w_3w_4$ with the corresponding set of independent runs $\{I_p\}_{p\in \ghostTrans(s_0,w_1)}$ satisfies the following: for every $p\in \ghostTrans(s_0,w_1)$ and prefix $x$ of $w_3$, consider the configuration $\vec{c_{w_2x}}=\xconf(\vec{c_p},w_2x)$, then for every $q\in \supp(\vec{c_{w_2x}})$ we have that $\vec{c_{w_2x}}(q)$ is within gap $G$ of a state in one of the independent runs in $I_p$ (at the prefix $w_2x$).
\end{definition}

\subsubsection{Existence of Increasing Infixes Assuming A $G$-Cover}
\label{sec:existence of inc inf assuming covered}
Fix some dip level $D$. We now turn to analyze the setting where a $G$-cover exists, and we show that under this assumption, a separated increasing infix exists. This is the first instance where we use a mild version of our \emph{Zooming} technique (\cref{sec:abs:zooming}). 
\begin{lemma}[\keyicon \lightbulbicon From $D$-Dip and Cover to Increasing Infix]
    \label{lem:unbounded long dip and G imply separated inc infix}
    Let $D,\ell,G\in \bbN$ such that there exists a function $\xwords$ that is $\ell$-sparse, $D$-dip and a $G$-cover. 
 
    Then there exists some $m\in \bbN$ with $\xwords(m)=w_1w_2w_3w_4$ and a decomposition $w_3=u'\cdot x\cdot y\cdot z'$ such that $u\cdot x\cdot y\cdot z$ is a separated increasing infix from 
    $\ghostTrans(s_0,w_1)$
    for $u=w_2\cdot u'$ and $z=z'\cdot w_4$.
\end{lemma}
We prove the lemma in the remainder of the section.
Let $G\in \bbN$ as assumed in the lemma. 
Our first step is to define coloring functions on infixes of infinitely many $w_3$ words, so that we are in our Ramsey framework (\cref{sec:ramsey}). 

We again start with some definitions. Consider a state $p$ and the configuration $\vec{c_p}$ that assigns $0$ to $p$ and $\infty$ elsewhere. 
Motivated by increasing infixes, we consider two characteristics of infixes.
The first characteristic, dubbed $\difftype_{p,G}(x)$, stores the order of the states in the configuration reached from $\vec{c_p}$ by $x$, and also stores which states are within distance $G$ from each other.

\begin{definition}[$\difftype$]
\label{def:difftype}
    For a constant $G\in \bbN$ and a word $x$, let $\vec{c_x}=\xconf(x,\vec{c_p})$. 
    The \emph{$G$-capped difference type} of $x$ from $p$ is 
    $\difftype_{p,G}(x)=\tup{\supp(\vec{c_x}),\le_x,\near_{p,G}(x)}$ where $\le_x\subseteq \supp(\vec{c_x})$ is a partial order on the reachable states defined by $s_1\le_x s_2$ if $\vec{c_x}(s_1)\le \vec{c_x}(s_2)$, and 
    \[\near_{x,G}=\{(s_1,i,s_2)\mid i=|\vec{c_x}(s_2)-\vec{c_x}(s_1)|\in \{0,\ldots,G\}\}\]
    We denote the set of all possible $\difftype$ by $\difftypeset_{p,G}=\{\difftype_{p,G}(x)\mid x\in (\Gamma_\infty^\infty)^*\}$
\end{definition}

Orthogonally to $\difftype$, our second characteristic measures the weight gained along a run, capped between two specific valued.
\begin{definition}[$\gaintype$]
    For a constant $D\in \bbN$ and words $u,v$, we define the \emph{$D$-capped gain type} of $(u,v)$ as $\gaintype_D(u,v)=\{(s_1,i,s_2)\in S\times \{-D,\ldots,D\cdot |S|\}\times S\mid s_1\in \ghostTrans(s_0,u)\wedge \exists \rho:s_1\runsto{v}s_2.\ \weight(\rho)=i\}$.
    
    We denote the set of all possible $\gaintype$ by $\gaintypeset_{D}=\{\gaintype_{D}(u,v)\mid u,v\in (\Gamma_\infty^\infty)^*\}$
\end{definition}
That is, the gain type consists of triples of the form $(s_1,i,s_2)$ where $s_1$ is reachable from $s_0$ via $u$, and there is a run that gains weight $i\in \{-D,\ldots,D\cdot |S|\}$ from $s_1$ to $s_2$.
Notice that here we allow a small decrease of at most $-D$, but no further, and we keep track of an increase up to $D\cdot |S|$.

Crucially, note that for a given $p,G$ and $D$, both $\difftypeset_{p.G}$ and $\gaintypeset_D$ are finite.

Recall that for increasing values of $m$, the length of the $w_3$ infix of $\xwords(m)$ also increases. 
We construct an infix-colored set (as in \cref{sec:ramsey}) as follows.
%
%
For every $n\in \bbN$, denote $\xwords(n)=w^n_1w^n_2w^n_3w^n_4$ and note that since $|w^n_3|=n$, then these words are unique\footnote{A nit-picking note: the word $w^n_1w^n_2w^n_3w^n_4$ is unique with respect to this particular decomposition (see \cref{rmk:w1w2w3w4 is a decomposition}. It may be equal to some other $w^k_1w^k_2w^k_3w^k_4$ as a concatenation.}

We now construct the infix-colored set $C=\{(w_3^n,\lambda_n)\}_{n\in \bbN}$ as follows. Our color set consists of tuples of a $\gaintype$ and two $\difftype$ for each reachable state (and a placeholder $\bot$ for unreachable states). More precisely, the color set is
\[
\begin{split}
\colset=\{(\chi,(\tau_{p,1},\tau_{p,2})_{p\in S})\mid \chi\in \gaintypeset_{D} 
\wedge \forall p\in S.\ \tau_{p,1},\tau_{p,2}\in \difftypeset_{p,G}\cup \{\bot\}\}
\end{split}
\]

For every $n\in \bbN$ we define a coloring function $\lambda_n:\interval_{|w_3^n|}\to \colset$ as follows: let $1\le i\le j\le |w_3^n|$ (i.e., $(i,j)\in \interval_{|w_3^n|}$ in the notation of \cref{sec:ramsey}), let $u=w_3[1, i-1]$ and $v=w_3[i,j]$ then 
$\lambda_n(i,j)=(\chi,(\tau_{p,1},\tau_{p,2})_{p\in S})$ where 
$ \chi = \gaintype_D(w_1w_2u,w_1w_2uv)$ and for every $p\in S$ we have \[(\tau_{p,1},\tau_{p,2})=(\difftype_{p,G}(w_2u),\difftype_{p,G}(w_2uv)) \text { if } p\in \ghostTrans(s_0,w_1)\]
 and $(\tau_{p,1},\tau_{p,2})=(\bot,\bot)$ otherwise.

Intuitively, each infix of $w_3$ is colored by three elements: the gain difference upon reading the infix, and the configurations before and after reading the infix for each state $p$ reachable after reading $w_1$ (including ghost states), described by their $G$-capped difference, so that we only have finitely many colors.
Recall that the assumption of the lemma that $\xwords$ is a $G$-cover states that every run from $p$ on $w_2x$ remains within gap $G$ of an independent run in $I_p$. Tracking $\difftype_{p,G}$ therefore suffices to track all relevant information about the runs from $p$.

We now invoke \cref{thm: our ramsey} with the letter-cost function $\cost$ of \cref{def:cost depth and sub cactus} and with the 
prefix-cost function $\precost(x)=2^{16(\bigM |S| \cost(x))^2}$.
(c.f., Item \textit{1} of \cref{def:separated increasing infix}). 
\acctodo{Changed}
We obtain that for every $L\in\bbN$ there exists a pair $(w_3',\lambda')$
and an index set $1\le i_0<\ldots <i_L\le |w_3'|+1$ such that there are infinitely many pairs $(w_3^n,\lambda_n)$ in our infix-colored set $C$ where $w_3'$ is a prefix of $w_3^n$ and the conditions of \cref{thm: our ramsey} hold. 

Fix $L>D\cdot |S|$.  We take such a pair $(w_3,\lambda)$ for our chosen value of $L$, and without loss of generality, we assume $(w_3,\lambda)$ is equal to an arbitrarily long $(w_3^n,\lambda_n)$ (by elongating it if needed, as it is a prefix of infinitely many such $(w_3^n,\lambda_n)$). 
In particular, we take $w_3$ long enough such that $\xwords(m)=w_1w_2w_3w_4$ for 
$m>2G$.

Denote $u'=w_3[1,i_0-1]$, $x=w_3[i_0,i_1-1]$, $y=w_3[i_1,i_L-1]$, and $z'=w_3[i_L,|w_3|]$, so that $w_3=u'xyz'$. Set $u=w_2\cdot u'$, $z= z'w_4$ and let $S'=\ghostTrans(s_0,w_1)$. 
We illustrate this word in \cref{fig:inc inf word decomposition}

\begin{figure}[ht]
    \centering
    \includegraphics[width=1\linewidth]{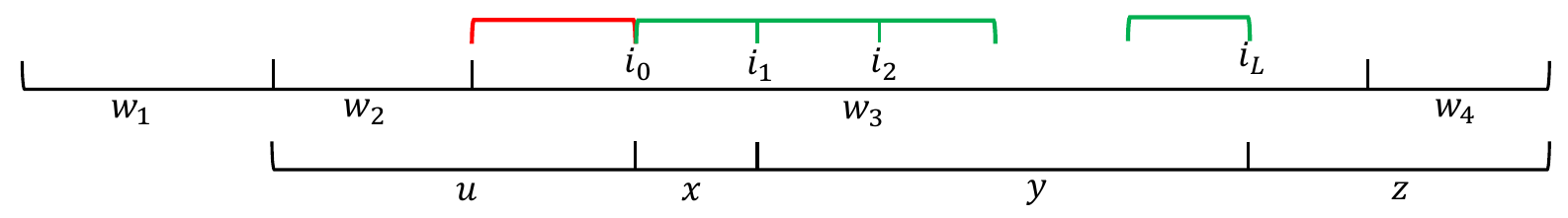}
    \caption{The increasing infix found by the colored Ramsey decomposition. The red and green intervals represent the coloring of the segments in the decomposition of $w_3$.}
    \label{fig:inc inf word decomposition}
\end{figure}

We claim that $uxyz$ is a separated increasing infix from $S'$, which would then conclude the proof.
To this end, we show each of the requirements of being a separated increasing infix (\cref{def:separated increasing infix from state,def:separated increasing infix}).

First, by our choice of 
the prefix-cost $\precost$
and by Item \textit{3} of \cref{thm: our ramsey} we have that 
$\cost(y)>2^{16(\bigM |S| \cost(x))^2}$
thus satisfying Requirement \textit{1} of \cref{def:separated increasing infix from state}.
\acctodo{Changed}

\begin{proposition}
\label{prop:exists inc inf Req 2}
Requirement \textit{2} of \cref{def:separated increasing infix from state} holds.

That is, denote $\vec{c_u}=\xconf(\vec{c_p},u)$, $\vec{c_{ux}}=\xconf(\vec{c_p},ux)$, and $\vec{c_{uxy}}=\xconf(\vec{c_p},uxy)$, then $\supp(\vec{c_u})=\supp(\vec{c_{ux}})=\supp(\vec{c_{uxy}})$.    
\end{proposition}
\begin{proof}
By Item \textit{2} of \cref{thm: our ramsey} we have that $\lambda(x)=\lambda(y)$ (we slightly abuse notation and color the infix word instead of the index pair). By the choice of our coloring function, this implies that
$\gaintype_D(w_1u,w_1ux)=\gaintype_D(w_1ux,w_1uxy)$ and for every $p\in \ghostTrans(s_0,w_1)$ we have $\difftype_{p,G}(w_1w_2u)=\difftype_{p,G}(w_1w_2ux)$ and $\difftype_{p,G}(w_2ux)=\difftype_{p,G}(w_2uxy)$. The latter two equalities therefore imply 
\[\difftype_{p,G}(w_2u)=\difftype_{p,G}(w_2ux)=\difftype_{p,G}(w_2uxy)\]
In particular, by \cref{def:difftype} we immediately get that the first components are equal, i.e.,  $\supp(\vec{c_u})=\supp(\vec{c_{ux}})=\supp(\vec{c_{uxy}})$.
\end{proof}

We now address Requirement \cref{itm:separated infix seamless baseline} from \cref{def:separated increasing infix from state}: 
let $p$ be a baseline state in $\ghostTrans(s_0,w_1)$, we want to show that the baseline run $\rho:p\runsto{uxyz}q$ (whose weight is constantly $0$) is seamless from $\vec{c_p}$. 
This, however, follows immediately from \cref{def:elongated words sequence} -- since $w_1w_2w_3w_4= \xwords(m)$, then its baseline run is seamless. Since $w_1w_2w_3w_4=w_1uxyz$ (as concatenations), then the baseline run $\rho:s_0 \runsto{w_1}p\runsto{uxyz} q$ is seamless, and in particular its infix $\rho:p\runsto{uxyz}q$ is seamless when starting from $\vec{c_p}$.

We turn to the more involved requirements, namely Requirement \textit{3} of \cref{def:separated increasing infix from state}, and showing that $(\ghostTrans(S',u),x)$ is a stable cycle, as per \cref{def:separated increasing infix}. 

Let $p\in \ghostTrans(s_0,w_1)$ and let $I_p$ be a set of independent runs as assumed by the $\ell$-sparsity of $\xwords$. 
From \cref{def:sparse words}, the gap between the independent runs in $I_p$ is at least $\sparseGapBound$, and since $|w_3|= m> 2G$, the gap is (much) greater 
than $2G$.

Consider a configuration $\vec{c}$ and states $p,q\in S$, we write $p\sim_{\vec{c}} q$ if $|\vec{c}(p)-\vec{c}(q)|\le 2G$. We start by claiming that for every prefix $v$ of $xy$, for $\vec{c_{uv}}=\xconf(\vec{c_p},uv)$ the relation $\sim_{\vec{c_{uv}}}$ is an equivalence relation.
Indeed, it is trivially reflexive and symmetric. Moreover, if $|\vec{c_{uv}}(p)-\vec{c_{uv}}(q)|\le 2G$, then by 
the assumption that $\xwords$ is a $G$-cover (\cref{def:G cover inc infix}),
there exists some states $r,r'$ that are along independent runs in $I_p$ satisfying $|\vec{c_{uv}}(p)-\vec{c_{uv}}(r)|\le G$ and $|\vec{c_{uv}}(q)-\vec{c_{uv}}(r')|\le G$. 
However, since the gap between independent runs in $I_p$ is much greater than $2G$, 
it follows that $r=r'$. 
That is, two equivalent states are within gap $G$ of a corresponding state in an independent run. In particular, this entails the transitivity of $\sim_{\vec{c_{uv}}}$.

Consider $B=\supp(\vec{c_u})=\supp(\vec{c_{ux}})=\supp(\vec{c_{uxy}})$, we partition $B$ according to the equivalence classes of these configurations, as follows.
Let $I_p=\{\rho_1,\ldots, \rho_k\}$ be the independent runs from $p$, and assume that $\rho_{i+1}$ is above $\rho_{i}$ for all $1\le i<k$. 
Denote by $\rho_i(v)$ the state of $\rho_i$ after reading the prefix $v$.
For $\xi\in \{u,ux,uxy\}$ we define a partition $B=V^\xi_1\cup \ldots \cup V^\xi_k$ by setting $V^\xi_i=\{q\mid q\sim_{\vec{c_{\xi}}} \rho_i(\xi)\}$, i.e., the set of states that are equivalent to $\rho_i(\xi)$. We now show that these partitions are not dependent on $\xi$.

\begin{proposition}
\label{prop:exists inc inf partitions}
The partitions induced for each $\xi\in \{u,ux,uxy\}$ are identical. That is, $V^{\xi}_i=V^{\zeta}_i$ for every $\xi,\zeta\in \{u,ux,uxy\}$.
\end{proposition}
\begin{proof}
    Recall that we obtain $x,y$ using our Ramsey argument, and due to its coloring we have
    \[\difftype_{p,G}(w_2u)=\difftype_{p,G}(w_2ux)=\difftype_{p,G}(w_2uxy)\]
    First, observe that $\sim_{\vec{c_{\xi}}}$ is the same relation for all $\xi$. Indeed, the relation $\sim_{\vec{c_{\xi}}}$ only depends on $\vec{c_{\xi}}$, and the way we obtained our observation above (that $\sim_{\vec{c_{uv}}}$ is an equivalence relation for all prefixes $v$ of $xy$) in fact shows that for every $\xi$, the states corresponding to $\{\rho_i(\xi)\}_{i=1}^k$ are representatives of the $k$ equivalence classes.

    Thus, all that remains to be proved is that there is no permutation in the order of the $V_i$ between the infixes. 
    That is, consider e.g., $V^u_1$ and $V^{ux}_1$. These sets are equivalence classes of $\sim_{\vec{\xi}}$ (which as we just showed is independent of $\xi$). However, it is not a-priori guaranteed that $V^u_1=V^{ux}_1$, rather it could be that e.g., $V^u_1=V^{ux}_3$.
    Intuitively, however, this cannot happen because the equivalence of $\difftype$ implies in particular an equivalent ordering of the states (i.e., $\le_u,\le_{ux}$ and $\le_{uxy}$, as per \cref{def:difftype}).
    We show the equivalence of $V^u_i$ and $V^{ux}_i$, the proof for $V^{ux}_i$ and $V^{uxy}_i$ is word-for-word identical, up to replacing $u$ by $ux$ and $ux$ by $uxy$.

    Assume by way of contradiction that there exists some permutation $\pi$ of $\{1,\ldots,k\}$ such that $\pi$ is not the identity permutation, and $V^u_i=V^{ux}_{\pi(i)}$ for all $i\in \{1,\ldots, k\}$.
    Since $\pi$ is not the identity, it has an inversion -- some $1\le j<h\le k$ such that $\pi(j)>\pi(h)$. 
    Let $q\in V^u_j$ and $r\in V^u_h$, then $q\sim_{\vec{c_\xi}} \rho_j(u)$ and $r\sim_{\vec{c_\xi}} \rho_h(u)$, and therefore $q\le_u r$ (since $\rho_j(u)\le_u \rho_h(u)$ and the gap between them is much greater than $2G$),
    \acctodo{Should be ok}
    and moreover there is a strict inequality in the weights: $\vec{c_u}(q)<\vec{c_u}(r)$.
    However, we then have $q\in V^{ux}_{\pi(j)}$ and $r\in V^{ux}_{\pi(h)}$, and with the same argument, namely $\rho_{\pi(h)}(ux)\le_{ux} \rho_{\pi(j)}(ux)$, we have $r\le_{ux} q$ (again with a strict inequality in weights).

    This is in contradiction to the fact that $\le_u\equiv \le_{ux}$, which concludes the proof.
\end{proof}

We are now ready to show that Item \cref{itm:separated infix gaps V} of  \cref{def:separated increasing infix from state} holds, for the partition $V_1\cup\ldots \cup V_k$ (we can omit the superscript due to \cref{prop:exists inc inf partitions}). 
To this end, consider $q\in V_j$ and $r\in V_{j+1}$ for some $1\le j<k$, we claim that for every $\xi\in \{u,ux,uxy\}$ it holds that 
$\vec{c_\xi}(r)-\vec{c_\xi}(q)>\incinfGapConst$. 

Recall again that the gap between runs in $I_p$ is at least $\sparseGapBound$ and that $|w_3|=m>2G$ and $|w_3|\ge |xy|$ (since $xy$ is an infix of $w_3$). Note that $\maxeff{xy}\le |xy|W_{\max}$, then we have that 
\begin{equation}
\label{eq:sparse gap lower bound}
    \begin{split}
&\sparseGapBound= 8\bigM(|S|+1)(mW_{\max})^2+ 8\bigM(|S|+1)(mW_{\max})^2\\
&> 8\bigM(|S|+1)(mW_{\max})^2 + 16G \ge 8\bigM(|S|+1)(|xy|W_{\max})^2 + 16G\ge \\
&8\bigM(|S|+1)\maxeff{xy}^2 + 16G
\end{split}
\end{equation}
\acctodo{Changed}
%
However, we have $|\vec{c_\xi}(q)-\vec{c_{\xi}}(\rho_{j}(\xi))|\le G$ and $|\vec{c_\xi}(r)-\vec{c_{\xi}}(\rho_{j+1}(\xi))|\le G$ (i.e., both $q$ and $r$ are close to their respective independent runs, as we observe above). Combining these inequalities (with coarse lower bounds) we conclude the claim:
\[\begin{split}
&\vec{c_\xi}(r)-\vec{c_\xi}(q)>8\bigM(|S|+1)\maxeff{xy}^2 + 16G - 2G>\incinfGapConst
\end{split}
\]
\acctodo{Changed}

We turn to show part \textit{(b)} in Requirement \textit{3}. That is, we need to show that for every $1\le j\le k$ there exist $k_{j,x},k_{j,y}$ such that for every $r\in V_j$ we have $\vec{c_{ux}}(r)=\vec{c_{u}}(r)+k_{j,x}$ and $\vec{c_{uxy}}(r)=\vec{c_{ux}}(r)+k_{j,y}$, and that in addition $k_{j,x}=0$ if and only if $k_{j,y}=0$.
The existence of \emph{integer} constants is easily implied by the equivalence of $\difftype$ and by \cref{prop:exists inc inf partitions}. The harder part here is showing that these constants are non-negative.

Recall our set of indices $1\le i_0<\ldots<i_L$ such that $x=w_3[i_0,i_1-1]$ and $y=w_3[i_1,i_L-1]$. We decompose $y$ to $y=y_1y_2\cdots y_{L-1}$ such that $y_j=[i_j,i_{j+1}-1]$. These infixes are obtained from \cref{thm: our ramsey}, and are therefore colored identically, and also identically to $x$. Also recall that $L>D\cdot |S|$.
Since the analysis carried out in \cref{prop:exists inc inf partitions} only relies on the equivalence in $\difftype$, we can conclude that for every $1\le t< L$, the partition induced by $\sim_{\vec{c_{ux\cdot y_1\cdots y_t}}}$ is also $V_1\cup\cdots \cup V_k$, with the same ordering as above.

Let $1\le j\le k$ and consider the ``minimal'' state $q\in V_j$, i.e., $\vec{c_{u}}(q)=\min\{\vec{c_{u}}(r)\mid r\in V_j\}$. Note that $q$ is also minimal after $ux$ and after $uxy_1\cdots y_t$ for every $t$, since the ordering of the states in the corresponding configurations are the same (due to the equivalence in $\difftype$).
 By the minimality of $q$, for every $r\in V_j$ we have $0\le \vec{c_{u}}(r)-\vec{c_{u}}(q)\le 2G$. By the equivalence in $\difftype$ (specifically, since the gaps are tracked up to size $2G$), we also have for every $\xi\in \{u,ux,uxy_1,\ldots, uxy_1\cdots y_{L-1}\}$ that $\vec{c_{\xi}}(r)-\vec{c_{\xi}}(q)=\vec{c_{u}}(r)-\vec{c_{u}}(q)$. By rearranging, we have
 $\vec{c_{\xi}}(r)=\vec{c_{u}}(r)+(\vec{c_{\xi}}(q)-\vec{c_{u}}(q))$. In particular, by setting $k_{j,x}=(\vec{c_{ux}}(q)-\vec{c_{u}}(q))$ (and similarly $k_{j,y}$), we have the result for \emph{integer} constants. It remains to prove the non-negativity of the constant, and the special condition for $0$.

To this end, we start by observing that for any infix $\xi$ as above and state $q\in V_j$, there are no runs leading to $q$ from any state in $V_{i}$ with $i<j$, and there is a minimal run to $q$ from some state in $V_j$ (as well as possibly-minimal runs from ``above'' $V_j$). We prove this in two parts.

Intuitively, if $i<j$, then $\rho_i$ is far below $\rho_j$ along all of $uxy$, and therefore no state in $\rho_i$ can yield a run to a corresponding state of $\rho_j$, as that would ``drag down'' $\rho_j$ too close to $\rho_i$. By extension, no state that is equivalent to $\rho_i$ can reach a state that is equivalent to $\rho_j$.
We formalize this intuition, the approach is illustrated in \cref{fig:independent run gaps proof}.

\begin{figure}[ht]
    \centering
    \includegraphics[width=0.5\linewidth]{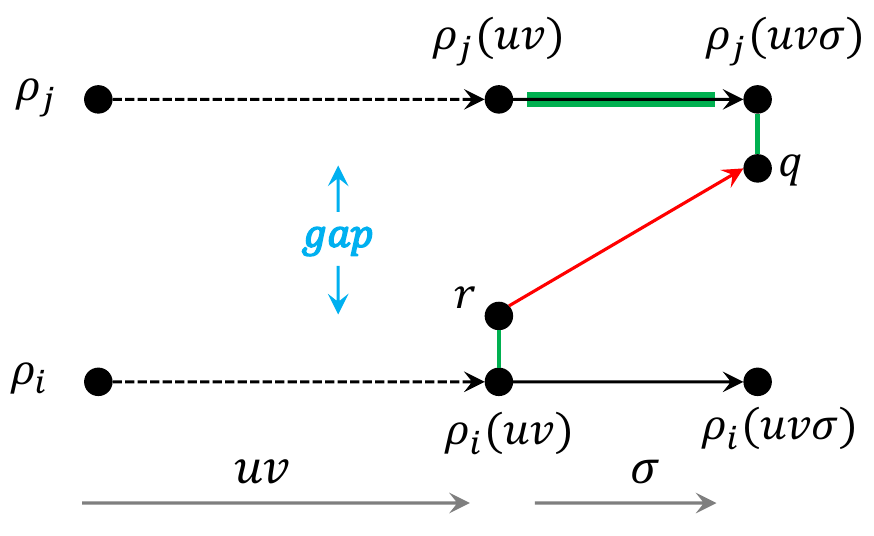}
    \caption{The red arrow is the assumption by contradiction. The green lines signify ``small'' differences. If the red arrow is a run, then it has a small gap, which pulls together $\rho_j$ and $\rho_i$, leading to a contradiction.}
    \label{fig:independent run gaps proof}
\end{figure}

\begin{proposition}
    \label{prop:exists inc inf independent runs no transitions}
    Let $\sigma\in \Gamma'$ be the next letter read after a prefix $v$ of $xy$, let $i<j$ and let $r,q\in S$ with $r\sim_{\vec{c_{uv}}}\rho_i(uv)$ and $q\sim_{\vec{c_{uv\sigma}}}\rho_j(uv\sigma)$, then there is no transition $(r,\sigma,c,q)\in \augTrans_\infty^\infty$ with $c<\infty$.
\end{proposition}
\begin{proof}
Let $r,q\in S$ with $r\sim_{\vec{c_{uv}}}\rho_i(uv)$ and $q\sim_{\vec{c_{uv\sigma}}}\rho_j(uv\sigma)$, and assume by way of contradiction that there is a transition $(r,\sigma,c,q)\in \augTrans_\infty^\infty$ with $c<\infty$ (and in particular $c\le W_{\max}$). It follows that $\vec{c_{uv\sigma}}(q)-\vec{c_{uv}}(r)\le W_{\max}$. Further assume that $r$ is the minimal state with this property with respect to $\vec{c_{uv}}$. It follows that $|\vec{c_{uv\sigma}}(q)-\vec{c_{uv}}(r)|\le W_{\max}$ (intuitively, we are merely avoiding the case where the absolute value grows because $q$ is dragged even further down below $r$, which would only simplify things for us).

Observe that $|\vec{c_{uv}}(\rho_j(uv))-\vec{c_{uv\sigma}}(\rho_j(uv\sigma))|\le W_{\max}$, as this is a single transition in a seamless run.
We then have
\[
\begin{split}
    &|\vec{c_{uv}}(\rho_j(uv))-\vec{c_{uv}}(\rho_i(uv))|=\\
    &|\vec{c_{uv}}(\rho_j(uv))-\vec{c_{uv\sigma}}(\rho_j(uv\sigma))+\vec{c_{uv\sigma}}(\rho_j(uv\sigma))-\vec{c_{uv}}(\rho_i(uv))|\le \\
    &|\vec{c_{uv}}(\rho_j(uv))-\vec{c_{uv\sigma}}(\rho_j(uv\sigma))|+|\vec{c_{uv\sigma}}(\rho_j(uv\sigma))-\vec{c_{uv}}(\rho_i(uv))|\le \\
    & W_{\max}+ |\vec{c_{uv\sigma}}(\rho_j(uv\sigma))-\vec{c_{uv}}(\rho_i(uv))|=\\
    & W_{\max}+ |\vec{c_{uv\sigma}}(\rho_j(uv\sigma))-\vec{c_{uv\sigma}}(q)+\vec{c_{uv\sigma}}(q)-\vec{c_{uv}}(r)+\vec{c_{uv}}(r)-\vec{c_{uv}}(\rho_i(uv))|\le \\
    & W_{\max}+ |\vec{c_{uv\sigma}}(\rho_j(uv\sigma))-\vec{c_{uv\sigma}}(q)|+|\vec{c_{uv\sigma}}(q)-\vec{c_{uv}}(r)|+|\vec{c_{uv}}(r)-\vec{c_{uv}}(\rho_i(uv))|\le \\
    & W_{\max}+ 2G+W_{\max}+2G\le 2W_{\max}+4G\\
\end{split}
\]
This, however, is a contradiction, since by the second inequality of \cref{eq:sparse gap lower bound} we have (in particular)
$|\vec{c_{uv}}(\rho_j(uv))-\vec{c_{uv}}(\rho_i(uv))|>8 W_{\max}+ 16 G$.
\acctodo{Checked}
\end{proof}

A corollary of \cref{prop:exists inc inf independent runs no transitions} is that any run on infix $\xi$ reaching $q\in V_j$ must start from some $r\in V_{i}$  with $i\ge j$. We now claim that at least one minimal run actually comes from $V_j$.

\begin{proposition}
    \label{prop:exists inc inf run from Vj to Vj}
    Let $\zeta,\xi\in \{u,ux,uxy_1,\ldots, uxy_1\cdots y_{L-1}\}$ such that $\xi$ is a strict prefix of $\zeta$ and $\zeta=\xi \cdot \chi$, and let $q\in V_j$ for some $j$. 
    There exists $r\in V_j$ and a run $\rho:p\runsto{\xi}r\runsto{\chi}q$ such that $\weight(\rho)=\minweight(\zeta,p\to q)$.
\end{proposition}
\begin{proof}
    \begin{figure}[ht]
        \centering
        \includegraphics[width=0.8\linewidth]{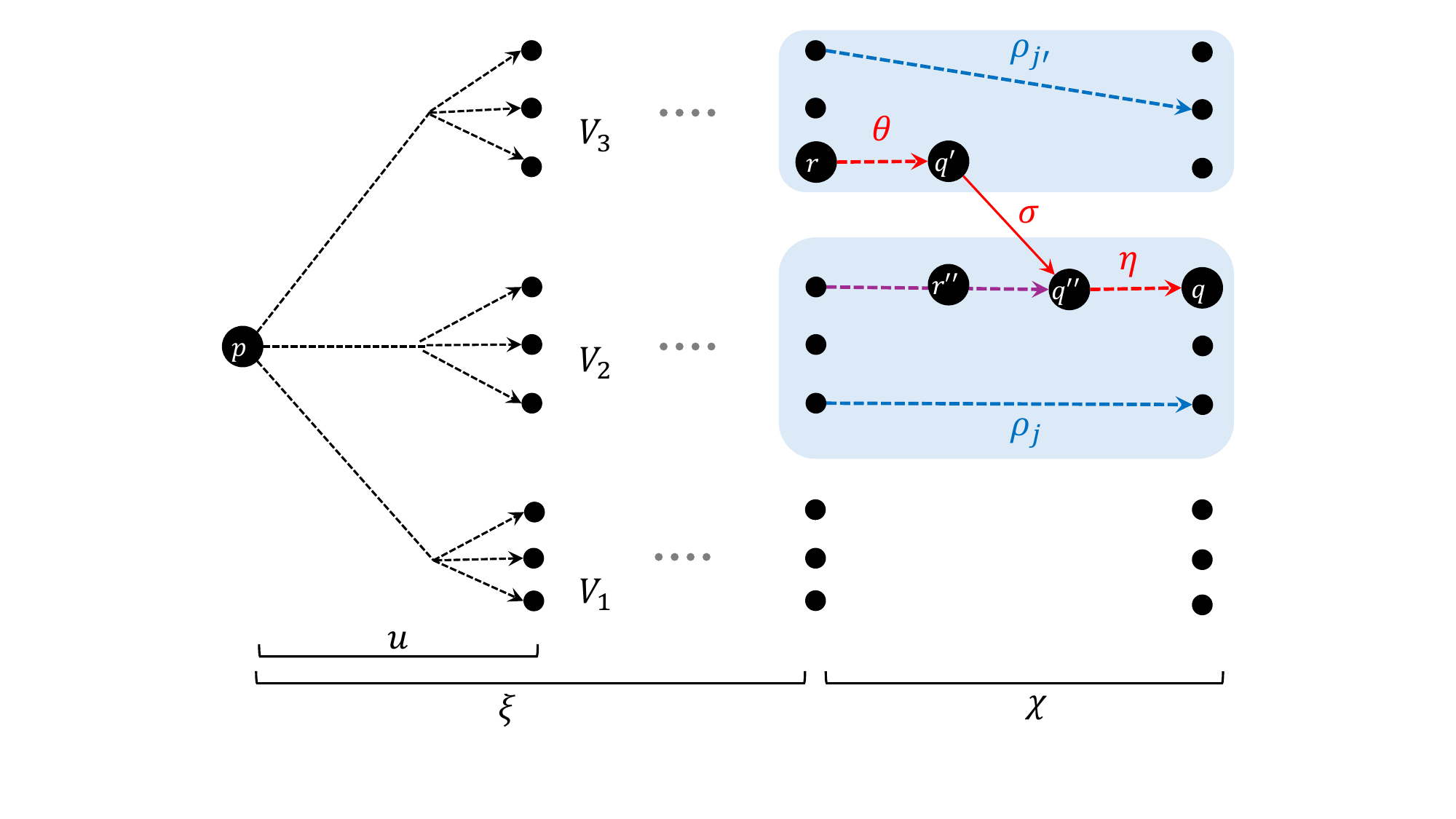}
        \caption{The highlighted areas are ``near'' $\rho_{j'}$ (above) and $\rho_j$ (below). The red run is $\mu$, and seems to cross on the letter $\sigma$ between the areas, which cannot happen. Thus, $q''$ is actually reachable from $r''$ within the lower area.}
        \label{fig:inc inf no minimal run from above}
    \end{figure}
    The proof idea is depicted in \cref{fig:inc inf no minimal run from above}.
    Assume by way of contradiction that the claim does not hold, and let $\mu:p\runsto{\xi}r'\runsto{\chi}q$ be a minimal-weight run to $q$ with $r'\in V_{j'}$ for $j'>j$. Recall that in every configuration along $uxy$, every state is within gap $G$ of an independent run in $I_p$ (since $\xwords$ is a $G$-cover). Since $r'\in V_{j'}$ and $q\in V_j$, then at some point along the run $\mu$ the states got into within gap $G$ of $\rho_j$ (whereas $\mu$ starts within gap $G$ of $\rho_{j'}$).
    Take $\mu$ to be the minimal-weight run with a maximal-length suffix near $\rho_j$. That is, $\mu$ gets within gap $G$ of $\rho_j$ the earliest among all minimal-weight runs. Write 
    $\mu:p\runsto{\xi}r'\runsto{\theta}q'\runsto{\sigma}q''\runsto{\eta}q$ where $q'$ is not near $\rho_j$, and the suffix $q''\runsto{\eta}q$ is near $\rho_j$. Note that $\sigma$ is the last letter before getting near $\rho_j$.

    Consider the transition $(q',\sigma,c,q'')$, then $|c|<W_{\max}$. However, by the gaps in $I_p$ from \cref{eq:sparse gap lower bound} we have $\vec{c_{\xi\theta}}(q')-\vec{c_{\xi\theta}}(\rho_j(\xi\theta))>W_{\max}+2G$.
    \acctodo{Checked}
    Thus, a single transition cannot decrease the weight enough to attain the value already assigned to $q''$ in $\vec{c_{\xi\theta}}$. It follows that there exists another run from $p$ that reaches $q''$ with weight $\vec{c_{\xi\theta}}$, and the last transition to $q''$ must be from some state $r''$ that is near $\rho_j$ (note that it cannot come from any lower independent run due to \cref{prop:exists inc inf independent runs no transitions}).

    This, however, contradicts the maximality of the suffix $\eta$, since the new run has a lower weight and a longer suffix.

    Thus, there is a minimal weight run whose states along the entire suffix $\chi$ remain near $\rho_j$, namely a run from $V_j$ to $V_j$.
\end{proof}

 \begin{proposition}
     \label{prop:exists inc inf no decreasing runs}
     Let $\xi,\zeta\in \{u,ux,uxy_1,\ldots, uxy_1\cdots y_{L-1}\}$ such that $\xi$ is a prefix of $\zeta$, and let $q\in \supp(\vec{c_\xi})$, then $\vec{c_{\zeta}}(q)-\vec{c_{\xi}}(q)\ge 0$. 
     
     Moreover, if $\vec{c_{\zeta}}(q)-\vec{c_{\xi}}(q)=0$, then $\vec{c_{\zeta'}}(q)-\vec{c_{\xi'}}(q)=0$
     for every $\xi',\zeta'$ such that $\xi'$ is a prefix of $\zeta'$.
 \end{proposition}
 \begin{proof}
     Denote $\zeta=\xi\chi$. Recall that $\xi$ and $\zeta$ have the same color, and in particular $\difftype_{p,G}(w_2\xi)=\difftype_{p,G}(w_2\zeta)$.
     Assume by way of contradiction that $\vec{c_{\zeta}}(q)-\vec{c_{\xi}}(q)< 0$ with $q\in V_j$ for some $1\le j\le k$. We can assume without loss of generality that $q$ is the minimal state in $V_j$, i.e., that $\vec{c_{\xi}}(q)=\min\{\vec{c_{\xi}}(r)\mid r\in V_j\}$. 
     Indeed, the gaps between the states in $V_j$ are identical in $\vec{c_{\xi}}$ and in $\vec{c_{\zeta}}$, and therefore for every $r\in V_j$ it holds that $\vec{c_{\zeta}}(r)-\vec{c_{\xi}}(r)< 0$, so we can take $q$ to be the minimal state.
     
     By \cref{prop:exists inc inf run from Vj to Vj}, there is some $r\in V_j$ and a minimal-weight run $\rho:r\runsto{\chi}q$. Since $q$ is minimal, it follows that  $\weight(\rho)=\vec{c_{\zeta}}(q)-\vec{c_{\xi}}(r)< 0$.

     Recall that $\xwords$ (and therefore $w_1w_2w_3w_4$) is $D$-dip (\cref{def:dip words}). 
     In particular $\weight(\rho)=\vec{c_{\zeta}}(q)-\vec{c_{\xi}}(r)> -D$. 
     Further recall that the $\gaintype$ component of our Ramsey-coloring tracks runs of weights in $\{-D,\ldots, D\cdot |S|\}$. Thus, we have $(r,\weight(\rho),q)\in \gaintype(w_1w_2\xi,w_1w_2,\xi\chi)$. 
     Looking at the $\gaintype$ component of the idempotent coloring, we have that for every infix $y_i$ for $1\le i <L$, since $y_i$ is colored the same as $\chi$, then $(r,\weight(\rho),q)\in \gaintype(w_1w_2uxy_1\cdots y_{i-1},w_1w_2uxy_1\cdots y_{i})$.
     
     We claim that this decreasing run induces a longer decreasing run from $V_j$ to $V_j$. Specifically, consider the configuration $\vec{c_{ux}}$. Since $\difftype_{p,G}(w_2ux)=\difftype_{p,G}(w_2\xi)$ then 
     \[\vec{c_{ux}}(r)-\vec{c_{ux}}(q)=\vec{c_{\xi}}(r)-\vec{c_{\xi}}(q)\]
     and reading $y_1$ from $\vec{c_{ux}}$ also admits the run $\rho$, so $\vec{c_{ux}}(q)>\vec{c_{uxy_1}}(q)$.
     But $\vec{c_{uxy_1}}(r)-\vec{c_{uxy_1}}(q)=\vec{c_{\xi}}(r)-\vec{c_{\xi}}(q)$, so we can repeat this argument for $y_2$, and by induction for all $1\le i<L$. 

     Since $L>D\cdot |S|>D$, it follows that $\vec{c_{uxy_1\cdots y_{L-1}}}(q)-\vec{c_{ux}}(q)<-D$ (since $q$ decreases by at least $1$ with every $y_i$). 
     By \cref{prop:exists inc inf run from Vj to Vj} and the minimality of $q$, we can trace a single run $\rho':r'\runsto{y_1\cdots y_{L-1}}q$ from some $r'\in V_j$ with $\weight(\rho')<-D$. This is in contradiction to the fact that $w_1w_2w_3w_4$ is $D$-dip. We therefore conclude that there are no such negative runs.

     Finally, consider the case where $\vec{c_{\zeta}}(q)-\vec{c_{\xi}}(q)= 0$. Following the same reasoning as above (i.e., taking a minimal run from $V_j$ to the minimal $q$ satisfying this), we can conclude by a similar induction that for every $\xi\in \{u,ux,uxy_1,\ldots, uxy_1\cdots y_{L-1}\}$ and $r\in V_j$, we have $\vec{c_{\xi}}(r)=\vec{c_{u}}(r)$, so the configuration restricted to $V_j$ remains unchanged throughout the run on $xy$, as required.
 \end{proof}
 \cref{prop:exists inc inf no decreasing runs} then concludes part \textit{(b)} in Requirement \textit{3} of \cref{def:separated increasing infix from state}.

 Finally, in order to conclude the requirements of \cref{def:separated increasing infix}, we need to show that $(\ghostTrans(S',u),x)$ is a stable cycle (\cref{def:stable cycle}). Recall that $S'=\ghostTrans(s_0,w_1)$. 
 \begin{proposition}
     \label{prop:exists inc inf x stable cycle}
     $(\ghostTrans(S',u),x)$ is a stable cycle.
 \end{proposition}
 \begin{proof}
     We start by showing that $(\ghostTrans(S',u),x)$ is a reflexive cycle, i.e., that $\booltrans(\ghostTrans(S',u),x)\subseteq \ghostTrans(S',u)$. 
     This is straightforward, but subtle. 
     
     Recall that the states of $\augA_\infty^\infty$ are tuples of the form $(p_1,p_2,T)$ where $p_2$ and $T$ behave deterministically. Since $\booltrans(S',u)=\booltrans(S',ux)$ (a corollary of \cref{prop:exists inc inf Req 2}, by taking union over all $p\in S'$), it follows that all states in this set share the $p_2$ and $T$ components.
     However, $\ghostTrans(S',u)$ essentially ``fills in'' the states with components $p_2$ and $T$ that are not already in $\booltrans(S',u)$ (see \cref{sec:reachable ghost states}), and therefore $\ghostTrans(S',u)=\ghostTrans(S',ux)$.

     Now, if $q\in \booltrans(\ghostTrans(S',u),x)$, then $q$ has the same $p_2$ and $T$ components as states in $\booltrans(S',ux)$, and therefore $q\in \ghostTrans(S',ux)=\ghostTrans(S',u)$.

     We now turn to show that $(\ghostTrans(S',u),x)$. That is, we need to show (c.f., \cref{def:stable cycle}) that the baseline state $s\in \ghostTrans(S',u)$ satisfies $\minweight(x^n,s\to s)=0$ for all $n\in \bbN$, and there are no negative cycles on states in $\ghostTrans(S',u)$.
     By the assumption of \cref{lem:unbounded long dip and G imply separated inc infix}, the entire word $w_1w_2w_3w_4$ has a seamless baseline run, which therefore has weight $0$ by definition. It follows that $\minweight(x,s\to s)=0$, and therefore $\minweight(x^n,s\to s)\le 0$ by induction for all $n\in \bbN$.
     We show that $\minweight(x^n,s\to s)=0$ and that there are no negative cycle in a single argument.

     Assume by way of contradiction that there is some $k \in \bbN$ and $p_1 \in \ghostTrans(S',u)$ with a negative reflexive run $\rho': p_1 \runsto{x^k} p_1$ with $\weight(\rho')=d<0$,
     We can assume without loss of generality that $k\le |S|$. Indeed, otherwise the run has the form $\rho':p_1\runsto{x^j}r\runsto{x^\ell}r\runsto{x^t}p_1$, and either the cycle on $r$ is also negative, in which case we can look at that cycle instead, or the cycle on $r$ is positive, in which case we can shorten $\rho'$ and still have a negative run.

     Write $\rho':p_1\runsto{x}p_2\runsto{x}\cdots\runsto{x}p_{k-1}\runsto{x}p_k$ with $p_k=p_1$. 
     Since $(\ghostTrans(S',u),x)$ is a reflexive cycle, then in particular $p_j\in \ghostTrans(S',u)$ for all $1\le j\le k$. Since $x$ is an infix of $w_3$, and $w_1w_2w_3w_4$ is $D$-dip, then in particular $\minweight(x,p_j\to p_{j+1})> -D$ for all $1\le j<k$. We claim that in addition it holds that 
     $\minweight(x,p_j\to p_{j+1})\le |S|D$. 
     Indeed, assume by way of contradiction that $\minweight(x,p_j\to p_{j+1})>|S|D$ for some $1\le j<k$. Since $k\le |S|$ and each sub-run on $x$ decreases by at most $D$, the total weight of $\rho'$ cannot be negative, in contradiction to the assumption.

     Thus, $-D\le \minweight(x,p_j\to p_{j+1})\le |S|D$, and it follows that $(p_j,\minweight(x,p_j\to p_{j+1}),p_{j+1})\in \gaintype(w_1w_2u,w_1w_2ux)$ for all $1\le j<k$.
     By the equivalent coloring of $x$ and each of the $y_i$, for every $1\le i<L$, it also holds that \[(p_j,\minweight(x,p_j\to p_{j+1}),p_{j+1})\in \gaintype(w_1w_2uxy_1\cdots y_{i-1},w_1w_2uxy_1\cdots y_i)\]
     This means that we can ``replicate'' the run $\rho'$ using a sequence of at most $|S|$ subwords from $y_i$. E.g., a run $p_1\runsto{y_1}p_2\runsto{y_3}\cdots \runsto{y_k}p_k$ can be constructed such that each transition on $y_j$ gains weight $\minweight(x,p_j\to p_{j+1})$.

     Now, since $L>|S|D$, we can concatenate at least $D$ copies of this replication of $\rho'$, so the corresponding run has weight at most $-D$ (since $\rho'$ is negative and has weight at most $-1$). This is a contradiction to $w_1w_2w_3w_4$ being $D$-dip.

     We therefore see that there are no negative reflexive cycles in $(\ghostTrans(S',u),x)$, and it is therefore a stable cycle.         
 \end{proof}

This concludes the proof of \cref{lem:unbounded long dip and G imply separated inc infix}. Admittedly, even the most tenacious reader may have (justifiably) forgotten by now that we are within a proof of a lemma.

\subsubsection{Existence of Increasing Infixes Without A Sparse $G$-Cover}

In \cref{lem:unbounded long dip and G imply separated inc infix} we establish the existence of a separated increasing infix under the assumption that there is a sparse $G$-cover.
In this section we show that this assumption in fact follows from the existence of a $D$-dip $\xwords$ function, and can therefore be dropped.

More precisely, the existence of a sparse $G$-cover for some $G$ can be obtained without any assumptions (and $D$-dip is an ``orthogonal'' assumption).
Intuitively, the idea is that if a seamless run stays far enough away from every independent run, then it induces another independent run. This allows us to apply an inductive argument to show the claim. 
We start by showing that sparsity is enough to ensure a $G$-cover.
\begin{lemma}[\lightbulbicon Zooming on $\xwords$]
    \label{lem:unbounded long dip imply separated inc infix}
    Let $\xwords$ be an $\ell$-sparse function for some $\ell\in \bbN$. 
    Then there exists a faithful restriction $\xwords'$ of $\xwords$ that is an $\ell'$-sparse $G$-cover for some $G,\ell'\in \bbN$. 
\end{lemma} 
\begin{proof}
    Recall that for every state $p$ and a set of independent runs $I_p$, it holds that $|I_p|\le |S|$. Also recall that the parameter $\ell$ in $\ell$-sparse is $\ell=|\bigcup_{p}I_p|$, and therefore $\ell\le |S|^2$. 

    We prove the lemma by reverse induction on $\ell$, starting from $|S|^2$ and down to $1$. Note that there must be at least one independent run, since the baseline run is seamless in the word $\xwords$ for all $m$ (\cref{def:elongated words sequence}).

    \paragraph{Base case: $\ell=|S|^2$.}
    Intuitively, in this case every seamless run is independent, and therefore every state along a seamless run is within gap $0$ of an independent run.

    More precisely, consider a word $w_1w_2w_3w_4=\xwords(m)$ for some $m$. 
    For every $p\in \ghostTrans(s_0,w_1)$ let $I_p$ be a set of independent runs. Since $\ell=|S|^2$, it holds that $|I_p|=|S|$. 
    For every prefix $x$ of $w_3$ consider the configuration $\vec{c_{w_2x}}=\xconf(\vec{c_p},w_2x)$, then it follows that $\supp(\vec{c_{w_2x}})=S$ (by size considerations) and for every $q\in S$ we have that $\vec{c_{w_2x}}(q)$ is along some seamless run from $p$, which is therefore one of the runs in $I_p$, so we conclude the lemma with $\xwords'=\xwords$ and $G=0$.


    \paragraph{Induction step: $\ell<|S|^2$.}
    We assume correctness of the claim for $\ell+1$, and we prove for $\ell$. We obtain $\xwords'$ by modifying $\xwords$, and the construction is split to two cases. We sketch the intuition before proceeding with the technical details.

    First, consider the case where $\xwords$ is already a $G$-cover for some $G\in \bbN$. In fact, we relax the definition further and only require that for every $m\in \bbN$ and $w_1w_2w_3w_4=\xwords(m)$, states visited after reading a prefix $x$ of $w_3$ of length at least $\frac12 m$ are $G$-covered. This seemingly-overcomplicated assumption then helps us in the second case, where we can require something more strict.
    We now split $w_3$ in the middle, obtaining $w_3=uv$ where $u=\lfloor\frac12 m \rfloor$ and $v=\lceil\frac12 m \rceil$, we denote $t=|v|$. 
    We then define $\xwords'(t)=w_1'w_2'w_3'w_4'$ where $w_1'=w_1$, $w_2'=w_2u$, $w_3'=v$ and $w_4'=w_4$. We show that $\xwords'$ is a $G$-cover, which follows trivially from $\xwords$ begin a $G$-cover. 
    Note that this case does not require the induction hypothesis.

    In the second case, we assume the complement of the first case. That is, for every $G\in \bbN$ there exists $m\in \bbN$ such that for $w_1w_2w_3w_4=\xwords(m)$ there exists a prefix $x$ of $w_3$ with $|x|\ge \frac12 m$ and states $p\in \ghostTrans(s_0,w_1)$ and $q$ such that after reading $w_2x$ from $p$, the state $q$ is not within gap $G$ from any independent run in $I_p$.
    
    We apply this assumption for a large enough constant $G$, which then allows us to observe that in order to create a gap of at least $G$ from all independent runs, the seamless run leading up to $q$ must be very far from all independent runs for a long infix. Therefore, this run is a new independent run in the infix, and we can now use the induction hypothesis.

    We now turn to the formal details, with depictions in \cref{fig:sparse induction first case,fig:sparse induction second case}.
    \begin{figure}[ht]
        \centering
        \includegraphics[width=1.0\linewidth]{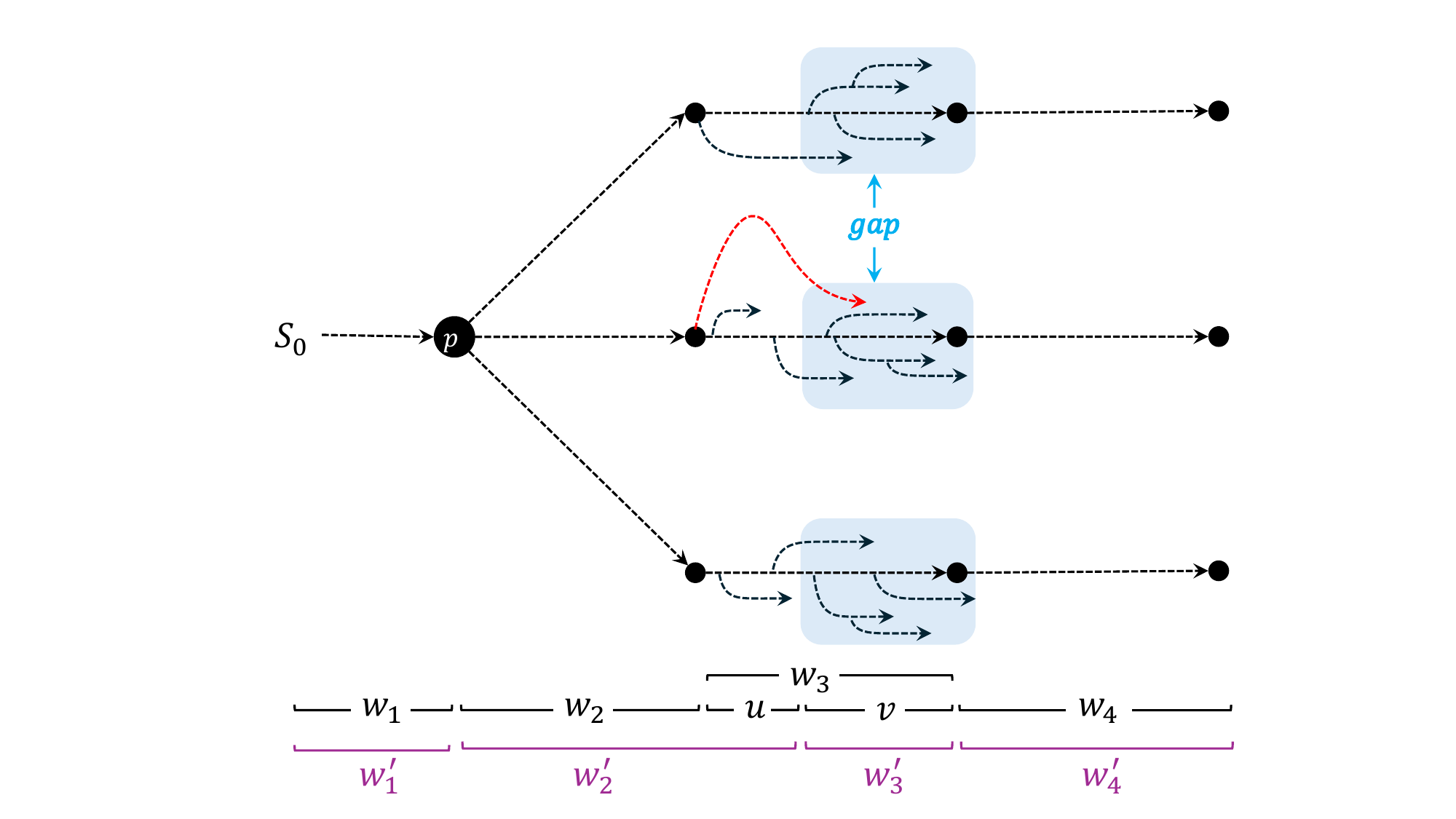}
        \caption{The first case in the proof of \cref{lem:unbounded long dip imply separated inc infix}. The highlighted regions are within $G$ of the independent runs. Note that the red run goes outside this region, but only in the first half of $w_3$.}
        \label{fig:sparse induction first case}
    \end{figure}
    \begin{figure}[ht]
        \centering
        \includegraphics[width=0.7\linewidth]{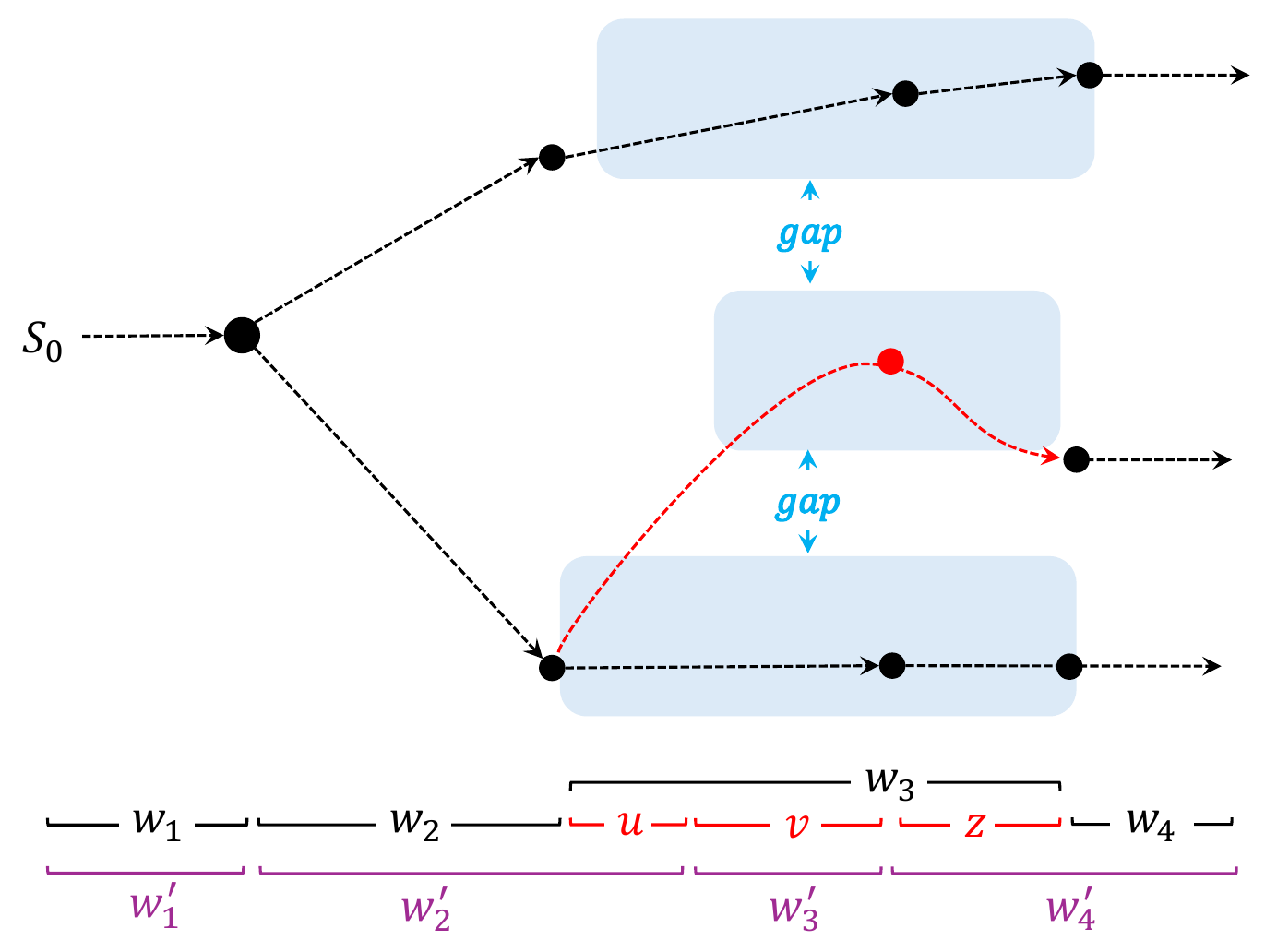}
        \caption{The second case in the proof of \cref{lem:unbounded long dip imply separated inc infix}. The red run goes far from the independent runs, and to do so it must be far for a long suffix ($v$). 
        }
        \label{fig:sparse induction second case}
    \end{figure}
    
    For the first case, assume that there exists $G\in \bbN$ such that for every $m\in \bbN$ and $w_1w_2w_3w_4=\xwords(m)$, for every state $p\in S$ and corresponding set of independent runs $I_p$, and for every prefix $x$ of $w_3$ with $|x|\ge \frac12|w_3|$, if $\vec{c_{w_2x}}=\xconf(w_2x,\vec{c_p})$ then for every state $q$, $\vec{c_{w_2x}}(q)$ is within gap $G$ of some independent run in $I_p$.

    Let $t\in \bbN$, we construct $\xwords'(t)$ as follows. Plugging $m=2t$, consider $w_1w_2w_3w_4=\xwords(2t)$ and write $w=uv$ with $|u|=|v|=t$.  We now set $\xwords'(t)=w'_1w'_2w'_3w'_4$ with $w_1'=w_1$, $w_2'=w_2u$, $w_3'=v$ and $w_4'=w_4$. We claim that $\xwords'$ is indeed a faithful restriction of $\xwords$ that is an $\ell$-sparse $G$-cover for some $G$.

    Indeed, by the decomposition (specifically, $w'_1=w_1$ and $w'_1w'_2w'_3w'_4=w_1w_2w_3w_4$) we have that $\xwords'$ is a faithful restriction of $w$. In particular, this implies that the baseline run is seamless.
    Next, for every $p\in \ghostTrans(s_0,w'_1)$ there same set $I_p$ of independent runs on $w_2w_3$ can serve as independent runs on $w'_2w'_3$, noticing that the gap requirement for $w_1w_2w_3w_4$ is $\sparseGapBound$, whereas for $w'_1w'_2w'_3w'_4$ it is 
    $16\bigM(|S|+1)(tW_{\max})^2$,
    \acctodo{If the constant changes, change this}
    but since $t<m$, the former gap suffices. 
    Thus, $\xwords'$ is indeed an $\ell$-sparse function.
    Finally, $\xwords'$ is a $G$-cover by the assumption: each state reachable over $w'_3=v$ after $w'_2=w_2u$ is reachable from $w_2$ by a prefix $x$ of $uv=w_3$, and is therefore within gap $G$ of a run in $I_p$.

    We now proceed to second case, which uses the induction hypothesis. We assume the complement of the previous assumption. That is, for every $G\in \bbN$ there exists $m\in \bbN$ such that for $w_1w_2w_3w_4=\xwords(m)$ there exists a prefix $x$ of $w_3$ with $|x|\ge \frac12 m$ and states $p\in \ghostTrans(s_0,w_1)$ and $q$ such that after reading $w_2x$ from $p$, the state $q$ is not within gap $G$ from any independent run in $I_p$.

    We define $\xwords'$ to be an $\ell+1$-sparse function as follows.
    Let $t\in \bbN$ and take $G>2W_{\max}t+16\bigM(|S|+1)(tW_{\max})^2$
    \acctodo{If the constant changes, change this}
    such that the corresponding $m$ in the assumption satisfies $m>2t$. 
    Note that we can indeed require $m>2t$, since for every $G$ there are only a finite number of values of $m$ for which there exist states  that are not within gap $G$.

    Let $w_1w_2w_3w_4=\xwords(m)$, then $|x_3|=m>2t$ and by the assumption there exists a prefix $x$ of $w_3$ with $|x|\ge \frac12 m>t$ and states $p\in \ghostTrans(s_0,w_1)$ and $q$ such that after reading $w_2x$ from $p$, the state $q$ is not within gap $G$ from any independent run in $I_p$. Write $x=uv$ where $|v|=t$ (which is possible since $|x|>t$), and write $w_3=uvz$. 
    We now define $w'_1=w_1$, $w'_2=w_2u$, $w'_3=v$ and $w'_4=zw_4$ and set $\xwords'(t)=w'_1w'_2w'_3w'_4$. 
    By identical considerations to the above, we get that $\xwords'$ is a faithful restriction of $\xwords$.

    Next, let $p\in \ghostTrans(s_0,w_1)$, consider the set $I_p$ that corresponds to $w_1w_2w_3w_4$, and let $q$ be a state such that $\vec{c_{w_2x}}(q)$ is not within gap $G$ of $I_p$ (which exists by our assumption).

    Since $G>2W_{\max}t+16\bigM(|S|+1)(tW_{\max})^2$, 
    \acctodo{If the constant changes, change this}
    and since in each transition the run to $q$, as well as any independent run, can change its weight by at most $W_{\max}$, it follows that for the seamless run $\rho:p\runsto{w_2x}q$, the $t$ transitions leading up to $q$, namely the suffix of $\rho$ on $v$, can change the difference between to runs by at most $2W_{\max}t$. In particular, the suffix of $\rho$ on $v$ is not within gap $16\bigM(|S|+1)(tW_{\max})^2$ of any independent run in $I_p$. 
    Thus, adding $\rho$ to $I_p$, we have a set of $\ell+1$ independent runs on $v$.
    \acctodo{Checked}

    Finally, by the induction hypothesis, the existence of the $\ell+1$-sparse function $\xwords'$ implies the existence of such a $G$-cover for some $G$ (possibly for some $\ell''$), and we are done.
\end{proof}    
Now, assume a $D$-dip $\xwords$ function exists and for every $m\in \bbN$, consider $\xwords(m)=w^m_1w^m_2w^m_3w^m_4$. For every state $p\in \ghostTrans(s_0,w_1)$, if there is a run $\rho:p\runsto{w_2w_3w_4}S$, then there is also such a seamless run. We arbitrarily select one seamless run $\rho_p^m$ for each $p$ and $m$. 
Since a single run is independent by definition, we therefore obtain a set of independent runs $I^m_p$ with $|I^m_p|=1$. Taking $\ell$ to be the maximal value such that $\ell=|\bigcup_{p}I^m_p|$ for infinitely many $m$, we can use \cref{prop:faithful restriction of words infinitely many ells} to obtain a faithful restriction of $\xwords$ that is $\ell$-sparse (for some $\ell\ge 1$). For brevity, we assume $\xwords$ itself is already $\ell$-sparse.

Crucially, note that taking a faithful restriction maintains the property of being a $D$-dip.

We can now apply \cref{lem:unbounded long dip imply separated inc infix}, to obtain a sparse $G$-cover $\xwords'$. Again, since $\xwords'$ is a faithful restriction of $\xwords$, then $\xwords'$ is also $D$-dip. Thus, we have the following.
\begin{corollary}
    \label{cor:there is a G cover}
    If there exists an elongated word sequence $\xwords$ that is a $D$-dip for some $D\in \bbN$, then there is a faithful restriction $\xwords'$ of $\xwords$ that is a $D$-dip $\ell$-sparse $G$-cover for some $\ell,G\in \bbN$. 
\end{corollary}
Then, by combining \cref{cor:there is a G cover} with \cref{lem:unbounded long dip and G imply separated inc infix} we have the following.
\begin{corollary}[\keyicon $D$-Dip to Increasing Infix]
    \label{cor:dip implies increasing infix}
    If there exists an elongated word sequence $\xwords$ that is $D$-dip for some $D\in \bbN$, then there exists $m\in \bbN$ with $\xwords(m)=w_1w_2w_3w_4$ and a decomposition $w_3=u'xyv'$ such that $uxyv$ is a separated increasing infix from $\ghostTrans(s_0,w_1)$ for $u=w_2u'$ and $v=v'w_4$.
\end{corollary}

\section{Cactus Budding}
\label{sec:cactus budding}
In this section we discuss how a cactus chain (\cref{def:cactus chain}) may grow a new cactus at its end (\emph{bud}, in cactus terms). The overall intuition is that a chain may bud a new inner cactus if doing so both decreases the cost of the chain (as per \cref{def:cost depth and sub cactus}), and does not decrease the potential. 
To this end, the main tool at our disposal is our study of Separated Increasing Infixes (\cref{def:separated increasing infix}).
As usual, there are several definitions and technical tools to develop before getting to the main result.
\begin{definition}[Superior Stable Cycle]
    \label{def:superior stable cycle}
    Consider two stable cycles $(S',u),(S',v)$ sharing the same first component $S'$. 
    We say that $(S',v)$ is \emph{superior} to $(S',u)$ if for every $q,p\in S'$ 
    we have $\minweight(\alpha_{S',u},q\to p)\le \minweight(\alpha_{S',v},q\to p)$.
\end{definition}
Our first (simple) observation is that a word that ``almost'' induces a stable cycle, and satisfies a superiority condition with respect to a stable cycle, actually induces a superior stable cycle.
\begin{proposition}
    \label{prop:superior word is a superior cycle}
    Consider a stable cycle $(S',u)$ and let $v\in (\Gamma_\infty^0)^*$ (i.e., without rebase and  jump letters) be a word such the following hold.
    \begin{enumerate}
        \item For every $q,p\in S'$ it holds that $\minweight(u,q\to p)\le \minweight(v,q\to p)$.
        \item For the baseline state $g \in S'$ it holds that $\minweight(v,g\to g)=0$.
    \end{enumerate}
    Then $(S',v)$ is a stable cycle, and it is superior to $(S',u)$.
\end{proposition}
\begin{proof}
    By the condition on the baseline state $g$
    and since $v$ does not have jump letters,
    the baseline run $\rho:g\runsto{v}g$ exists and $\weight(\rho)=0$. 
    Therefore in order to show that $(S',v)$ is a stable cycle, all that remains is to prove that there are no negative cycles on $v^k$ for any $k\in \bbN$, and that $\booltrans(S',v)\subseteq S'$.
    
    The latter is simple: if $q\in \booltrans(S',v)$ then $\minweight(v,S'\to q)<\infty$, but then $\minweight(u,S'\to q)<\infty$, so $q\in \booltrans(S',u)\subseteq S'$ (since $(S',u)$ is a stable cycle), so $\booltrans(S',v)\subseteq S'$.
    Before addressing the main claim, we prove a slightly more general property: for every $k\in \bbN$ and $r_0,r_k\in S'$ it holds that $\minweight(u^k,r_0\to r_k)\le \minweight(v^k,r_0\to r_k)$. 
    Indeed, let $\rho:r_0\runsto{v^k}r_k$ such that $\weight(\rho)=\minweight(v^k,r_0\to r_k)$, and write $\rho:r_0\runsto{v}r_1\cdots r_{k-1}\runsto{v}r_k$. By the premise we have
    \[
    \begin{split}
    &\minweight(v^k,r_0\to r_k)=\weight(\rho)=\sum_{i=0}^{k-1}\minweight(v,r_{i}\to r_{i+1})\ge\\
    &\sum_{i=0}^{k-1}\minweight(u,r_{i}\to r_{i+1})\ge \minweight(u^k,r_0\to r_k).
    \end{split}
    \]

    Now, assume by way of contradiction that $\minweight(s \to s, v^k)<0$ for some $k\in \bbN$, then by the property above we have $\minweight(s \to s, u^k)<0$ in contradiction to the fact that $(S',u)$ is a stable cycle.

    We conclude that $(S',v)$ is a stable cycle. 
    Similarly, we can now prove that it is superior to $(S',u)$. Consider $p,q\in S'$. If $\minweight(\alpha_{S',v},p\to q)=\infty$, there is nothing to prove. Thus, assume $\minweight(\alpha_{S',v},p\to q)<\infty$. By \cref{def:stabilization}, we have that $(p,q)\in \GroundPairs(S',v)$ and there exists $g\in S'$ such that 
    $\minweight(\alpha_{S',v},p\to q)=\minweight(v^{2\bigM},p\runsto{v^\bigM}g\runsto{v^\bigM} q)$.
    By the property above, we now have
    \[
    \begin{split}
    &\minweight(\alpha_{S',v},p\to q)=\minweight(v^{2\bigM},p\runsto{v^\bigM}g\runsto{v^\bigM} q)=
    \minweight(v^{\bigM},p\to g)+\minweight(v^{\bigM},g\to q)\ge\\
    &    \minweight(u^{\bigM},p\to g)+\minweight(u^{\bigM},g\to q)=
    \minweight(u^{2\bigM},p\runsto{u^\bigM}g\runsto{u^\bigM} q)\ge
    \minweight(\alpha_{S',u},p\to q)
    \end{split}
    \]
    where the last inequality is by \cref{def:stabilization}, keeping in mind that the grounding state for $p,q$ with $u$ may differ from $g$, but that $g\in \MinRefStates(S',u^\bigM)$ (that is, it is a possible grounding state), so we indeed have the inequality.
    We thus have that $(S',v)$ is superior to $(S',u)$.
\end{proof}
Our next observation is that superiority ``propagates'' through to outer-cactus letters.

\begin{proposition}
\label{prop:superiority propagates to outer cactus}
    Consider a stable cycle $(S',w)$ such that $w=w_1  \alpha_{B,u} w_2$ with $B=\ghostTrans(s_0,w_1)$. Let $(B,v)$ be a superior stable cycle to $(B,u)$, then $(S',w_1\alpha_{B,v}  w_2)$ is a superior stable cycle to $(S',w)$.
\end{proposition}
\begin{proof}
Write $w'=w_1  \alpha_{B,v} w_2$, we prove that $w'$ satisfies the conditions of \cref{prop:superior word is a superior cycle} with respect to $(S',w)$.
Let $p,q\in S'$. If $\minweight(w_1\alpha_{B,v}w_2,p\to q)=\infty$, there is nothing to prove. Otherwise, let $r,t\in B$ such that 
\[\minweight(w_1\alpha_{B,v}w_2)=\minweight(w_1\alpha_{B,v}w_2,p\runsto{w_1}r\runsto{\alpha_{B,v}}t\runsto{w_2}q)\]
Since $(B,v)$ is superior to $(B,u)$, we now have
\[
\begin{split}
&\minweight(w_1\alpha_{B,v}w_2)=\minweight(w_1\alpha_{B,v}w_2,p\runsto{w_1}r\runsto{\alpha_{B,v}}t\runsto{w_2}q)\\
&=\minweight(w_1,p\to r)+\minweight(\alpha_{B,v},r\to t)+\minweight(w_2,t\to q)\\
&\ge \minweight(w_1,p\to r)+\minweight(\alpha_{B,u},r\to t)+\minweight(w_2,t\to q)\\
& = \minweight(w_1\alpha_{B,v}w_2,p\runsto{w_1}r\runsto{\alpha_{B,u}}t\runsto{w_2}q) \ge \minweight(w_1\alpha_{B,u}w_2)
\end{split}
\]
satisfying the first item in \cref{prop:superior word is a superior cycle}.

Finally, since $\alpha_{B,v}$ is a stable cycle, and since $w_1\alpha_{B,u}w_2$ has a seamless baseline run with weight $0$, then so does $w_1\alpha_{B,v}w_2$ (since $\alpha_{B,v}$ also has a $0$ baseline transition).
We conclude that $w'=w_1\alpha_{B,v}w_2$ satisfies the conditions in \cref{prop:superior word is a superior cycle}, and therefore $(S',w')$ is a superior stable cycle to $(S',w)$.
\end{proof}

Next, we show that the charge \emph{decreases} with superiority. Intuitively, this is because the charge measures how negative runs can be come, and superiority increases all the runs.

\begin{proposition}
    \label{prop:charge decreases superior}
    Consider stable cycles $(B,u),(B,v)$ such that $(B,v)$ is superior to $(B,u)$. Then for every $x,y$ we have that $\charge(x\alpha_{B,v}y)\le \charge(x\alpha_{B,u}y)$.
\end{proposition}
\begin{proof}
    Let $q\in S$ such that $\rho:s_0\runsto{x\alpha_{B,v}y}q$ satisfies $\weight(\rho)=-\charge(x\alpha_{B,v}y)$ and $\rho$ is seamless.
    Write $\rho:s_0\runsto{x}r\runsto{\alpha_{B,v}}t\runsto{y}q$, then by superiority and since $\rho$ is seamless, we have
    \[
    \begin{split}
        &-\charge(x\alpha_{B,v}y)=\weight(\rho)=
        \minweight(x,s_0\to r)+\minweight(\alpha_{B,v},r\to t)+\minweight(y,t\to q)\ge\\
        &\minweight(x,s_0\to r)+\minweight(\alpha_{B,u},r\to t)+\minweight(y,t\to q)\ge \minweight(x\alpha_{B,u}y,s_0\to q)\ge -\charge(x\alpha_{B,u}y)
    \end{split}
    \]
    So $\charge(x\alpha_{B,v}y)\le \charge(x\alpha_{B,u}y)$
    and we are done.
\end{proof}
Our next major goal is to show a dual result for the potential (namely that the potential increases with superiority), and moreover -- that this carries through cactus chains (\cref{def:cactus chain}).
Specifically, we do not prove this claim in full generality. Rather, we focus on cactus chains that end in an increasing infix (\cref{def:separated increasing infix}), and show how they induce superior cactus chains with increasing potential.

Recall that separated increasing infixes (\cref{def:separated increasing infix}) can be ``folded'' into stable cycles, as per \cref{lem:increasing infix budding}. 
We make use of this fact in order to ``bud'' a new cactus from an increasing infix, while maintaining nice properties, as follows (see depiction in \cref{fig:inc inf budding potential}).
\begin{figure}[ht]
    \centering
    \includegraphics[width=0.75\linewidth]{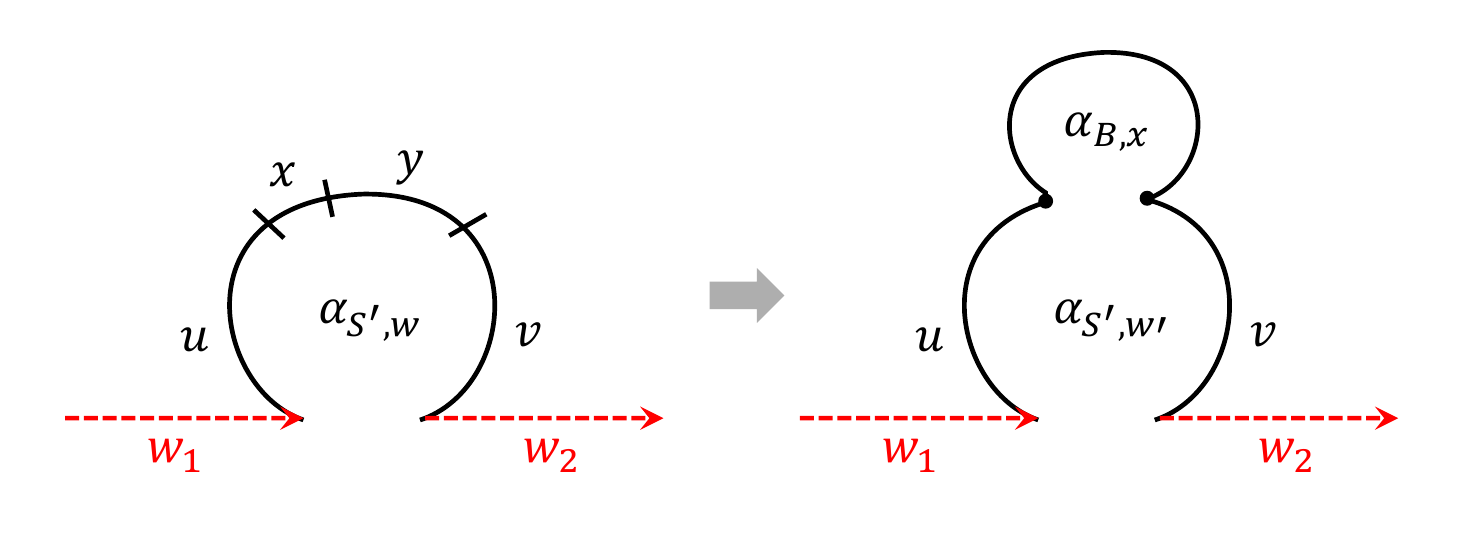}
    \caption{Budding a cactus from an increasing infix. Observe that $y$ is depicted as longer than $x$, to reflect its higher cost.}
    \label{fig:inc inf budding potential}
\end{figure}
\begin{lemma}[\lightbulbicon Increasing Infix Budding]
    \label{lem:increasing infix budding in chain base case}
    Either $\augA_{\infty}^\infty$ has a type-0 witness, or the following holds.
    Consider a cactus letter $\alpha_{S',w}\in (\Gamma_\infty^0)^*$ where $w=uxyz$ is a separated increasing infix from $S'$. Let $w'=u\alpha_{B,x}v$ where $B=\ghostTrans(S',u)$. Then $(S',w')$ is a stable cycle superior to $(S',w)$, and for every $w_1,w_2\in (\Gamma_\infty^0)^*$ we have $\pot(w_1\alpha_{S',w}w_2)\le \pot(w_1\alpha_{S',w'}w_2)$, if the potential is defined.
\end{lemma}
\begin{proof}
    The fact that $(S',w')$ is a superior stable cycle follows directly from \cref{lem:increasing infix budding}: we show there that either $\augA_{\infty}^\infty$ has a type-0 witness, or for every $r\in S'$ and $t\in S$ we have $\minweight(w,r\to t)\le \minweight(w',r\to t)$. Moreover, the baseline run exists in $(S',w')$, since $\alpha_{B,x}$ maintains the baseline run that is present in $xy$. 
    Therefore, we can use \cref{prop:superior word is a superior cycle} and we have that $(S',w')$ is a stable cycle superior to $(S',w)$.

    The interesting part of the lemma is obtaining the potential inequality. 
    Intuitively, we first unfold both $w_1\alpha_{S',w}w_2$ and $w_1\alpha_{S',w'}w_2$ with a high enough constant. This leaves us with many repetitions of $w$ and of $w'$, respectively. By \cref{lem:unfolding maintains potential}, we know that the potential is maintained by this unfolding. 
    We then use \cref{lem:increasing infix budding} again, this time deriving the potential inequality $\pot(w'_1ww'_2)\le \pot(w'_1w'w'_2)$ for every $w'_1,w'_2$.
    We then use this to replace each $w$ by $w'$ in the unfolding, while maintaining the potential inequality, which allows us to conclude the proof. We refer to the latter argument as \emph{synchronous unfolding}, and reuse it later.
    We turn to the precise details.

    Consider $w_1,w_2\in (\Gamma_\infty^0)^*$ such  that $\pot(w_1\alpha_{S',w}w_2)$ is defined. Note that this only means that the baseline run is seamless, and this is preserved in all the words we consider henceforth, so we do not mention this every time we discuss the potential.
    
    Let $F\ge \max\{2\maxeff{w_1\alpha_{S',w}w_2},2\maxeff{w_1\alpha_{S',w'}w_2}\}$, and consider $\unfold(w_1,\alpha_{S',w},w_2 \wr F)=w_1w^{2\bigM M_0}w_2$ and $\unfold(w_1,\alpha_{S',w'},w_2 \wr F)=w_1w'^{2\bigM M_0}w_2$. By \cref{rmk:increasing repetitions in unfolding} we can indeed assume that $M_0$ is the same for both unfoldings (by taking the maximum of the two unfolding constants).

    By \cref{lem:unfolding maintains potential}, either $\cA_\infty^\infty$ has a type-0 witness, in which case we are done, or we have 
    \[\pot(w_1\alpha_{S',w}w_2)=\pot(w_1w^{2\bigM M_0}w_2)\quad \text{ and }\quad \pot(w_1\alpha_{S',w'}w_2)=\pot(w_1w'^{2\bigM M_0}w_2)\]

    By \cref{lem:increasing infix budding}, for every $w'_1,w'_2\in (\Gamma_\infty^0)^*$ such that $\booltrans(s_0,w'_1)\subseteq S'$, it holds that $\pot(w'_1ww'_2)\le \pot(w'_1w'w'_2)$. 

    We claim that 
$\pot(w_1w^{2\bigM M_0}w_2)\le \pot(w_1w'^{2\bigM M_0}w_2)$. In order to show this, define for every $0\le j\le 2\bigM M_0$ the word $e_j=w_1 w^j w'^{2\bigM M_0-j}w_2$, then this is equivalent to showing that $\pot(e_{2\bigM M_0})\le \pot(e_0)$. We prove that for every $0< j\le 2\bigM M_0$ it holds that $\pot(e_{j})\le \pot(e_{j-1})$, which implies the claim.

Indeed, consider such $j$ and write $w'_1=w_1 w^{j-1}$ and $w'_2= w'^{2\bigM M_0-j}w_2$.  Note that $e_j=w'_1 w w'_2$ whereas $e_{j-1}=w'_1 w' w'_2$.
Since $(S',w)$ is a stable cycle and $\booltrans(s_0,w_1)\subseteq S'$, then in particular $\booltrans(s_0,w'_1)\subseteq S'$. Therefore, as we mention above, we have
$\pot(w'_1 w w'_2)\le \pot(w'_1w'w'_2)$, i.e., $\pot(e_{j})\le \pot(e_{j-1})$, concluding the claim.
We now conclude the proof with: 
\[\pot(w_1\alpha_{S',w}w_2)=\pot(w_1w^{2\bigM M_0}w_2)\le \pot(w_1w'^{2\bigM M_0}w_2)=\pot(w_1\alpha_{S',w'}w_2)\]
We remark that the latter argument is the \emph{synchronous unfolding} mentioned above.
\end{proof}

The results in this section so far pertain to stable cycles, and can be viewed as if they occur at the ``apex'' of a cactus \emph{chain} (recall \cref{def:cactus chain}).
We now turn to lift these results to the entire cactus chain.

\begin{definition}[Superior Cactus Chain]
    \label{def:superior cactus chain}
    Consider two cactus chains $\Theta_1=\alpha_{S'_1,w_1},...,\alpha_{S'_n,w_n}$ and  $\Theta_2=\beta_{S'_1,w'_1},...,\beta_{S'_n,w'_n}$. 
    We say that $\Theta_2$ is \emph{superior} to $\Theta_1$ if $\beta_{S'_i,w'_i}$ is superior to $\alpha_{S'_i,w_i}$
for every $1\le i\le n$.  
\end{definition}

In case a cactus chain ends with a separated increasing infix, we call it \emph{pre-bud}, with the intuition that we can now bud a new cactus letter at its end. The cactus chain obtained by budding this increasing infix is referred to as a \emph{post-bud} chain, as follows. 
\begin{definition}[Pre-bud and Post-bud Cactus Chain]
    \label{def: pre and post bud chain}
    A cactus chain $\alpha_{S'_1,w_1},...,\alpha_{S'_n,w_n}$ is \emph{pre-bud} if $w_n=uxyv$ is a separated increasing infix from $S'_n$. Denote $w_i=u_i\alpha_{S'_{i+1},w_{i+1}}v_{i}$ for every $1\le i<n$ as per \cref{def:cactus chain}.

    The corresponding \emph{post-bud} cactus chain $\alpha_{S'_1,w'_1},...,\alpha_{S'_n,w'_n}$ is defined as follows:
    \begin{itemize}
        \item $w'_n=u\alpha_{B,x}v$ with $B=\ghostTrans(S'_n,u)$.
        \item For every $1\le i<n$ we define $w'_i=u_i\alpha_{S'_{i+1},w'_{i+1}}v_i$.
    \end{itemize}
\end{definition}
It is easy to see (by induction) that the post-bud cactus chain is indeed a valid cactus chain, as per \cref{def:cactus chain}. The only non trivial part is the base case, which is also immediate by the definition of an increasing infix from $S'$ (\cref{def:separated increasing infix}).

Intuitively, the post-bud cactus chain can be elongated by placing $\alpha_{B,x}$ as a new element, but at this point we do not do that, in order to keep the number of elements in the pre-bud and post-bud the same, so that they can be compared by the superiority relation. 
Indeed, our main result of this section is that the post-bud chain is superior to the pre-bud chain, and also has increased potential.
\begin{lemma}[\keyicon The Post-Bud Chain is Superior]
    \label{lem:post bud chain is superior higher potential}
    Either $\augA_{\infty}^\infty$ has a type-0 witness, or the following holds.
    Let $\Theta_1=\alpha_{S'_1,w_1},...,\alpha_{S'_n,w_n}$ be a pre-bud cactus chain and $\Theta_2=\alpha_{S'_1,w'_1},...,\alpha_{S'_n,u\alpha_{B,x}v}$ its post-bud cactus chain, then $\Theta_2$ is superior to $\Theta_1$ and for every $z_1,z_2$ with $\booltrans(s_0,z_1)\subseteq S'_1$ it holds that 
    $\pot(z_1 \alpha_{S'_1,w_1}z_2)\le \pot(z_1 \alpha_{S'_1,w'_1}z_2)$, if the potential is defined.
\end{lemma}
\begin{proof}
    The proof is by induction over the length of the chain. In a nutshell, the base case is exactly \cref{lem:increasing infix budding in chain base case}, whereas the inductive step follows by a similar ``synchronized unfolding'' argument as in the proof of \cref{lem:increasing infix budding in chain base case}.

    Before starting the induction, we assume that $\augA_{\infty}^\infty$ does not have a type-0 witness. If it does, we are done. Accordingly, whenever we use results that are prefixed by this assumption, we ignore it (c.f., \cref{rmk:the either witness prefix}).

    \paragraph{Induction base: $n=1$} In this case $\Theta_1=\alpha_{S'_1,w_1}$ with $w_1=uxyv$ an increasing infix from $S'_1$. By \cref{lem:increasing infix budding in chain base case} for $w'_1=u\alpha_{B,x}$ we have that $(S'_1,w'_1)$ is a superior stable cycle to $(S'_1,w_1)$ and for every $z_1,z_2\in (\Gamma_\infty^0)^*$ with $\booltrans(s_0,z_1)\subseteq S'_1$ it holds that $\pot(z_1\alpha_{S'_1,w_1}z_2)\le \pot(z_1\alpha_{S'_1,w'_1}z_2)$, if the potential is defined. This concludes the base case.

    \paragraph{Inductive step:} Assume correctness for $n$, we prove for $n+1$.
    Consider the pre-bud chain $\Theta_1=\alpha_{S'_0,w_0},\alpha_{S'_1,w_1},...,\alpha_{S'_n,w_n}$ and its post-bud chain $\Theta_2=\alpha_{S'_0,w'_0}\alpha_{S'_1,w_1},...,\alpha_{S'_n,u\alpha_{B,x}v}$. 

    Observe that \cref{def: pre and post bud chain} is inductive from the end of the chain to its start. In particular, we have that $\Theta'_1=\alpha_{S'_1,w_1},...,\alpha_{S'_n,w_n}$ is a pre-bud chain and its corresponding post-bud chain is conveniently $\Theta'_2=\alpha_{S'_1,w_1},...,\alpha_{S'_n,u\alpha_{B,x}v}$. 
    By the induction hypothesis, it holds that $\Theta'_2$ is superior to $\Theta'_1$. In particular, $(S'_1,w'_1)$ is superior to $(S'_1,w_1)$. By the definition of the post-bud chain (\cref{def: pre and post bud chain}), we have that $w_0=u_0 \alpha_{S'_1,w_1} v_0$ and $w'_0=u_0 \alpha_{S'_1,w'_1} v_0$. 
    By the superiority-propagation of \cref{prop:superiority propagates to outer cactus} it follows that $(S'_0,w'_0)$ is superior to $(S'_0,w_0)$, as required.  

    It remains to show the potential inequality. Let $z_1,z_2\in (\Gamma_\infty^0)^*$ with $\booltrans(s_0,z_1)\subseteq S'_0$. We need to show that $\pot(z_1\alpha_{S'_0,w_0}z_2)\le \pot(z_1\alpha_{S'_0,w'_0}z_2)$.
    We proceed with the synchronous unfolding argument, similarly to \cref{lem:increasing infix budding in chain base case}, with the main difference being that the ``inner'' potential inequality is now obtained by the induction hypothesis.

    Let $F\ge \max\{ 2 \maxeff{z_1\alpha_{S'_0,w_0}z_2},2\maxeff{z_1\alpha_{S'_0,w'_0}z_2}\}$, 
    and consider $\unfold(z_1,\alpha_{S'_0,w_0},z_2 \wr F)=z_1w_0^{2\bigM M_0}z_2$ 
    and $\unfold(z_1,\alpha_{S'_0,w'_0},z_2 \wr F)=z_1 {w'_0}^{2\bigM {M_0}}z_2$. 
    By \cref{rmk:increasing repetitions in unfolding} we can indeed assume that $M_0$ is the same for both unfoldings (by taking the maximum of the two unfolding constants).

    By \cref{lem:unfolding maintains potential} we have 
    \[
    \pot(z_1\alpha_{S'_0,w_0}z_2)=\pot(z_1w_0^{2\bigM M_0}z_2)\quad \text{ and }\quad
    \pot(z_1\alpha_{S'_0,w'_0}z_2)=\pot(z_1{w'_0}^{2\bigM M_0}z_2)
    \]
    It is therefore enough to show that $\pot(z_1w_0^{2\bigM M_0}z_2)\le \pot(z_1{w'_0}^{2\bigM M_0}z_2)$.
    Recall that $w_0=u_0 \alpha_{S'_1,w_1} v_0$ and $w'_0=u_0 \alpha_{S'_1,w'_1} v_0$ (as per \cref{def: pre and post bud chain}). 
    For every $0\le j\le 2\bigM M_0$, define the word $e_j=z_1(u_0w_0v_0)^j(u_0w'_0v_0)^{2\bigM M_0-j}z_2$, then we want to prove $\pot(e_{2\bigM M_0})\le \pot(e_0)$. 
    We prove that for every $0<j\le 2\bigM M_0$ it holds that $\pot(e_j)\le \pot(e_{j-1})$, which implies the claim.

    Indeed, consider such $j$ and write $z'_1=z_1(u_0w_0v_0)^{j-1}$ and $z'_2=(u_0w'_0v_0)^{2\bigM M_0- j}z_2$. 
    Note that $e_j=z'_1w_0 z'_2$ and $e_{j-1}=z'_1w'_0 z'_2$.
    Further note that since $\booltrans(S'_0,u_0)\subseteq S'_1$, and $w_0$ and $w'_0$ are stable cycles on $S'_1$, and $\booltrans(S'_1,v_0)\subseteq S'_0$, and  $\booltrans(s_0,z_1)\subseteq S'_0$ by the assumption, it holds that $\booltrans(s_0,z'_1)\subseteq S'_1$. 
    Thus, by the induction hypothesis (applied with $z'_1$ and $z'_2$), we have that 
    $\pot(z'_1 w_0 z'_2)\le \pot(z'_1 w'_0 z'_2)$. That is, $\pot(e_j)\le \pot(e_{j-1})$, from which we conclude that $\pot(z_1w_0^{2\bigM M_0}z_2)\le \pot(z_1{w'_0}^{2\bigM M_0}z_2)$.
    Thus, we conclude the proof with:
    \[
    \pot(z_1\alpha_{S'_0,w_0}z_2)=\pot(z_1w_0^{2\bigM M_0}z_2)\le 
    \pot(z_1{w'_0}^{2\bigM M_0}z_2)=\pot(z_1\alpha_{S'_0,w'_0}z_2)
    \]    
\end{proof}

\section{On Charge and Unbounded Potential}
\label{sec:discharging and unbounded potential}
In this section we continue our examination of organized structures in runs on arbitrarily long words, akin to \cref{sec:existence of separated increasing infixes} (which we also rely on). 
Our focus now is on the behavior of the charge ($\charge$) and its implication on the potential ($\pot$) (see \cref{def:charge,def:potential}).

The structure of this section resembles that of \cref{sec:existence of separated increasing infixes}: we consider words of the form $w_1w_2w_3w_4$ (with an implicit decomposition as in \cref{rmk:w1w2w3w4 is a decomposition}) and define a new specialization of elongated words sequences, as well as concepts of independent runs, types and dips, and some new characteristics. 
We then prove certain properties of words that have these characteristics. 
The proofs of this section, however, are more involved than those of \cref{sec:existence of separated increasing infixes}, and make use of baseline shifts (\cref{sec: baseline shift}) and flattening (\cref{sec: cactus unfolding}).

In this section we fix a finite alphabet $\Gamma'\subseteq \Gamma^0_\infty$ -- that is, without rebase letters (and also without jump letters). In addition, we require that $\Gamma_0^0\subseteq \Gamma'$ (i.e., $\Gamma'$ contains the original letters of $\augA$).
We discuss these assumption when it becomes relevant.
Let $W_{\max}$ be the maximal weight appearing in a transition over $\Gamma'$.

\subsection{A Condition for Unbounded Potential}
\label{sec: closing runs condition implies unbounded potential}
Our starting point is now an elongated words sequence $\words$ (note that we drop the assumption that $|w_3|=m$ added in \cref{def:exact word sequence}). Unlike \cref{sec:existence of inc inf assuming covered}, in this section we do not focus on each state reachable after $w_1$, but rather on the configurations visited after $w_1w_2$.

Our first definition is a variant of independent runs. 
\begin{definition}[Configuration-Independent Runs]
    \label{def:configuration independent runs}
    Consider a word $w_1w_2w_3w_4$ and two seamless runs $\rho_1,\rho_2$ over it. We say that $\rho_1$ and $\rho_2$ are \emph{configuration-independent on $w_3$} if there exists $G\in \bbN$ such that for every prefix $v$ of $w_3$ we have $\rho_2(w_1w_2v)-\rho_1(w_1w_2v)>G$. 
    In this case we say that $\rho_2$ is \emph{above $\rho_1$ with gap $G$}.

    We say that $\rho_2$ and $\rho_1$ are \emph{not within gap $G$} when the order is not important.

    A set $\{\rho_1,\ldots, \rho_k\}$ of runs is \emph{configuration-independent with gap $G$} if the runs are independent in pairs with gap $G$. A singleton is independent by definition.
\end{definition}
Note that since we now focus on all the runs starting from $\xconf(s_0,w_1w_2)$, it follows that the maximal number of independent runs is at most $|S|$.

We proceed with an analogue of sparse words. Since we now look at general $\words$, we define the gaps between runs as a function of $|w_3|$ instead of a function of $m$. 
In particular, when we apply the definition, we define the gap explicitly as an expression in $|w_3|$ (see e.g., \cref{lem: decomposed seq with cover sparse no ghosts implies unbounded potential} below).
\begin{definition}[$\ell$-Sparse $K$-Gap Words]
\label{def:discharging ell sparse K gap}
For $\ell\in \bbN$ and $K:\bbN\to \bbN$ we say that an elongated word sequence $\words$ is \emph{$\ell$-sparse} with $K$-gap if for every $m\in \bbN$ with $\words(m)=w_1w_2w_3w_4$ there is a set $I$ of independent runs with gap $K(|w_3|)$
such that $|I|=\ell$.
\end{definition}

In the following, we consider \emph{decompositions} of $w_3$. Intuitively, a decomposition of $w_3$ is simply a concatenation $w_3=u_1u_2\cdots u_k$ for some $k\in \bbN$, such that the number of segments increases with $m$. We later impose further conditions on these decompositions, but for now they only serve as parts of definitions.
\begin{definition}[Decomposed Words Sequence]
    \label{def:decomposed words sequence}
    A sequence $\words$ is \emph{decomposed} if for every $m\in \bbN$ with $\words(m)=w_1w_2w_3w_4$ there is an associated decomposition $w_3=u_1u_2\cdots u_{d(m)}$ such that $\lim_{m\to \infty}d(m)=\infty$ and all the $u_i$ are nonempty.
\end{definition}
As usual, we assume that a decomposed word sequence has a fixed associated decomposition function, so we can refer to ``the decomposition of $\words(m)$'' for every $m\in \bbN$.

Our next definitions are variants of a $G$-cover (\cref{def:G cover inc infix}). 
\begin{definition}[Configuration $G$-Cover]
    \label{def: config G cover words}
    Consider an $\ell$-sparse decomposed sequence $\words$. 
    For every $m\in \bbN$ and $\words(m)=w_1w_2w_3w_4$ let $I=\{\rho_1,\ldots, \rho_\ell\}$ be the corresponding $\ell$-sparse independent runs. 
    
    For $G\in \bbN$, we say that $\words$ is a \emph{$G$-cover} if for every prefix $u$ of $w_3$ and configuration $\vec{c}=\xconf(s_0,w_1w_2u)$, for every state $q\in \booltrans(s_0,w_1w_2u)$ we have $\min\{|\vec{c}(q)-\weight(\rho_i(w_1w_2u))|: 1\le i\le \ell\}\le G$. 
\end{definition}
The next cover definition differs from \cref{def: config G cover words} in that we require the cover to hold only at end points of the decomposition segments.
\begin{definition}[Decomposition $G$-Cover]
    \label{def: G cover decomposed words}
    Consider an $\ell$-sparse decomposed sequence $\words$. 
    For every $m\in \bbN$ and $\words(m)=w_1w_2w_3w_4$ with decomposition $w_3=u_1\cdots u_{d(m)}$ let $I=\{\rho_1,\ldots, \rho_\ell\}$ be the corresponding $\ell$-sparse independent runs. 
    
    For $G\in \bbN$, we say that $\words$ is a \emph{decomposition $G$-cover} if for every $0\le j\le k$ and configuration $\vec{c_{j}}=\xconf(s_0,w_1w_2u_1\cdots u_j)$, for every state $q\in \booltrans(s_0,w_1w_2u_1\cdots u_j)$ we have $\min\{|\vec{c_j}(q)-\weight(\rho_i(w_1w_2u_1\cdots u_j))|: 1\le i\le \ell\}\le G$. 
\end{definition}

As mentioned above, our definitions here focus on reachable states. Nonetheless, the presence of ghost runs still introduces complications. The following definition (when applicable) enables us to mitigate the effect of these runs.
\begin{definition}[Ghost-Free Decomposition]
\label{def:ghost free decomposition}
 Consider a decomposed sequence $\words$.
 For every $m\in \bbN$ and $\words(m)=w_1w_2w_3w_4$ with decomposition $w_3=u_1\cdots u_{d(m)}$ we say that $\words$ is \emph{ghost-free} if there are no ghost runs from $w_1w_2u_1\cdots u_{i-1}$ to $w_1w_2u_1\cdots u_i$ for every $1\le i< d(m)$.
 (note that there is no requirement on the last segment, which can be thought of as a ``remainder'').
\end{definition}
Note the this requirement is not only that there are no ghost runs spanning the entire $w_3$, but that even when focusing on a single segment in the decomposition, there is no seamless ghost run on it. In particular, a ghost-free  decomposition has no ghost runs on any concatenation $u_i\cdots u_{j}$ for any $i\le j$ (this follows from the fact that an infix of a ghost run is also a ghost run, as noted after \cref{def:ghost run}).
Still, there could be ghost runs that start e.g., in the middle of segment $u_i$ and are cut in the middle of $u_{i+1}$.

Our next requirement of the decomposition is more involved, and assumes there exists a run whose gap from the minimal states gets smaller. 
\begin{definition}[\keyicon Decreasing-Gap Fair Decomposition]
    \label{def:decreasing gap fair decomposition}
    Consider a decomposed word sequence $\words$ and assume there exists a function $b:\bbN\to \bbN$ with $\lim_{n\to \infty}b(n)=\infty$ such that for every $m\in \bbN$ with $\words(m)=w_1w_2w_3w_4$, there is some seamless run $\rho$ on $w_1w_2w_3$ on such that   $\weight(\rho(w_1w_2))+\charge(w_1w_2)>\weight(\rho(w_1w_2w_3))+\charge(w_1w_2w_3)+b(m)$.

    We say that the decomposition is \emph{decreasing-gap fair} (with respect to $\rho$ and $b(n)$) if there is a function $\sqrtb:\bbN\to \bbN$ with $\lim_{n\to \infty}\sqrtb(n)=\infty$ such that the decomposition $w_3=u_1\cdots u_{d(m)}$ satisfies for every $1\le i<d(m)$ that
    \[ \weight(\rho(w_1w_2u_1\cdots u_{i-1}))+\charge(w_1w_2u_1\cdots u_{i-1})>\weight(\rho(w_1w_2u_1\cdots u_{i}))+\charge(w_1w_2u_1\cdots u_{i})+\sqrtb(m)
    \]
    (note that there is no requirement on the last segment, which can be thought of as a ``remainder'').
\end{definition}
Intuitively, a decreasing-gap fair decomposition guarantees that not only the number of segments increases, but also the gap between the minimal state and $\rho$ decreases significantly with every segment (except possibly the last segment). We remark that the expressions ``$+\charge$'' should be read as ``difference from the minimal weight'', since the charge is the inverse of the minimal weight. This idea is depicted in~\cref{fig:fair decomp}

\begin{figure}[ht]
    \centering
    \includegraphics[width=0.75\linewidth]{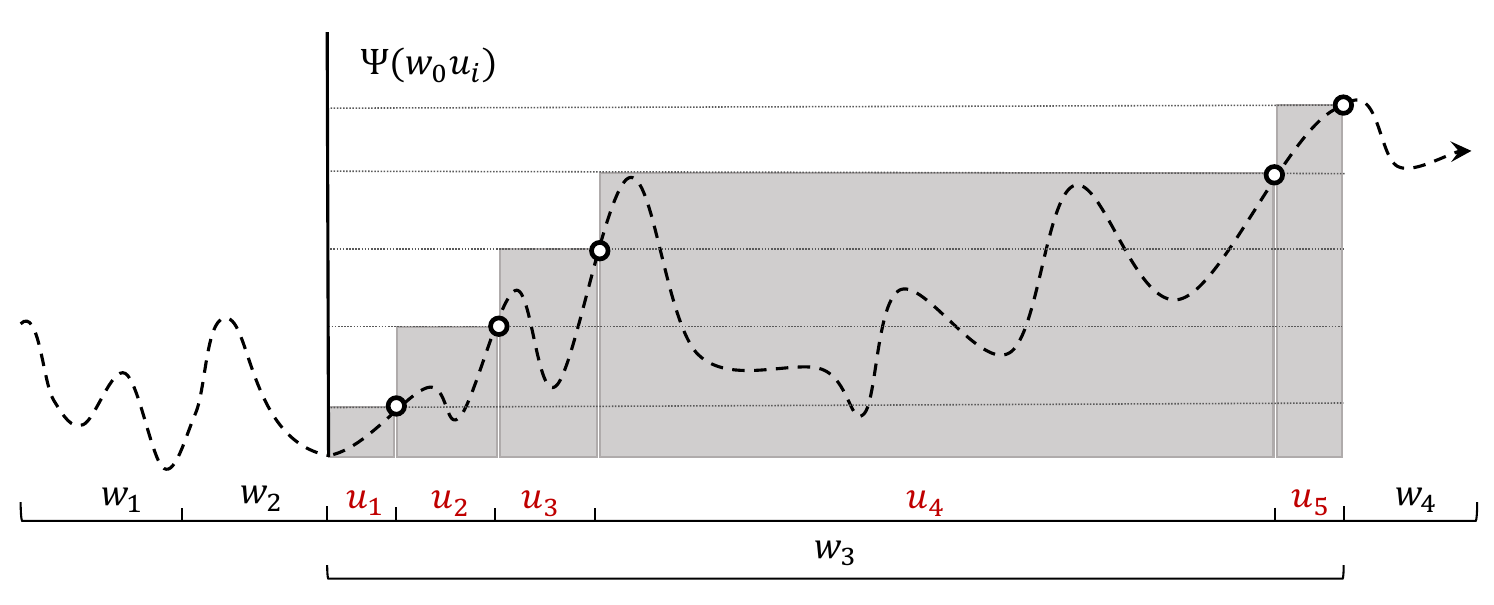}
    \caption{For visual simplicity, we assume $\rho$ is constant $0$. The overall charge in every segment decreases (with respect to $\rho$), where the minimal charge for every segment is at its end. The segments are marked by the first time the charge crosses the threshold.}
    \label{fig:fair decomp}
\end{figure}

In this section we also make use of the $\difftype$, always starting from $s_0$. 
We remind that for a word $x$ we have $\difftype_{s_0,G}(x)=\tup{\supp(\vec{c_x}),\le_x,\near_{s_0,G}(x)}$ where $\le_x$ is the relative order of the states according to $\vec{c_x}$ and $\near_{s_0,G}$ specifies which states are assigned values within gap $G$ of each other by $\vec{c_x}$. Recall that the (finite) set of all $\difftype$ is denoted $\difftypeset$.

\paragraph{Main result of this section -- overview}
Consider a decomposed sequence $\words$ that satisfies enough of the ``nice'' properties above (i.e., it is sparse with a large gap, has a cover, is ghost-free, and has a decreasing-gap fair decomposition). 
Our main result in this section is that under this assumption, the potential ($\pot$, see \cref{def:potential}) over $\Gamma_0^0$ is unbounded. 

A high-level overview of this result is given in \cref{sec:abs:from decomposition to result horror}, and is depicted in \cref{fig:proofHorror}.
We now give a slightly more technical overview, followed by the complete proof. 

Recall that for a word $w$ with a seamless baseline run, the \emph{potential} $\pot(w)$ is the gap between the maximal dominant state and the baseline, which in turn means that there is some suffix $z$ (over $\Gamma_\infty^\infty$) such that reading $z$ from the dominant state yields finite weight, but all states lower than the dominant state (after reading $w$) do not yield finite-weight run on $z$.

Under the assumptions on $\words$, consider $\words(m)=w_1w_2w_3w_4$ for some very large $m$. In particular, we take $m$ large enough so that the decomposition has a repetition of $\difftype$ in some infix $x$ of $w_3$ (this is possible since the number of segments increases with $m$).

By our assumption, the minimal run $\rho_{\min}$ gets close to some seamless run $\rho$. For intuition purposes, it's useful to think of the minimal run as increasing and on $\rho$ as decreasing (\cref{fig:abs:proofStep1}).
Thus, there are significantly-decreasing runs, and in particular we find such runs on $x$ (note that for this we need the properties of the fair decomposition).

We then \emph{shift} our perspective so that the most decreasing seamless run $\rhoss$ becomes baseline (see Step 2). We observe that $\rhoss$ is then far above the minimal run, which is now highly increasing. Moreover, there are now no decreasing runs (\cref{fig:abs:proofStep2}).

We then show (Step 3) that under sufficient assumptions, $x$ actually induces a stable cycle $\alpha_{S',x}$ (this is the more technical part of the proof). 
Now, since the minimal run increases so much, it follows that it has no grounded pairs, and therefore attains weight $\infty$ upon reading $\alpha_{S',x}$ (see Step 4\footnote{In \cref{fig:proofHorror} we combine Steps 3 and 4, so there is an offset in the numbering, i.e. Step 5 here is Step 4 in the figure.}, depicted in \cref{fig:abs:proofStep3}). 
Intuitively, we are already close to showing that the potential is unbounded, since we have a large gap between runs, and the lower runs jump to $\infty$. However, the precise definition of potential requires the gap to be from the baseline run, not the minimal run (indeed, if e.g., $\rhoss$ is the baseline and contains the dominant state, the potential would be $0$). 

We therefore want to make the minimal run the baseline. Thus, in Step 5 we \emph{shift} again, this time making the prefix of the minimal run up to $x$ the baseline (\cref{fig:abs:proofStep4}). Then, all that remains (almost) is to construct a suffix that would attain finite value from the shifted $\rhoss$, but not from minimal run. This is achieved by using a \emph{jump letter} and $\alpha_{S',x}$.

A caveat to the above is that it uses the general alphabet $\Gamma'$, whereas we want to show the unboundedness of the potential already over $\Gamma_0^0$. Therefore, already in Step 1 we \emph{flatten} the prefix before $x$. This guarantees that all the letters before $\alpha_{S',x}$ are from $\Gamma_0^0$, a property that is maintained by shifts.

\begin{lemma}[\keyicon \lightbulbicon Decomposed Words to Unbounded Potential]
    \label{lem: decomposed seq with cover sparse no ghosts implies unbounded potential}
    Assume there exists $G,\ell\in \bbN$ and a decomposed word sequence $\words$ that is $\ell$-sparse with gap 
    $4G + (|S|+1)4\bigM W_{\max} |w_3|$ and a ghost-free decomposition $G$-cover.
    \acctodo{ACCOUNTING}
    In addition, assume the decomposition is decreasing-gap fair with respect to a seamless run $\rho$ and function $b:\bbN\to \bbN$.
    Then $\sup\{\pot(w)\mid w\in (\Gamma^0_0)^*\}=\infty$.
\end{lemma}
\begin{proof}
Let $P\in \bbN$, we construct a word over $\Gamma_0^0$ whose potential is at least $P$. Let $G,\ell\in \bbN$ and $\words$ be as in the premise of the lemma, and $b,\sqrtb:\bbN\to \bbN$ and $\rho$ as per \cref{def:decreasing gap fair decomposition}. 
Let $d:\bbN\to \bbN$ be the decomposition length function, i.e., for $m\in \bbN$ and $\words(m)=w_1w_2w_3w_4$ we have $w_3=u_1\cdots u_{d(m)}$.

Take $m\in \bbN$ large enough so that $m>P$,  $d(m)>|\difftypeset_{s_0,G}|+1$ and $\sqrtb(m)>4G$. Since both $d$ and $\sqrtb$ are increasing, this is possible.
Consider the word $w_1w_2w_3w_4=\words(m)$ and let 
$w_3=u_1\cdots u_{d(m)}$ be its decreasing-gap fair decomposition.

By the pigeonhole principle, there exist $0\le i_1<i_2<d(m)$ such that 
\[\difftype_{s_0,G}(w_1w_2u_1\cdots u_{i_1})=\difftype_{s_0,G}(w_1w_2u_1\cdots u_{i_2})\]
(Observe that we take $d(m)$ large enough so that we can ignore the last segment, which is a remainder).

Denote 
$\widehat{x_0}=w_1w_2u_1\cdots u_{i_1}$ (the reason for the hat is that we modify $\widehat{x_0}$ to obtain a ``simpler'' $x_0$ below), $x=u_{i_1+1}\cdots u_{i_2}$ and $x_1=u_{i_2+1}\cdots u_{d(m)}w_4$. 

\paragraph{Step 1: flatten the prefix before $x$.}
Recall that $\Gamma'\subseteq \Gamma_\infty^0$, i.e., has no rebase letters. It follows that the letters in $w_1w_2w_3w_4$ are either in $\Gamma_0^0$ or are nested cacti. 
Let $F>2\maxeff{\widehat{x_0}x}$, we now define $x_0=\flatten(\widehat{x_0} \wr F)$.
By \cref{def:cactus rebase flattening}, it follows that in order to obtain $x_0$, we only apply the cactus unfolding of \cref{def:cactus extension}. This implies that no jump letters are introduced in $x_0$ (indeed, those are introduced only in the case of rebase removal). Thus, $x_0\in (\Gamma_0^0)^*$. 
In the following we wish to reason about $x_0xx_1$ instead of $w_1w_2w_3w_4$. 
However, the main difference between the words is that due to the flattening, we now have $\booltrans(s_0,x_0)=\ghostTrans(s_0,\widehat{x_0})$. 
Consider the configurations $\vec{c'_0}=\xconf(\vec{c_{\init}},\widehat{x_0})$ and $\vec{c_0}=\xconf(\vec{c_\init},x_0)$, and similarly $\vec{c'_x}=\xconf(\vec{c_{\init}},\widehat{x_0}x)$ and $\vec{c_x}=\xconf(\vec{c_\init},x_0x)$.
By \cref{lem:flattening configuration characterization} for every $s,s'\in \booltrans(s_0,\widehat{x_0})$ we have $\vec{c'_0}(s)-\vec{c'_0}(s')=\vec{c_0}(s)-\vec{c_0}(s')$. Additionally, for $r\in \ghostTrans(s_0,\widehat{x_0})\setminus \booltrans(s_0,\widehat{x_0})$ and $r\in \booltrans(s_0,\widehat{x_0})$ we have 
$\vec{c_0}(r)-\vec{c_0}(s)\ge 2\maxeff{\widehat{x_0}x}$.
\acctodo{ACCOUNTING}
And in particular ghost states attain higher values than states reachable by $\widehat{x_0}$. In the following, we only use the fact that these are higher, and not the concrete gap size.

Let $\rho_{\min}$ 
be the minimal-weight independent run on $w_3$ among the assigned set of independent runs on $w_1w_2w_3w_4$ (recall that by  \cref{def:configuration independent runs} the independence is measured on the $w_3$ infix). In particular, $\rho_{\min}$ is seamless.
Note that since $x$ is a concatenation of segments from the fair decomposition, then the gap between the minimal states and $\rho$ 
decreases significantly upon reading $x$. 
We claim that this implies the gap between $\rho$ and $\rho_{\min}$ also significantly decreases. 
Intuitively, this holds because $\rho_{\min}$ remains close to the minimal state (by the $G$-cover property).

Formally, let $q\in \arg\min\{\vec{c_0}(p)\mid p\in S\}$. By 
\cref{lem:flattening configuration characterization}
we also have $q\in \arg\min\{\vec{c'_0}(p)\mid p\in S\}$. In particular, $\charge(x_0)=-\vec{c_0}(q)$. 
Since we assume that the independent runs form a $G$-cover, then the minimal state $q$ must by within distance $G$ from $\rho_{\min}$. That is, $|\vec{c_0}(q)-\rho_{\min}(x_0)|\le G$. Similarly, $\charge(x_0x)=-\vec{c_x}(q')$ where $q'\in \arg\min\{\vec{c_x}(p)\mid p\in S\}$ and $|\vec{c_x}(q')-\rho_{\min}(x_0x)|\le G$.

Since $\weight(\rho(x_0))+\charge(x_0)>\weight(\rho(x_0x))+\charge(x_0x)+\sqrtb(m)$
we have by the triangle inequality that 
\[\weight(\rho(x_0))-\weight(\rho_{\min}(x_0))>\weight(\rho(x_0x))-\weight(\rho_{\min}(x_0x))+\sqrtb(m)-2G\]
(note that we replace e.g., $\charge(x_0)$ with $-\weight(\rho_{\min}(x_0))$, since the charge is the inverse of the minimal weight).
%
In fact we can increase the bound to $(i_2-i_1)\sqrtb(m)$, but the bound we state suffices.
Rearranging, we get that 
\[
\weight(\rho_{\min}(x_0x))-\weight(\rho_{\min}(x_0))>\weight(\rho(x_0x))-\weight(\rho(x_0))+\sqrtb(m)-2G
\]
i.e., $\rho_{\min}$ is significantly increasing on $x$ relatively to $\rho$.

\paragraph{Step 2: baseline shift to sharpest decrease.}
Define $\rhoss$ to be the seamless run that has the sharpest decrease on $x$. 
Specifically, $\rhoss$ is the seamless run that minimizes $\weight(\eta(x_0x))-\weight(\eta(x_0))$ among all seamless runs $\eta$ on $x_0x$. Note that $\rhoss$ is not $\rho_{\min}$, as at least $\rho$ is certainly decreasing more sharply than $\rho_{\min}$ by the above (in case of several such runs, we choose one arbitrarily).

We now perform a baseline shift (see \cref{sec: baseline shift}) on $\rhoss$. Observe that $\rhoss$ is an entire run on $w_1w_2w_3w_4$, and therefore this baseline shift acts on all runs on $w_1w_2w_3w_4$. 
Since $\rhoss$ has the sharpest decrease on $x$, and since shifts maintain the gaps between runs (\cref{cor:baseline shift maintains gaps}), we have that all the shifted runs are non-decreasing. That is, for every run $\eta$ on $\baseshift{x_0x}{\rhoss}$ it holds that 
\begin{equation}
\label{eq: shifted runs are nonnegative in discharging}
    \weight(\baseshift{\eta}{\rhoss}(x_0x))-\weight(\baseshift{\eta}{\rhoss}(x_0))\ge 0
\end{equation}
Indeed, by \cref{cor:baseline shift to seamless run} we now have that $\baseshift{\rhoss}{\rhoss}$ is the baseline run (and in particular remains at weight $0$). Therefore, by \cref{cor:baseline shift maintains gaps} we have
\[\weight(\eta(x_0x))-\weight(\rhoss(x_0x))=\weight(\baseshift{\eta}{\rhoss}(x_0x))-0\] 
and similarly 
\[\weight(\eta(x_0))-\weight(\rhoss(x_0))=\weight(\baseshift{\eta}{\rhoss}(x_0))-0\] 
By subtracting the equations we get 
\[
\begin{split}
&\weight(\baseshift{\eta}{\rhoss}(x_0x))-\weight(\baseshift{\eta}{\rhoss}(x_0))=\\
&\weight(\eta(x_0x))-\weight(\rhoss(x_0x)) - \weight(\eta(x_0)) + \weight(\rhoss(x_0))=\\
&\weight(\eta(x_0x))- \weight(\eta(x_0))- (\weight(\rhoss(x_0x))  - \weight(\rhoss(x_0)))\ge 0
\end{split}
\]
Where the last inequality is because $\rhoss$ minimizes $(\weight(\eta(x_0x))  - \weight(\eta(x_0)))$ by definition.

To avoid the cumbersome notation of baseline shifts in the following, we wish to reuse the names $x$, $w_1w_2w_3w_4$, and  the runs $\rho_{\min}, \rhoss$ and $\rho$. 
To justify this abuse, we notice that since baseline shifts preserve the run structure (\cref{cor:baseline shift maintains gaps}), then the shifts of the independent runs of $w_1w_2w_3w_4$ retain the properties of being $\ell$-sparse with gap $4G+(|S|+1)|w_3|$ and being a ghost-free $G$-cover.
In particular, we now have that $\rhoss$ is the baseline run (and in particular remains at weight $0$). 

\paragraph{Step 3: $x$ induces a stable cycle.}
Since all the runs on $x$ are nonnegative (\cref{eq: shifted runs are nonnegative in discharging}) and $x$ yields a repetition of configurations, we gain the intuition that $x$ should induce a stable cycle. We show this is indeed the case. The difficult part is to show that there are no negative cycles induced by $x^n$ for any $n\in \bbN$. 
The proof relies on the fact that there are large gaps between the runs, and that the runs are a $G$-cover. Broadly, this is a similar proof to \cref{prop:exists inc inf x stable cycle}.

Let $S'=\ghostTrans(s_0,x_0)$, we claim that $(S',x)$ is a stable cycle. Let $\vec{c_0}=\xconf(s_0,x_0)$ as before and $\vec{c_x}=\xconf(s_0,x_0x)$. 
Since $\difftype_{s_0,G}(\widehat{x_0})=\difftype_{s_0,G}(\widehat{x_0}\cdot x)$, it follows that $S'=\ghostTrans(s_0,x_0 \cdot x)$. Indeed, we have $\booltrans(s_0,x_0)=\supp(\vec{c_0})=\supp(\vec{c_x})=\booltrans(s_0,x_0\cdot x)$, and by \cref{sec:reachable ghost states} this also implies the equivalent ghost states. Thus, we have in particular that $\booltrans(S',x)\subseteq S'$, so $(S',x)$ is a reflexive cycle.

The next step (as per \cref{def:stable cycle}) is to show that for the baseline state $s\in S'$ we have that $s\in\MinRefStates(S',x^n)$ for all $n\in \bbN$, and  that $\minweight(x^n,s\to s)=0$. 
Recall that after the baseline shift, we have that $\rho_\ssearrow$ is a baseline run on $x$ and has weight $0$, in particular $\weight(\rho_\ssearrow)=0$, and therefore $\weight(\rho^n_\ssearrow)=0$, so $x^n$ has a cycle of weight $0$. 

It remains to show that there are no negative cycles in $(S',x^n)$ for any $n$. That is, for every $s'\in S'$, we want to show that $\minweight(x^n,s'\to s')\ge 0$.
We start with an overview of how this is proved. 
Recall that our baseline shift operation ensures that $\rho_\ssearrow$ has weight $0$, and $\rho_\ssearrow$ before the rebase has the strongest decrease, so all seamless runs after rebase are non-negative (on $x$). 
This, however, does not immediately extend to $x^n$. Indeed, it could conceptually be that there are e.g., (non-seamless) negative runs $r_1\runsto{x} r_2$ and $r_2\runsto{x} r_1$, but then their concatenation yields a negative seamless cycle on $x^2$.

We therefore take an approach similar to the proof of \cref{lem:unbounded long dip and G imply separated inc infix}. Specifically, we notice that the equivalence of $\difftype_{s_0,G}$ before and after reading $x$ implies that we can partition these configurations to ordered sub-configurations $V_1,\cdots V_k$ such that all the ``interesting'' transitions are from $V_i$ are to $V_i$. We then show that after reading $x$, the actual values of the elements in the configurations of each $V_i$ does not decrease (due to $\rho_\ssearrow$ being the baseline). This allows us to conclude there are no negative cycles. 
For the ghost-states, we use the ghost-freeness assumption to conclude there cannot be negative cycles on them.
We proceed with the formal details.

We remind of the following notations of Step 1. $\vec{c'_0}=\xconf(s_0,\widehat{x_0})$ and $\vec{c_0}=\xconf(s_0,x_0)$. Also denote $\vec{c'_x}=\xconf(s_0,\widehat{x_0}x)$ and $\vec{c_x}=\xconf(s_0,x_0x)$. 
Strictly speaking, note that in Step 1 these configurations are defined before our baseline shift of Step 2. However, the gaps are maintained by the shift, so this overloading is justified.

Since $\difftype_{s_0,G}(\widehat{x_0})=\difftype_{s_0,G}(\widehat{x_0}x)$, and by the preservation of the reachable set, we have that for every $s\in \booltrans(s_0,\widehat{x_0})=\booltrans(s_0,\widehat{x_0}x)$ that $\vec{c_0}(s)=\minweight(x_0,s_0\to s)$ and $\vec{c_x}(s)=\minweight(x_0x,s_0\to s)$.
Put simply -- the configurations correctly track the weights of all states reachable by $\widehat{x_0}$ (we deal later with the remaining states, namely $S'\setminus\booltrans(s_0,\widehat{x_0})$).

We partition $\booltrans(s_0,\widehat{x_0})$ as $V_1\cup \ldots \cup V_k$ according to $\difftype_{s_0,G}(\widehat{x_0})$, similarly to the proof of \cref{lem:unbounded long dip and G imply separated inc infix}. Specifically, we write $p\equiv q$ if $|\vec{c_0}(p)-\vec{c_0}(q)|\le 2G$. Since $w_1w_2w_3w_4$ is a $G$-cover, this is an equivalence relation. We then define $V_1,\ldots,V_k$ to be the equivalence classes of $\sim$, ordered by $\le_{\vec{c_0}}$, i.e., all the states in $V_1$ are lower than those of $V_2$, etc.
By the equivalence of $\difftype$, we have that the partition induces by $\vec{c_x}$ is the same as that of $\vec{c_0}$.

Still similarly to \cref{lem:unbounded long dip and G imply separated inc infix}, since $w_1w_2w_3w_4$  is $\ell$-sparse with a huge gap, we have that if $p\in V_i$, $q\in V_j$ and there is a run $\eta:p\runsto{x} q$, then $i\ge j$. That is, there are no runs that go ``up'' between the sets (as the gap is to large to be bridged by $x$, so this would ``drag down'' $q$).
Moreover, if $q\in V_j$ there exists some $p\in V_j$ such that $p\runsto{x} q$, i.e., there must exist a predecessor of $q$ in $V_j$ (otherwise all predecessors of $q$ come from higher sets, so $q$ would be higher). Additionally, there must exist such $p\in V_j$ such that $\eta:p\runsto{x} q$ is seamless (indeed, runs from higher sets to $q$ cannot be seamless).

We claim that for every $q\in \booltrans(s_0,\widehat{x_0})$ it holds that $\vec{c_0}(q)\le \vec{c_x}(q)$. Intuitively, $x$ only induces increasing behaviors.
To show this, let $1\le i\le k$ such that $q\in V_i$. Let $q^0_i$ be the minimal state in $V_i$ (i.e., $q^0_i\in \arg\min\{\vec{c_0}(q)\mid q\in V_i\}$).

Let $p^0_i$ be a state in $V_i$ such that there exists a seamless run $\eta:p^0_i\runsto{x}q^0_i$. Recall that our baseline shift ensures there are no seamless negative runs. It follows that $\vec{c_x}(q^0_i)=\vec{c_0}(p^0_i)+\weight(\eta)\ge \vec{c_0}(p^0_i)\ge \vec{c_0}(q^0_i)$, where the last inequality is because we know that $q^0_i$ is the minimal in $V_i$. 
Therefore, the claim holds for $q^0_i$. 
For every other state $q\in V_i$ we have that $q\sim q^0_i$, and therefore $\vec{c_x}(q)-\vec{c_x}(q^0_i)=\vec{c_0}(q)-\vec{c_0}(q^0_i)$, so $\vec{c_x}(q)-\vec{c_0}(q)=\vec{c_x}(q^0_i)-\vec{c_0}(q^0_i)\ge 0$.

Notice that the argument above only assumes that the gaps are larger than $|x|$. Recall, however, that we originally start with gaps larger than $(|S|+1) |w_3|\ge  (|S|+1)|x|$ (this includes the gap from the ghost states, which is even larger). It therefore follows that we can repeat this argument for up to $|S|$ times, and obtain that the weights of each state is increased after reading $x^n$ for any $n\le |S|$.  

Now, assume by way of contradiction that there is a negative cycle $\eta:q\runsto{x^n}q$ such that $q\in \booltrans(s_0,\widehat{x_0})$ and $n\in\bbN$. We first notice that we can assume $n\le |S|$ (otherwise we can find an inner $x^{n'}$-cycle that is either itself a negative cycle, or is positive and we can shorten the original cycle). This, however, implies that $\minweight(x_0,s_0\to q)> \minweight(x_0 x^n,s_0\to q)$, which is a contradiction to our observation above.

It remains to show that there are no negative cycles stemming from ghost-states. 
In fact, we prove that there are no cycles (on $x^n$) on ghost states at all, negative or not.
This follows from ghost-freeness assumed in the lemma, as follows.
Assume by way of contradiction that there exists $g\in \ghostTrans(s_0,\widehat{x_0})\setminus \booltrans(s_0,\widehat{x_0})$ and $\eta:g\runsto{x^n}g$. We decompose $\eta$ as $\eta:g\runsto{x}g'\runsto{x^{n-1}}g$ according to the state reached after $x$.
Since $x$ spans at least one $u_i$ segment by definition, it follows from ghost-freeness (\cref{def:ghost free decomposition}) that $\rho$ is not seamless already before $g'$, i.e., we can write $x=y_1y_2$ such that $\eta:g\runsto{y_1}h\runsto{y_2}g'\runsto{x^{n-1}}g$ and we can assume $h\in \booltrans(s_0,\widehat{x_0}y_1)$. 

There is a slightly delicate point here: it could be that $\eta$ is not seamless but that the run that ``undercuts'' it is also from a ghost-state. However, we can then repeat this argument until eventually reaching a ``minimal'' ghost run, after which the undercutting run indeed uses a state in $\booltrans(s_0,\widehat{x_0}y_1)$.

Since $\eta$ is a run, and in particular has finite weight, it follows that its suffix $\eta':h\runsto{y_2}g'\runsto{x^{n-1}}g$ is also finite weight, and since $h\in \booltrans(s_0,\widehat{x_0}y_1)$, we have that $g\in \booltrans(s_0,\widehat{x_0}x^n)$, but $x$ preserves the set of reachable states (by the equivalence of $\difftype$), so it follows that $g\in \booltrans(s_0,\widehat{x_0})$, in contradiction to the assumption that $g\notin \booltrans(s_0,\widehat{x_0})$.

We conclude that $(S',x)$ is a stable cycle, and so $\alpha_{S',x}$ is a cactus letter.  

\paragraph{Step 4: replace $x$ with $\alpha_{S',x}$.}
We now focus on the word $x_0\alpha_{S',x}$. 
Our goal is to show that the minimal run after $x_0$, upon reading $\alpha_{S',x}$, ``jumps'' to $\infty$. 
Intuitively, due to the 
increase of the minimal run, we can show that the minimal class $V_1$ (of states equivalent to the minimal run) does not contain any grounded pairs of $\alpha_{S',x}$ (see \cref{def:grounded pairs}). Hence, there are no finite-value transitions on $\alpha_{S',x}$ from this set.

Formally, recall from Step 1 that the minimal independent run (on $x$) $\rho_{\min}:s_0\runsto{x_0}s_1\runsto{x}s_2$ satisfies 
\begin{equation}
\label{eq: lower bound on rho min gap on x}
  \weight(\rho_{\min}(x_0x))-\weight(\rho_{\min}(x_0))>\weight(\rho(x_0x))-\weight(\rho(x_0))+\sqrtb(m)-2G>2G  
\end{equation}
Where the last inequality is since after the baseline shift we have that $\rho$ is non-decreasing, so  $\weight(\rho(x_0x))-\weight(\rho(x_0))\ge 0$, and by our choice of $\sqrtb(m)> 4G$.

Let $V_1$ be the minimal equivalence class of states defined in Step 3. Due to the $G$-cover property, we have that $s_1\in V_1$ and for each $s'\in V_1$ we have $|\vec{c_0}(s_1)-\vec{c_0}(s')|\le G$ and similarly $|\vec{c_x}(s_1)-\vec{c_x}(s')|\le G$. Consequently, for every $s',s''\in V_1$ we have by the triangle inequality and \cref{eq: lower bound on rho min gap on x} that
$\minweight(x,s'\to s'')>2G-2G=0$. 

We now claim that for every $s',s''\in V_1$ we have that $(s',s'')\notin \GroundPairs(S',x)$. In order to prove this, we use the contrapositive of \cref{lem:pumping grounded pairs}: it suffices to show that 
$\minweight(x^{4\bigM},s'\to s'')>\minweight(x^{2\bigM},s'\to s'')$, since if $(s',s'')$ were a grounded pair, these weights would be equal.

By our observation above, we \emph{almost} have that $\minweight(x^{k+1},s'\to s'')> \minweight(x^k,s'\to s'')$. Indeed, we already have that any runs on $x$ between states in $V_1$  are strictly increasing. However, a-priori it could be that after enough repetitions of $x$, runs from e.g., $V_2$ would decrease and become the minimal on $x^k$.
However, we now recall that the gap between the independent runs is greater than $(|S|+1) 4\bigM W_{\max} |w_3|>  (|S|+1)4\bigM W_{\max}|x|$. It follows that with each repetition of $x$, the gap between runs from $V_1$ and from $V_2$ (or higher) can become closer by at most $2|x| W_{\max}$. In particular, after $4\bigM$ iterations, the gap is still at least 
\[(|S|+1)4\bigM W_{\max}|x|-8\bigM W_{\max}|x|>0\]
where we assume $|S|\ge 2$ (which holds if the original WFA $\cA$ has at least one state).
and therefore no runs starting at $V_2$ or higher can become minimal. We thus conclude that $(s',s'')\notin \GroundPairs(S',x)$.

It follows (by \cref{def:stabilization}) that $\minweight(\alpha_{S',x},V_1\to S)=\infty$, i.e., all states in $V_1$ ``jump to infinity'' upon $\alpha_{S',x}$.




\paragraph*{Step 5: baseline shift to minimal run}
After completing Step 4, we are now at the following setting: upon reading the prefix $x_0\in (\Gamma_0^0)^*$, the next letter is $\alpha_{S',x}$, and all states in $V_1$ cannot read it, i.e, jump to $\infty$.

Intuitively, this almost shows that the potential of $x_0$ is at least $P$ -- there is a run much higher than $V_1$, and all the states from $V_1$ go to $\infty$. However, our precise definition of potential requires that this happens with the baseline run. We therefore need to perform some baseline shifts. Unfortunately, this baseline shift makes $\alpha_{S',x}$ ``unreadable'' after the shift, as it now has the wrong format. We fix this by introducing jump letters. We proceed with the details.

Consider the seamless run $\mu$ that ends with minimal weight on $x_0$, and perform a baseline shift \emph{only on $x_0$} with respect to $\mu$. Note that $\mu$ ends in $V_1$ (since the minimal state is there).

We now have that $\baseshift{\mu}{\mu}$ is the baseline run on $\baseshift{x_0}{\mu}$, and we denote $\vec{c_1}=\xconf(\vec{c_\init},\baseshift{x_0}{\mu})$.
Recall that $m>P$ (this is chosen before Step 1). 

In the following, we still wish to reason about the $V_j$ partition. However, following the baseline shift, the state have changed. Understanding the change and showing that it is very manageable requires us to dig slightly deeper into the states reached after $x_0$ and after $\baseshift{x_0}{\mu}$.

Specifically, since $x_0$ is obtained by flattening $\widehat{x_0}$, there are no ghost states. That is,
the states reachable after $x_0$ are $S'=\{(p, q,T)\mid p\in T\}$  for some $q,T$ with $q\in T\subseteq Q$ (as per the construction of $\augA$ in \cref{sec:augmented construction}). Note that this is the same $S'$ as in the cactus letter $\alpha_{S',x}$. 

After baseline shift (\cref{prop:baseline shift run bijection}), the states reachable when reading $\baseshift{x_0}{\mu}$ are 
$S''=\{(p,q',T)\mid p\in T\}$ for some $q'\in T$. 

Apart from the change in the baseline component, however, there is no change in the state components or their gaps. We can therefore denote $\baseshift{V_j}{\mu}$ when we wish to refer to the set of states corresponding to $V_j$ after the baseline shift.
We can now proceed with showing the potential is unbounded.

Since the gap between independent runs is much larger than $m$, then in particular all states in $\baseshift{V_j}{\mu}$ with $j>1$ are at least $P$ above any state in $\baseshift{V_1}{\mu}$, after reading $\baseshift{x_0}{\mu}$. 
An important detail here is that any reachable state after $x_0$ that is not in any of the $V_j$, i.e., in $\ghostTrans(s_0,x_0)$, is still reachable, and the minimal weight to it is above all the $V_j$, and in particular these states also have a huge gap from $V_1$.

In order to conclude that $\pot(\baseshift{x_0}{\mu})>P$, it suffices (see \cref{def:potential,def:dominant state}) to find a suffix $z$ such that for some state $s\in S''$ such that 
$\minweight(z,s\to S)<\infty$ and for every $s'\in S''$ with $\vec{c_1}(s')<\vec{c_1}(s)$ we have 
$\minweight(z,s\to S)=\infty$.  
Crucially, the suffix $z$ need not be jump-free or over $\Gamma^0_0$.

Note that $\alpha_{S',x}$ has no transitions of finite weight from states in $S''$ due to the ``wrong'' baseline component $q'$.
Consider therefore the jump letter $\jl_{(\cdot,q',T)\to (\cdot,  q,T)}$ as per \cref{def:jump letters} and the suffix $z=\jl_{(\cdot,q',T)\to (\cdot,  q,T)}\alpha_{S',x}$. 

We first claim that there are no finite-weight runs from any $s'\in \baseshift{V_1}{\mu}$ on $z$. Indeed, for every state $s'\in \baseshift{V_1}{\mu}$, write $s'=(p,q',T)$, then after reading $\jl_{(\cdot,q',T)\to (\cdot,  q,T)}$, the next state is (deterministically) $(p,q,t)$, and the weight of the transition is $0$. However, this is now a state in $V_1$, and therefore has no finite-weight transition on $\alpha_{S',x}$, as per Step 4.

All that remains is to show that there is some reachable state $s$ for which 
$\minweight(z,s\to S)<\infty$.
Indeed, we can then take the minimal such state, and conclude the potential is above $P$.

The immediate candidate for $s$ is the baseline state of $S'$. That is, take $(q,q,T)\in S'$ and consider the corresponding state $(q,q',T)\in S''$, then after reading $\jl_{(\cdot,q',T)\to (\cdot,  q,T)}$ we arrive at $(q,q,T)$ with the same weight, and since $\alpha_{S',x}$ is a cactus letter and its baseline state is $(q,q,T)$, then $\minweight((q,q,T),\alpha_{S',x})=0$. In particular, we have that $(q,q,T)\notin V_1$, so $(q,q',T)\notin \baseshift{V_1}{\mu}$. We conclude that $\minweight((q,q',T),\jl_{(\cdot,q',T)\to (\cdot,  q,T)}\alpha_{S',x})=0<\infty$, and therefore as mentioned above, we conclude that $\pot(\baseshift{x_0}{\mu})>P$, where $\baseshift{x_0}{\mu}\in (\Gamma_0^0)^*$, as required. 
\end{proof}

We conclude this section with one more definition, to be used later on, and some simple corollaries of it. 
It is an analogue of $D$-dip words (\cref{def:dip words}), now requiring that there are no large dips only on an infix of $w_3$, as follows.
\begin{definition}[Infix $D$-dip]
\label{def:infix D dip}
For $D\in \bbN$, we say that $\words$ is \emph{Infix $D$-dip} if for every $n\in \bbN$ there exists some $m\in \bbN$ and $\words(m)=w_1w_2w_3w_4$ such that we can write $w_3=xvy$ with $|v|>n$ and for every $p\in \ghostTrans(s_0,w_1w_2x)$ and run $\rho:p\runsto{x}S$ it holds that $\weight(\rho)>-D$.
\end{definition}
Clearly there is a lot of similarity between Infix $D$-dip and the $D$-dip words of \cref{def:dip words}. As we now show, the former implies the latter. Intuitively, this is simple: just select the words that witness the infix $D$-dip property, and rewrite them as concatenations so that the infix $v$ becomes $w_3'$.

Note that the central result of \cref{sec:separated increasing infix}, namely \cref{cor:dip implies increasing infix}, shows that $D$-dip words implies the existence of a separated increasing infix. These observations are central tools in our final proof (\cref{sec:final nondet implies witness}).

\begin{proposition}
\label{prop:infix D dip implied D dip}
    Let $D\in \bbN$. If $\words$ is Infix $D$-dip, then there exists a faithful restriction (see \cref{def:elongated words sequence}) $\xwords$ of $\words$ such that $\xwords$ is an exact words sequence (\cref{def:exact word sequence}) and is $D$-dip.
\end{proposition}
\begin{proof}
    We define $\xwords$ as follows. For every $n\in \bbN$, by \cref{def:infix D dip} let $m\in \bbN$ such that $\words(m)=w_1w_2w_3w_4$ and we can write $w_3=xvy$ with $|v|>n$, where for every $p\in \ghostTrans(s_0,w_1w_2x)$ and run $\rho:p\runsto{x}S$ it holds that $\weight(\rho)>-D$.
    Define $\xwords(n)=w'_1w'_2w'_3w'_4$ with $w'_1=w_1$, $w'_2=w_2x$, $w'_3=v[1,n]$ and $w'_4=v[n,|v|]yw_4$. By definition, $\xwords$ is a faithful restriction of $\words$. Since we define $w'_3=v[1,n]$, we have $|w_3|=n$, so $\xwords$ is indeed an exact words sequence. In addition, the property of $v$ is exactly the $D$-dip property of \cref{def:dip words} (and in particular holds for any prefix of $v$, namely $w'_3$), so we are done.
\end{proof}
Combining \cref{prop:infix D dip implied D dip} with \cref{cor:dip implies increasing infix}, we have the following.
\begin{corollary}
    \label{cor:infix D dip implies increasing infix}
    Let $D\in \bbN$. If $\words$ is Infix $D$-dip, then there exists $m\in \bbN$ with $\words(m)=w_1w_2w_3w_4$ and a decomposition $w_3=u'xyv'$ such that $uxyv$ is a separated increasing infix from $\ghostTrans(s_0,w_1)$ for $u=w_2u'$ and $v=v'w_4$.
\end{corollary}

\subsection{Leveled Words and Discharging Words}
\label{sec:leveled and discharging words}
In this section we focus on two types of $\words$ sequences. The first, dubbed \emph{leveled words}, models the case where the charge remains bounded. 
It's dual, dubbed \emph{discharging words}, models the case where the charge significantly decreases along $w_3$ (i.e., the minimal run is increasing).

We show that the existence of either type of words sequence (under very mild assumptions) implies that either the potential is unbounded, or the sequence has the $D$-dip property. These properties are central to our reasoning in the final proof (\cref{sec:final nondet implies witness}).

\subsubsection{Leveled Words Sequences}
\label{sec:leveled words sequences}

An elongated words sequence $\words$ is leveled if, intuitively, the charge along the $w_3$ infix remains in some constant width ``band''.
\begin{definition}[$\kappa$-Leveled Words Sequence]
    \label{def:leveled words sequence}
    For $\kappa\in \bbN$, a function from $\bbN$ to $\Gamma'^*$ is a \emph{$\kappa$-leveled words sequence}, denoted $\levwords:\bbN\to \Gamma'^*$, if it is an elongated words sequence, and for every $m\in \bbN$ with $\levwords(m)=w_1w_2w_3w_4$ and prefix $u$ of $w_3$ it holds that 
    $|\charge(w_1w_2)-\charge(w_1w_2u)|\le \kappa$.
\end{definition}
That is, in a $\kappa$-leveled word sequence, the charge on the $w_3$ infix remains within a band of ``width'' $\kappa$ around the charge at $w_2$. Note that in particular, the charge at $w_1w_2w_3$ is also within this band. Also note that this is in particular an elongated words sequence, so the length of $w_3$ increases. The behavior of leveled words is depicted in \cref{fig:leveled}. 

\begin{figure}[ht]
    \centering
    \includegraphics[width=0.9\linewidth]{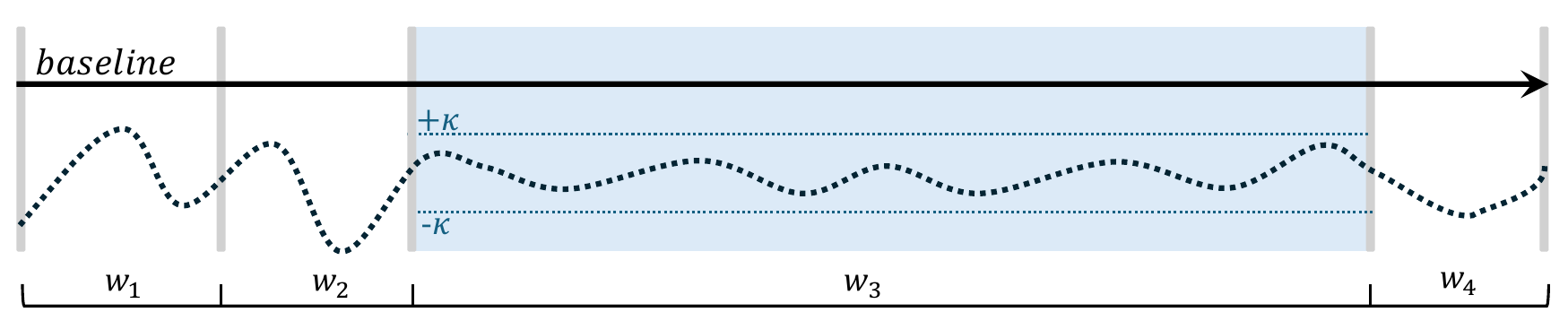}
    \caption{A leveled word $w_1w_2w_3w_4$. The charge upon reading the infix $w_3$ stayes bounded within $\kappa$ of its starting point.}
    \label{fig:leveled}
\end{figure}

Our main technical result in this section that the existence of leveled words with large gap between independent runs implies that either the potential is unbounded, or the sequence has the infix $D$-dip property (\cref{def:infix D dip}).
This result is proved by induction on the number of independent runs. Before proving the claim, we present a subcase of it, used in the induction.

\begin{proposition}
    \label{prop:leveled with large gaps no ghost G cover implies unbounded potential or D dip}
    Let $\ell,G\in \bbN$ and consider a leveled words sequence $\levwords$ that is $\ell$-sparse with gap $4G+(|S|+1)4\bigM W_{\max}|w_3|$ and is a configuration $G$-cover. Further assume there exists $l\in \bbN$ such that for every $m$ with $\levwords(m)=w_1w_2w_3w_4$, all ghost runs on $w_3$ are of length at most $l$.
    Then either $\sup\{\pot(w)\mid w\in (\Gamma_0^0)^*\}=\infty$, or $\levwords$ has the Infix $D$-dip property for some $D\in \bbN$.
\end{proposition}
\begin{proof}
    If $\levwords$ has the Infix $D$-dip property for some $D\in \bbN$, we are done. We therefore assume this is not the case. Thus, by the converse of \cref{def:infix D dip}, for every $D,n\in \bbN$, there exists $m$ with $\levwords(m)=w_1w_2w_3w_4$ such that we can write $w_3=xvy$ with $|v|>n$ and there exists $p\in \ghostTrans(s_0,w_1w_2x)$ and a run $\rho:p\runsto{x}S$ with $\weight(\rho)\le -D$.

    Our goal now is to obtain from $\levwords$ another leveled sequence $\levwords'$ that has a decreasing-gap fair decomposition (\cref{def:decreasing gap fair decomposition}). Intuitively, we obtain the decreasing gap using the decreasing runs on $v$ above.

    We define $\levwords'$ as follows. For every $m\in \bbN$, let $m'>m$ such that for $\levwords(m')=w_1w_2w_3w_4$ we can write $w_3=xvy$ with $|v|>(|S|+1)m+l$  (recall that $l$ is the bound on the length of a ghost run) and there exists $p\in \ghostTrans(s_0,w_1w_2x)$ and a run $\rho:p\runsto{v}S$ with $\weight(\rho)\le -W_{\max}((|S|+1)m+l)$. 
    Since there are no ghost runs of length greater than $l$ over $w_3$, it follows that $\rho$ reaches a reachable state $q$ within $l$ steps from $p$. We then still have a run $\rho'$ over a suffix $v'$ of $v$ such that $\rho':q\runsto{v'}S$ with $\weight(\rho')\le -W_{\max}(|S|+1)m$. Indeed, $\rho$ can lose at most $W_{\max}l$ weight within $l$ steps. 
    Recall that for decreasing-gap fair decompositions we need a \emph{seamless} decreasing run (\cref{def:decreasing gap fair decomposition}). Intuitively, the existence of a long-enough decreasing run clearly implies the existence of a decreasing seamless run, since a run cannot decrease without ``dragging down'' with it some seamless run. 
    
    Formally, if $\rho'$ has an infix of length at least $m$ that is decreasing and seamless, we redefine $v'$ as that infix and we are done. Otherwise, decompose $\rho'$ into $|S|+1$ segments by the first time $\rho'$ loses an additional $W_{\max}m$ weight.     
    Since there are at most $|S|$ seamless runs, it follows that $\rho'$ intersects the same seamless runs at least twice at the end of these segments, but then this seamless run is decreasing at least as much as $\rho'$, so there is a seamless run decreasing by at least $W_{\max}m$, as required. We reuse $v'$ to denote the corresponding suffix, and write 
 $v=v_0v'$.

    Define $\levwords'(m)=w'_1w'_2w'_3w'_4$ with $w'_1=w_1,w'_2=w_2xv_0,w'_3=v'$ and $w'_4=yw_4$.
    Denote by $\kappa$ the charge band of $\levwords$.
    Observe that $\levwords'$ is a $2\kappa$-leveled word sequence: the infixes $v'$ have increasing length, and since the words in $\levwords'$ are equal as concatenations to words from $\levwords$, then the bounded-charge property is preserved, with the modification that previously all charges were within gap $\kappa$ of $\charge(w_1w_2)$, whereas now it is possible that e.g., $\charge(w'_1w'_2)=\charge(w_1w_2)-\kappa$, so the remaining charge is within distance at most $2\kappa$. 
    Moreover, $\levwords'$ is $\ell$-sparse with gap $4G+(|S|+1)4\bigM W_{\max}|w_3|$ and a configuration $G$-cover. Indeed, since we choose $v'$ that are shorter than their corresponding $|w_3|$, the gap requirement that is already met in $\levwords$ is even larger with respect to $\levwords'$, and the decomposition $G$-cover is implied by the configuration $G$-cover.

    It remains to show that $\levwords'$ has a decreasing-gap fair decomposition. To this end, for every $m\in \bbN$ consider $\levwords'(m)=w'_1w'_2w'_3w'_4$ as above, where by construction there exists $q\in \booltrans(s_0,w'_1w'_2)$ and a run $\rho:q\runsto{w'_3}S$ with $\weight(\rho)\le -W_{\max}m$ and $|w'_3|>m$. 
    We decompose $w'_3=u_1u_2\cdots u_{d(m)+1}$ as follows. 
    Let $d(m)$ be the maximal number such that the following ordered indices are distinct: 
    $0=i_0\le i_1\le \ldots \le i_{d(m)}<|w_3|$ where for every $1\le k\le d(m)$ we have
    \[i_k=\min\{j\mid \weight(\rho(w'_1w'_2w'_3[1,i_{k-1}]))-\weight(\rho(w'_1w'_2w'_3[1,j]))\ge \sqrt{m}W_{\max}\}\]
    We then set $u_{k}=w'_3[i_{k-1}+1,i_k]$ for all $1\le k\le d(m)$ and $u_{d(m)+1}$ the remaining suffix.

    Observe that for every $1\le k\le d(m)$ we now have by definition that $\weight(\rho(w'_1w'_2u_1\cdots u_{k-1}))-\weight(\rho(w'_1w'_2u_1\cdots u_{k}))\ge \sqrt{m}W_{\max}$. Therefore, since the charge remains within a $2\kappa$ band, we have that 
    \[
    \begin{split}
    &\weight(\rho(w'_1w'_2u_1\cdots u_{k-1}))+\charge(\rho(w'_1w'_2u_1\cdots u_{k-1}))\ge \\
    &\weight(\rho(w'_1w'_2u_1\cdots u_{k}))+\sqrt{m}W_{\max}+\charge(w'_1w'_2u_1\cdots u_{k})-2\kappa
    \end{split}
    \]
    Moreover, since with each letter the run $\rho$ loses at most $W_{\max}$ weight, then the number of segments in the decomposition is (roughly) $\sqrt{m}$, i.e., $\lim_{m\to \infty}d(m)=\infty$.
    
    Thus, $\levwords'$ has a decreasing-gap fair decomposition with $\sqrtb(m)=\sqrt{m}W_{\max}-2\kappa$.   
    We conclude that $\levwords'$ satisfies the conditions of \cref{lem: decomposed seq with cover sparse no ghosts implies unbounded potential}, and therefore $\sup\{\pot(w)\mid w\in (\Gamma_0^0)^*\}=\infty$, so we are done.
\end{proof}
    
We are now ready for the main inductive argument.

\begin{lemma}[\lightbulbicon Leveled Words and Gap to Unbounded Potential (or $D$-Dip)]
    \label{lem:leveled with large gaps implies unbounded potential or D dip}
    Let $\ell\in \bbN$ and consider a leveled words sequence $\levwords$ that is $\ell$-sparse with gap $(|S|+1)4\bigM W_{\max}|w_3|$. Then either $\sup\{\pot(w)\mid w\in (\Gamma_0^0)^*\}=\infty$, or $\levwords$ has the Infix $D$-dip property for some $D\in \bbN$.
\end{lemma}
\begin{proof}
    The proof is by reverse induction on $\ell$, and shares common elements with the proof of \cref{lem:unbounded long dip imply separated inc infix}.

    \paragraph*{Base case: $\ell=|S|$} Recall that the maximal number of configuration-independent runs is $|S|$ (see \cref{def:configuration independent runs}). 
    In this case, for every $m\in \bbN$ with $\levwords(m)=w_1w_2w_3w_4$ we have that in every configuration in $w_3$, each state belongs to one of the independent runs. 
    In particular, $\levwords$ is a configuration $G$-cover with $G=0$, and is also $\ell$-sparse by assumption, with gap $4G+(|S|+1)4\bigM W_{\max}|w_3|$ (since $G=0$).

    Moreover, since $\ell=|S|$, then in particular every state is reachable along any prefix of $w_3$, and therefore there are no ghost runs at all along $w_3$. 
    Thus, the conditions of \cref{prop:leveled with large gaps no ghost G cover implies unbounded potential or D dip} are met, so either $\levwords$ has the infix $D$-dip property for some $D$, or $\sup\{\pot(w)\mid w\in (\Gamma_0^0)^*\}=\infty$, and we are done.

    \paragraph*{Inductive case: $\ell<|S|$}
    The induction is split into three subcases. 
    \subparagraph*{Case 1: long ghost runs.}
    In this case we assume that there are unboundedly long ghost runs. Specifically, assume that for every $l\in \bbN$ there is some $m\in \bbN$ with $\levwords(m)=w_1w_2w_3w_4$ and we can write $w_3=xuy$ such that $|u|\ge l$ and there exists $p\in \ghostTrans(s_0,w_1w_2x)\setminus \booltrans(s_0,w_1w_2x)$ and a ghost run $\rho:p\runsto{u}S$.
    Intuitively, while this ghost run is not a real run on $w_1w_2w_3w_4$, it becomes a real run if we flatten the prefix up to it. Moreover, since flattening guarantees that ghost runs are very high, this implies the existence of another independent run, so we can apply the induction.

    Formally, we define a new sequence $\levwords'$ as follows. For every $m'\in \bbN$ let $m>m'$ be such that $\levwords(m)=w_1w_2w_3w_4$ has a ghost run as above with $|u|>m'$. Define $\levwords'(m')=w'_1w'_2w'_3w'_4$ with
    $w'_1=\flatten(w_1w_2x \wr (|S|+1)4\bigM W_{\max}|w_1w_2w_3w_4|), w'_2=\epsilon, w'_3=y$ and $w'_4=yw_4$.\footnote{We remark that unlike \cref{sec:separated increasing infix}, the roles of $w_1$ and $w_2$ here are combined, so it is fine selecting $w_2=\epsilon$ and moving its content to $w_1$ (or vice versa). We keep the four-part formalism for uniformity.} 

    Similarly to \cref{lem: decomposed seq with cover sparse no ghosts implies unbounded potential}, since our alphabet does not contain rebase letters, by \cref{def:cactus rebase flattening} the resulting prefix $x'_1$ is over $\Gamma_0^0$. 
    Denote $\vec{c}=\xconf(\vec{c_{init}},w_1w_2x)$ and $\vec{c'}=\xconf(\vec{c_{init}},w'_1w'_2)$, then by 
    \cref{lem:flattening configuration characterization}
    we have that $\vec{c}(q)=\vec{c'}(q)$ for every $q\in \booltrans(s_0,w_1w_2)$, and 
    $\vec{c'}(p)\ge (|S|+1)4\bigM W_{\max}|w_1w_2w_3w_4| + \max\{\vec{c}(p')\mid p'\in \supp(\vec{c})\}$.
    Specifically, the entire run $\rho$ from $p$ is above all the existing $\ell$ independent runs with gap at least $(|S|+1)4\bigM W_{\max}|w_3|$.
    It follows that $\rho$ can be added as an independent run. Thus, $\levwords'$ has $\ell+1$ independent runs. Moreover, note that $\levwords'$ remains a leveled words sequence, since the minimal states are always reachable and therefore maintain their weight after flattening. 
    We can therefore apply the induction hypothesis on $\levwords'$. By the hypothesis, if $\sup\{\pot(w)\mid w\in (\Gamma_0^0)^*\}=\infty$, then we are done.
    
    Otherwise, we get that $\levwords'$ has the $D$-dip property. Recall, however, that we need to show $\levwords$ has the $D$-dip property. Fortunately, observe that $D$-dip is ``local'' in that it only considers runs over the relevant infix, and since the $D$-dip infixes of $\levwords'$ also appear in $\levwords$, so we can indeed show this as follows.

    For every $n\in \bbN$, consider $m'$ such that $\levwords(m')=w'_1w'_2w'_3w'_4$ has an infix $u$ with $w'_3=xuy$ with $|u|\ge n$ and for every $p\in \ghostTrans(s_0,w'_1w'_2x)$ and run $\rho:p\runsto{u}S$ we have $\weight(\rho)>-D$ (these exist by the $D$-dip property, see \cref{def:infix D dip}). 
    By the construction of $\levwords'$, there exists $m>m'$ such that $\levwords(m)=w_1w_2w_3w_4$ and $w'_3$ is an infix of $w_3$. Moreover, while $w'_1w'_2x$ differs from the corresponding prefix $z$ of $\levwords(m)$ due to the flattening, it still holds that $\ghostTrans(s_0,z)=\ghostTrans(s_0,w'_1w'_2x)$ by the flattening procedure.
    It follows that the same condition on $u$ holds also in $w_1w_2w_3w_4$. Therefore, $\levwords$ also has the infix $D$-dip property, and we are done.

    In the following cases, we therefore assume there is a bound $l\in \bbN$ such that all ghost runs on infixes of any $w_3$ are of length at most $l$.

    \subparagraph*{Case 2: a diverging run on a suffix of $w_3$.}
    In this case we essentially assume that the independent runs are not a $G$-cover for any $G$, and use this to construct a new independent run from states that stay far away from all independent runs. We then apply the induction hypothesis. This is very similar to the setting in \cref{fig:sparse induction second case}, where instead of starting from a single state, we start from a configuration. Specifically, we assume the divergence occurs in some long-enough prefix of $w_3$.

    Formally, assume that for every $G\in \bbN$ there exists $m>2G$ with $\levwords(m)=w_1w_2w_3w_4$ and we can write $w_3=uv$ with $|u|\ge \frac12|w_3|$ such that for $\vec{c}=\xconf(s_0,w_1w_2u)$ there exists a state $q\in \supp(\vec{c})$ for which $\vec{c}(q)$ is not withing gap $G$ of any independent run.

    We define a new sequence $\levwords'$ as follows. For every $m'\in \bbN$, take some $G>2W_{\max}m'+(|S|+1)4\bigM W_{\max}m'$ 
    and let $m>2m'$ large enough such that the gap of $q$ as above is at least $G$.
    Write $\levwords(m)=w_1w_2w_3w_4$ with $w_3=uv$ and $q$ as above. 
    We can now use the fact that all seamless runs change their weight by at most $W_{\max}$ at each step. Consider the seamless run $\rho$ leading up to $q$, and write $u=u_1u_2$ with $|u_2|=m'$ (this is possible since $|u|\ge \frac12|w_3|\ge \frac12 m>m'$).
    Since $q$ is not within gap $G$ of any independent run, for every prefix $u'$ of $u_2$ it holds that $\weight(\rho(w_1w_2u_1u'))$ is not within gap $G-2W_{\max}m'>(|S|+1)4\bigM W_{\max}m'$ of any independent run.

    We can therefore define $\levwords(m')=w'_1w'_2w'_3w'_4$ by $w'_1=w_1,w'_2=w_2u_1,w'_3=u_2,w'_4=vw_4$. As above, this is a leveled words sequence, since the charge requirement remains valid. Moreover, we now have $\ell+1$ independent runs satisfying the gap constraints. We can therefore apply the induction hypothesis. As in Case 1, if the potential is unbounded then we are done, and if $\levwords'$ has the infix $D$-dip property, this lifts back to $\levwords$ (here we do not even have to consider the flattening).

    \subparagraph*{Case 3: $G$-cover on suffix of $w_3$.}
    We now consider the complement of Case 2, still under the assumption that all ghost runs are of length at most $l$.
    Specifically, we assume that there exists $G\in \bbN$ such that for every $m>2G$ with $\levwords(m)=w_1w_2w_3w_4$, for every decomposition $w_3=uv$ with $|u|\ge \frac12|w_3|$ and for every state $q\in \booltrans(s_0,w_1w_2u)$, the weight of $q$ is within gap $G$ of some independent run. We therefore have a $G$-cover of the ``second half'' of $w_3$. 
    Our goal is now to apply \cref{prop:leveled with large gaps no ghost G cover implies unbounded potential or D dip}. For this, we only need to define a sequence where the $w_3$ infix corresponds to these ``second halves''. A slight caveat is that \cref{prop:leveled with large gaps no ghost G cover implies unbounded potential or D dip} requires a larger gap than $\levwords$ has. This, however, is an ``artificial'' problem: the gaps in $\levwords$ grow with $m$, so we just need to pick words that are further away, so we have larger gaps.
    
    Formally, we proceed as follows. For every $m'\in \bbN$, let $m=2(m'+G)$ with $\levwords(m)=w_1w_2w_3w_4$ and write $w_3=uv$ with $|v|=m'$ (in particular $|u|\ge \frac12 |w_3|$). Define $\levwords(m')=w'_1w'_2w'_3w'_4$ by $w'_1=w_1,w'_2=w_2u,w'_3=v,w'_4=w_4$. We observe that $\levwords'$ is still a leveled words sequence, since the charge constraints do not change. We claim that $\levwords'$ is $\ell$-sparse with gap $4G+(|S|+1)4\bigM W_{\max}|w_3|$. Indeed, recalling that $m=2(m'+G)$, the gap guaranteed by $\levwords$ is 
    \[
    \begin{split}
    &(|S|+1)4\bigM W_{\max}|w_3|\ge (|S|+1)4\bigM W_{\max}2(m'+G)\ge \\
    &4G+ (|S|+1)4\bigM W_{\max}m'=4G+ (|S|+1)4\bigM W_{\max}|w'_3|    
    \end{split}
    \]
    Finally, $\levwords'$ has a configuration $G$-cover, as we discuss above.
    
    We can therefore apply \cref{prop:leveled with large gaps no ghost G cover implies unbounded potential or D dip} to $\levwords'$ and we conclude that either $\levwords'$ has the infix $D$-dip property, or $\sup\{\pot(w)\mid w\in (\Gamma_0^0)^*\}=\infty$. Again, the infix $D$-dip property lifts back to $\levwords$, so we are done.
    Note that the last case is not inductive, but direct (indeed, it can be thought of as a different base case).
    \end{proof}

Recall by the definition of sparse words (\cref{def:sparse words}), that if we consider a set with a single independent run (i.e., $\ell=1$), then this set is vacuously sparse with any size of gap. We can then already apply \cref{lem:leveled with large gaps implies unbounded potential or D dip}. Indeed, this aligns precisely with the inductive proof above: we start with just the seamless baseline run. If it is a $G$-cover, we are done. Otherwise show that there is another independent run, and so on. Thus, in order to use \cref{def:sparse words}, we only need the leveled sequence, and not the gap requirement, giving us the main tool of this section.
\begin{corollary}[\keyicon Leveled Words to Unbounded Potential or $D$-Dip]
    \label{cor:leveled implies unbounded potential or D dip}
    Consider a $\levwords$ sequence, then either $\sup\{\pot(w)\mid w\in (\Gamma_0^0)^*\}=\infty$, or $\levwords$ has the Infix $D$-dip property for some $D\in \bbN$. 
\end{corollary}

\subsubsection{Discharging Words Sequences}
\label{sec:discharging word sequences}
A counterpart to leveled words are \emph{discharging words}, where the charge significantly decreases.
\begin{definition}[\keyicon Discharging Words Sequence]
\label{def:discharging word sequence}
An elongated words sequence $\words$ is a \emph{discharging words sequence}, denoted $\diswords:\bbN\to \Gamma'^*$, if for every $m\in \bbN$ with $\diswords(m)=w_1w_2w_3w_4$ the following hold:
\begin{itemize}
    \item $\charge(w_1w_2)-\charge(w_1w_2w_3)>m$ (the minimal states get much closer to the baseline run upon reading $w_3$).
    \item For every prefix $v$ of $w_3$ we have $\charge(w_1w_2v)\ge \charge(w_1w_2w_3)$ (the minimal state stays below its value at $w_1w_2w_3$).
\end{itemize}
\end{definition}


We illustrate the ``shape'' of $\diswords$ in \cref{fig:discharging}.

\begin{figure}[ht]
    \centering
    \includegraphics[width=0.8\linewidth]{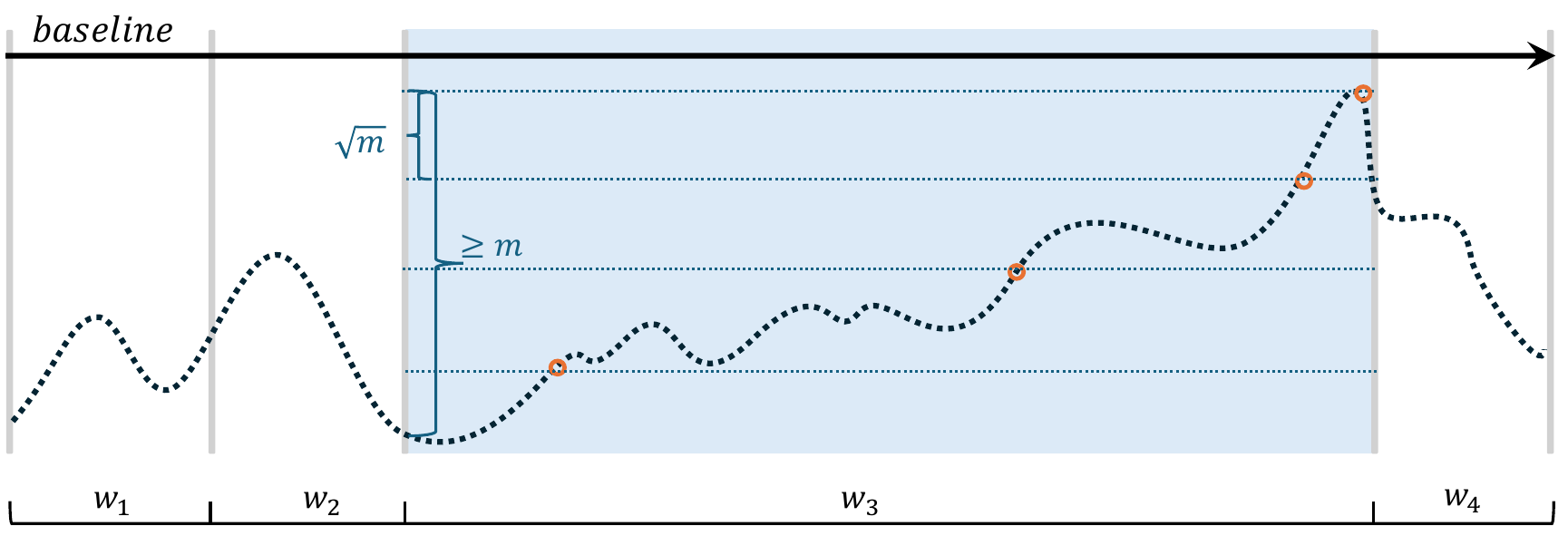}
    \caption{A discharging word $w_1w_2w_3w_4$. The infix $w_3$ brings the charge much closer to the baseline, being at its closest at the end of $w_3$. 
    The orange dots signify each segment in the decomposition, i.e., the first time the charge decreases by another $\sqrt{m}$.}
    \label{fig:discharging}
    
\end{figure}

Unless stated otherwise, in the following when we write $\diswords$ we implicitly assume the function is a discharging words sequence.

Discharging words sequences have a natural notion of decomposition: since the charge between $w_1w_2$ and $w_1w_2w_3$ decreases by at least $m$, we can mark each index where it decreases by $\sqrt{m}$, and thus obtain roughly $\sqrt{m}$ segments, each of length roughly $\sqrt{m}$. This is similar to the proof of \cref{prop:leveled with large gaps no ghost G cover implies unbounded potential or D dip}.
Intuitively, each segment in the decomposition can be thought of as ``zooming in'' on $w_3$ and finding a self-similar structure. Indeed, the characteristic of the $w_3$ infixes is that the charge decreases along them, while they get longer, and this is exactly the property satisfied by the infixes in the decomposition.

Unfortunately, things are not quite so simple, since the charge can behave in nasty ways. Specifically, it could be that $\charge(x)$ is much larger than $\charge(x\sigma)$ for some $x$, in case the minimal run on $x$ cannot continue with $\sigma$, thus making the minimal state much higher. Put more succinctly, this behavior is the opposite of \cref{def:charge bounded decrease}.
In such cases, the decomposition we propose may have much less than $\sqrt{m}$ segments.

Nonetheless, we define this decomposition, and show that under sufficient assumptions (namely \cref{def:charge bounded decrease}), it behaves well, i.e., it is a fair decreasing-gap decomposition. 
Keeping in mind that the general goal of this section is to ultimately prove that the potential is unbounded, then by \cref{lem:charge no bounded decrease then potential unbounded} we already know that if the charge does not have bounded decrease, the potential is unbounded. Therefore, we need not be worried about this too much.

\begin{definition}[\keyicon Discharging Decomposition]
    \label{def:discharging decomposition}
    For $m\in \bbN$ and $\diswords(m)=w_1w_2w_3w_4$, let $d(m)$ be the maximal number such that the following ordered indices are distinct:
    $0=i_0 \le i_1\le i_2\le \ldots\le i_{d(m)}< |w_3|$ where for every $1\le k\le d(m)$ we have
    $i_k=\min\{j\mid \charge(w_1w_2w_3[1,i_{k-1}])-\charge(w_1w_2w_3[1,j])>\sqrt{m}\}$.
    
    The \emph{discharging decomposition} of $w_3$ is then $w_3=u_1u_2\cdots u_{d(m)+1}$ where $u_k=w_3[i_{k-1}+1,i_k]$ for every $1\le k\le d(m)$, and $u_{d(m)+1}=w_3[i_{d(m)+1},|w_3|]$.

    For every $1\le k\le d(m)$ we then have: 
    \begin{itemize}
        \item $\charge(w_1w_2u_1\cdots u_{k-1})-\charge(w_1w_2u_1\cdots u_{k})> \sqrt{m}$.
        \item For every prefix $v$ of $u_1\cdots u_k$ it holds that 
        $\charge(w_1w_2v)\ge \charge(w_1w_2u_1\cdots u_{k})$
    \end{itemize}
\end{definition}
Observe that the indices in \cref{def:discharging decomposition} are indeed ordered, since we always take the first index after which the charge decreases by a certain amount.

We now show that if $\charge$ behaves nicely, then the discharging decomposition is a decreasing-gap fair decomposition. Specifically, we require that the charge has bounded decrease (as per \cref{def:charge bounded decrease}). 
Recall that otherwise, we have by \cref{lem:charge no bounded decrease then potential unbounded} that the potential is unbounded, which we utilize later on in the proof.
\begin{proposition}
    \label{prop:discharging words with bounded decrease charge is fair decomposition}
    Let $B\in \bbN$. Assume $\diswords$ is such that for every $m\in \bbN$ and $\diswords(m)=w_1w_2w_3w_4$ we have that $\charge$ has $B$-Bounded Decrease on the infix $w_3$ (as per \cref{def:charge bounded decrease}).
    Then the discharging decomposition is a decreasing-gap fair decomposition (as per \cref{def:decreasing gap fair decomposition}).
\end{proposition}
\begin{proof}
    Let $m\in \bbN$ and $\diswords(m)=w_1w_2w_3w_4$. Write $w_3=u_1\cdots u_{d(m)+1}$.

    We claim that the discharging decomposition is a decreasing-gap fair decomposition with respect to the baseline run $\rho_0$. Note that the baseline run has weight $0$ in all prefixes. Therefore, we need to show that the length of the $u_i$ increases with $m$, that $d(m)$ increases with $m$, and that $\charge$ decreases by some function $\sqrtb(m)$ between each segment (as per \cref{def:decreasing gap fair decomposition}, instantiated with $\rho_0$) except for the last ``remainder'' segment. 

    The latter requirement is trivial by the \cref{def:discharging decomposition} -- the segments are defined such that 
    \[
     \charge(w_1w_2u_1\cdots u_{i-1})-\charge(w_1w_2u_1\cdots u_i)>\sqrt{m}
    \]
    i.e, their charge gap is at least $\sqrt{m}$.

    By the bounded-decrease charge property of \cref{def:charge bounded decrease}, we inductively have that for every segment $u_i$ (indeed, every infix of $w_3$) it holds that
    \[
    \charge(w_1w_2u_1\cdots u_{i-1})-\charge(w_1w_2u_1\cdots u_i)<|u_i|B
    \]
    Combining this with the charge gap above, we have that $|u_i|> \sqrt{m}/B$, and in particular the length of $u_i$ increases with $m$.

    We turn to show that $d(m)$ increases. We initially claim that 
    \[
     \charge(w_1w_2u_1\cdots u_{i-1})-\charge(w_1w_2u_1\cdots u_i)\le \sqrt{m}+B
    \]
    Indeed write $u_i=v\cdot \sigma$ where $\sigma$ is the last letter, then 
     \[
    \begin{split}
        &\charge(w_1w_2u_1\cdots u_{i-1})-\charge(w_1w_2u_1\cdots u_{i-1}v\sigma)=\\
        &\charge(w_1w_2u_1\cdots u_{i-1})-\charge(w_1w_2u_1\cdots u_{i-1}v)+ \charge(w_1w_2u_1\cdots u_{i-1}v) -\charge(w_1w_2u_1\cdots u_{i-1}v\sigma)\le \\
        &\sqrt{m} + B
    \end{split}
    \]
    Where the last inequality follows from \cref{def:charge bounded decrease} and from the fact that $u_i=v\sigma$ is minimal with the gap property in \cref{def:discharging decomposition} (i.e., $\charge(w_1w_2u_1\cdots u_{i-1})-\charge(w_1w_2u_1\cdot u_{i-1}v)\le \sqrt{m}$ otherwise $v$ would have been selected instead of $u_i$).

    We notice that for segment $u_{d(m)+1}$ we do not provide bounds, as it is essentially a ``remainder''. Trivially, we have $\charge(w_1w_2u_1\cdots u_{d(m)})-\charge(w_1w_2u_1\cdots u_{d(m)+1})\le \sqrt{m}$ (otherwise we would have another segment).

    By the definition of $\diswords$ we now have that 
    \[
    \sum_{i=1}^{d(m)+1}(\charge(w_1w_2u_1\cdots u_{i-1})-\charge(w_1w_2u_1\cdots u_i))=\charge(w_1w_2)-\charge(w_1w_2w_3)\ge m
    \]
    and on the other hand by the above we have
    \[
    \sum_{i=1}^{d(m)+1}(\charge(w_1w_2u_1\cdots u_{i-1})-\charge(w_1w_2u_1\cdots u_i))\le d(m)(\sqrt{m}+B)+\sqrt{m}
    \]
    We therefore have
    $d(m)\ge \frac{m-\sqrt{m}}{\sqrt{m}+B}$
    so $\lim_{m\to \infty}d(m)=\infty$, as required.
\end{proof}

\cref{prop:discharging words with bounded decrease charge is fair decomposition} lays the ground to use \cref{lem: decomposed seq with cover sparse no ghosts implies unbounded potential}: if the charge has bounded decrease, then the decomposition is fair, so under the appropriate conditions, the potential is unbounded. That is, we have the following.
\begin{corollary}
    \label{cor: discharging bounded charge cover sparse no ghosts implies unbounded potential}
    If there exists $G,\ell,B\in \bbN$ and a discharging word sequence $\diswords$ that satisfies the following:
    \begin{enumerate}
        \item $\charge$ has $B$-bounded decrease on $w_3$ for every $\diswords(m)=w_1w_2w_3w_4$, and
        \item $\diswords$ (with the discharging decomposition) is a ghost-free decomposition $G$-cover and is $\ell$-sparse with gap $4G + (|S|+1)4\bigM W_{\max} |w_3|$,
        \acctodo{ACCOUNTING}
    \end{enumerate}
    then $\sup\{\pot(w)\mid w\in (\Gamma^0_0)^*\}=\infty$.
\end{corollary}

Our main result of the section is again an inductive lemma (similar to \cref{lem:leveled with large gaps implies unbounded potential or D dip}), guaranteeing either a $D$-dip or unbounded potential, given a discharging words sequence.
\begin{lemma}[\keyicon \lightbulbicon Discharging and Gap to Unbounded Potential (or $D$-Dip)]
    \label{lem:discharge with large gaps implies unbounded potential or D dip}
    Let $\ell,B\in \bbN$ and consider a discharging words sequence $\diswords$ that is $\ell$-sparse with gap $(|S|+1)4\bigM W_{\max}|w_3|$. 
    Also assume $\charge$ has $B$-bounded decrease on $w_3$ for every $\diswords(m)=w_1w_2w_3w_4$.
    Then either $\sup\{\pot(w)\mid w\in (\Gamma_0^0)^*\}=\infty$ or $\diswords$ has the infix $D$-dip property for some $D\in \bbN$.
\end{lemma}
\begin{proof}
    We prove the lemma by reverse induction on $\ell$. Again, there are similarities to the proofs of \cref{lem:leveled with large gaps implies unbounded potential or D dip,lem:unbounded long dip imply separated inc infix}.
    
    \paragraph*{Base case: $\ell=|S|$} Recall that the maximal number of configuration-independent runs is $|S|$ (see \cref{def:configuration independent runs}). 
    In this case, for every $m\in \bbN$ with $\diswords(m)=w_1w_2w_3w_4$ we have that in every configuration in $w_3$, each state belongs to one of the independent runs. 
    In particular, $\diswords$ is a configuration $G$-cover with $G=0$, and is also $\ell$-sparse by assumption, with gap $4G+(|S|+1)4\bigM W_{\max}|w_3|$ (since $G=0$).

    Moreover, since $\ell=|S|$, then in particular every state is reachable along any prefix of $w_3$, and therefore there are no ghost runs at all along $w_3$. 
    Thus, the conditions of \cref{cor: discharging bounded charge cover sparse no ghosts implies unbounded potential} are met, so $\sup\{\pot(w)\mid w\in (\Gamma_0^0)^*\}=\infty$, and we are done.

    As a side-note, observe that for this case we do not actually need the assumption that $\charge$ has bounded decrease, as this readily follows from the fact that there are $|S|$ independent runs.

    \paragraph*{Inductive case: $\ell<|S|$}
    For the inductive case, we start by decomposing $\diswords$ to its discharging decomposition (\cref{def:discharging decomposition}) $w_3=u_1\cdots u_{d(m)}u_{d(m)+1}$. 
    Since $\charge$ has $B$-bounded decrease, then by \cref{prop:discharging words with bounded decrease charge is fair decomposition}, this decomposition is a decreasing-gap fair decomposition.
    We now split to three subcases. 
    
    \subparagraph*{Case 1: a diverging run on a segment}
    In this case we essentially assume that the independent runs are not a decomposition $G$-cover for any $G$, and use this to construct a new independent run from states that stay far away from all independent runs. We then apply the induction hypothesis. This starts similar to the analogous case in \cref{lem:leveled with large gaps implies unbounded potential or D dip}, but differs in that the charge constraints are not necessarily met, and requires further case analysis.

    Formally, assume that for every $G\in \bbN$ there exists $m>2G$ with $\diswords(m)=w_1w_2w_3w_4$ and $w_3=u_1\cdots u_{d(m)}u_{d(m)+1}$ such that there exist $1\le i\le d(m)$ where for 
    $\vec{c_i}=\xconf(s_0,w_1w_2u_1\cdots u_i)$ there is a state $q\in \supp(\vec{c_i})$ for which $\vec{c_i}(q)$ is not withing gap $G$ of any independent run.
    That is, we assume that $\diswords(m)$ does not have a decomposition $G$-cover (\cref{def: G cover decomposed words}).
    
    By the assumption, for every $n\in \bbN$ there exists $m\in \bbN$ such that for $1\le i\le d(m)$ and $q$ above it holds that $q$ is not within gap $2W_{\max}n+(|S|+1)4\bigM W_{\max}n$ of any independent run, and $|u_i|>n$. 
    In particular, there is a seamless runs $\rho:s_0\runsto{w_1w_2u_1\cdots u_i}q$ such that for every split $u_i=v_0v$ with $|v|\le n$ we have that $\weight(\rho(w_1w_2u_1\cdots u_{i-1}v_0))$ is not within gap $(|S|+1)4\bigM W_{\max}n$ of any independent run. 

    With these notations, we divide into cases, depending on the behavior of $\charge$ on $u_i$.

    \subparagraph*{Case 1a: increasing charge}
    In this case, we assume that for every $n\in \bbN$ there exists $m>n$ and  $\diswords(m)=w_1w_2w_3w_4$ with $u_i$, $q$ and $\rho$ as above such that there is a split $u_i=v_0v$ with
    $\charge(w_1w_2u_1\cdots u_{i-1}v_0)-\charge(w_1w_2u_1\cdots u_{i-1}v_0v)>n$. That is, there is a suffix $v$ of $w_3$ on which the charge decreases significantly. Note the similarity to the condition of a discharging words sequence (\cref{def:discharging word sequence}). 

    We therefore construct a discharging words sequence $\diswords'$ by defining $\diswords'(n)=w'_1w'_2w'_3w'_4$ with $w'_1=w_1, w'_2=w_2u_1\cdots u_{i-1}v_0, w'_3=v$ and $w'_4=u_{i+1}\cdots u_{d(m)+1}w_4$.

    Since $\diswords$ is already a discharging words sequence, and $u_i$ is a segment in the decomposition, then by the minimality criterion of \cref{def:discharging decomposition}, we have that the charge along $u_i$ does not exceed $\charge(w_1w_2u_1\cdots u_i)$. Therefore, combined with the condition on $v$ above, we have that $\diswords'$ satisfies both the conditions of \cref{def:discharging word sequence}, so it is a discharging words sequence.

    Moreover, as we observed above, the run $\rho$ on the suffix $v$ can be added as an independent run with the gap required in the induction assumption. Note that the remaining runs certainly have large enough gaps, since $m>n$. We therefore satisfy the induction assumption, so by the induction hypothesis, either the potential is unbounded (and we are done), or $\diswords'$ has the Infix $D$-dip property for some $D\in \bbN$. 

    By an identical argument to \cref{lem:leveled with large gaps implies unbounded potential or D dip}, we can lift this $D$-dip property to $\diswords$ itself, and we are done.

    \subparagraph*{Case 1b: non-increasing charge}
    Complementing Case 1a, we now assume that there exists $\kappa$ such that for every $m>\kappa$ and $\diswords(m)=w_1w_2w_3w_4$ with $u_i$, $q$ and $\rho$ as above, every split $u_i=v_0v$ satisfies
    \[\charge(w_1w_2u_1\cdots u_{i-1}v_0)-\charge(w_1w_2u_1\cdots u_{i-1}v_0v)\le \kappa\] 

    Our goal in this case is to obtain from these segments a leveled words sequence. Indeed, the condition above is almost the condition for leveled words, but is missing the absolute value. 
    However, since $u_i$ is a \emph{minimal} infix of $w_3$ upon which the charge decreases enough, as per \cref{def:discharging decomposition}, it follows that $\charge(w_1w_2u_1\cdots u_{i-1}v_0)-\charge(w_1w_2u_1\cdots u_{i-1}v_0v)>0$, so we can add the absolute value and assume $|\charge(w_1w_2u_1\cdots u_{i-1}v_0)-\charge(w_1w_2u_1\cdots u_{i-1}v_0v)|\le \kappa$. 
    
    We construct a sequence $\levwords$ as follows: for every $n\in \bbN$ take $m$ for which $|u_i|>n$ (recall that the segments increase in length by \cref{prop:discharging words with bounded decrease charge is fair decomposition}). Define $\levwords(n)=w'_1w'_2w'_3w'_4$ with $w'_1=w_1,w'_2=w_2u_1\cdots u_{i-1}$, $w'_3=u_i$ and $w'_4=u_{i+1}\cdots u_{d(m)+1}w_4$. 
    By the above, we have that $\levwords$ is indeed a leveled words sequence. Moreover, it inherits from $\diswords$ the properties of being $\ell$-sparse and having gap $(|S|+1)4\bigM W_{\max}|w_3|$. 
    We can therefore invoke \cref{lem:leveled with large gaps implies unbounded potential or D dip}, and we obtain that either the potential is unbounded, or $\levwords$ has the Infix $D$-dip property for some $D\in \bbN$. As with the case above, an identical argument to \cref{lem:leveled with large gaps implies unbounded potential or D dip} shows that this lifts to $\diswords$, so we are done.

    \subparagraph*{Case 2: long ghost runs}
    The next case is not mutually exclusive to Case 1, so we assume Case 1 does not hold (even though we do not use this fact).
    We assume that there are infinitely many $m\in \bbN$ with $\diswords(m)=w_1w_2w_3w_4$ and decomposition $w_3=u_1\cdots u_{d(m)+1}$ such that there exists a ghost run over an entire segment $u_i$. More precisely, there exists $p\in \ghostTrans(s_0,w_1w_2u_1\cdots u_{i-1})$ and a ghost run $\rho:p\runsto{u_i}S$. 

    Our approach in this case is similar to that taken in \cref{lem:leveled with large gaps implies unbounded potential or D dip} -- we flatten the prefix up to $u_i$, thus making the ghost run a real run that satisfies the gap constraints. Thus, we can apply induction.

    Formally, we define a new sequence $\diswords'$ as follows. For every $m'\in \bbN$ let $m>m'$ be such that $\diswords(m)=w_1w_2w_3w_4$ has a ghost run on $u_i$, and we require
    \begin{equation}
    \label{eq:discharging induction condition on ui}
    \charge(w_1w_2u_1\cdots u_{i-1})- \charge(w_1w_2u_1\cdots u_{i})>m'
    \end{equation}
    Note that we can require this by since $u_i$ is a segment in a discharging decomposition, as per \cref{def:discharging decomposition}.
    Define $\diswords'(m')=w'_1w'_2w'_3w'_4$ with
    $w'_1=\flatten(w_1w_2u_1\cdots u_{i-1} \wr (|S|+1)4\bigM W_{\max}|w_1w_2w_3w_3|)$, $w'_2=\epsilon$, $w'_3=u_i$ and $w'_4=u_i\cdots u_{d(m)+1}w_4$.

    Since our alphabet does not contain rebase letters, by \cref{def:cactus rebase flattening} the resulting prefix $x'_1$ is over $\Gamma_0^0$. 
    Denote $\vec{c}=\xconf(\vec{c_{init}},w_1w_2u_1\cdots u_{i-1})$ and $\vec{c'}=\xconf(\vec{c_{init}},w'_1w'_2)$, then by 
    \cref{lem:flattening configuration characterization}
    we have that $\vec{c}(q)=\vec{c'}(q)$ for every $q\in \booltrans(s_0,w_1w_2)$, and $\vec{c'}(p)$ is much larger than any finite entry in $\vec{c}$ (according to the flattening constant above). 
    Specifically, the entire run $\rho$ from $p$ is above all the existing $\ell$ independent runs with gap at least $(|S|+1)4\bigM W_{\max}|w_3|$.
    It follows that $\rho$ can be added as an independent run. Thus, $\diswords'$ has $\ell+1$ independent runs. Moreover, note that $\diswords'$ remains a discharging words sequence. Indeed, 
    since the minimal states and baseline states are always reachable and therefore maintain their weight after flattening. In particular, by \cref{eq:discharging induction condition on ui} $u_i$ satisfies the discharge constraints of \cref{def:discharging word sequence} (to be more precise, it satisfies the second requirement by the minimality constraint of \cref{def:discharging decomposition}).
    
    We can therefore apply the induction hypothesis on $\diswords'$. By the hypothesis, if $\sup\{\pot(w)\mid w\in (\Gamma_0^0)^*\}=\infty$, then we are done. Otherwise, we get that $\diswords'$ has the $D$-dip property, and we again repeat the argument in \cref{lem:leveled with large gaps implies unbounded potential or D dip} to lift this to $\diswords$ (note that here the argument is the involved case therein, since we use flattening) and we are done.

    \subparagraph*{Case 3: Ghost free $G$-cover}
    Our final case assumes both Case 1 and Case 2 do not hold. Since Case 2 does not hold, there are at most finitely many $m$ for which there are ghost runs over segments in the decomposition. 
    We slightly strengthen this to assume there are no ghost runs over any segments at all. This is possible since we can replace $\diswords$ by duplicating $\diswords(m)$ for some very large $m$ into $\diswords(i)$ for all $i\le m$. Note that this does not interfere with the Infix $D$-dip property, which anyway refers to infinitely many infixes and ignores the replacement of finitely many elements.
    We therefore assume that the discharging decomposition of $\diswords$ is Ghost Free, as per \cref{def:ghost free decomposition}.

    Next, since Case 1 does not hold, then there exists $G\in \bbN$ such that for every $m>2G$ with $\diswords(m)=w_1w_2w_3w_4$ and $w_3=u_1\cdots u_{d(m)+1}$, for every $1\le i\le d(m)$ and $\vec{c_i}=\xconf(\vec{c_{\init}},w_1w_2u_1\cdots u_i)$, every state $q\in \supp(\vec{c_i})$ is within gap $G$ from some independent run. 
    As above, we can assume this holds in fact for every $m$, by possibly replacing finitely many elements of $\diswords$.
    Thus, the discharging decomposition of $\diswords$ is a decomposition $G$-cover (as per \cref{def: G cover decomposed words}).

    We can now almost apply \cref{cor: discharging bounded charge cover sparse no ghosts implies unbounded potential}, except the gap required there is slightly bigger than the gap in $\diswords$. Identically to \cref{lem:leveled with large gaps implies unbounded potential or D dip}, this is an artificial problem, and can be resolved by shifting the elements of $\diswords$, i.e., defining $\diswords'(m)=\diswords(m+G)$, which increases the bound by the necessary $4G$. 

    We can therefore invoke \cref{cor: discharging bounded charge cover sparse no ghosts implies unbounded potential} to conclude that the potential is unbounded, and we are done.
\end{proof}

In order to eliminate the dependency on bounded-decrease of charge, we observe that if for every $B\in \bbN$ there exists $m\in \bbN$ such that for $\diswords(m)=w_1w_2w_3w_4$ the charge upon reading $w_3$ does not have $B$-bounded decrease, then we can apply \cref{lem:charge no bounded decrease then potential unbounded} with the word $u=w_1w_2x$ and $\sigma\in \Gamma_\infty^0$ for an appropriate prefix $x\sigma$ of $w_3$ (upon which the charge decreases by $B$). We then get the the potential over $\Gamma_0^0$ is unbounded.
Combining this with \cref{lem:discharge with large gaps implies unbounded potential or D dip} we conclude with the following.
\begin{corollary}
\label{cor:discharge with sparse gap implies unbounded potential}
    Let $\ell\in \bbN$. If there exists a discharging words sequences $\diswords$ that is $\ell$-sparse with gap $(|S|+1)4\bigM W_{\max}|w_3|$. Then either $\sup\{\pot(w)\mid w\in (\Gamma_0^0)^*\}=\infty$ or $\diswords$ has the infix $D$-dip property for some $D\in \bbN$.
\end{corollary}

Similarly to \cref{cor:leveled implies unbounded potential or D dip} for leveled words, we can start with a set of independent runs of size $1$, that vacuously has a large gap, and obtain our main tool of this section.
\begin{corollary}[\keyicon Discharging Words to Unbounded Potential or $D$-Dip]
    \label{cor:discharge implies unbounded potential or D dip}
    Consider a $\diswords$ sequence, then either $\sup\{\pot(w)\mid w\in (\Gamma_0^0)^*\}=\infty$, or $\diswords$ has the Infix $D$-dip property for some $D\in \bbN$. 
\end{corollary}

\section{On Potential and Type-1 Witnesses}
\label{sec:potential to type 1 witness}
In \cref{sec:discharging and unbounded potential} we discuss the behavior of the charge, and relate it to unbounded potential. 
In this section, we complete the picture by relating unbounded potential to the existence of a type-1 witness.
In broad strokes, this section is very similar to \cref{sec:discharging and unbounded potential}, and we reuse its proof structures (as well as specific techniques).
However, there are many small differences that present new complications, and are therefore treated differently. We emphasize them as they occur.
Due to the similarities, we reuse some terms from \cref{sec:discharging and unbounded potential}, but add a ``$\pot$'' symbol to them, to emphasize their connection to potential.

Fix a finite alphabet $\Gamma'\subseteq \Gamma^0_\infty$ -- that is, without rebase letters (and also without jump letters). In addition, we require that $\Gamma_0^0\subseteq \Gamma'$ (i.e., $\Gamma'$ contains the original letters of $\augA$).
Let $W_{\max}$ be the maximal weight appearing in a transition over $\Gamma'$.

\subsection{A Condition for Type-1 Witness}
\label{sec:condition for type 1 witness}
In this section we introduce a notion of fair decompositions with respect to potential, and use it to establish a condition for the existence of a type-1 witness, analogously to \cref{sec: closing runs condition implies unbounded potential}. 
Our starting point remains elongated words sequence (\cref{def:elongated words sequence}) and its related definitions (e.g., sparsity, independent runs, gaps, etc.).

Our first notion is an analogue of decreasing-gap fair decompositions (\cref{def:decreasing gap fair decomposition}). In this case, however, the potential is increasing (with respect to some run). Note that unlike charge, the potential is generally not attained by the minimal run, and therefore the intuition of ``gap'' somewhat loses its meaning. Indeed, this later causes some case analysis that is not present in the charge setting.

\begin{definition}[\keyicon Increasing-Potential Fair Decomposition]
    \label{def:increasing potential fair decomposition}
    Consider a decomposed word sequence
    $\words$ and assume there exists a function $b:\bbN\to \bbN$ with $\lim_{n\to \infty}b(n)=\infty$ such that for every $m\in \bbN$ with $\words(m)=w_1w_2w_3w_4$, there is some seamless run $\rho$ on $w_1w_2w_3$ on such that $\pot(w_1w_2w_3)-\weight(\rho(w_1w_2w_3))>\pot(w_1w_2)+b(m)-\weight(\rho(w_1w_2))$.

    We say that the decomposition is \emph{increasing-potential fair} (with respect to $\rho$ and $b(n)$) if there is a function $\sqrtb:\bbN\to \bbN$ with $\lim_{n\to \infty}\sqrtb(n)=\infty$ such that the decomposition $w_3=u_1\cdots u_{d(m)}$ satisfies for every $1\le i<d(m)$ that
    \[ 
        \pot(w_1w_2u_1\cdots u_{i-1})+\sqrtb(m)- \weight(\rho(w_1w_2u_1\cdots u_{i-1}))\le \pot(w_1w_2u_1\cdots u_{i})-\weight(\rho(w_1w_2u_1\cdots u_{i}))
    \]
    (there is no requirement on the last ``remainder''segment).
\end{definition}

We prove an analogue of \cref{lem: decomposed seq with cover sparse no ghosts implies unbounded potential}, this time showing that an increasing-potential fair decomposition (with enough ``nice'' properties) implies the existence of a type-1 witness.

\begin{lemma}[\keyicon \lightbulbicon Increasing Potential Decomposition to Witness]
    \label{lem: decomposed inc pot with cover sparse no ghosts implies type 1 witness}
    Assume there exists $G,\ell\in \bbN$ and a decomposed word sequence $\words$ that is $\ell$-sparse with gap 
    $4G + (|S|+1)4\bigM W_{\max} |w_3|$ and a ghost-free decomposition $G$-cover.
    \acctodo{ACCOUNTING}
    In addition, assume the decomposition is increasing-potential fair with respect to a seamless run $\rho$ and function $b:\bbN\to \bbN$.
    Then $\augA_\infty^\infty$ has a type-1 witness.
\end{lemma}
\begin{proof}
Let $G,\ell\in \bbN$ and $\words$ be as in the premise of the lemma, and $b,\sqrtb:\bbN\to \bbN$ and $\rho$ as per \cref{def:decreasing gap fair decomposition}. 
Let $d:\bbN\to \bbN$ be the decomposition length function, i.e., for $m\in \bbN$ and $\words(m)=w_1w_2w_3w_4$ we have $w_3=u_1\cdots u_{d(m)}$.

Take $m\in \bbN$ large enough so that $d(m)>|\difftypeset_{s_0,G}|+1$ and $\sqrtb(m)>4G$. Since both $d$ and $\sqrtb$ are increasing, this is possible.
Consider the word $w_1w_2w_3w_4=\words(m)$ and let 
$w_3=u_1\cdots u_{d(m)}$ be its decreasing-gap fair decomposition.

By the pigeonhole principle, there exist $0\le i_1<i_2<d(m)$ such that 
\[\difftype_{s_0,G}(w_1w_2u_1\cdots u_{i_1})=\difftype_{s_0,G}(w_1w_2u_1\cdots u_{i_2})\]
(Observe that we take $d(m)$ large enough so that we can ignore the last segment, which is a remainder).

Denote 
$\widehat{x_0}=w_1w_2u_1\cdots u_{i_1}$ (the reason for the hat is that we modify $\widehat{x_0}$ to obtain a ``simpler'' $x_0$ below), $x=u_{i_1+1}\cdots u_{i_2}$ and $x_1=u_{i_2+1}\cdots u_{d(m)}w_4$.

\paragraph{Step 1: flatten the prefix before $x$.}
Recall that $\Gamma'\subseteq \Gamma_\infty^0$, i.e., has no rebase letters. It follows that the letters in $w_1w_2w_3w_4$ are either in $\Gamma_0^0$ or are (nested) cacti. Let $F>2\maxeff{\widehat{x_0}x}$, we define 
$x_0=\flatten(\widehat{x_0} \wr F)$. 
By \cref{def:cactus rebase flattening}, it follows that in order to obtain $x_0$, we only apply the cactus unfolding of \cref{def:cactus extension}. This implies that no jump letters are introduced in $x_0$ (indeed, those are introduced only in the case of rebase removal). Thus, $x_0\in (\Gamma_0^0)^*$. 
In the following we wish to reason about $x_0xx_1$ instead of $w_1w_2w_3w_4$. 
However, the main difference between the words is that due to the flattening, we now have $\booltrans(s_0,x_0)=\ghostTrans(s_0,\widehat{x_0})$. 
Consider the configurations $\vec{c'_0}=\xconf(\vec{c_{\init}},\widehat{x_0})$ and $\vec{c_0}=\xconf(\vec{c_\init},x_0)$, and similarly $\vec{c'_x}=\xconf(\vec{c_{\init}},\widehat{x_0}x)$ and $\vec{c_x}=\xconf(\vec{c_\init},x_0x)$.
By \cref{lem:flattening configuration characterization} for every $s,s'\in \booltrans(s_0,\widehat{x_0})$ we have $\vec{c'_0}(s)-\vec{c'_0}(s')=\vec{c_0}(s)-\vec{c_0}(s')$. Additionally, for $r\in \ghostTrans(s_0,\widehat{x_0})\setminus \booltrans(s_0,\widehat{x_0})$ and $r\in \booltrans(s_0,\widehat{x_0})$ we have 
$\vec{c_0}(r)-\vec{c_0}(s)\ge 2\maxeff{\widehat{x_0}x}$
\acctodo{ACCOUNTING}
and in particular ghost states attain higher values than states reachable by $\widehat{x_0}$. In the following, we only use the fact that these are higher, and not the concrete gap size.

We claim that for every prefix $x'$ of $x$ we have that $\pot(\widehat{x_0}x')=\pot(x_0x')$ (and in particular the potential is defined in both).
Indeed, since flattening in our case (of $\Gamma_\infty^0$) amounts to repeated unfolding, then by \cref{lem:unfolding maintains potential} either $\augA_\infty^\infty$ has a type-0 witness (and in particular a type-1 witness), in which case we are done, or $\pot(\widehat{x_0}x')=\pot(x_0x')$.

\paragraph{Step 2: baseline shift to sharpest decrease}
In this step, we use the fact that the potential is increasing with respect to $\rho$ upon reading $x$, and we shift the runs so that the potential becomes increasing in value (i.e., with respect to the baseline). Moreover, our shifting ensures that all runs are non-decreasing, which is used in Step 3.

Define $\rhoss$ to be the seamless run that has the sharpest decrease on $x$. Specifically, $\rhoss$ is the seamless run that minimizes $\weight(\eta(x_0x))-\weight(\eta(x_0))$ among all seamless runs $\eta$ on $x_0x$. 
If there are several such runs, we arbitrarily choose one. 

Recall that the increasing-potential assumption with respect to $\rho$ means that $\rho$ is decreasing relatively to the potential. Therefore, $\rhoss$ is also decreasing relatively to the potential, as it has the sharpest decrease among all runs.

We now perform a baseline shift of $x_0x$ on $\rhoss$, identically to \cref{lem: decomposed seq with cover sparse no ghosts implies unbounded potential}.
Since $\rhoss$ has the sharpest decrease on $x$, and since shifts maintain the gaps between runs (\cref{cor:baseline shift maintains gaps}), we have that all the shifted runs are non-decreasing. That is, for every run $\eta$ on $\baseshift{x_0x}{\rhoss}$ it holds that 
\begin{equation}
\label{eq: shifted runs are nonnegative in increasing potential}
    \weight(\baseshift{\eta}{\rhoss}(x_0x))-\weight(\baseshift{\eta}{\rhoss}(x_0))\ge 0
\end{equation}
Indeed, by \cref{cor:baseline shift to seamless run} we now have that $\baseshift{\rhoss}{\rhoss}$ is the baseline run (and in particular remains at weight $0$). Therefore, by \cref{cor:baseline shift maintains gaps} we have
\[\weight(\eta(x_0x))-\weight(\rhoss(x_0x))=\weight(\baseshift{\eta}{\rhoss}(x_0x))-0\] 
 and similarly 
\[\weight(\eta(x_0))-\weight(\rhoss(x_0))=\weight(\baseshift{\eta}{\rhoss}(x_0))-0\] 
 By subtracting the equations we get 
\[
\begin{split}
&\weight(\baseshift{\eta}{\rhoss}(x_0x))-\weight(\baseshift{\eta}{\rhoss}(x_0))=\\
&\weight(\eta(x_0x))-\weight(\rhoss(x_0x)) - \weight(\eta(x_0)) + \weight(\rhoss(x_0))=\\
&\weight(\eta(x_0x))- \weight(\eta(x_0))- (\weight(\rhoss(x_0x))  - \weight(\rhoss(x_0)))\ge 0
\end{split}
\]
Where the last inequality is because $\rhoss$ minimizes $\weight(\eta(x_0x))  - \weight(\eta(x_0))$ by definition.

Moreover, since the potential is also shifted (see \cref{prop:potential over baseline shift}), we have in particular that for every prefix $x'$ of $x$ it holds that $\pot(\baseshift{x_0x'}{\rhoss})=\pot(x_0x')-\weight(\rhoss(x_0x'))$. 
Recall that by the increasing-potential fair decomposition and by our choice of $\sqrtb(m)>4G$ we have 
$\pot(x_0x)-\weight(\rho(x_0x))>\pot(x_0)-\weight(\rho(x_0))+4G$, and by rearranging: $\pot(x_0x)-\pot(x_0)>\weight(\rho(x_0x))-\weight(\rho(x_0))+4G$.
Since $\rhoss$ has the sharpest decrease, then by the gap-preservation of baseline shift (\cref{cor:baseline shift maintains gaps}) and by the above we have
\[
\begin{split}
&\pot(\baseshift{x_0x}{\rhoss})-\pot(\baseshift{x_0}{\rhoss}) = \\
& \pot(x_0x)-\weight(\rhoss(x_0x))-\pot(x_0)+\weight(\rhoss(x_0)) >\\
& \weight(\rho(x_0x))-\weight(\rho(x_0))+4G -\weight(\rhoss(x_0x))+\weight(\rhoss(x_0))\ge 4G
\end{split}
\]
where the last inequality is obtained by plugging in $\eta=\rho$ in the inequalities above.
We conclude that
\[
\pot(\baseshift{x_0x}{\rhoss})-\weight(\baseshift{\rho(x_0x)}{\rhoss})>\pot(\baseshift{x_0}{\rhoss})-\weight(\baseshift{\rho(x_0)}{\rhoss})+4G
\]
That is, $\baseshift{x_0x}{\rhoss}$ maintains the property that its potential increases upon reading $x$ from $x_0$ by at least $4G$ with respect to $\baseshift{\rho}{\rhoss}$.

Finally, note that since $\baseshift{\rho}{\rhoss}$ is non-decreasing (as all the shifted runs are), then in particular $\weight(\baseshift{\rho(x_0x)}{\rhoss})-\weight(\baseshift{\rho(x_0)}{\rhoss})\ge 0$, and therefore by the inequality above we conclude that the potential increases (a lot) upon reading $\baseshift{x_0x}{\rhoss}$.
\begin{equation}
    \label{eq:type 1 existence potential is increasing}
    \pot(\baseshift{x_0x}{\rhoss})>\pot(\baseshift{x_0}{\rhoss})+4G
\end{equation}

A crucial detail to notice at this point is that since $x$ may contain cactus letters, then its baseline shift may contain \emph{rebase letters} (as per \cref{sec: baseline shift}). Thus, the alphabet after the shift is $\Gamma_\infty^1$, which is why the result of this section is a type-1 witness, rather than type-0.

To avoid the cumbersome notation of baseline shifts in the following, we reuse the names $x_0,x$, $w_1w_2w_3w_4$, and  the runs $\rhoss$ and $\rho$ for their baseline-shifted counterpart. This is sound, since the potential property is maintained, as we show above, and
the properties of being $\ell$-sparse with gap $4G+(|S|+1)4\bigM W_{\max}|w_3|$ and being a ghost-free decomposition $G$-cover are retained by the preservation of run structure as per \cref{cor:baseline shift maintains gaps}.

\paragraph{Step 3: $x$ induces a stable cycle.}
We remark that this step is identical, word-for-word, with the same step in \cref{lem: decomposed seq with cover sparse no ghosts implies unbounded potential}. The reason for the duplication is that we utilize the details of the proof in Step 4 (in particular, the partition to $V_j$ sub-configurations).

Since all the runs on $x$ are nonnegative (\cref{eq: shifted runs are nonnegative in increasing potential}) and $x$ yields a repetition of configurations, we gain the intuition that $x$ should induce a stable cycle. We show this is indeed the case. The difficult part is to show that there are no negative cycles induced by $x^n$ for any $n\in \bbN$. 
The proof relies on the fact that there are large gaps between the runs, and that the runs are a $G$-cover. Broadly, this is a similar proof to \cref{prop:exists inc inf x stable cycle}.

Let $S'=\ghostTrans(s_0,x_0)$, we claim that $(S',x)$ is a stable cycle. Let $\vec{c_0}=\xconf(s_0,x_0)$ as before and $\vec{c_x}=\xconf(s_0,x_0x)$. 
Since $\difftype_{s_0,G}(\widehat{x_0})=\difftype_{s_0,G}(\widehat{x_0}\cdot x)$, it follows that $S'=\ghostTrans(s_0,x_0 \cdot x)$. Indeed, we have $\booltrans(s_0,x_0)=\supp(\vec{c_0})=\supp(\vec{c_x})=\booltrans(s_0,x_0\cdot x)$, and by \cref{sec:reachable ghost states} this also implies the equivalent ghost states. Thus, we have in particular that $\booltrans(S',x)\subseteq S'$, so $(S',x)$ is a reflexive cycle.

The next step (as per \cref{def:stable cycle}) is to show that for the baseline state $s\in S'$ we have that $s\in\MinRefStates(S',x^n)$ for all $n\in \bbN$, and  that $\minweight(x^n,s\to s)=0$. 
Recall that after the baseline shift, we have that $\rho_\ssearrow$ is a baseline run on $x$ and has weight $0$, in particular $\weight(\rho_\ssearrow)=0$, and therefore $\weight(\rho^n_\ssearrow)=0$, so $x^n$ has a cycle of weight $0$. 

It remains to show that there are no negative cycles in $(S',x^n)$ for any $n$. That is, for every $s'\in S'$, we want to show that $\minweight(x^n,s'\to s')\ge 0$.
We start with an overview of how this is proved. 
Recall that our baseline shift operation ensures that $\rho_\ssearrow$ has weight $0$, and $\rho_\ssearrow$ before the rebase has the strongest decrease, so all seamless runs after rebase are non-negative (on $x$). 
This, however, does not immediately extend to $x^n$. Indeed, it could conceptually be that there are e.g., (non-seamless) negative runs $r_1\runsto{x} r_2$ and $r_2\runsto{x} r_1$, but then their concatenation yields a negative seamless cycle on $x^2$.

We therefore take an approach similar to the proof of \cref{lem:unbounded long dip and G imply separated inc infix}. Specifically, we notice that the equivalence of $\difftype_{s_0,G}$ before and after reading $x$ implies that we can partition these configurations to ordered sub-configurations $V_1,\cdots V_k$ such that all the ``interesting'' transitions are from $V_i$ are to $V_i$. We then show that after reading $x$, the actual values of the elements in the configurations of each $V_i$ does not decrease (due to $\rho_\ssearrow$ being the baseline). This allows us to conclude there are no negative cycles. 
For the ghost-states, we use the ghost-freeness assumption to conclude there cannot be negative cycles on them.
We proceed with the formal details.

We remind of the following notations of Step 1. $\vec{c'_0}=\xconf(s_0,\widehat{x_0})$ and $\vec{c_0}=\xconf(s_0,x_0)$. Also denote $\vec{c'_x}=\xconf(s_0,\widehat{x_0}x)$ and $\vec{c_x}=\xconf(s_0,x_0x)$. 
Strictly speaking, note that in Step 1 these configurations are defined before our baseline shift of Step 2. However, the gaps are maintained by the shift, so this overloading is justified.

Since $\difftype_{s_0,G}(\widehat{x_0})=\difftype_{s_0,G}(\widehat{x_0}x)$, and by the preservation of the reachable set, we have that for every $s\in \booltrans(s_0,\widehat{x_0})=\booltrans(s_0,\widehat{x_0}x)$ that $\vec{c_0}(s)=\minweight(x_0,s_0\to s)$ and $\vec{c_x}(s)=\minweight(x_0x,s_0\to s)$.
Put simply -- the configurations correctly track the weights of all states reachable by $\widehat{x_0}$ (we deal later with the remaining states, namely $S'\setminus\booltrans(s_0,\widehat{x_0})$).

We partition $\booltrans(s_0,\widehat{x_0})$ as $V_1\cup \ldots \cup V_k$ according to $\difftype_{s_0,G}(\widehat{x_0})$, similarly to the proof of \cref{lem:unbounded long dip and G imply separated inc infix}. Specifically, we write $p\equiv q$ if $|\vec{c_0}(p)-\vec{c_0}(q)|\le 2G$. Since $w_1w_2w_3w_4$ is a $G$-cover, this is an equivalence relation. We then define $V_1,\ldots,V_k$ to be the equivalence classes of $\sim$, ordered by $\le_{\vec{c_0}}$, i.e., all the states in $V_1$ are lower than those of $V_2$, etc.
By the equivalence of $\difftype$, we have that the partition induces by $\vec{c_x}$ is the same as that of $\vec{c_0}$.

Still similarly to \cref{lem:unbounded long dip and G imply separated inc infix}, since $w_1w_2w_3w_4$  is $\ell$-sparse with a huge gap, we have that if $p\in V_i$, $q\in V_j$ and there is a run $\eta:p\runsto{x} q$, then $i\ge j$. That is, there are no runs that go ``up'' between the sets (as the gap is to large to be bridged by $x$, so this would ``drag down'' $q$).
Moreover, if $q\in V_j$ there exists some $p\in V_j$ such that $p\runsto{x} q$, i.e., there  exist a predecessor of $q$ in $V_j$ (otherwise all predecessors of $q$ come from higher sets, so $q$ would be higher). Additionally, there must exist such $p\in V_j$ such that $\eta:p\runsto{x} q$ is seamless (indeed, runs from higher sets to $q$ cannot be seamless).

We claim that for every $q\in \booltrans(s_0,\widehat{x_0})$ it holds that $\vec{c_0}(q)\le \vec{c_x}(q)$. Intuitively, $x$ only induces increasing behaviors.
To show this, let $1\le i\le k$ such that $q\in V_i$. Let $q^0_i$ be the minimal state in $V_i$ (i.e., $q^0_i\in \arg\min\{\vec{c_0}(q)\mid q\in V_i\}$).

Let $p^0_i$ be a state in $V_i$ such that there exists a seamless run $\eta:p^0_i\runsto{x}q^0_i$. Recall that our baseline shift ensures there are no seamless negative runs. It follows that $\vec{c_x}(q^0_i)=\vec{c_0}(p^0_i)+\weight(\eta)\ge \vec{c_0}(p^0_i)\ge \vec{c_0}(q^0_i)$, where the last inequality is because we know that $q^0_i$ is the minimal in $V_i$. 
Therefore, the claim holds for $q^0_i$. 
For every other state $q\in V_i$ we have that $q\sim q^0_i$, and therefore $\vec{c_x}(q)-\vec{c_x}(q^0_i)=\vec{c_0}(q)-\vec{c_0}(q^0_i)$, so $\vec{c_x}(q)-\vec{c_0}(q)=\vec{c_x}(q^0_i)-\vec{c_0}(q^0_i)\ge 0$.

Notice that the argument above only assumes that the gaps are larger than $|x|$. Recall, however, that we originally start with gaps larger than $(|S|+1) |w_3|\ge  (|S|+1)|x|$ (this includes the gap from the ghost states, which is even larger). It therefore follows that we can repeat this argument for up to $|S|$ times, and obtain that the weights of each state is increased after reading $x^n$ for any $n\le |S|$.  

Now, assume by way of contradiction that there is a negative cycle $\eta:q\runsto{x^n}q$ such that $q\in \booltrans(s_0,\widehat{x_0})$ and $n\in\bbN$. We first notice that we can assume $n\le |S|$ (otherwise we can find an inner $x^{n'}$-cycle that is either itself a negative cycle, or is positive and we can shorten the original cycle). This, however, implies that $\minweight(x_0,s_0\to q)> \minweight(x_0 x^n,s_0\to q)$, which is a contradiction to our observation above.

It remains to show that there are no negative cycles stemming from ghost-states. 
In fact, we prove that there are no cycles (on $x^n$) on ghost states at all, negative or not.
This follows from ghost-freeness assumed in the lemma, as follows.
Assume by way of contradiction that there exists $g\in \ghostTrans(s_0,\widehat{x_0})\setminus \booltrans(s_0,\widehat{x_0})$ and $\eta:g\runsto{x^n}g$. We decompose $\eta$ as $\eta:g\runsto{x}g'\runsto{x^{n-1}}g$ according to the state reached after $x$.
Since $x$ spans at least one $u_i$ segment by definition, it follows from ghost-freeness (\cref{def:ghost free decomposition}) that $\rho$ is not seamless already before $g'$, i.e., we can write $x=y_1y_2$ such that $\eta:g\runsto{y_1}h\runsto{y_2}g'\runsto{x^{n-1}}g$ and we can assume $h\in \booltrans(s_0,\widehat{x_0}y_1)$. 

There is a slightly delicate point here: it could be that $\eta$ is not seamless but that the run that ``undercuts'' it is also from a ghost-state. However, we can then repeat this argument until eventually reaching a ``minimal'' ghost run, after which the undercutting run indeed uses a state in $\booltrans(s_0,\widehat{x_0}y_1)$.

Since $\eta$ is a run, and in particular has finite weight, it follows that its suffix $\eta':h\runsto{y_2}g'\runsto{x^{n-1}}g$ is also finite weight, and since $h\in \booltrans(s_0,\widehat{x_0}y_1)$, we have that $g\in \booltrans(s_0,\widehat{x_0}x^n)$, but $x$ preserves the set of reachable states (by the equivalence of $\difftype$), so it follows that $g\in \booltrans(s_0,\widehat{x_0})$, in contradiction to the assumption that $g\notin \booltrans(s_0,\widehat{x_0})$.

We conclude that $(S',x)$ is a stable cycle, and so $\alpha_{S',x}$ is a cactus letter.

\paragraph*{Step 4: obtaining a type-$1$ witness.}
Continuing with the notations of Step 3, consider a maximal-dominant state $q_x$ of $\vec{c_x}=\xconf(s_0,x_0x)$, and let $z$ be a corresponding suffix such that 
$\minweight(z,q_x\to S)<\infty$ and for every $q'$ with $\vec{c_x}(q')<\vec{c_x}(q)$ it holds that $\minweight(z,q'\to S)=\infty$.
Assume $q_x\in V_i$ for some $1\le i\le k$. 
Consider also a maximal-dominant state $q_0$ of $\vec{c_0}=\xconf(s_0,x_0)$.

Intuitively, this step proceeds as follows. We first show that we can assume $q_0=q_x$, and in particular $q_0\in V_i$ as well. That is, the potential before reading $x$ stems from the same $V_i$ sub-configuration as the potential after $x$. In fact, we show that the same suffix $z$ works for $q_0$.
Then, since the potential increases on $x$ (by Step 1), we conclude that the sub-configuration corresponding to $V_i$ also increases. We use this to show by a pumping argument that all the runs on $x$ from $V_i$ to $V_i$ are between non-grounded pairs of $(S',x)$. 
This means that after replacing $x$ with $\alpha_{S',x}$, the configuration $V_i$ no longer yields runs to $V_i$, only to ``lower'' $V_j$. We then consider the suffix $\alpha_{S',x}z$ from $\vec{c_0}$, and show that with this suffix, all the states up to and including $V_i$ attain weight $\infty$. 
Finally, this shows that either $q_0$ is not maximal-dominating (which is a contradiction) or all states go to $\infty$, meaning that $(x_0,x,z)$ is a (type-1) witness. 
We now proceed with the formal details.

We start by claiming that we can assume $q_0=q_x$. 
That is, we show that $q_x$ is maximal-dominant in $\vec{c_0}$.
We already have that $\minweight(z,q_x\to S)<\infty$. Consider $q'$ with $\vec{c_0}(q')<\vec{c_0}(q_x)$. Since $\vec{c_0}$ and $\vec{c_x}$ have the same $\difftype_{s_0,G}$, then $\vec{c_x}(q')<\vec{c_x}(q_x)$, and therefore we also have $\minweight(z,q'\to S)=\infty$. Thus, $q_x$ is dominant in $\vec{c_0}$. To show it is also maximal dominant, we repeat the same argument only going from $\vec{c_0}$ to $\vec{x}$: assume $q''$ satisfies $\vec{c_0}(q'')>\vec{c_0}(q_x)$, but $q''$ is dominant with suffix $z''$. Then an identical argument shows that $q''$ is dominant in $\vec{c_x}$, and by the order preservation, it is also higher than $q_x$ in $\vec{c_x}$, contradicting the fact that $q_x$ is maximal dominant.

Recall by \cref{eq:type 1 existence potential is increasing} that (under the baseline-shifted notations for $x_0,x$) we have $\pot(x_0x)>\pot(x_0)+4G$. Specifically, let $P=\pot(x_0x)-\pot(x_0)$, then $P>4G$ and $\vec{c_x}(q_x)=\pot(c_0)(q_x)+P$. We now intuitively show that pumping $x$ must keep increasing the sub-configuration $V_i$, and therefore there are no grounded pairs from $V_i$ to itself.
To show this, consider a minimal-weight run
$\rho:s_0\runsto{x_0}p_0\runsto{x}p_1\runsto{x}\cdots \runsto{x}p_{k-1}\runsto{x}p_k$ for $p_k\in V_i$ on $x^k$ with $k\le 4\bigM$. 
From Step 3 we know that there are no transitions on $x$ from $V_j$ to $V_i$ if $j<i$, and that for every state in $V_i$ there is a transition from some state in $V_i$. 
Since $\rho$ is minimal-weight to $p_k$, it is in particular seamless. By the gap and cover criteria in the lemma, for every $p'\in V_j$ with $j>i$ we have 
\[
\begin{split}
&\vec{c_0}(p')-\vec{c_0}(p_0)>4G+(|S|+1)4\bigM W_{\max}|w_3|-2G=\\
&2G+(|S|+1)4\bigM W_{\max}|w_3|>
8\bigM W_{\max}|x| \ge 2\maxeff{x^k}
\end{split}
\]
\acctodo{ACCOUNTING}
Where the $-2G$ in the first inequality comes because we measure the gap between $V_i$ and $V_j$, but the gap in the premise is between independent runs (which form a $G$-cover), and the last inequality is because $k\le 4\bigM$. 

We now see that all the $p_m$ are in $V_i$, otherwise there is a run from a higher $V_j$, but such a run cannot reach as low as $V_i$ due to the gap.
More precisely, since $\rho$ is seamless, we have that $p_m\in V_i$ for every $1\le m\le k$ (as there is a run through the $V_i$, and any run through higher $V_j$ cannot be seamless due to the gap). 

Assume by way of contradiction that there are $s,r\in V_i$ such that $(s,r)\in \GroundPairs(S',x)$. By \cref{lem:pumping grounded pairs}, this means that $\minweight(x^{4\bigM},s\to r)=\minweight(x^{2\bigM},s\to r)$. In particular, $\xconf(\vec{c_0},x^{4\bigM})(r)\le 2G + \xconf(\vec{c_0},x^{2\bigM})(r)$, since in the worst case the minimal run to $r$ starts from a state other than $s$ that is lower by $2G$ than $\vec{c_0}(s)$ (due to the $G$-cover, as it also starts from $V_i$).

But due to the repetition of $\difftype_{s_0,G}$, and by our observation that runs from higher $V_j$ are not seamless on $x^k$, combined with \cref{eq:type 1 existence potential is increasing} (namely that the sub-configuration of $V_i$ grows by exactly $P$ with $x$), we have
$\xconf(\vec{c_0},x^{4\bigM})(r)= \vec{c_0}(r)+ 4\bigM P$ and $\xconf(\vec{c_0},x^{2\bigM})(r)= \vec{c_0}(r)+ 2\bigM P$.
Plugging these into the inequality above, we have:
$\vec{c_0}(r)+4\bigM P\le 2G+\vec{c_0}(r)+ 2\bigM P$, so $2\bigM P\le 2G$, but $P>4G$, so this is a contradiction. 
Intuitively, this is similar to the situation depicted in \cref{fig:increasing infix run characterization}

We conclude that for every $s,r\in V_i$ we have $(s,r)\notin \GroundPairs(S',w)$.
Now consider the suffix $\alpha_{S',x}z$ from $\vec{c_0}$. For every $p\in V_j$ with $j<i$, we have that $\booltrans(p,\alpha_{S',x})\subseteq \bigcup_{j'\le j}V_{j'}$, and since $z$ has no finite-weight runs from any $V_{j'}$ lower than $V_i$, we have 
\[\minweight(\alpha_{S',x}z,p\to S)=\minweight(\alpha_{S',x}z,p\runsto{\alpha_{S',x}} \bigcup_{j'\le j}V_j\runsto{z}S)=\infty\]
Moreover, for every $q\in V_i$ we have $\booltrans(q,\alpha_{S',x})\in \bigcup_{j'<i} V_{j'}$ (note the strict inequality). This is because there are no grounded pairs from $V_i$ to itself, so all the transitions on $\alpha_{S',x}$ go to strictly lower $V_{j'}$ sets. In particular, we again have
$\minweight(\alpha_{S',x}z,q\to S)=\infty$.

Recall that $q_0\in V_i$ is maximal dominant. Therefore, if for some $q''\in V_{i'}$ for $i'>i$ we have $\minweight(\alpha_{S',x}z,q''\to S)<\infty$, this would contradict the maximal-dominance of $q_0$. We therefore conclude that for every $q\in \supp(\vec{c_0})$ we have  $\minweight(\alpha_{S',x}z,q\to S)=\infty$.

We thus have the following setting: $\minweight(x_0\cdot x\cdot z,s_0\to S)<\infty$, and $\minweight(s_0\alpha_{S',x}z,s_0\to S)=\infty$, thus fulfilling Requirement 4 of \cref{def:witness}, showing that $(x_0,x,z)$ is a witness (the other properties are immediate by the fact that $x$ is a stable cycle and preserves the reachability set).

As mentioned in Step 2, due to the baseline shift we might have $x\in (\Gamma_\infty^1)^*$. Since $x_0\in (\Gamma_0^0)^*$ due to the flattening in Step 1, we conclude that this is a type-1 witness.
\end{proof}

\subsection{$\pot$-Leveled Words and $\pot$-Discharging Words}
\label{sec:potleveled and potdischarging words}
In this section we present the potential analogue to leveled and discharging words, as well as the corresponding theorems of \cref{sec:leveled and discharging words}. Again, the overall structure is the same, but there are technical differences.

\subsubsection{$\pot$-Leveled Words Sequences}
\label{sec:potleveled words sequences}
An elongated words sequence $\words$ is $\pot$-leveled if, intuitively, the potential along the $w_3$ infix remains in some constant width ``band''.
\begin{definition}[$\kappa$-$\pot$-Leveled Words Sequence]
    \label{def:potleveled words sequence}
    For $\kappa\in \bbN$, a function from $\bbN$ to $\Gamma'^*$ is a \emph{$\kappa$-$\pot$-leveled words sequence}, denoted $\potlevwords:\bbN\to \Gamma'^*$, if it is an elongated words sequence, and for every $m\in \bbN$ with $\potlevwords(m)=w_1w_2w_3w_4$ and prefix $u$ of $w_3$ it holds that 
    $|\pot(w_1w_2)-\pot(w_1w_2u)|\le \kappa$.
\end{definition}
We say that a sequence is $\pot$-leveled if it is $\kappa$-$\pot$-leveled for some $\kappa$. 
 
\begin{proposition}
    \label{prop:potleveled with large gaps no ghost G cover implies type 1 or D dip}
    Let $\ell,G\in \bbN$ and consider a $pot$-leveled words sequence $\potlevwords$ that is $\ell$-sparse with gap $4G+(|S|+1)4\bigM W_{\max}|w_3|$ and is a configuration $G$-cover. Further assume there exists $l\in \bbN$ such that for every $m$ with $\potlevwords(m)=w_1w_2w_3w_4$, all ghost runs on $w_3$ are of length at most $l$.
    Then either $\augA_{\infty}^\infty$ has a type-1 witness, or $\potlevwords$ has the Infix $D$-dip property for some $D\in \bbN$.
\end{proposition}
\begin{proof}
    The proof is nearly identical to that of \cref{prop:leveled with large gaps no ghost G cover implies unbounded potential or D dip}, with only minor changes between charge and potential, and with eventually calling upon \cref{lem: decomposed inc pot with cover sparse no ghosts implies type 1 witness} to obtain a type-1 witness.

    If $\potlevwords$ has the Infix $D$-dip property for some $D\in \bbN$, we are done. We therefore assume this is not the case. That is, by the converse of \cref{def:infix D dip}, for every $D,n\in \bbN$, there exists $m$ with $\potlevwords(m)=w_1w_2w_3w_4$ such that we can write $w_3=xvy$ with $|v|>n$ and there exists $p\in \ghostTrans(s_0,w_1w_2x)$ and a run $\rho:p\runsto{x}S$ with $\weight(\rho)\le -D$.

    Our goal now is to obtain from $\potlevwords$ another $\pot$-leveled sequence $\potlevwords'$ that has an increasing-potential fair decomposition (\cref{def:increasing potential fair decomposition}). Intuitively, we obtain the increasing potential using the decreasing runs on $v$ above. 
    Note that despite the counterintuitive phrasing (increasing v.s. decreasing), this is still similar to the $\charge$ case, as there also a decreasing charge essentially means an increasing minimal run.

    We define $\potlevwords'$ as follows. For every $m\in \bbN$, let $m'>m$ such that for $\potlevwords(m')=w_1w_2w_3w_4$ we can write $w_3=xvy$ with $|v|>(|S|+1)m+l$  (recall that $l$ is the bound on the length of a ghost run) and there exists $p\in \ghostTrans(s_0,w_1w_2x)$ and a run $\rho:p\runsto{v}S$ with $\weight(\rho)\le -W_{\max}((|S|+1)m+l)$. 
    Since there are no ghost runs of length greater than $l$ over $w_3$, it follows that $\rho$ reaches a reachable state $q$ within $l$ steps from $p$. We then still have a run $\rho'$ over a suffix $v'$ of $v$ such that $\rho':q\runsto{v'}S$ with $\weight(\rho')\le -W_{\max}(|S|+1)m$. Indeed, $\rho$ can lose at most $W_{\max}l$ weight within $l$ steps. 
    Recall that for increasing-potential fair decompositions we need a \emph{seamless} run (\cref{def:increasing potential fair decomposition}). Intuitively, the existence of a long-enough decreasing run clearly implies the existence of a decreasing seamless run, since a run cannot decrease without ``dragging down'' with it some seamless run. 
    
    Formally, if $\rho'$ has an infix of length at least $m$ that is decreasing and seamless, we redefine $v'$ as that infix and we are done. Otherwise, decompose $\rho'$ into $|S|+1$ segments by the first time $\rho'$ loses an additional $W_{\max}m$ weight.     
    Since there are at most $|S|$ seamless runs, it follows that $\rho'$ intersects the same seamless runs at least twice at the end of these segments, but then this seamless run is decreasing at least as much as $\rho'$, so there is a seamless run decreasing by at least $W_{\max}m$, as required. We reuse $v'$ to denote the corresponding suffix, and write 
 $v=v_0v'$.

    Define $\potlevwords'(m)=w'_1w'_2w'_3w'_4$ with $w'_1=w_1,w'_2=w_2xv_0,w'_3=v'$ and $w'_4=yw_4$.
    Denote by $\kappa$ the potential band of $\potlevwords$.
    Observe that $\potlevwords'$ is a $2\kappa$-$\pot$-leveled word sequence: the infixes $v'$ have increasing length, and since the words in $\potlevwords'$ are equal as concatenations to words from $\potlevwords$, then the bounded-charge property is preserved, with the modification that previously all potentials were within gap $\kappa$ of $\pot(w_1w_2)$, whereas now it is possible that e.g., $\pot(w'_1w'_2)=\pot(w_1w_2)-\kappa$, so the remaining charge is within distance at most $2\kappa$. 
    Moreover, $\potlevwords'$ is $\ell$-sparse with gap $4G+(|S|+1)4\bigM W_{\max}|w_3|$ and a configuration $G$-cover. Indeed, since we choose $v'$ that are shorter than their corresponding $|w_3|$, the gap requirement that is already met in $\potlevwords$ is even larger with respect to $\potlevwords'$, and the decomposition $G$-cover is implied by the configuration $G$-cover.

    It remains to show that $\potlevwords'$ has an increasing-potential fair decomposition. 
    Intuitively, this follows because the potential stays within a band, whereas $\rho$ is decreasing, and therefore the difference of $\rho$ from the potential must increase.
    Formally, for every $m\in \bbN$ consider $\potlevwords'(m)=w'_1w'_2w'_3w'_4$ as above, where by construction there exists $q\in \booltrans(s_0,w'_1w'_2)$ and a run $\rho:q\runsto{w'_3}S$ with $\weight(\rho)\le -W_{\max}m$ and $|w'_3|>m$. 
    We decompose $w'_3=u_1u_2\cdots u_{d(m)+1}$ as follows. 
    Let $d(m)$ be the maximal number such that the following ordered indices are distinct: 
    $0=i_0\le i_1\le \ldots \le i_{d(m)}<|w_3|$ where for every $1\le k\le d(m)$ we have
    \[i_k=\min\{j\mid \weight(\rho(w'_1w'_2w'_3[1,i_{k-1}]))-\weight(\rho(w'_1w'_2w'_3[1,j]))\ge \sqrt{m}W_{\max}\}\]
    We then set $u_{k}=w'_3[i_{k-1}+1,i_k]$ for all $1\le k\le d(m)$ and $u_{d(m)+1}$ the remaining suffix.

    Observe that for every $1\le k\le d(m)$ we now have by definition that \[\weight(\rho(w'_1w'_2u_1\cdots u_{k-1})-\weight(\rho(w'_1w'_2u_1\cdots u_{k}))\ge \sqrt{m}W_{\max}\] 
    or equivalently 
    \[-\weight(\rho(w'_1w'_2u_1\cdots u_{k-1})+\sqrt{m}W_{\max}\le -\weight(\rho(w'_1w'_2u_1\cdots u_{k}))\]
    Also, since the potential remains within a $2\kappa$ band, we have that 
    $\pot(w'_1w'_2u_1\cdots u_{k-1})-2\kappa \le \pot(w'_1w'_2u_1\cdots u_{k})$. Thus, we have
    \[
    \begin{split}
    &\pot(w'_1w'_2u_1\cdots u_{k-1})+\sqrt{m}W_{\max}-2\kappa -\weight(\rho(w'_1w'_2u_1\cdots u_{k-1}))\le  \\
    &\pot(w'_1w'_2u_1\cdots u_{k})-\weight(\rho(w'_1w'_2u_1\cdots u_{k}))
    \end{split}
    \]
    Moreover, since with each letter the run $\rho$ loses at most $W_{\max}$ weight, then the number of segments in the decomposition is (roughly) $\sqrt{m}$, i.e., $\lim_{m\to \infty}d(m)=\infty$.
    
    Thus, $\potlevwords'$ has an increasing-potential fair decomposition with $\sqrtb(m)=\sqrt{m}W_{\max}-2\kappa$.   
    We conclude that $\potlevwords'$ satisfies the conditions of \cref{lem: decomposed inc pot with cover sparse no ghosts implies type 1 witness}, and therefore $\augA_\infty^\infty$ has a type-1 witness, and we are done.
\end{proof}
    
We are now ready for the main inductive argument.

\begin{lemma}
    \label{lem:potleveled with large gaps implies type 1 or D dip}
    Let $\ell\in \bbN$ and consider a $\pot$-leveled words sequence $\potlevwords$ that is $\ell$-sparse with gap $(|S|+1)4\bigM W_{\max}|w_3|$. Then either $\augA_\infty^\infty$ has a type-1 witness, or $\potlevwords$ has the Infix $D$-dip property for some $D\in \bbN$.
\end{lemma}
\begin{proof}
    The proof is by reverse induction on $\ell$, and is very similar to that of \cref{lem:leveled with large gaps implies unbounded potential or D dip}, with the exception of the fact that for flattening in the inductive step, we must use \cref{lem:unfolding maintains potential}, which is less straightforward than the case of charge.

    \paragraph*{Base case: $\ell=|S|$} Recall that the maximal number of configuration-independent runs is $|S|$ (see \cref{def:configuration independent runs}). 
    In this case, for every $m\in \bbN$ with $\potlevwords(m)=w_1w_2w_3w_4$ we have that in every configuration in $w_3$, each state belongs to one of the independent runs. 
    In particular, $\potlevwords$ is a configuration $G$-cover with $G=0$, and is also $\ell$-sparse by assumption, with gap $4G+(|S|+1)4\bigM W_{\max}|w_3|$ (since $G=0$).

    Moreover, since $\ell=|S|$, then in particular every state is reachable along any prefix of $w_3$, and therefore there are no ghost runs at all along $w_3$. 
    Thus, the conditions of \cref{prop:potleveled with large gaps no ghost G cover implies type 1 or D dip} are met, so either $\potlevwords$ has the infix $D$-dip property for some $D$, or $\augA_\infty^\infty$ has a type-1 witness, and we are done.

    \paragraph*{Inductive case: $\ell<|S|$}
    The induction is split into three subcases. 
    \subparagraph*{Case 1: long ghost runs.}
    In this case we assume that there are unboundedly long ghost runs. Specifically, assume that for every $l\in \bbN$ there is some $m\in \bbN$ with $\potlevwords(m)=w_1w_2w_3w_4$ and we can write $w_3=xuy$ such that $|u|\ge l$ and there exists $p\in \ghostTrans(s_0,w_1w_2x)\setminus \booltrans(s_0,w_1w_2x)$ and a ghost run $\rho:p\runsto{u}S$.
    Intuitively, while this ghost run is not a real run on $w_1w_2w_3w_4$, it becomes a real run if we flatten the prefix up to it. Moreover, since flattening guarantees that ghost runs are very high, this implies the existence of another independent run, so we can apply the induction.

    Formally, we define a new sequence $\potlevwords'$ as follows. For every $m'\in \bbN$ let $m>m'$ be such that $\potlevwords(m)=w_1w_2w_3w_4$ has a ghost run as above with $|u|>m'$. Define $\potlevwords'(m')=w'_1w'_2w'_3w'_4$ with
    $w'_1=\flatten(w_1w_2x \wr (|S|+1)4\bigM W_{\max}|w_1w_2w_3w_4|), w'_2=\epsilon, w'_3=y$ and $w'_4=yw_4$.\footnote{We remark that unlike \cref{sec:separated increasing infix}, the roles of $w_1$ and $w_2$ here are combined, so it is fine selecting $w_2=\epsilon$ and moving its content to $w_1$ (or vice versa). We keep the four-part formalism for uniformity.} 

    Similarly to \cref{lem: decomposed inc pot with cover sparse no ghosts implies type 1 witness}, since our alphabet does not contain rebase letters, by \cref{def:cactus rebase flattening} the resulting prefix $x'_1$ is over $\Gamma_0^0$. 
    Denote $\vec{c}=\xconf(\vec{c_{init}},w_1w_2x)$ and $\vec{c'}=\xconf(\vec{c_{init}},w'_1w'_2)$, then by 
    \cref{lem:flattening configuration characterization}
    we have that $\vec{c}(q)=\vec{c'}(q)$ for every $q\in \booltrans(s_0,w_1w_2)$, and 
    $\vec{c'}(p)\ge (|S|+1)4\bigM W_{\max}|w_1w_2w_3w_4| + \max\{\vec{c}(p')\mid p'\in \supp(\vec{c})\}$.
    Specifically, the entire run $\rho$ from $p$ is above all the existing $\ell$ independent runs with gap at least $(|S|+1)4\bigM W_{\max}|w_3|$.
    It follows that $\rho$ can be added as an independent run. Thus, $\potlevwords'$ has $\ell+1$ independent runs. 

    In order to apply the induction we need to show that $\potlevwords'$ remains a $\pot$-leveled words sequence. This, however, might not be the case. Indeed, the flattening may have changed the potential. Fortunately, since our alphabet is $\Gamma'\subseteq \Gamma_\infty^0$, then in particular it does not contain rebase or jump letters. This means that flattening amounts to repeated unfolding of cactus letters (\cref{def:unfolding function,def:cactus rebase flattening}). Thus, we can apply \cref{lem:unfolding maintains potential} and get that either $\augA_\infty^\infty$ has a type-0 witness (and in particular a type-1 witness), or we have that the potential remains unchanged after flattening, so $\potlevwords'$ remains a $\pot$-leveled words sequence.
    We can therefore apply the induction hypothesis on $\potlevwords'$. By the hypothesis, if $\augA_\infty^\infty$ has a type-1 witness, then we are done.
    
    Otherwise, we get that $\potlevwords'$ has the $D$-dip property. Recall, however, that we need to show $\potlevwords$ has the $D$-dip property. Fortunately, observe that $D$-dip is ``local'' in that it only considers runs over the relevant infix, and since the $D$-dip infixes of $\potlevwords'$ also appear in $\potlevwords$, so we can indeed show this as follows.

    For every $n\in \bbN$, consider $m'$ such that $\potlevwords(m')=w'_1w'_2w'_3w'_4$ has an infix $u$ with $w'_3=xuy$ with $|u|\ge n$ and for every $p\in \ghostTrans(s_0,w'_1w'_2x)$ and run $\rho:p\runsto{u}S$ we have $\weight(\rho)>-D$ (these exist by the $D$-dip property, see \cref{def:infix D dip}). 
    By the construction of $\potlevwords'$, there exists $m>m'$ such that $\potlevwords(m)=w_1w_2w_3w_4$ and $w'_3$ is an infix of $w_3$. Moreover, while $w'_1w'_2x$ differs from the corresponding prefix $z$ of $\potlevwords(m)$ due to the flattening, it still holds that $\ghostTrans(s_0,z)=\ghostTrans(s_0,w'_1w'_2x)$ by the flattening procedure.
    It follows that the same condition on $u$ holds also in $w_1w_2w_3w_4$. Therefore, $\potlevwords$ also has the infix $D$-dip property, and we are done.

    In the following cases, we therefore assume there is a bound $l\in \bbN$ such that all ghost runs on infixes of any $w_3$ are of length at most $l$.

    \subparagraph*{Case 2: a diverging run on a suffix of $w_3$.}
    In this case we essentially assume that the independent runs are not a $G$-cover for any $G$, and use this to construct a new independent run from states that stay far away from all independent runs. We then apply the induction hypothesis. This is very similar to the setting in \cref{fig:sparse induction second case}, where instead of starting from a single state, we start from a configuration. Specifically, we assume the divergence occurs in some long-enough prefix of $w_3$.

    Formally, assume that for every $G\in \bbN$ there exists $m>2G$ with $\potlevwords(m)=w_1w_2w_3w_4$ and we can write $w_3=uv$ with $|u|\ge \frac12|w_3|$ such that for $\vec{c}=\xconf(s_0,w_1w_2u)$ there exists a state $q\in \supp(\vec{c})$ for which $\vec{c}(q)$ is not withing gap $G$ of any independent run.

    We define a new sequence $\potlevwords'$ as follows. For every $m'\in \bbN$, take some $G>2W_{\max}m'+(|S|+1)4\bigM W_{\max}m'$ 
    and let $m>2m'$ large enough such that the gap of $q$ as above is at least $G$.
    Write $\potlevwords(m)=w_1w_2w_3w_4$ with $w_3=uv$ and $q$ as above. 
    We can now use the fact that all seamless runs change their weight by at most $W_{\max}$ at each step. Consider the seamless run $\rho$ leading up to $q$, and write $u=u_1u_2$ with $|u_2|=m'$ (this is possible since $|u|\ge \frac12|w_3|\ge \frac12 m>m'$).
    Since $q$ is not within gap $G$ of any independent run, for every prefix $u'$ of $u_2$ it holds that $\weight(\rho(w_1w_2u_1u'))$ is not within gap $G-2W_{\max}m'>(|S|+1)4\bigM W_{\max}m'$ of any independent run.

    We can therefore define $\potlevwords(m')=w'_1w'_2w'_3w'_4$ by $w'_1=w_1,w'_2=w_2u_1,w'_3=u_2,w'_4=vw_4$. As above, this is a $\pot$-leveled words sequence, since the potential requirement remains valid on infixes. 
    Moreover, we now have $\ell+1$ independent runs satisfying the gap constraints. We can therefore apply the induction hypothesis. As in Case 1, if $\augA_\infty^\infty$ has a type-1 witness then we are done, and if $\potlevwords'$ has the infix $D$-dip property, this lifts back to $\potlevwords$ (here we do not even have to consider the flattening).

    \subparagraph*{Case 3: $G$-cover on suffix of $w_3$.}
    We now consider the complement of Case 2, still under the assumption that all ghost runs are of length at most $l$.
    Specifically, we assume that there exists $G\in \bbN$ such that for every $m>2G$ with $\potlevwords(m)=w_1w_2w_3w_4$, for every decomposition $w_3=uv$ with $|u|\ge \frac12|w_3|$ and for every state $q\in \booltrans(s_0,w_1w_2u)$, the weight of $q$ is within gap $G$ of some independent run. We therefore have a $G$-cover of the ``second half'' of $w_3$. 
    Our goal is now to apply \cref{prop:potleveled with large gaps no ghost G cover implies type 1 or D dip}. For this, we only need to define a sequence where the $w_3$ infix corresponds to these ``second halves''. A slight caveat is that \cref{prop:potleveled with large gaps no ghost G cover implies type 1 or D dip} requires a larger gap than $\potlevwords$ has. This, however, is an ``artificial'' problem: the gaps in $\potlevwords$ grow with $m$, so we just need to pick words that are further away, so we have larger gaps.
    
    Formally, we proceed as follows. For every $m'\in \bbN$, let $m=2(m'+G)$ with $\potlevwords(m)=w_1w_2w_3w_4$ and write $w_3=uv$ with $|v|=m'$ (in particular $|u|\ge \frac12 |w_3|$). Define $\potlevwords(m')=w'_1w'_2w'_3w'_4$ by $w'_1=w_1,w'_2=w_2u,w'_3=v,w'_4=w_4$. We observe that $\potlevwords'$ is still a leveled words sequence, since the potential constraints do not change. We claim that $\potlevwords'$ is $\ell$-sparse with gap $4G+(|S|+1)4\bigM W_{\max}|w_3|$. Indeed, recalling that $m=2(m'+G)$, the gap guaranteed by $\potlevwords$ is 
    \[
    \begin{split}
    &(|S|+1)4\bigM W_{\max}|w_3|\ge (|S|+1)4\bigM W_{\max}2(m'+G)\ge \\
    &4G+ (|S|+1)4\bigM W_{\max}m'=4G+ (|S|+1)4\bigM W_{\max}|w'_3|    
    \end{split}
    \]
    Finally, $\potlevwords'$ has a configuration $G$-cover, as we discuss above.
    
    We can therefore apply \cref{prop:potleveled with large gaps no ghost G cover implies type 1 or D dip} to $\potlevwords'$ and we conclude that either $\potlevwords'$ has the infix $D$-dip property, or $\augA_\infty^\infty$ has a type-1 witness. Again, the infix $D$-dip property lifts back to $\potlevwords$, so we are done.
    Note that the last case is not inductive, but direct (indeed, it can be thought of as a different base case).
    \end{proof}

As a potential-counterpart to \cref{cor:leveled implies unbounded potential or D dip} for leveled words, we have the main tool for $\potlevwords$ (by applying \cref{lem:potleveled with large gaps implies type 1 or D dip} with $\ell=1$, for which the gap is unbounded).
\begin{corollary}
    \label{cor:potleveled implies witness or D dip}
    Consider a $\potlevwords$ sequence, then either $\augA_\infty^\infty$ has a type-1 witness, or $\potlevwords$ has the Infix $D$-dip property for some $D\in \bbN$. 
\end{corollary}

\subsubsection{$\pot$-Discharging Words Sequences}
\label{sec:potdischarging word sequences}
A counterpart to $\pot$-leveled words are \emph{$\pot$-discharging words}, where the potential significantly increases. We remark that while the naming scheme may seem confusing (i.e., ``discharging'' is used for ``increasing potential''), recall that in the charge case, a decreasing charge amounts to an increasing minimal run (intuitively), so in this respect the scenarios are similar. We therefore keep the name, to maintain the analogy.
\begin{definition}[$\pot$-Discharging Words Sequence]
\label{def:potdischarging word sequence}
An elongated words sequence $\words$ is a \emph{$\pot$-discharging words sequence}, denoted $\potdiswords:\bbN\to \Gamma'^*$, if for every $m\in \bbN$ with $\potdiswords(m)=w_1w_2w_3w_4$ the following hold:
\begin{itemize}
    \item $\pot(w_1w_2w_3)-\pot(w_1w_2)>m$ (the potential significantly increases upon reading $w_3$).
    \item For every prefix $v$ of $w_3$ we have $\pot(w_1w_2v)\le \pot(w_1w_2w_3)$ (the potential does not exceed its level on $w_1w_2w_3$).
\end{itemize}
\end{definition}
Similarly to the discharging words setting (\cref{sec:discharging word sequences}), we introduce a specialized notion of decomposition for $\pot$-discharging words, whereby in each segment the potential increases.
Recall that for the charge setting, when using the decomposition in \cref{def:discharging decomposition}, we encounter the problem that the charge may have unbounded ``jumps''. To overcome this, we resort there to assuming the bounded charge decrease property of \cref{def:charge bounded decrease}.

Fortunately, in this section things are simpler: the potential always has bounded increase, as per \cref{lem:bounded growth potential}. This means that we get simplified assumptions when reasoning about $\pot$-discharging decompositions, as follows.

\begin{definition}[$\pot$-Discharging Decomposition]
    \label{def:potdischarging decomposition}
    For $m\in \bbN$ and $\potdiswords(m)=w_1w_2w_3w_4$, let $d(m)$ be the maximal number such that the following ordered indices are distinct:
    $0=i_0 \le i_1\le i_2\le \ldots\le i_{d(m)}< |w_3|$ where for every $1\le k\le d(m)$ we have
    $i_k=\min\{j\mid \pot(w_1w_2w_3[1,j])-\pot(w_1w_2w_3[1,i_{k-1}])>\sqrt{m}\}$.
    
    The \emph{discharging decomposition} of $w_3$ is then $w_3=u_1u_2\cdots u_{d(m)+1}$ where $u_k=w_3[i_{k-1}+1,i_k]$ for every $1\le k\le d(m)$, and $u_{d(m)+1}=w_3[i_{d(m)+1},|w_3|]$.

    For every $1\le k\le d(m)$ we then have: 
    \begin{itemize}
        \item $\pot(w_1w_2u_1\cdots u_{k})-\pot(w_1w_2u_1\cdots u_{k-1})> \sqrt{m}$.
        \item For every prefix $v$ of $u_1\cdots u_k$ it holds that 
        $\pot(w_1w_2v)\le \pot(w_1w_2u_1\cdots u_{k})$
    \end{itemize}
\end{definition}
As mentioned above, due to the bounded growth of the potential, we readily have that the $\pot$-discharging decomposition is an increasing potential decomposition, as follows.
\begin{proposition}
    \label{prop:potdischarging words has fair decomposition}
    Let $\potdiswords$ be a $\pot$-discharging sequence, then its $\pot$-discharging decomposition is an increasing-potential fair decomposition (as per \cref{def:increasing potential fair decomposition}).
\end{proposition}
\begin{proof}
    Let $m\in \bbN$ and $\potdiswords(m)=w_1w_2w_3w_4$. Write $w_3=u_1\cdots u_{d(m)+1}$.

    We claim that the $\pot$-discharging decomposition is an increasing-potential fair decomposition with respect to the baseline run $\rho_0$. Note that the baseline run has weight $0$ in all prefixes. Therefore, we need to show that the length of the $u_i$ increases with $m$, that $d(m)$ increases with $m$, and that $\pot$ increases by some function $\sqrtb(m)$ between each segment (as per \cref{def:increasing potential fair decomposition}, instantiated with $\rho_0$) except for the last ``remainder'' segment. 

    The latter requirement is trivial by the \cref{def:potdischarging decomposition} -- the segments are defined such that 
    \[
     \pot(w_1w_2u_1\cdots u_{i})-\pot(w_1w_2u_1\cdots u_{i-1})>\sqrt{m}
    \]
    i.e, their potential increase is at least $\sqrt{m}$.

    By \cref{lem:bounded growth potential}, the potential can grow by at most $B$ upon reading a letter, for some constant $B$ (that depends on the finite alphabet $\Gamma'$).
    We then inductively have that for every segment $u_i$ (indeed, every infix of $w_3$) it holds that
    \[
    \pot(w_1w_2u_1\cdots u_{i})-\pot(w_1w_2u_1\cdots u_{i-1})< |u_i|B
    \]
    Combining this with the potential difference above, we have that $|u_i|> \sqrt{m}/B$, and in particular the length of $u_i$ increases with $m$.

    We turn to show that $d(m)$ increases. We initially claim that 
    \[
     \pot(w_1w_2u_1\cdots u_{i})-\pot(w_1w_2u_1\cdots u_{i-1})\le \sqrt{m}+B
    \]
    Indeed write $u_i=v\cdot \sigma$ where $\sigma$ is the last letter, then 
     \[
    \begin{split}
        &\pot(w_1w_2u_1\cdots u_{i-1}v\sigma)-\pot(w_1w_2u_1\cdots u_{i-1}v)=\\
        &\pot(w_1w_2u_1\cdots u_{i-1}v\sigma)-\pot(w_1w_2u_1\cdots u_{i-1}v)+ \pot(w_1w_2u_1\cdots u_{i-1}v) -\pot(w_1w_2u_1\cdots u_{i-1})\le \\
        &\sqrt{m} + B
    \end{split}
    \]
    Where the last inequality follows from \cref{lem:bounded growth potential} and from the fact that $u_i=v\sigma$ is minimal with the gap property in \cref{def:potdischarging decomposition} (i.e., $\pot(w_1w_2u_1\cdots u_{i-1}v)-\pot(w_1w_2u_1\cdot u_{i-1})\le \sqrt{m}$ otherwise $v$ would have been selected instead of $u_i$).

    We notice that for segment $u_{d(m)+1}$ we do not provide bounds, as it is essentially a ``remainder''. 
    Trivially, we have $\pot(w_1w_2u_1\cdots u_{d(m)+1})-\pot(w_1w_2u_1\cdots u_{d(m)})\le \sqrt{m}$ (otherwise we would have another segment).

    By the definition of $\potdiswords$ we now have that 
    \[
    \sum_{i=1}^{d(m)+1}(\pot(w_1w_2u_1\cdots u_{i})-\pot(w_1w_2u_1\cdots u_{i-1}))=\pot(w_1w_2w_3)-\pot(w_1w_2w_3)\ge m
    \]
    and on the other hand by the above we have
    \[
    \sum_{i=1}^{d(m)+1}(\pot(w_1w_2u_1\cdots u_{i})-\pot(w_1w_2u_1\cdots u_{i-1}))\le d(m)(\sqrt{m}+B)+\sqrt{m}
    \]
    We therefore have
    $d(m)\ge \frac{m-\sqrt{m}}{\sqrt{m}+B}$
    so $\lim_{m\to \infty}d(m)=\infty$, as required.
\end{proof}

\cref{prop:potdischarging words has fair decomposition} allows us to weaken the conditions in \cref{lem: decomposed inc pot with cover sparse no ghosts implies type 1 witness} to require only that the sequence is $\pot$-discharging, since this now implies an increasing-potential fair decomposition. Thus, we have the following.
\begin{corollary}
    \label{cor: potdischarging cover sparse no ghosts implies type 1}
    If there exists $G,\ell\in \bbN$ and a $\pot$-discharging word sequence $\potdiswords$ whose discharging decomposition is a ghost-free decomposition $G$-cover and is $\ell$-sparse with gap $4G + (|S|+1)4\bigM W_{\max} |w_3|$,
    \acctodo{ACCOUNTING}
    then $\augA_\infty^\infty$ has a type-1 witness.
\end{corollary}

Our main result of the section is again an inductive lemma analogous to \cref{lem:discharge with large gaps implies unbounded potential or D dip}, guaranteeing either a $D$-dip or a type-1 witness, given a $\pot$-discharging words sequence.
\begin{lemma}[\keyicon \lightbulbicon Potential Discharge and Gap to Witness or $D$-Dip]
    \label{lem:potdischarge with large gaps implies type 1 or D dip}
    Let $\ell\in \bbN$ and consider a $\pot$-discharging words sequence $\potdiswords$ that is $\ell$-sparse with gap $(|S|+1)4\bigM W_{\max}|w_3|$. 
    Then $\cA_\infty^\infty$ has a type-1 witness, or $\potdiswords$ has the infix $D$-dip property for some $D\in \bbN$.
\end{lemma}
\begin{proof}
    We prove the lemma by reverse induction on $\ell$. Again, the proof is nearly identical to that of \cref{lem:discharge with large gaps implies unbounded potential or D dip}, but has some tweaks where potential differs from charge (in particular in Case 2, which requires flattening).
    
    \paragraph*{Base case: $\ell=|S|$} Recall that the maximal number of configuration-independent runs is $|S|$ (see \cref{def:configuration independent runs}). 
    In this case, for every $m\in \bbN$ with $\potdiswords(m)=w_1w_2w_3w_4$ we have that in every configuration in $w_3$, each state belongs to one of the independent runs. 
    In particular, $\potdiswords$ is a configuration $G$-cover with $G=0$, and is also $\ell$-sparse by assumption, with gap $4G+(|S|+1)4\bigM W_{\max}|w_3|$ (since $G=0$).

    Moreover, since $\ell=|S|$, then in particular every state is reachable along any prefix of $w_3$, and therefore there are no ghost runs at all along $w_3$. 
    Thus, the conditions of \cref{cor: potdischarging cover sparse no ghosts implies type 1} are met, 
    so $\augA_\infty^\infty$ has a type-1 witness.
    \paragraph*{Inductive case: $\ell<|S|$}
    For the inductive case, we start by decomposing $\potdiswords$ to its discharging decomposition (\cref{def:potdischarging decomposition}) $w_3=u_1\cdots u_{d(m)}u_{d(m)+1}$, and by \cref{prop:potdischarging words has fair decomposition} this decomposition is an increasing-potential fair decomposition.
    We now split to three subcases. 
    
    \subparagraph*{Case 1: a diverging run on a segment}
    In this case we essentially assume that the independent runs are not a decomposition $G$-cover for any $G$, and use this to construct a new independent run from states that stay far away from all independent runs. We then apply the induction hypothesis. 
    
    This starts similar to the analogous case in \cref{lem:potleveled with large gaps implies type 1 or D dip}, but differs in that the potential constraints are not necessarily met, and requires further case analysis.

    Formally, assume that for every $G\in \bbN$ there exists $m>2G$ with $\potdiswords(m)=w_1w_2w_3w_4$ and $w_3=u_1\cdots u_{d(m)}u_{d(m)+1}$ such that there exist $1\le i\le d(m)$ where for 
    $\vec{c_i}=\xconf(s_0,w_1w_2u_1\cdots u_i)$ there is a state $q\in \supp(\vec{c_i})$ for which $\vec{c_i}(q)$ is not withing gap $G$ of any independent run.
    That is, we assume that $\potdiswords(m)$ does not have a decomposition $G$-cover (\cref{def: G cover decomposed words}).
    
    By the assumption, for every $n\in \bbN$ there exists $m\in \bbN$ such that for $1\le i\le d(m)$ and $q$ above it holds that $q$ is not within gap $2W_{\max}n+(|S|+1)4\bigM W_{\max}n$ of any independent run, and $|u_i|>n$. 
    In particular, there is a seamless runs $\rho:s_0\runsto{w_1w_2u_1\cdots u_i}q$ such that for every split $u_i=v_0v$ with $|v|\le n$ we have that $\weight(\rho(w_1w_2u_1\cdots u_{i-1}v_0))$ is not within gap $(|S|+1)4\bigM W_{\max}n$ of any independent run. 

    With these notations, we divide into cases, depending on the behavior of $\pot$ on $u_i$.

    \subparagraph*{Case 1a: increasing potential}
    In this case, we assume that for every $n\in \bbN$ there exists $m>n$ and  $\potdiswords(m)=w_1w_2w_3w_4$ with $u_i$, $q$ and $\rho$ as above such that there is a split $u_i=v_0v$ with
    $\pot(w_1w_2u_1\cdots u_{i-1}v_0v)-\pot(w_1w_2u_1\cdots u_{i-1}v_0)>n$. 
    That is, there is a suffix $v$ of $w_3$ on which the potential increases significantly. Note the similarity to the condition of a $\pot$-discharging words sequence (\cref{def:potdischarging word sequence}). 

    We therefore construct a $\pot$-discharging words sequence $\potdiswords'$ by defining $\potdiswords'(n)=w'_1w'_2w'_3w'_4$ with $w'_1=w_1, w'_2=w_2u_1\cdots u_{i-1}v_0, w'_3=v$ and $w'_4=u_{i+1}\cdots u_{d(m)+1}w_4$.

    Since $\potdiswords$ is already a $\pot$-discharging words sequence, and $u_i$ is a segment in the decomposition, then by the minimality criterion of \cref{def:potdischarging decomposition}, we have that the charge along $u_i$ does not exceed $\pot(w_1w_2u_1\cdots u_i)$. Therefore, combined with the condition on $v$ above, we have that $\potdiswords'$ satisfies both the conditions of \cref{def:potdischarging word sequence}, so it is a discharging words sequence.

    Moreover, as we observed above, the run $\rho$ on the suffix $v$ can be added as an independent run with the gap required in the induction assumption. Note that the remaining runs certainly have large enough gaps, since $m>n$. We therefore satisfy the induction assumption, so by the induction hypothesis, either $\augA_\infty^\infty$ has a type-1 witness (and we are done), or $\potdiswords'$ has the Infix $D$-dip property for some $D\in \bbN$. 

    By an identical argument to \cref{lem:potleveled with large gaps implies type 1 or D dip}, we can lift this $D$-dip property to $\potdiswords$ itself, and we are done.

    \subparagraph*{Case 1b: non-increasing potential}
    Complementing Case 1a, we now assume that there exists $\kappa$ such that for every $m>\kappa$ and $\potdiswords(m)=w_1w_2w_3w_4$ with $u_i$, $q$ and $\rho$ as above, every split $u_i=v_0v$ satisfies
    $\pot(w_1w_2u_1\cdots u_{i-1}v_0v)-\pot(w_1w_2u_1\cdots u_{i-1}v_0)\le \kappa$. 

    Our goal in this case is to obtain from these segments a $\pot$-leveled words sequence. Indeed, the condition above is almost the condition for $\pot$-leveled words, but is missing the absolute value. 
    However, since $u_i$ is a \emph{minimal} infix of $w_3$ upon which the potential increases enough, as per \cref{def:potdischarging decomposition}, it follows that $\pot(w_1w_2u_1\cdots u_{i-1}v_0v)-\pot(w_1w_2u_1\cdots u_{i-1}v_0)>0$, so we can in fact add the absolute value and assume 
    \[|\pot(w_1w_2u_1\cdots u_{i-1}v_0v)-\pot(w_1w_2u_1\cdots u_{i-1}v_0)|\le \kappa\] 
    
    We construct a sequence $\potlevwords$ as follows: for every $n\in \bbN$ take $m$ for which $|u_i|>n$ (recall that the segments increase in length by \cref{prop:potdischarging words has fair decomposition}). Define $\potlevwords(n)=w'_1w'_2w'_3w'_4$ with $w'_1=w_1,w'_2=w_2u_1\cdots u_{i-1}$, $w'_3=u_i$ and $w'_4=u_{i+1}\cdots u_{d(m)+1}w_4$. 
    By the above, we have that $\potlevwords$ is indeed a leveled words sequence. Moreover, it inherits from $\potdiswords$ the properties of being $\ell$-sparse and having gap $(|S|+1)4\bigM W_{\max}|w_3|$. 
    We can therefore invoke \cref{lem:potleveled with large gaps implies type 1 or D dip}, and we obtain that $\augA_\infty^\infty$ has a type-1 witness, or $\potlevwords$ has the Infix $D$-dip property for some $D\in \bbN$. As with the case above, an identical argument to \cref{lem:potleveled with large gaps implies type 1 or D dip} shows that this lifts to $\potdiswords$, so we are done.

    \subparagraph*{Case 2: long ghost runs}
    The next case is not mutually exclusive to Case 1, so we assume Case 1 does not hold (even though we do not use this fact).
    We assume that there are infinitely many $m\in \bbN$ with $\potdiswords(m)=w_1w_2w_3w_4$ and decomposition $w_3=u_1\cdots u_{d(m)+1}$ such that there exists a ghost run over an entire segment $u_i$. More precisely, there exists $p\in \ghostTrans(s_0,w_1w_2u_1\cdots u_{i-1})$ and a ghost run $\rho:p\runsto{u_i}S$. 

    Our approach in this case is similar to that taken in \cref{lem:potleveled with large gaps implies type 1 or D dip} -- we flatten the prefix up to $u_i$, thus making the ghost run a real run that satisfies the gap constraints. Thus, we can apply induction.

    Formally, we define a new sequence $\potdiswords'$ as follows. For every $m'\in \bbN$ let $m>m'$ be such that $\potdiswords(m)=w_1w_2w_3w_4$ has a ghost run on $u_i$, and we require
    \begin{equation}
    \label{eq:potdischarging induction condition on ui}
    \pot(w_1w_2u_1\cdots u_{i})- \pot(w_1w_2u_1\cdots u_{i-1})>m'
    \end{equation}
    Note that we can require this by since $u_i$ is a segment in a $\pot$-discharging decomposition, as per \cref{def:potdischarging decomposition}.
    Define $\potdiswords'(m')=w'_1w'_2w'_3w'_4$ with
    $w'_1=\flatten(w_1w_2u_1\cdots u_{i-1} \wr (|S|+1)4\bigM W_{\max}|w_1w_2w_3w_3|)$, $w'_2=\epsilon$, $w'_3=u_i$ and $w'_4=u_i\cdots u_{d(m)+1}w_4$.

    Since our alphabet does not contain rebase letters, by \cref{def:cactus rebase flattening} the resulting prefix $x'_1$ is over $\Gamma_0^0$. 
    Denote $\vec{c}=\xconf(\vec{c_{init}},w_1w_2u_1\cdots u_{i-1})$ and $\vec{c'}=\xconf(\vec{c_{init}},w'_1w'_2)$, then by 
    \cref{lem:flattening configuration characterization}
    we have that $\vec{c}(q)=\vec{c'}(q)$ for every $q\in \booltrans(s_0,w_1w_2)$, and $\vec{c'}(p)$ is much larger than any finite entry in $\vec{c}$ (according to the flattening constant above). 
    Specifically, the entire run $\rho$ from $p$ is above all the existing $\ell$ independent runs with gap at least $(|S|+1)4\bigM W_{\max}|w_3|$.
    It follows that $\rho$ can be added as an independent run. Thus, $\potdiswords'$ has $\ell+1$ independent runs. 

    Now, however (and unlike \cref{lem:discharge with large gaps implies unbounded potential or D dip}), it is not necessarily the case that the potential is maintained after flattening. We therefore resort again to \cref{lem:unfolding maintains potential} and obtain that either $\augA_\infty^\infty$ has a type-0 witness (and in particular type-1), or indeed the potential is unchanged by the flattening.
    In the latter case, we have that $\potdiswords'$ remains a discharging words sequence. In particular, by \cref{eq:potdischarging induction condition on ui} $u_i$ satisfies the discharge constraints of \cref{def:potdischarging word sequence} (to be more precise, it satisfies the second requirement by the minimality constraint of \cref{def:potdischarging decomposition}).
    
    We can therefore apply the induction hypothesis on $\potdiswords'$. By the hypothesis, if $\augA_\infty^\infty$ has a type-1 witness, then we are done. Otherwise, we get that $\potdiswords'$ has the $D$-dip property, and we again repeat the argument in \cref{lem:potleveled with large gaps implies type 1 or D dip} to lift this to $\potdiswords$ (note that here the argument is the involved case therein, since we use flattening) and we are done.

    \subparagraph*{Case 3: Ghost free $G$-cover}
    Our final case assumes both Case 1 and Case 2 do not hold. Since Case 2 does not hold, there are at most finitely many $m$ for which there are ghost runs over segments in the decomposition. 
    We slightly strengthen this to assume there are no ghost runs over any segments at all. This is possible since we can replace $\potdiswords$ by duplicating $\potdiswords(m)$ for some very large $m$ into $\potdiswords(i)$ for all $i\le m$. Note that this does not interfere with the Infix $D$-dip property, which anyway refers to infinitely many infixes and ignores the replacement of finitely many elements.
    We therefore assume that the $\pot$-discharging decomposition of $\potdiswords$ is Ghost Free, as per \cref{def:ghost free decomposition}.

    Next, since Case 1 does not hold, then there exists $G\in \bbN$ such that for every $m>2G$ with $\potdiswords(m)=w_1w_2w_3w_4$ and $w_3=u_1\cdots u_{d(m)+1}$, for every $1\le i\le d(m)$ and $\vec{c_i}=\xconf(\vec{c_{\init}},w_1w_2u_1\cdots u_i)$, every state $q\in \supp(\vec{c_i})$ is within gap $G$ from some independent run. 
    As above, we can assume this holds in fact for every $m$, by possibly replacing finitely many elements of $\potdiswords$.
    Thus, the $\pot$-discharging decomposition of $\potdiswords$ is a decomposition $G$-cover (as per \cref{def: G cover decomposed words}).

    We can now almost apply \cref{cor: potdischarging cover sparse no ghosts implies type 1}, except the gap required there is slightly bigger than the gap in $\potdiswords$. Identically to \cref{lem:potdischarge with large gaps implies type 1 or D dip}, this is an artificial problem, and can be resolved by shifting the elements of $\potdiswords$, i.e., defining $\potdiswords'(m)=\potdiswords(m+G)$, which increases the bound by the necessary $4G$. 

    We can therefore invoke \cref{cor: potdischarging cover sparse no ghosts implies type 1}, and conclude that  $\augA_\infty^\infty$ has a type-1 witness, and we are done.
\end{proof}

As a potential-counterpart to \cref{cor:discharge implies unbounded potential or D dip} for discharge words, we have the main tool for $\potdiswords$ (by applying \cref{lem:potdischarge with large gaps implies type 1 or D dip} with $\ell=1$, for which the gap is unbounded).
\begin{corollary}[\keyicon Potential Discharge to Witness or $D$-Dip]
    \label{cor:potdischarge implies type 1 or D dip}
    Consider a $\potdiswords$ sequence, then either $\augA_\infty^\infty$ has a type-1 witness, or $\potdiswords$ has the Infix $D$-dip property for some $D\in \bbN$. 
\end{corollary}

\section{A Nondeterminizable WFA Has a Witness}
\label{sec:final nondet implies witness}
We are finally ready to prove our main result, namely that if $\cA$ is a nondeterminizable WFA, then $\augA_\infty^\infty$ has a witness. We then use this to conclude the proof that determinizability of WFA is decidable.
\begin{theorem}
    \label{thm:nondet to witness}
    Consider a WFA $\cA$. If $\cA$ is nondeterminizable, then $\augA_\infty^\infty$ has a type-1 witness.
\end{theorem}
We prove this theorem in the remainder of this section.
Technically, this result combines all the tools developed thus far. For an overview, see \cref{sec:abs:finale}.

Henceforth, fix a WFA $\cA$, and assume $\cA$ is nondeterminizable.

\subsection{The Potential is Unbounded (Or There is a Witness)}
\label{sec:final nondet implies unbounded potential}
The first step of our proof is to show that if $\cA$ is nondeterminizable, then the potential is unbounded. More precisely, we prove the following.
\begin{lemma}[\keyicon \lightbulbicon Nondeterminizability implies Unbounded Potential or Witness]
    \label{lem:A is nondet then potential is unbounded}
    Consider a WFA $\cA$. If $\cA$ is nondeterminizable, then either $\sup\{\pot(w)\mid w\in (\Gamma_0^0)^*\}=\infty$, or $\cA$ has a type-0 witness.
\end{lemma}
We prove the lemma in the remainder of the section. 
since $\cA$ is nondeterminizable, then by \cref{lem:A det iff augA det} we have that $\augA$ is also nondeterminizable.
Our starting point is to use the gap-criterion of \cref{thm:det iff bounded gap}. Specifically, for every $n\in \bbN$ there exists a $n$-gap witness $x_n,y_n\in (\Gamma_0^0)^*$ (\cref{def: B gap witness}). 
That is, for every $n\in \bbN$ there exist seamless runs $\rho:s_0\runsto{x_ny_n}S$ and $\mu:s_0\runsto{x_n}S$ such that $\minweight(x_ny_n,s_0\to S)=\weight(\rho)$ while $\minweight(x_n,s_0\to S)=\weight(\mu)$ and $\weight(\rho(x_n))-\weight(\mu(x_n))\ge n$.

For every $n\in \bbN$, we now perform a baseline shift (\cref{sec: baseline shift}) so that $\rho$ becomes the baseline (see \cref{fig:final baseline shift}). Since the alphabet is $\Gamma_0^0$, then this is maintained by the baseline shift. Specifically, we consider $\rho'=\baseshift{\rho}{\rho}$, $\mu'=\baseshift{\mu}{\rho}$, $x_n'=\baseshift{x_n}{\rho}$ and $y_n'=\baseshift{y_n}{\rho}$.
By \cref{cor:baseline shift maintains gaps,cor:baseline shift to seamless run}, we have that the gap property above is maintained, and $\rho'$ is the baseline run (and therefore has weight $0$ on all prefixes). 
That is, $\minweight(x'_ny'_n,s_0\to S)=\weight(\rho')=0$ while $\minweight(x'_n,s_0\to S)=\weight(\mu')$ and $\weight(\rho(x'_n))-\weight(\mu(x'_n))=-\weight(\mu(x'_n))\ge n$. 
\begin{figure}[ht]
    \centering
    \includegraphics[width=0.8\linewidth]{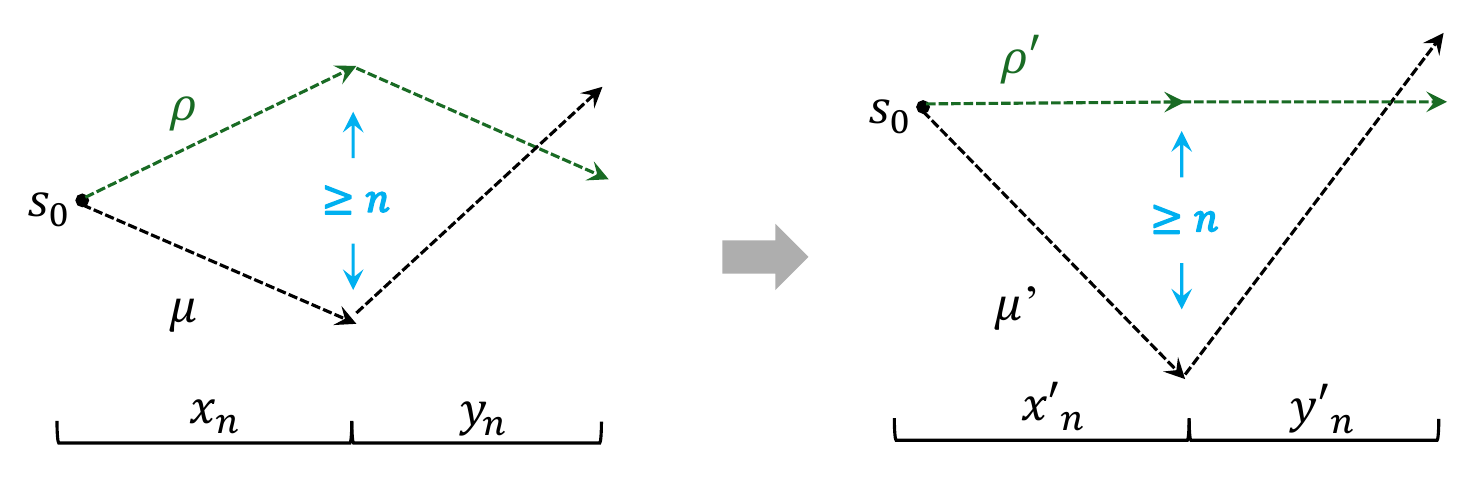}
    \caption{Baseline shift on $\rho$. Note that we do not assume $\mu'$ actually continues after $x'_n$, but this is the more interesting case.}
    \label{fig:final baseline shift}
\end{figure}
By the gap criterion, and since $\mu'$ is a minimal run on $x'_n$, we now see that $\charge(x'_n)=-\weight(\mu(x'_n))\ge n$. However, since $\rho'$ is minimal on $x'_ny'_n$, we also have $\charge(x'_ny'_n)=0$. That is, there is a significant discharge upon reading $y_n$.
Intuitively, we would like to obtain from this sequence a discharging words sequence, as these almost guarantee unbounded potential (\cref{cor:discharge implies unbounded potential or D dip}).
The major barrier we face is that instead of unbounded potential, we might also get a $D$-dip, at which point we would be stuck. 

To avoid this problem, we need to work much harder, and use our framework of cactus budding, as follows. 
For every $n\in \bbN$ we replace $y'_n$ by a word $z'_n\in (\Gamma_0^\infty)^*$ (i.e., containing cactus letters of any rank, but not rebase or jump letters), such that $\charge(x'_nz'_n)=0$ and
$\cost(z'_n)$ \emph{is minimal} with this property.
Note that since $y'_n$ satisfies the charge constraint, then there indeed exist such words $z'_n$ for every $n\in \bbN$ (possibly $y'_n$ itself, or otherwise some more elaborate words over $\Gamma_0^\infty$). 
Intuitively, we allow ``folding'' infixes of $y'_n$ into cacti, as long as this reduces the overall cost of the obtained word. 

Looking a few steps ahead, the idea is that $z'_n$ cannot have $D$-dip infixes, as those would imply the existence of separated increasing infixes, which in turn imply that we can perform cactus budding inside $z'_n$, thus reducing the cost. 
Unfortunately, since we now work over $(\Gamma_0^\infty)^*$, there is no guarantee that the $z'_n$ are constructed from a finite alphabet, which is a requirement for discharging sequences. So we need to first reduce the setting to a finite alphabet. 

To this end, we define a function that selects sub-cacti of a certain depth from the root of a cactus, as follows.
\begin{definition}
\label{def:deep sub cacti selection}
For a word $w\in (\Gamma_\infty^0)^*$ and $k\in \bbN$ we define its \emph{$k$-deep sub-cacti} $\subcact_k(w)\subseteq \Gamma_{\infty}^0$ inductively (both on $k$ and on the structure of $w$) as follows.
    \begin{itemize}
        \item For $k=0$, if $w=\sigma \in \Gamma_\infty^0$ is a letter (either cactus or not) then $\subcact_0(\sigma)=\sigma$.
        \item For $k>0$, if  $w=\sigma \in \Gamma_0^0$ is a (non-cactus) letter, then $\subcact_k(\sigma)=\emptyset$.
        \item For $k>0$, if $w=\alpha_{S',x}\in \Gamma_\infty^0$ is a cactus letter, then $\subcact_k(\alpha_{S',x})=\subcact_{k-1}(x)$.       
        \item If $w=\sigma_1\cdots \sigma_n\in (\Gamma_\infty^0)^*$ then  
        $\subcact_k(\sigma_1\cdots \sigma_n)=\bigcup_{i=1}^n\subcact_k(\sigma_i)$.
    \end{itemize}
\end{definition}
Intuitively, we ``dig'' into the cactus letters for exactly $k$ times, and select the set of letters that appear at \emph{exactly} that depth (cactus or not). If we cannot expand a certain branch to depth $k$, we ignore it (output $\emptyset$ for that branch).

\begin{example}
\label{xmp:sub cacti selection}
To illustrate this selection, recall \cref{xmp:cactus letters}, then in the notations therein we have:
\[
\begin{split}
&\subcact_0(ab\alpha_3 aa)= \{a,b,\alpha_3\}\\
&\subcact_1(ab\alpha_3 aa)= \subcact_0(\alpha_3)=\{a,b,\alpha_4,\alpha_2\}\\
&\subcact_2(ab\alpha_3 aa)= \subcact_1(\alpha_3)=\subcact_0(\alpha_4)\cup \subcact_0(\alpha_2)=\{b\}\cup \{a,b,\alpha_1\}=\{a,b,\alpha_1\}\\
&\subcact_3(ab\alpha_3 aa)= \subcact_2(\alpha_3)=\subcact_1(\alpha_4)\cup \subcact_1(\alpha_2)=\emptyset \cup \subcact_0(\alpha_1)=\{a\}
\end{split}
\]
\end{example}

Now, since all the $z'_n$ have minimal cost, then by \cref{cor:minimal cost cactus chain has bounded depth} we have that every cactus letter appearing in $z'_n$ has depth at most $|S|$. Indeed, otherwise there is some cactus letter $\alpha_{S',w}$ appearing in $z'_n$ that is the root of a depth $|S|+1$ cactus chain, but then we can replace it with a less-deep cactus chain with the same transition weights, and reduce the overall cost of $z'_n$.
In particular, this means that $\subcact_{|S|}(z'_n)\subseteq \Gamma_0^0$, as it does not contain cactus letters.
Thus, $|\bigcup_{n=1}^\infty\subcact_{|S|}(z'_n)|<\infty$.
We now take $k\le |S|$ to be the minimal such that $|\bigcup_{n=1}^\infty\subcact_{k}(z'_n)|<\infty$. 
That is, we dig deep enough so that the alphabet of the inner cacti (in all the $z'_n$) becomes finite.

Before proceeding, we wish to claim that the $z'_n$ have unbounded length, as $n$ increases. This, however, might not be the case (e.g., it could be that all $z'_n$ are a single letter, that does not allow the minimal run to continue).
However, in this case we claim that $\charge$ does not have bounded decrease, and therefore we are under the conditions of \cref{lem:charge no bounded decrease then potential unbounded}, which immediately implies that the potential is unbounded, concluding the proof of \cref{lem:A is nondet then potential is unbounded}.

More precisely, assume that $|z'_n|< C$ for some fixed $C\in \bbN$, and assume by way of contradiction that $\charge$ has the $B$ bounded-decrease property of \cref{def:charge bounded decrease}. Thus, we have $\charge(x'_n)-\charge(x'_nz'_n)<B|z'_n|\le BC$ for all $n\in \bbN$, but this is in contradiction to the fact that $\charge(x'_n)-\charge(x'_nz'_n)\ge n$ by the definition of $x'_n$ and $z'_n$.
We therefore conclude that the sequence of $z'_n$ is unbounded in length.
We now split to two cases:

\paragraph*{First case: $k=0$.}
In this case, all the $z'_n$ are already constructed from some finite alphabet $\Gamma'\subseteq \Gamma_\infty^0$. This makes things easy for us. We define the function $\diswords$ by choosing for every $m\in \bbN$ a large enough $n\in \bbN$ such that $|z'_n|>n$ (as we show above, the $z'_n$ are unbounded in length), and defining $\diswords(m)=w_1w_2w_3w_4$ with $w_1=x'_n$, $w_2=\epsilon$, $w_3=z'_n$ and $w_4=\epsilon$.

We now apply our main result for discharging sequences, namely \cref{cor:discharge implies unbounded potential or D dip}, and get that either $\sup\{\pot(w)\mid w\in (\Gamma_0^0)^*\}=\infty$, concluding the proof of \cref{lem:A is nondet then potential is unbounded}, or $\diswords$ has the Infix $D$-dip property for some $D\in \bbN$. In the latter case, we proceed to either reach a contradiction, or obtain a witness, as follows.

There is no need to remember what the $D$-dip property is at this point, as we immediately apply 
\cref{prop:infix D dip implied D dip} and then \cref{cor:dip implies increasing infix} to conclude that $\diswords$ has a separated increasing infix. More precisely, there exists $m\in \bbN$ with $\diswords(m)=w_1w_2w_3w_4$ and a decomposition $w_3=u'xyv'$ such that $uxyv$ is a separated increasing infix from $S'=\ghostTrans(s_0,w_1)$ for $u=w_2u'$ and $v=v'w_4$. 
By the construction of $\diswords$, we have in fact that $w_1=x'_n$ and $w_3=z'_n$, and $u'=u,v'=v$.
There is also no need to remember what an increasing infix is, only that we can now apply cactus budding, as per \cref{lem:increasing infix budding}. 
The conditions of \cref{lem:increasing infix budding} then tell us that either $\augA_\infty^\infty$ has a type-0 witness, in which case we conclude the proof of \cref{thm:nondet to witness}, or we can proceed to replace $z'_n$ with the word $z''_n=u\alpha_{B,x}v$ where $B=\ghostTrans(S',u)$, and we have that $\cost(z'_n)>\cost(z''_n)$, and that for every $r\in S',t\in S$ it holds that $\minweight(z'_n,r\to t)\le \minweight(z''_n,r\to t)$. 

Since $\minweight(x'_nz'_n,s_0\to S)=0$, and since $\alpha_{B,x}$ maintains a seamless baseline run, then by the inequality above we have that 
\[
\begin{split}
&0\ge \minweight(x'_nz''_n,s_0\to S)=\minweight(x'_nz''_n,s_0\runsto{x'_n} S'\runsto{z''_n}S)=\\
&\minweight(x'_n,s_0\to S')+\minweight(z''_n,S'\to S)\ge
\minweight(x'_n,s_0\to S')+\minweight(z'_n,S'\to S)=\\
&\minweight(x'_nz''_n,s_0\runsto{x'_n} S'\runsto{z'_n}S)=\minweight(x'_nz'_n,s_0\to S) =0
\end{split}
\]
so $\minweight(x'_nz''_n,s_0\to S)=0$ and in particular $\charge(x'_nz''_n)=0$. This, however, is a contradiction to the choice of $z'_n$ as minimal cost with this property (since $\cost(z'_n)>\cost(z''_n)$). 
We conclude that $\diswords$ cannot have the Infix $D$-dip property, so the potential is unbounded in this case.

\paragraph{Second case: $k>0$.}
Recall that $k$ is the minimal number such that 
$|\bigcup_{n=1}^\infty\subcact_{k}(z'_n)|<\infty$. Since $k>0$, then by its minimality, we have that $|\bigcup_{n=1}^\infty\subcact_{k-1}(z'_n)|=\infty$. 
In particular, since $\Gamma_0^0$ is finite, we have that $\Upsilon=\left(\bigcup_{n=1}^\infty\subcact_{k-1}(z'_n)\right)\setminus \Gamma_0^0$ is an infinite set of cactus letters. Consider the set of ``inner words'' of these cactus letters, i.e., the set $\Xi=\{w\mid \exists S'\subseteq S,\ \alpha_{S',w}\in \Upsilon\}$. By \cref{def:deep sub cacti selection} we have that $\bigcup_{w\in \Xi}\subcact_0(w)\subseteq \bigcup_{n=1}^\infty\subcact_{k}(z'_n)$, and in particular, this is a finite set, meaning that the alphabet of the words in $\Xi$ is finite. Since $\Xi$ itself is an infinite set (since $\Upsilon$ is infinite), it follows that there are arbitrarily long words in $\Xi$. 
Moreover, since there are only finitely many options for sets $S'$, it follows that there exists a fixed $S'\subseteq S$ such that for every $n\in \bbN$ there is a word $\alpha_{S',w_n}\in \Upsilon$ with $|w_n|>n$. 

We claim that $S'$ is also reachable by some word $v_{\init}$, i.e., $\booltrans(s_0,v_{\init})=S'$. 
Intuitively, this is because $\alpha_{S',w_n}$ is read along some $z'_m$, so all we need is the prefix up to it. However, this only guarantees reaching a subset of $S'$. We fix this by flattening.

Formally, each $\alpha_{S',w_n}$ is a letter appearing in some $z'_m$, and there are finite-weight runs on $z'_m$. Thus, the prefix $v''$ up to the occurrence of $w_n$ in some $z'_m$ reaches a state $s\in S'$ from which $\alpha_{S',w_n}$ is read. In particular, we have $\ghostTrans(s_0,v'')=S$ (by \cref{def:ghost states}, of ghost states). 
Then, we can define $v_{\init}=\flatten(v''\wr F)$ for $F>2\maxeff{v''}$, and obtain by \cref{lem:flattening configuration characterization} that $\booltrans(s_0,v_{\init})=S'$.

We therefore arrive in the following setting: there is a sequence of words $\{(v_{\init}\alpha_{S',w_n})\}_{n=1}^\infty$ with $v_{\init}\in (\Gamma_0^0)^*$, as well as $|w_n|> n$ and $w_n\in \Gamma'$ for all $n\in \bbN$, where $\Gamma'$ is a finite alphabet (namely $\Gamma'= \bigcup_{n=1}^\infty\subcact_{k}(z'_n)$).

We now unfold each $\alpha_{S',w_n}$ as follows. Let $n\in \bbN$, and set $F>2\maxeff{v_{\init}\alpha_{S',w_n}}$. Define $v_0=\unfold(v_\init,\alpha_{S',w_n},\epsilon\wr F)=v_{\init}w_n^{2\bigM M_0}$.
By \cref{prop:unfolding maintains charge}, we have that $\charge(v_0)=\charge(v_{\init}\alpha_{S',w_n})$.
By \cref{rmk:increasing repetitions in unfolding}, we can arbitrarily increase $M_0$ and the following still holds. It particular, for $M_0+1$ we also have $\charge(v_0w_n^{2\bigM})=\charge(v_{\init}\alpha_{S',w_n})$, and therefore $\charge(v_0)=\charge(v_0w_n^{2\bigM})$.
Consider therefore the sequence of prefixes
\[v_0, v_0w_n,v_0w_n^2,\ldots,v_0w_n^{2\bigM}\]
then there exists $0\le i<2\bigM$ such that $\charge(v_0w_n^i)\ge \charge(v_0w_n^{i+1})$. Indeed, otherwise the charge strictly increases, contradicting the equality of charge between $v_0$ and $v_0w_n^{2\bigM}$.

We again construct an elongating words sequence $\words$ as follows. For every $m\in \bbN$, let $n$ be large enough such that $|w_n|>n$ (recall that we prove above that there is such $n$), and choose $i$ as above, i.e., such that $\charge(v_0w_n^i)\ge \charge(v_0w_n^{i+1})$. 
Define $\words(m)=w'_1w'_2w'_3w'_4$ with $w'_1=v_0w_n^i$, $w'_2=\epsilon$, $w'_3=w_n$ and $w'_4=\epsilon$.

We now analyze two possible scenarios.
\begin{itemize}
    \item If $\words$ is a leveled sequence (\cref{def:leveled words sequence}), then by \cref{cor:leveled implies unbounded potential or D dip}, either the potential is unbounded, concluding the proof of \cref{lem:A is nondet then potential is unbounded}, or $\words$ has the Infix $D$-dip property. In this case, we can reach a similar contradiction as in the $k=0$ case, as follows.

    Since $\words$ has the Infix $D$-dip property, then by \cref{cor:infix D dip implies increasing infix} it has a separated increasing infix. Specifically, there exists $m\in \bbN$ with $\words(m)=v_0w_n^i\cdot \epsilon\cdot w_n\cdot \epsilon$ and a decomposition $w_n=uxyv$ that is a separated increasing infix from $\ghostTrans(s_0,v_0w_n^i)=S'$ (as above, this follows from $w'_2=w'_4=\epsilon$).
    Recall that $w_n$ is a word such that\footnote{$\alpha_{S',w_n}$ actually appears in a sub-cactus of some $z'_{n'}$, and it may be that $n'\neq n$. For clarity we unify them, as neither have any implication on the proof.} $\alpha_{S',w_n}\in \subcact_{k-1}(z'_n)$.

    By expanding \cref{def:deep sub cacti selection}, we can write $z'_n$ as a cactus chain (\cref{def:cactus chain}), i.e., there is a cactus chain
    \[\Theta_1=\alpha_{S'_1,y_1},\ldots,\alpha_{S'_{k-1},y_k}\] 
    such that $z'_n=x_1\alpha_{S'_1,y_1}z_1$ for some words $x_1,y_1,z_1$, and $(S'_{k-1},y_k)=(S',w_n)$. 
    Then, since $w_n=uxyz$ is an increasing infix, then $\Theta_1$ is a \emph{Pre-bud chain} (\cref{def: pre and post bud chain}). 
    We can therefore invoke \cref{lem:post bud chain is superior higher potential}, and obtain that the corresponding Post-bud chain $\Theta_2$ is superior to $\Theta_1$.
    In particular, the first element of $\Theta_2$, denoted $\alpha_{S'_1,y''_1}$ is superior to $\alpha_{S'_1,y_1}$. 

    Then, \cref{prop:charge decreases superior} gives us that $0=\charge(x'_n z'_n)=\charge(x'_n\cdot x_1\alpha_{S'_1,y_1}y_1)\ge \charge(x'_n\cdot x_1 \alpha_{S'_1,y''_1})y_1)\ge 0$ (since the charge is non-negative).
    Thus, $\charge(x'_n\cdot x_1 \alpha_{S'_1,y''_1})y_1)=0$.
    Note, however, that $\alpha_{S'_1,y''_1}$ is obtained by propagating the replacement of $w_n$ with $u\alpha_{B,x}y$ through the chain, and by \cref{lem:increasing infix budding} we have $\cost(w_n)>\cost(u\alpha_{B,x}y)$ (which also propagates through the chain), meaning that $\cost(z'_n)>\cost(x_1\alpha_{S'_1,y''_1}y_1)$. 
    But this is a contradiction to the assumption that $z'_n$ has minimal cost with this property, so we are done.    

    \item The second case is when $\words$ is not a leveled sequence. We claim that in this case we can find a faithful restriction of $\words$ (\cref{def:elongated words sequence}) that is a discharging sequence, and then we can follow the same reasoning as above. Intuitively, not being leveled means that the charge fluctuates ``wildly'', and in particular becomes very high. Since it is non-negative and must end in $0$, then either it drops abruptly, violating the $B$-bounded decrease property, in which case we have unbounded potential, or it decreases slowly in which case we have a discharging sequence.

    Formally, since $\words$ is not leveled, then for every $B\in \bbN$ there exists $m\in \bbN$ with $\words(m)=v_0w_n^i\cdot \epsilon\cdot w_n\cdot \epsilon$ for some $n\in \bbN$, and a prefix $w'_n$ of $w_n$ such that $|\charge(v_0w_n^i)-\charge(v_0w_n^iw'_n)|>B$. Write $w_n=w'_n w''_n$.
    \begin{itemize}
        \item If for every $B$ there exists such $m,n$ for which $\charge(v_0w_n^i)-\charge(v_0w_n^iw'_n)>B$, then we can readily construct a discharging words sequence $\diswords$ by setting $\diswords(m')=v_0w_n^i\cdot \epsilon\cdot w'_n\cdot w''_n$, where we take $n$ large enough so that the discharge is greater than $m'$. Note that we can assume $w'_n$ increase in length, since the charge is assumed to have bounded-decrease, as discussed above.
        \item If for every $B$ there exists such $m,n$ for which $\charge(v_0w_n^iw'_n)-\charge(v_0w_n^i)>B$, recall that by our choice of $i$, we have $\charge(v_0w_n^i)\ge \charge(v_0w_n^{i+1})=\charge(v_0w_n^{i}w'_nw''_n)$, and therefore 
        we have $\charge(v_0w_n^iw'_n)-\charge(v_0w_n^iw'_nw''_n)>B$.
        Since the charge has bounded-decrease, as we already argue before the $k=0$ case, we have that the $w''_n$ suffixes in this case have unbounded length.
        Then we can again construct a discharging\footnote{A nit-picking details is that we need to take $w''_n$ to be the suffix upon which the discharge is maximal, to conform with \cref{def:discharging word sequence}.} words sequence $\diswords$ by setting $\diswords(m')=v_0w_n^iw'_n\cdot \epsilon\cdot w''_n\cdot \epsilon$, where we take $n$ large enough so that the discharge is greater than $m'$.
    \end{itemize}
    Therefore, we have a discharging sequence $\diswords$ such that for every $m\in \bbN$ with $\diswords(m)=w'_1w'_2w'_3w'_4$ it holds that $w'_3$ is an infix (in fact, either a suffix or a prefix) of $w_n$ for some $n\in \bbN$. 
    By our main result for discharging sequences, namely \cref{cor:discharge implies unbounded potential or D dip}, we again have that either the potential is unbounded, concluding the proof of \cref{lem:A is nondet then potential is unbounded}, or $\diswords$ has the Infix $D$-dip property. 
    In the latter case, we proceed identically to the leveled-words case above: we obtain an increasing infix within $w_n$, therefore we have a Pre-bud chain of lower cost, which contradicts the cost-minimality of $z'_n$.    
\end{itemize}

All in all, we conclude that in all cases that did not reach a contradiction, either we have a type-0 witness, or $\sup\{\pot(w)\mid w\in (\Gamma_0^0)^*\}=\infty$. This completes the proof of \cref{lem:A is nondet then potential is unbounded}. 

\subsection{$\augA_\infty^\infty$ Has a Witness}
\label{sec:final pot unbounded then witness}
We are now ready for the second (and more substantial) part of the proof of \cref{thm:nondet to witness}. 
From a high-level perspective, this part is very similar to the first part in \cref{sec:final nondet implies unbounded potential}. That is, we start with a sequence of words over $\Gamma_0^0$ that has some property (previously unbounded gaps, now unbounded potential). We then replace them with their minimal-cost variants over $\Gamma_0^\infty$. 
We proceed to find a finite alphabet deep enough in these words, and from there obtain discharging or leveled sequences. 
While the proof of this part does not seem more complicated, we emphasize that the lemma used in this part all pertain to potential, which uses the more difficult parts of our toolbox (e.g., \cref{sec:potleveled and potdischarging words}, and most of \cref{sec:cactus budding}).

\begin{lemma}[\keyicon \lightbulbicon Unbounded Potential Implies Witness]
    \label{lem:potential is unbounded then type 1}
    Consider a WFA $\cA$. If $\sup\{\pot(w)\mid w\in (\Gamma_0^0)^*\}=\infty$, then  $\cA$ has a type-1 witness.
\end{lemma}
We prove the lemma in the remainder of the section.

Since the potential is unbounded, there is a sequence of words $\{z_n\}_{n=1}^\infty$ such that $\pot(w_z)>n$ for all $n\in \bbN$. 
Intuitively, we would like to obtain from this sequence a potential-discharging words sequence (\cref{def:potdischarging word sequence}), as these almost guarantee a witness (\cref{cor:potdischarge implies type 1 or D dip}).
We again face the barrier that instead of a witness, we might also get a $D$-dip. However, we are now experienced and ready for the challenge.

Adhering to the naming scheme of \cref{sec:final pot unbounded then witness}, for every $n\in \bbN$ we replace $z_n$ by a word $z'_n\in (\Gamma_0^\infty)^*$ (i.e., containing cactus letters of any rank, but not rebase or jump letters), such that $\pot(z'_n)>n$ and $\cost(z'_n)$ \emph{is minimal} with this property.
Note that since $z_n$ satisfies the potential constraint, then there indeed exist such words $z'_n$ for every $n\in \bbN$ (possibly $z_n$ itself, or otherwise some more elaborate words over $\Gamma_0^\infty$). 
Intuitively, we allow ``folding'' infixes of $z_n$ into cacti, as long as this reduces the overall cost of the obtained word. 

Now, since all the $z'_n$ have minimal cost, then by \cref{cor:minimal cost cactus chain has bounded depth} we have that every cactus letter appearing in $z'_n$ has depth at most $|S|$. Indeed, otherwise there is some cactus letter $\alpha_{S',w}$ appearing in $z'_n$ that is the root of a depth $|S|+1$ cactus chain, but then we can replace it with a less-deep cactus chain with the same transition weights, and reduce the overall cost of $z'_n$.
In particular, this means that $\subcact_{|S|}(z'_n)\subseteq \Gamma_0^0$, as it does not contain cactus letters.
Thus, $|\bigcup_{n=1}^\infty\subcact_{|S|}(z'_n)|<\infty$.
We now take $k\le |S|$ to be the minimal such that $|\bigcup_{n=1}^\infty\subcact_{k}(z'_n)|<\infty$. 
That is, we dig deep enough so that the alphabet of the inner cacti (in all the $z'_n$) becomes finite.

Unlike \cref{sec:final nondet implies unbounded potential}, we cannot claim at this point that the lengths of the $z'_n$ are unbounded. 
Indeed, the bounded-growth property of the potential (\cref{lem:bounded growth potential}) applies only over a finite alphabet, which is not the case at this point.

We split to identical cases as in \cref{sec:final nondet implies unbounded potential}.
\paragraph*{First case: $k=0$.}
In this case, all the $z'_n$ are already constructed from some finite alphabet $\Gamma'\subseteq \Gamma_\infty^0$. In particular, we can now apply the bounded-growth property of \cref{lem:bounded growth potential}, and get that the $z'_n$ are unbounded in length as $n$ increases.

We define the function $\potdiswords$ (\cref{def:potdischarging word sequence}) by choosing for every $m\in \bbN$ a large enough $n\in \bbN$ such that $|z'_n|>n$ and defining $\diswords(m)=w_1w_2w_3w_4$ with $w_1=\epsilon$, $w_2=\epsilon$, $w_3=z'_n$ and $w_4=\epsilon$. Note that \cref{def:potdischarging word sequence} also requires that the potential does not exceed $\pot(z'_n)$ on any prefix of $w_3$. 
This requirement is satisfied because $z'_n$ is assumed to have minimal cost, and if we could drop a suffix of it while increasing the potential, this would decrease the cost.

We now apply our main result for potential-discharging sequences, namely 
\cref{cor:potdischarge implies type 1 or D dip} and get that either $\augA_\infty^\infty$ has a type-1 witness, concluding the proof of \cref{lem:potential is unbounded then type 1} or $\potdiswords$ has the Infix-$D$-dip property.

In the latter case, we proceed to either reach a contradiction, or obtain a witness, as follows.
We apply \cref{prop:infix D dip implied D dip} and then \cref{cor:dip implies increasing infix} to conclude that $\potdiswords$ has a separated increasing infix. More precisely, there exists $m\in \bbN$ with $\diswords(m)=w_1w_2w_3w_4$ and a decomposition $w_3=uxyv$ that is a separated increasing infix from $S'=\ghostTrans(s_0,w_1)$ (again, the format of the infix is due to $w_2=w_4=\epsilon$).

We can now apply cactus budding, as per \cref{lem:increasing infix budding}. 
The conditions of \cref{lem:increasing infix budding} then tell us that either $\augA_\infty^\infty$ has a type-0 witness, in which case we conclude the proof of \cref{thm:nondet to witness}, or we can proceed to replace $z'_n$ with the word $z''_n=u\alpha_{B,x}v$ where $B=\ghostTrans(S',u)$, and we have that $\cost(z'_n)>\cost(z''_n)$, as well as $\pot(z'_n)\le \pot(z''_n)$. This yields a contradiction to the minimal cost of $z'_n$, and we are done.

We remark that the entire framework of increasing infixes (\cref{sec:separated increasing infix}) is developed to obtain these two inequalities.

\paragraph{Second case: $k>0$.}
Recall that $k$ is the minimal number such that 
$|\bigcup_{n=1}^\infty\subcact_{k}(z'_n)|<\infty$. Since $k>0$, then by its minimality, we have that $|\bigcup_{n=1}^\infty\subcact_{k-1}(z'_n)|=\infty$. 
In particular, since $\Gamma_0^0$ is finite, we have that $\Upsilon=\left(\bigcup_{n=1}^\infty\subcact_{k-1}(z'_n)\right)\setminus \Gamma_0^0$ is an infinite set of cactus letters. Consider the set of ``inner words'' of these cactus letters, i.e., the set $\Xi=\{w\mid \exists S'\subseteq S,\ \alpha_{S',w}\in \Upsilon\}$. By \cref{def:deep sub cacti selection} we have that $\bigcup_{w\in \Xi}\subcact_0(w)\subseteq \bigcup_{n=1}^\infty\subcact_{k}(z'_n)$, and in particular, this is a finite set, meaning that the alphabet of the words in $\Xi$ is finite. Since $\Xi$ itself is an infinite set (since $\Upsilon$ is infinite), it follows that there are arbitrarily long words in $\Xi$. 
Moreover, since there are only finitely many options for sets $S'$, it follows that there exists a fixed $S'\subseteq S$ such that for every $n\in \bbN$ there is a word $\alpha_{S',w_n}\in \Upsilon$ with $|w_n|>n$. 

We claim that $S'$ is also reachable by some word $v_{\init}$, i.e., $\booltrans(s_0,v_{\init})=S'$. 
Intuitively, this is because $\alpha_{S',w_n}$ is read along some $z'_m$, so all we need is the prefix up to it. However, this only guarantees reaching a subset of $S'$. We fix this by flattening. 

Formally, each $\alpha_{S',w_n}$ is a letter appearing in some $z'_m$, and there are finite-weight runs on $z'_m$. Thus, the prefix $v''$ up to the occurrence of $w_n$ in some $z'_m$ reaches a state $s\in S'$ from which $\alpha_{S',w_n}$ is read. In particular, we have $\ghostTrans(s_0,v'')=S$ (by \cref{def:ghost states}, of ghost states). 
Then, we can define $v_{\init}=\flatten(v''\wr F)$ for $F>2\maxeff{v''}$, and obtain by \cref{lem:flattening configuration characterization} that $\booltrans(s_0,v_{\init})=S'$.

We therefore arrive in the following setting: there is a sequence of words $\{(v_{\init}\alpha_{S',w_n})\}_{n=1}^\infty$ with $v_{\init}\in (\Gamma_0^0)^*$, as well as $|w_n|> n$ and $w_n\in \Gamma'$ for all $n\in \bbN$, where $\Gamma'$ is a finite alphabet (namely $\Gamma'= \bigcup_{n=1}^\infty\subcact_{k}(z'_n)$).

We now unfold each $\alpha_{S',w_n}$ (but we are wary of the differences between charge and potential, and recall that potential does not behave quite as nicely with unfolding). 
Let $n\in \bbN$, and set $F>2\maxeff{v_{\init}\alpha_{S',w_n}}$. Define $v_0=\unfold(v_\init,\alpha_{S',w_n},\epsilon\wr F)=v_{\init}w_n^{2\bigM M_0}$.
By \cref{lem:unfolding maintains potential}, we have that either $\augA_\infty^\infty$ has a type-0 witness, concluding the proof of \cref{lem:potential is unbounded then type 1}, or $\pot(v_0)=\pot(v_{\init}\alpha_{S',w_n})$.
By \cref{rmk:increasing repetitions in unfolding}, we can arbitrarily increase $M_0$ and the following still holds. It particular, for $M_0+1$ we also have $\pot(v_0w_n^{2\bigM})=\pot(v_{\init}\alpha_{S',w_n})$, and therefore $\pot(v_0)=\pot(v_0w_n^{2\bigM})$.
Consider therefore the sequence of prefixes
\[v_0, v_0w_n,v_0w_n^2,\ldots,v_0w_n^{2\bigM}\]
then there exists $0\le i<2\bigM$ such that $\pot(v_0w_n^i)\le \pot(v_0w_n^{i+1})$. Indeed, otherwise the potential strictly decreases, contradicting the equality of potential between $v_0$ and $v_0w_n^{2\bigM}$.

We again construct an elongating words sequence $\words$ as follows. For every $m\in \bbN$, let $n$ be large enough such that $|w_n|>n$ (recall that we prove above that there is such $n$), and choose $i$ as above, i.e., such that $\pot(v_0w_n^i)\le \pot(v_0w_n^{i+1})$. 
Define $\words(m)=w'_1w'_2w'_3w'_4$ with $w'_1=v_0w_n^i$, $w'_2=\epsilon$, $w'_3=w_n$ and $w'_4=\epsilon$.

We now analyze two possible scenarios.
\begin{itemize}
    \item If $\words$ is a $\pot$-leveled sequence (\cref{def:potleveled words sequence}), then by \cref{cor:potleveled implies witness or D dip}, either $\augA_\infty^\infty$ has a type-1 witness, concluding the proof of \cref{lem:potential is unbounded then type 1}, or $\words$ has the Infix $D$-dip property. In this case, we can reach a similar contradiction as in the $k=0$ case, as follows.

    Since $\words$ has the Infix $D$-dip property, then by \cref{cor:infix D dip implies increasing infix} it has a separated increasing infix. Specifically, there exists $m\in \bbN$ with $\words(m)=v_0w_n^i\cdot \epsilon\cdot w_n\cdot \epsilon$ and a decomposition $w_n=uxyv$ that is a separated increasing infix from $\ghostTrans(s_0,v_0w_n^i)=S'$ (as above, this follows from $w'_2=w'_4=\epsilon$).
    Recall that $w_n$ is a word such that\footnote{$\alpha_{S',w_n}$ actually appears in a sub-cactus of some $z'_{n'}$, and it may be that $n'\neq n$. For clarity we unify them, as neither have any implication on the proof.} $\alpha_{S',w_n}\in \subcact_{k-1}(z'_n)$.

    By expanding \cref{def:deep sub cacti selection}, we can write $z'_n$ as a cactus chain (\cref{def:cactus chain}), i.e., there is a cactus chain
    \[\Theta_1=\alpha_{S'_1,y_1},\ldots,\alpha_{S'_{k-1},y_k}\] 
    such that $z'_n=x_1\alpha_{S'_1,y_1}z_1$ for some words $x_1,y_1,z_1$, and $(S'_{k-1},y_k)=(S',w_n)$. 
    Then, since $w_n=uxyz$ is an increasing infix, then $\Theta_1$ is a \emph{Pre-bud chain} (\cref{def: pre and post bud chain}). 
    We can therefore invoke \cref{lem:post bud chain is superior higher potential}, and obtain that the corresponding Post-bud chain, which starts with $\alpha_{S'_1,y'_1}$, satisfies $\pot(z'_n)=\pot(x_1\alpha_{S'_1,y_1}z_1)\le \pot(x_1\alpha_{S'_1,y'_1}z_1)$.

    Note, however, that $\alpha_{S'_1,y'_1}$ is obtained by propagating the replacement of $w_n$ with $u\alpha_{B,x}y$ through the chain, and by \cref{lem:increasing infix budding} we have $\cost(w_n)>\cost(u\alpha_{B,x}y)$ (which also propagates through the chain), meaning that $\cost(z'_n)>\cost(x_1\alpha_{S'_1,y'_1}z_1)$. 
    But this is a contradiction to the assumption that $z'_n$ has minimal cost with this property, so we are done.    

    \item The second case is when $\words$ is not a $\pot$-leveled sequence. We claim that in this case we can find a faithful restriction of $\words$ (\cref{def:elongated words sequence}) that is a $\pot$-discharging sequence, and then we can follow the same reasoning as above. Intuitively, not being leveled means that the potential fluctuates ``wildly'' (and can become very large or very low, unlike charge). Recall that $\words$ is over a finite alphabet (due to the minimality assumption on $k$, see $\Gamma'$ above). In particular, we have that the potential satisfies the bounded-growth property of \cref{lem:bounded growth potential}. This becomes useful below.

    Since $\words$ is not leveled, then for every $B\in \bbN$ there exists $m\in \bbN$ with $\words(m)=v_0w_n^i\cdot \epsilon\cdot w_n\cdot \epsilon$ for some $n\in \bbN$, and a prefix $w'_n$ of $w_n$ such that $|\pot(v_0w_n^i)-\pot(v_0w_n^iw'_n)|>B$. Write $w_n=w'_n w''_n$.
    \begin{itemize}
        \item If for every $B$ there exists such $m,n$ for which $\pot(v_0w_n^iw'_n)-\pot(v_0w_n^i)>B$, then we can readily construct a $\pot$-discharging words sequence $\potdiswords$ by setting $\potdiswords(m')=v_0w_n^i\cdot \epsilon\cdot w'_n\cdot w''_n$, where we take $n$ large enough so that the discharge is greater than $m'$. Specifically, we can assume $w'_n$ grows unboundedly due to the bounded-growth property of the potential (\cref{lem:bounded growth potential}).
        \item If for every $B$ there exists such $m,n$ for which $\pot(v_0w_n^i)-\pot(v_0w_n^iw'_n)>B$, then by the choice of $i$ we have $\pot(v_0w_n^i)\le \pot(v_0w_n^{i+1})=\pot(v_0w_n^{i}w'_nw''_n)$.
        Thus, we have $\pot(v_0w_n^iw'_nw''_n)-\pot(v_0w_n^iw'_n)>B$.
        Since the potential has bounded-growth over this finite alphabet, as we already argue before the $k=0$ case, we have that the $w''_n$ suffixes in this case have unbounded length.
        Then we can again construct a $\pot$-discharging\footnote{A nit-picking details is that we need to take $w''_n$ to be the suffix upon which the discharge is maximal, to conform with \cref{def:potdischarging word sequence}.} words sequence $\potdiswords$ by setting $\potdiswords(m')=v_0w_n^iw'_n\cdot \epsilon\cdot w''_n\cdot \epsilon$, where we take $n$ large enough so that the discharge is greater than $m'$.
    \end{itemize}
    Therefore, we have a $\pot$-discharging sequence $\potdiswords$ such that for every $m\in \bbN$ with $\potdiswords(m)=w'_1w'_2w'_3w'_4$ it holds that $w'_3$ is an infix (in fact, either a suffix or a prefix) of $w_n$ for some $n\in \bbN$. 
    By our main result for $\pot$-discharging sequences, namely \cref{cor:potdischarge implies type 1 or D dip}, we again have that either $\augA_\infty^\infty$ has a type-1 witness, concluding the proof of \cref{lem:potential is unbounded then type 1}, or  $\potdiswords$ has the Infix $D$-dip property. 
    In the latter case, we proceed identically to the $\pot$-leveled words case above: we obtain an increasing infix within $w_n$, therefore we have a Pre-bud chain of lower cost, which contradicts the cost-minimality of $z'_n$.    
\end{itemize}

All in all, we conclude that in all cases that did not reach a contradiction, $\augA_\infty^\infty$ has a type-1 witness and the proof of \cref{lem:potential is unbounded then type 1}. 
Combining this with \cref{lem:A is nondet then potential is unbounded}, we conclude the proof of \cref{thm:nondet to witness}.

\subsection{Determinizability of WFAs is Decidable}
\label{sec:decidability}
Having proved a sound and complete characterization of nondeterminizable WFAs by witnesses, we can now use the fact that checking a witness is decidable, to conclude the decidability of determinization.

\begin{theorem}
\label{thm:decidability of determinization}
    Determinizability of WFA is decidable.
\end{theorem}
\begin{proof}
    We show that both determinizability and its complement are in RE, and so the determinizability problem is in $\mathrm{RE}\cap \mathrm{coRE}=R$, i.e., decidable.

    \paragraph{Determinizability is in $\mathrm{RE}$}
    Given a WFA $\cA$, we can decide whether it is determinizable by enumerating all deterministic WFAs, and for every deterministic WFA $\cD$, check whether $\cA$ and $\cD$ are equivalent. The latter is decidable by \cite{Almagor2020Whatsdecidableweighted} (the so-called \emph{$N,D$ Equality 
 Problem} therein. Crucially relying on one of the automata being deterministic).

    If we find an equivalent deterministic automaton, we are done. Otherwise we continue.

    We remark that we can improve upon this brute force algorithm, where instead of trying every deterministic automaton, we enumerate every size of gap $G$, and construct a deterministic automaton as per the proof of \cref{thm:det iff bounded gap}. Essentially, we try bigger and bigger gaps, and if $\cA$ is determinizable then eventually we will reach its gap and obtain the equivalent deterministic automaton.

    \paragraph{Determinizability is in $\mathrm{coRE}$}
    By \cref{thm:nondet to witness}, if $\cA$ is nondeterminizable, then it has a type-1 witness. We therefore enumerate every tuple $(w_1,w_2,w_3)$ with $w_1,w_2\in \Gamma_\infty^1$ and $w_3\in \Gamma_\infty^\infty$ (including jump letters in $w_3$), and for every such tuple check whether it is a witness using \cref{thm:verifying witness is decidable}.

    Note that in order to enumerate all tuples, we need to increase the \emph{length} of the words, the \emph{cactus depth}, and the \emph{rebase rank}. However, this can be done in parallel (equivalently, the tuple can be guessed nondeterministically).
\end{proof}

\begin{figure}[ht]
    \centering
    \includegraphics[width=0.3\linewidth]{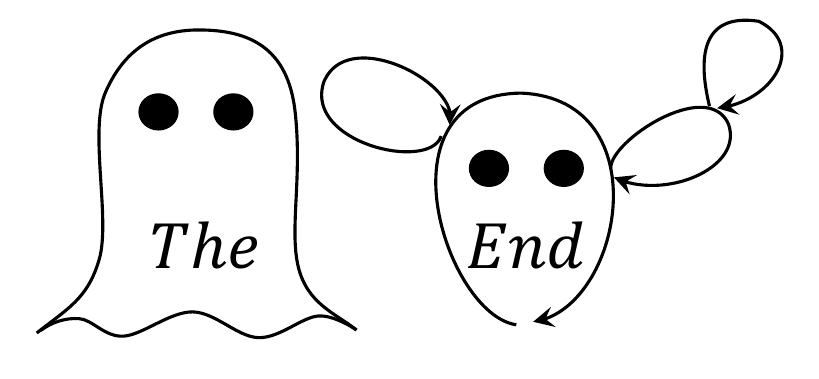}
    \label{fig:the end}
\end{figure}

\section*{Funding}
This research was supported by the ISRAEL SCIENCE FOUNDATION (grant No. 989/22).

\bibliographystyle{plain} 
\bibliography{main}

\pagebreak
\appendix

\section{A Characterization of Determinizability by Gaps}
\label{sec:apx det iff bounded gap}
\newcommand{\weighti}{\weight^{+\mathsf{i}}}
\newcommand{\weightf}{\weight^{+\mathsf{f}}}
\newcommand{\weightif}{\weight^{+\mathsf{if}}}
\newcommand{\minweighti}{\mathsf{m}\weight^{+\mathsf{i}}}
\newcommand{\minweightf}{\mathsf{m}\weight^{+\mathsf{f}}}
\newcommand{\minweightif}{\mathsf{m}\weight^{+\mathsf{if}}}
\newcommand{\normalize}{\mathsf{shift_0}}

In this section we prove a characterization of determinizability by means of ``gaps'' between runs. 
Since we use this result in \cref{sec:no need for init and fin and acc} for proving that we can omit initial and final weights (as well as accepting states), we prove it also for the general model of WFAs with initial and final weights.
Note that our model can be thought of as having initial weights in $\{0,\infty\}^Q$ (with only a single state weighted $0$), and the final weights are just $\vec{0}$.

We now define explicitly and precisely the model of WFA with initial and final weights, and give the full details of the reduction.

\paragraph{Weighted Automata with Initial and Final Weights}
\newcommand{\WFAif}{WFA$_\text{i-f}$\xspace}
A \emph{$(\min,+)$ Weighted Automaton with initial and final weights} (\WFAif for short) is a tuple $\cA= \tup{Q,\Sigma, \init, \Delta ,\fin}$ with the same components as a WFA, except that 
$\init,\fin\in \bbZ_\infty^Q$ are $Q$-indexed vectors denoting for each state its \emph{initial weight} and \emph{final weight}, respectively. 

For the purpose of this section, we assume $\Sigma$ is \emph{finite}. We do, however, comment on what happens for infinite $\Sigma$ in \cref{rmk: det iff bound infinite alphabet}.
    
A \WFAif is \emph{deterministic} if $|\{q\mid \init(q)\neq \infty\}|=1$ and for every $\sigma\in \Sigma$ and $p\in Q$ there exists at most one transition $(p,\sigma,c,q)$ with $c\neq \infty$.

For a run $\rho:p\runsto{x}q$ in a \WFAif, the sum of its weights along the transitions is denoted $\weight(\rho)$ (similarly to WFAs). We also introduce some additional notations for \WFAif, as follows.
\begin{itemize}
    \item $\weighti(\rho)=\init(p)+\weight(\rho)$ and $\minweighti(x,Q_1\to Q_2)=\min\{\weighti(\rho)\mid \rho:Q_1\runsto{x}Q_2\}$.
    \item $\weightf(\rho)=\weight(\rho)+\fin(q)$ and $\minweightf(x,Q_1\to Q_2)=\min\{\weightf(\rho)\mid \rho:Q_1\runsto{x}Q_2\}$.
    \item $\weightif(\rho)=\init(p)+\weight(\rho)+\fin(q)$ and $\minweightif(x,Q_1\to Q_2)=\min\{\weightif(\rho)\mid \rho:Q_1\runsto{x}Q_2\}$.
\end{itemize}
We denote $Q_0=\{q\mid \init(q)<\infty\}$ and $F=\{q\mid \fin(q)<\infty\}$.
As usual, we assume all automata we deal with a trim. For \WFAif, this means that every state is reachable from $Q_0$ and in addition can reach $F$. 

Intuitively, we show that $\cA$ is determinizable if and only if there is a bound of how large a run can grow above the current minimal run, and still become minimal with an appropriate suffix. The intuition for the proof is that a deterministic equivalent must, in some sense, keep track of the differences in values between all runs that may still become minimal.

Our main result in this section is the following.
\begin{theorem}
    \label{thm: gaps iff det in WFAif}
    A trim \WFAif is determinizable if and only if there exists a bound $B\in \bbN$ such that for every word $x\in \Sigma^*$ and state $q\in Q$, if there exists $y\in \Sigma^*$ such that $\minweightif(x\cdot y,Q_0\to F)=\minweightif(x\cdot y,Q_0\runsto{x}q\runsto{y}F)<\infty$, then $\minweighti(x,Q_0\to q)-\minweighti(x,Q_0\to Q)\le B$.
\end{theorem}

\begin{remark}
    \label{rmk: definition of det in gap characterization}
    It is easy to see that the statement of \cref{thm: gaps iff det in WFAif} coincides with that of \cref{thm:det iff bounded gap} in when the \WFAif is actually a WFA. 
    It should be noted, however, that when we say ``determinizable'' for a WFA, we mean that the equivalent deterministic is also a WFA (rather than a \WFAif). Otherwise one may worry that the determinization of a WFA in this characterization is a \WFAif. We explain why this holds in the proof. 
\end{remark}
\begin{remark}[On Infinite Alphabets]
\label{rmk: det iff bound infinite alphabet}
The proof that the existence of a bound implies determinizability does not rely on the alphabet begin finite, and therefore holds for infinite alphabets as well. We use this in \cref{cor:aug inf inf is det iff A is det}.

The converse direction, however, does not hold. Yet, the proof in the converse direction can be adapted to state that if $\cA$ is determinizable, then for every finite subset of the alphabet, there exists a bound for words over this subset. While we do not make use of this, it is interesting to note.
\end{remark}

We prove \cref{thm: gaps iff det in WFAif} in the remainder of this section. We remark that the final weights, combined with our sub-goal in \cref{rmk: definition of det in gap characterization}, add some annoying technicalities to the proof. In order to understand the proof, we recommend assuming that all final weights are $0$, and applying this whenever they are discussed.

\subsection{First Direction: a Bound Implies Determinizability}
\label{sec:bound to det gap characterization}
For the first direction, assume there exists a bound $B$ as specified in the theorem. We construct an equivalent deterministic \WFAif  $\cD=\tup{Q',\Sigma,\init',\Delta',\fin'}$. 
We start with the intuition for the construction. Consider a word $x\in \Sigma^*$. The \WFAif $\cD$ keeps track of minimal run of $\cA$ on $x$, but also keeps track of the difference between the minimal run and higher runs, provided that these runs are no more than $B$ above the minimal one. The condition in the theorem ensures that once a run becomes higher than $B$ above the minimal one, it can no longer become minimal on any suffix, and can therefore be safely ignored.
In order to track the runs, the minimal run is always normalized to $0$, and this normalization is accounted for by the weight of the transition. 

We start with some definitions and notations.  Let $\max\fin=\max\{|\fin(q)|\mid q\in Q\}$.
For a function $f:Q\to \bbNinf$, we define $f|_B\in Q'$ by setting $f|_B(q)=f(q)$ if $f(q)\le B$, and $f|_B(q)=\infty$ otherwise, and we define $\normalize(f)$ to be the function $\normalize(f)(q)=f(q)-\min\{f(p)\mid p\in Q\}$, i.e., we shift all values of $f$ such that the minimum is $0$. We now turn to formally define the construction.
\begin{itemize}
    \item The state space is $Q'=(\{0,\ldots, M\}\cup \{\infty\})^Q\times \{-\max\fin,\ldots,\max\fin\}$. Intuitively, each state corresponds to a configuration, shifted to $0$ (and ignoring states whose minimal run thus far has weight more than $B$ above the minimum), as well as a number representing the final weight that was incurred in the last step.
    \item In order to define the initial state, we consider two cases. 
    \begin{itemize}
        \item If every state $q$ with $\init(q)<\infty$ has $\inf(q)=\infty$ (i.e., no accepting initial states), then 
    The initial vector $\init'$ assigns weight $\init(q)$ to the single state $(\normalize(\init)|_B,0)$.
    \item If there exists $q\in Q$ such that $\init(q)<\infty$ and $\fin(q)<\infty$ (i.e., the word $\epsilon$ receives a finite weight), let $q\in Q$ be the state for which $\init(q)+\fin(q)$ is minimal. 
    The initial vector $\init'$ assigns weight $\init(q)+\fin(q)$ to the single state $(\normalize(\init)|_B,\fin(q))$.
    \end{itemize}
    
    \item For $(g,c)\in Q'$, the final weight is $0$ if there exists $q\in Q$ such that $\fin(q)<\infty$ and $g(q)<\infty$ (i.e., if $g$ ``contains'' an accepting state), and otherwise the final weight is $\infty$.
    
    Observe that if $\cA$ is a WFA, then $\init'$ has entries in $\{0,\infty\}$ with only one entry of $0$ (and $\fin'$ always has entries in $\{0,\infty\}$ by our construction). 
    Thus, $\cD$ is also a WFA in this case (c.f., \cref{rmk: definition of det in gap characterization}).
    \item The transitions are defined as follows. 
    Consider a state $(g,c)\in Q'$ and letter $\sigma\in \Sigma$. We first define an intermediate function $g':Q\to \bbZinf$ by setting for every $q\in Q$
    \[g'(q)=\min\{g(p)+d\mid p\in Q,\ (p,\sigma,d,q)\in \Delta\}\]
    We then consider $g''=\normalize(g')|_B$, and split to cases similarly to the initial states. 
    \begin{itemize}
        \item If $\fin(q)=\infty$ for every $q$ with $g''(q)<\infty$, we have the transition 
        \[((g,c),\sigma,\min\{g'(q)\mid q\in Q\}-c,(g'',0))\in \Delta'\] 

        \item If there exists $q\in Q$ with $\fin(q)<\infty$ and $g''(q)<\infty$, let $q\in Q$ be the state for which $g''(q)+\fin(q)$ is minimal.  We have the transition 
        \[((g,c),\sigma,\min\{g'(q)\mid q\in Q\}-c+\fin(q),(g'',\fin(q)))\in \Delta'\]        
    \end{itemize}
    (all other transitions from $g$ with $\sigma$ have weight $\infty$, so $\cD$ is deterministic).

\end{itemize}

We explain the intuition and prove the correctness of the construction (namely that $\cD$ is equivalent to $\cA$). The proof proceeds by induction on the length of a give word $x$.

For the base case, $x=\epsilon$ and we consider the initial state. By definition, the configuration at the initial state is $g_0=\normalize(\init)|_B$. 
If $\weightif_\cA(\epsilon)=\infty$, then the initial weight is $0$, but the initial state $(g_0,0)$ has final weight $\infty$.

Otherwise, $\weightif_\cA(\epsilon)=\infty$ and the initial state $(g_0,\fin(q))$ has weight $\init(q)+\fin(q)=\weightif_\cA(\epsilon)$ for the appropriate state $q$.

From this, we extract our induction hypothesis. Consider the run $\rho_x$ of $\cD$ on $x$, then:
\begin{enumerate}
    \item $\rho_x$ ends in state $(\normalize(g_x)|_B,c_x)$ with $g_x(q)$ being the minimal value of a run on $x$ that ends in $q$, provided this run is at most $B$ above the overall minimal run thus far.
    \item If $\weightif_\cA(x)<\infty$, then $c_x=\fin(q)$ for the state $q$ such that $g_x(q)+\fin(q)=\weightif_\cA(x)$ (i.e., the state which induces the minimal accepting run, including final weight), and in addition $\weighti_\cD(\rho)=g_x(q)+c_x$.
    \item If $\weightif_\cA(x)=\infty$, then $c_x=0$ and $\weighti_\cD(\rho)=\min\{g_x(q)\mid q\in Q\}$.
\end{enumerate}
By the above, the induction hypothesis holds for $\epsilon$. The definition of the transition relation easily shows that this carries through the induction step, as we now show.

Consider a word $x\sigma$, and assume the hypothesis above for $\rho_x$, which ends in $(\normalize(g_x)|_B,c_x)$. Note that we cannot directly refer to $g_x$, as it has been normalized and clipped at $B$. Let $(h,d)$ be the state reached by $\cD$ after reading $\sigma$ from $(\normalize(g_x)|_B,c_x)$.
Consider the intermediate function $g'$ defined in $\Delta'$ above. Observe that $g'(q)$ is the minimal added weight from $g(q')$ via a transition to $q$. By the premise of the theorem, states whose minimal run is at more than $B$ above the minimal run, cannot yield minimal runs on any suffix. By the induction hypothesis, the states for which $g_x(q)=0$ represent the minimal runs, and therefore all states more than $B$ above are marked as $\infty$ in $g_x$ anyway. It follows that $g'(q)$ represents the minimal weight of a run that ends in $q$ and might still become a minimal run. 

The remaining two conditions also follow from the definition of $\Delta'$: if there is a state with finite final weight, then $h$ accumulates the final weight of the state yielding the minimal run. Otherwise $h$ accumulates $0$ (on top of the weight lost by shifting).
Observe that in $\Delta'$, the transition accounts for $-c$, since if $g$ already accounted for some final weight, this has to be removed (as the run did not end in $g$).
Put simply: $\cD$ tracks the correct minimal weight to any run that might still become minimal.
We then have that $(h,d)$ satisfies the conditions of the induction hypothesis, which concludes the claim.

By the hypothesis, when the run ends, the accumulated weight is exactly $\weight_\cA(x)$, so we are done.

\subsection{Second Direction: Determinizability Implies a Bound}
\label{sec:det to bound gap characterization}
For this section, assume $\cA=\tup{Q,\Sigma,\init,\Delta,\fin}$ has an equivalent (trim) deterministic \WFAif $\cD=\tup{Q',\Sigma,\init',\Delta',\fin'}$, we show that there exists a bound $B$ as in the theorem.
Recall that we denote by $Q_0,F$ (resp. $Q'_0,F'$) the initial and final states (i.e., those with finite initial and final weight) in $\cA$ and $\cD$ respectively.

Denote by $M$ the maximal absolute value of a weight appearing in either $\cA$ or $\cD$ (as a transition, initial or final weight).

We start by showing that $\cD$ must ``roughly'' track the same weights as the minimal run of $\cA$ on any word.
\begin{proposition}
    \label{prop:det iff gap: det must track min run}
    For every word $x\in \Sigma^*$, if $\minweighti_\cD(x,Q'_0\to Q')<\infty$, we have that $|\minweighti_\cD(x,Q'_0\to Q')-\minweighti_\cA(x,Q_0\to Q)|\le M\cdot (|Q|+1)$.
\end{proposition}
\begin{proof}
    Intuitively, there exists a short suffix from which $xz$ must be accepted with finite weight (since the automata are trim), and since $\cD$ is equivalent to $\cA$, it cannot deviate more than $|z|M$ away from the value assigned by $\cA$.

    Formally, let $q\in Q$ such that $\minweighti_\cA(x,Q_0\to Q)=\minweighti_\cA(x,Q_0\to q)$, and let $z\in \Sigma^*$ be a minimal-length word such that $\minweightf_\cA(z,q\to F)<\infty$. Such a word induces a cycle-free run (otherwise it is not minimal), and $|z|<|Q|$ (the strict inequality is because $z$ induces a run with $|z|+1$ states).

    We then have that 
    \[\minweightif_\cD(xz,Q'_0\to Q')=\minweightif_\cA(xz,Q_0\to F)\le  \minweighti_\cA(x,Q_0\to q)+\minweightf_\cA(z,q\to F)\]
    Observe that $|\minweightf_\cA(z,q\to F)|\le M\cdot |Q|$ (since $|z|<|Q|$, but one more weight is incurred as the final weight).
    Thus, we have $\minweightif_\cD(xz,Q'_0\to Q')\le \minweighti_\cA(x,Q_0\to q)+M\cdot |Q|$.
    
    On the other hand, by our choice of $q$, for every other $q'\in Q$ we have that $\minweighti_\cA(x,Q_0\to q')\ge \minweight_\cA(x,Q_0\to q)$, and in particular for the state $q'$ along a minimal run of $\cA$ on $xz$ we have
    \[\minweightif_\cA(xz,Q_0\to F)=\minweighti_\cA(x,Q_0\to q')+\minweightf_\cA(z,q'\to F)\ge \minweight_\cA(x,Q_0\to q) - M\cdot |Q|\]
    so $\minweightif_\cD(xz,Q'_0\to Q')\le \minweighti_\cA(x,Q_0\to q)-M\cdot |Q|$.
    
    Combining the two inequalities, we have
    $|\minweightif_\cD(x,Q'_0\to Q')-\minweighti_\cA(x,Q_0\to Q)|\le M\cdot |Q|$.
    The last step is to notice that since $\cD$ is deterministic, we have that 
    $|\minweighti_\cD(x,Q'_0\to Q')-\minweightif_\cD(x,Q'_0\to Q')|\le M$ (as they differ by at most the final weight).
    By the triangle inequality, we conclude that
    $|\minweighti_\cD(x,Q'_0\to Q')-\minweighti_\cA(x,Q_0\to Q)|\le M\cdot (|Q|+1)$
\end{proof}

We proceed with the proof of the existence of a bound $B$.
Assume by way of contradiction that there does not exist such a bound $B$. 
That is, for every $B\in \bbN$ there exists a word $x\in \Sigma^*$, a state $q\in Q$ and $y\in \Sigma^*$ such that 
\begin{itemize}
    \item $\minweightif_\cA(x\cdot y,Q_0\to F)=\minweightif_\cA(x\cdot y,Q_0\runsto{x}q\runsto{y}F)<\infty$, but 
    \item $\minweighti_\cA(x,Q_0\to q)-\minweighti_\cA(x,Q_0\to Q)>B$.
\end{itemize}

Let $B>(2^{|Q|}\cdot |Q'|+(|Q|+1))\cdot 2\cdot M$ be some large constant, and take $x,q,y$ as above.
Let $p\in Q'$ the state of $\cD$ after reading $x$, i.e., 
\begin{equation}
\label{eq:det gap: D run xy}
  \minweightif_\cD(xy,Q'_0\to F')=\minweighti_\cD(x,Q'_0\to p)+\minweightf_\cD(y,p\to F')    
\end{equation}
  
By our choice of $q$, we have that 
\begin{equation}
\label{eq:det gap: A run xy}
\minweightif_\cA(x y,Q_0\to F)=\minweighti_\cA(x,Q_0\to q)+\minweightf_\cA(y,q\to F)
\end{equation}

Combining these equations with \cref{prop:det iff gap: det must track min run}, we have the following:
\begin{equation}
\label{eq:det gap: D lower bound y}
\begin{split}
    &\minweightf_\cD(y,p\to F')=_{(1)}\minweightif_\cD(xy,Q'_0\to F')-\minweighti_\cD(x,Q'_0\to p)\\
    &=_{(2)}\minweightif_\cA(xy,Q_0\to F)-\minweighti_\cD(x,Q'_0\to p)\\
    &=_{(3)}\minweighti_\cA(x,Q_0\to q)+\minweightf_\cA(y,q\to F)-\minweighti_\cD(x,Q'_0\to p)\\
    &\ge_{(4)} \minweighti_\cA(x,Q_0\to q)+\minweightf_\cA(y,q\to F)-\minweighti_\cA(x,Q_0\to Q)-M(|Q|+1)\\
    &>_{(5)} \minweightf_\cA(y,q\to F)+B-M(|Q|+1)
\end{split}
\end{equation}
Where 
\begin{itemize}[noitemsep]
    \item[(1)] is just a rearrangement of \cref{eq:det gap: D run xy}.
    \item[(2)] is because $\cD$ is equivalent to $\cA$.
    \item[(3)] is by \cref{eq:det gap: A run xy}.
    \item[(4)] is an application of \cref{prop:det iff gap: det must track min run}.
    \item[(5)] is our assumption on $q$ above (by way of contradiction).
\end{itemize}

Next, consider the set $S_x=\delta_{\bbB}(x)$ of states reachable (by finite weight runs) on $x$ in $\cA$ and recall that the run of $\cD$ on $x$ reaches state $p$. Since $B>(2^{|Q|}\cdot |Q'|+(|Q|+1))\cdot 2\cdot M$, it follows that $|x|>2^{|Q|}\cdot |Q'|+(|Q|+1)>2^{|Q|}\cdot |Q'|$. Indeed, with each letter read, the minimal and maximal runs can diverge by at most $2M$, so our assumption on the size of the gap cannot be captured with a shorter $x$.

It follows, however, that there exists $x'\in \Sigma^*$ with $|x|\le 2^{|Q|}\cdot |Q'|$ such that $\delta_{\bbB}(x')=S_x$ and the run of $\cD$ on $x'$ reaches $p$. This is by e.g., removing simultaneous cycles in $\cD$ and the (Boolean) subset-construction of $\cA$.

We can now readily reach a contradiction as follows. On the one hand we use the lower bound on $\minweightf_D(y,p\to F')$ in \cref{eq:det gap: D lower bound y} to get
\begin{equation}
\label{eq:det gap: D lower bound x'y}
\begin{split}
&\minweightif_\cD(x'y,Q'_0\to F')= \minweighti_\cD(x',Q'_0\to p)+\minweightf_\cD(y,p\to F')\\
&\ge \minweighti_\cD(x',Q'_0\to p)+\minweightf_\cA(y,q\to F)+B-M(|Q|+1)
\end{split}
\end{equation}
On the other hand we have
\[
\begin{split}
&\minweightif_\cD(x'y,Q'_0\to F')=_{(1)} \minweightif_\cA(x'y,Q_0\to F)\\
&\le_{(2)}  \minweighti_\cA(x',Q_0\to q)+\minweightf_\cA(y,q\to F)\\
&\le_{(3)}  \minweighti_\cA(x',Q_0\to Q)+2^{|Q|}\cdot |Q'|\cdot 2\cdot M+\minweightf_\cA(y,q\to F)\\
&\le_{(4)} \minweighti_\cD(x',Q'_0\to p)+M(|Q|+1)+2^{|Q|}\cdot |Q'|\cdot 2\cdot M+\minweightf_\cA(y,q\to F) 
\end{split}
\]
where
\begin{itemize}[noitemsep]
    \item[(1)] is because $\cA$ and $\cD$ are equivalent.
    \item[(2)] is because the minimal run is at most the one that goes through $q$.
    \item[(3)] is because the run to $q$ can diverge by at most $|x'|\cdot 2\cdot M$ from the minimal run.
    \item[(4)] is by \cref{prop:det iff gap: det must track min run}.
\end{itemize}

Combining this inequality with \cref{eq:det gap: D lower bound x'y} and cancelling equal terms, we get
\[B\le 2M(|Q|+1)+2^{|Q|}\cdot |Q'|\cdot 2\cdot M\]

which contradicts our choice of $B$.
We conclude that there exists a bound $B$ on the gap, as per the theorem statement.

\section{No Need for Initial and Final Weights, nor Accepting States}
\newcommand{\slet}{\mathsf{s}}
\newcommand{\flet}{\mathsf{f}}
\label{sec:no need for init and fin and acc}
In this section we show that the determinization problem for WFAs with initial and final weights (\WFAif) reduces to the determination problem for our model, which only has one initial state (and in particular $\init\in \{0,\infty\}^Q$) and all states are accepting (i.e., $\fin=\{0\}^Q$). 
The precise definition of \WFAif can be found in \cref{sec:apx det iff bounded gap} (we also rely on the results proved there, namely \cref{thm: gaps iff det in WFAif}).

The intuition for getting rid of the initial and final weights is simple: we add two letters $\{\slet,\flet\}$ to the alphabet, to denote the start and the finish of a word, respectively. Upon reading $\slet$ and $\flet$, the automaton incurs the weights described by the initial and final states, respectively.
This, however, still leaves us with accepting states, which turn out to pose a bigger obstacle. 
Nonetheless, we show that if all states are made accepting after the aforementioned additions, the obtained WFA is still equi-determinizable to the original \WFAif (but they are not equivalent in their quantitative language).

\begin{lemma}
    \label{lem:init final weight reduces to no init and final}
    The determinization problem for \WFAif is reducible (in logarithmic space) to the determinization problem for WFAs.
\end{lemma}
We prove the lemma in the remainder of this section.
Consider a \WFAif $\cA= \tup{Q,\Sigma, \init, \Delta ,\fin}$, and assume $\cA$ is trim (otherwise remove non-reachable states and states that cannot reach $F$).

We construct a WFA $\cB= \tup{Q\cup \{s_0,s_f\},\Sigma\cup \{\slet,\flet\}, s_0, \eta }$ where $s_0,s_f\notin Q$ and $\slet,\flet\notin \Sigma$. The transitions are defined as follows:
    \begin{itemize}
        \item For $p,q\in Q$ and $\sigma\in \Sigma$ we have $(p,\sigma,c,q)\in \eta$ iff $(p,\sigma,c,q)\in \Delta$.
        \item For $q\in Q$ we have $(s_0,\slet,\init(q),q)\in \eta$ and $(q,\flet,\fin(q),s_f)\in \eta$.
        \item The remaining transitions are with weight $\infty$.
    \end{itemize}
    This construction can clearly be implemented in logarithmic space.
    We prove that $\cA$ is determinizable if and only if $\cB$ is determinizable.

    Observe that by the construction of $\cB$, we have for every word $w\in \Sigma^*$ that
    \begin{equation}
    \label{eq:init fin weight equiv}
    \weight_\cA(w)=\weight_\cB(\slet \cdot w\cdot  \flet)    
    \end{equation}

   \subsection{If $\cB$ is Determinizable, then $\cA$ is Determinizable.}
   \label{sec:no need for init and fin: B det to A det}
    For the first (and easy) direction, assume $\cB$ is determinizable, and let $\cD=\tup{Q_\cD,\Sigma\cup\{\slet,\flet\}, q_0, \Delta_\cD}$ be an equivalent deterministic trim WFA. 
    Note that we can assume $q_0$ has no incoming transitions (otherwise a word with more than one $\slet$ can have finite weight, unless all words have cost $\infty$, which is a degenerate case).
    Similarly, we can assume that once $\flet$ is read, a unique state $q_f$ is reached, from which there are no outgoing transitions (indeed, no final weights can be accumulated upon leaving $q_f$). We can assume $q_0\neq q_f$, otherwise the only accepted word in $\cA$ is $\epsilon$, which is again degenerate (since $q_f$ has no outgoing transitions).

    Define a deterministic \WFAif  $\cD'=\tup{Q'_{\cD},\Sigma, \init', \Delta'_\cD ,\fin'}$ as follows. The states are $Q'_\cD=Q_\cD\setminus \{q_0,q_f\}$.
    Let $q\in Q'_\cD$ be the unique state such that $(q_0,\slet,c,q)\in \Delta$ with $c\neq \infty$ (since $\cD'$ is deterministic, and recall that $q\neq q_0$).
    Define $\init'(q)=c$, and for every $q'\neq q$ we set $\init'(q')=\infty$. 
    For the final vector, for every $q\in Q'_\cD$
    set $\fin(q)=c$ where $(q,\flet,c,q_f)\in \Delta$.
    For the remaining transitions, we have that $\Delta'_\cD=\Delta\cap (Q'_\cD\times \Sigma\times \bbZinf\times Q'_\cD)$ (i.e., we keep only transitions on $Q'_\cD$ and over $\Sigma$).

    We claim that $\weight_{\cD'}(w)=\weight_{\cA}(w)$ for every $w\in \Sigma^*$. By \cref{eq:init fin weight equiv}, it is enough to prove that $\weight_{\cD'}(w)=\weight_\cB(\slet\cdot w\cdot \flet)$, but the latter is immediate from the construction, since reading $\slet$ is simulated by the initial weight, and reading $\flet$ by the final weights. Note that even though we made all states accepting in $\cB$, this does not play a role here, since we only consider (in $\cB$) words padded with $\slet$ and $\flet$. 
    
    \subsection{If $\cA$ is Determinizable, then $\cB$ is Determinizable}
    \label{sec:no need for init and fin: A det to B det}
    We now turn to the harder direction. Assume $\cA$ is determinizable and consider an  equivalent deterministic  \WFAif $\cD=\tup{Q'_\cD,\Sigma,\init',\Delta_\cD,\fin'}$. Let $q^0_\cD$ be the single state with $\init'(q^0_\cD)<\infty$, i.e., the initial state of $\cD$, and $F_\cD$ the accepting states (i.e., those with finite $\fin'$).
    Denote $Q'=Q\cup \{s_0,s_f\}$, and assume by way of contradiction that $\cB$ is not determinizable. 
    
    By the characterization in \cref{thm: gaps iff det in WFAif}, we have that for every $B\in \bbN$ there exist words $x,y\in (\Sigma\cup \{\slet,\flet\})^*$ and a state $q\in Q'$ such that
    \begin{itemize}
        \item $\minweight_\cB(x\cdot y,s_0\to Q')=\minweight_\cB(x\cdot y,s_0\runsto{x}q\runsto{y}Q')<\infty$, but 
        \item $\minweight_\cB(x,s_0\to q)-\minweight_\cB(x,s_0\to Q')>B$.
    \end{itemize}
    Observe that we do not need to specify whether the expressions above refer to initial/final weights, since $\cB$ is a WFA. Moreover, we look at reachability to $Q'$ instead of $F$, since in a WFA all states are accepting.
    We refer to $x,y,q$ above as a \emph{$B$-gap witness for $\cB$}.

    The proof follows similar lines to \cref{sec:det to bound gap characterization}, with some crucial changes.

    Denote by $M$ be the maximal absolute value of a weight appearing in $\cA,\cB$ or $\cD$. 
    Since $\cA$ is determinizable, then by \cref{thm: gaps iff det in WFAif}, it has a bound $B'$ on the gaps.
    
    Let $B>\max\{B'+1,2M(2|Q|+1+2^{|Q|}\cdot |Q_\cD|)\}$, and let $x,y,q$ be a $B$-gap witness for $\cB$.
 
    If there is such a witness $x,y,q$ with $x=\slet u$ and $y=v\flet$ with $u,v\in \Sigma^*$, then $u,v,q$ is also a $B$-gap witness for $\cA$. Indeed, every run of $\cB$ on $\slet u$ corresponds to a run of $\cA$ on $u$ (i.e., the runs are identical up to omitting $s_0$ in the former), and every run of $\cB$ on $v\flet$ corresponds to a run of $\cA$ on $v$ that reaches an accepting state (up to omitting $s_f$ in the former). This, however, is a contradiction, since $B>B'$.

    Assuming the gap witness is not of the form above, we note that any $B$-gap witnesses for $\cB$ must have $x=\slet u$ for some $u\in \Sigma^*$. Indeed, otherwise either $x=\epsilon$ or $x$ gets weight $\infty$ in $\cB$. Therefore, $y\in \Sigma^*$ (i.e., cannot contain $\slet$ or $\flet$), otherwise $x\cdot y$ would get weight $\infty$ in $\cB$. 

    Since $\cA$ is trim, there exists a word $z\in \Sigma^*$ with $|z|<|Q|$ such that $\minweight_\cA(yz,q\to F)<\infty$. Thus, we have $|\minweight_\cB(z\flet,Q'\to Q')|<|Q|M$. 
    Since $x,y,q$ is a $B$-gap witness for $\cB$, we have 
    \[
    \begin{split}
        &\minweight_\cB(xyz\flet,s_0\to Q')\ge \minweight_\cB(xy,s_0\to Q')+\minweight_\cB(z\flet,Q'\to Q')\\
        &\ge \minweight_\cB(x,s_0\to q)+\minweight_\cB(y,q\to Q')-|Q|M
    \end{split}
    \]
    By the construction of $\cB$ from $\cA$ (and recall that $x=\slet u$), this implies that 
    \[
    \begin{split}
        &\minweightif_\cA(uyz,Q_0\to F)\ge \minweighti_\cA(uy,Q_0\to Q)+\minweightf_\cA(z,Q\to Q)\\
        &\ge \minweighti_\cA(u,Q_0\to q)+\minweight_\cA(y,q\to Q)-|Q|M
    \end{split}
    \]
    On the other hand,
    \[
    \begin{split}
        &\minweightif_\cA(uyz,Q_0\to F)\le \minweighti_\cA(u,Q_0\to q)+\minweightf_\cA(yz,q\to F)\\
        &\le \minweighti_\cA(u,Q_0\to q)+\minweight_\cA(y,q\to Q)+|Q|M
    \end{split}
    \]
    Combining the two inequalities, we have
    \begin{equation}
    \label{eq:no need for init and fin:uyz via q}
    |\minweightif_\cA(uyz,Q_0\to F) - (\minweighti_\cA(u,Q_0\to q)+\minweight_\cA(y,q\to Q))|\le |Q|M    
    \end{equation}
    Let $p\in Q_\cD$ be the state of $\cD$ after reading $u$. Thus,
     \begin{equation}
    \label{eq:no need for init and fin:D on uyz}
    \minweightif_\cD(uyz,q^0_\cD \to F_\cD)=\minweighti_\cD(u,q^0_\cD\to p)+\minweightf(yz,p\to F_\cD)
    \end{equation}
    Combining these equations with \cref{prop:det iff gap: det must track min run}, we have the following:
    \begin{equation}
    \label{eq:no need for init and fin:D on yz lower}
    \begin{split}
    &\minweightf_\cD(yz,p\to F_\cD)=_{(1)}\minweightif_\cD(uyz,q^0_\cD\to F_\cD)-\minweighti_\cD(u,q^0_\cD\to p)\\
    &=_{(2)}\minweightif_\cA(uyz,Q_0\to F)-\minweighti_\cD(u,q^0_\cD\to p)\\
    &\ge_{(3)}\minweighti_\cA(u,Q_0\to q)+\minweight_\cA(y,q\to Q)-|Q|M-\minweighti_\cD(u,q^0_\cD\to p)\\
    &\ge_{(4)}\minweighti_\cA(u,Q_0\to q)+\minweight_\cA(y,q\to Q)-|Q|M- \minweighti_\cA(u,Q_0\to Q)-M(|Q|+1)\\
    &>_{(5)} \minweight_\cA(y,q\to Q)-|Q|M+B-M(|Q|+1)\\
    &=_{(6)}  \minweight_\cA(y,q\to Q)+B-M(2|Q|+1)
    \end{split}
    \end{equation}
    Where 
    \begin{itemize}[noitemsep]
        \item[(1)] is just a rearrangement of \cref{eq:no need for init and fin:D on uyz}.
        \item[(2)] is because $\cD$ is equivalent to $\cA$.
        \item[(3)] is by \cref{eq:no need for init and fin:uyz via q} as a lower bound.
        \item[(4)] is an application of \cref{prop:det iff gap: det must track min run}.
        \item[(5)] is because $(x,y,q)$ is a $B$-gap witness for $\cB$, and $\minweighti_\cA(u,Q_0\to q)=\minweight_\cB(x,s_0\to q)$ and $\minweighti_\cA(u,Q_0\to Q)=\minweight_\cB(x,s_0\to Q')$.
        \item[(6)] is just rearranging.
    \end{itemize}

Next, consider the set $S_u=\delta_{\bbB}(u)$ of states reachable (by finite weight runs) on $u$ in $\cA$ and recall that the run of $\cD$ on $u$ reaches state $p$. 
Since $B>2M(2|Q|+1+2^{|Q|}\cdot |Q_\cD|)$,
it follows in particular that $|u|>2^{|Q|}\cdot |Q_\cD|$. Indeed, with each letter read, the minimal and maximal runs can diverge by at most $2M$, so our assumption on the size of the gap cannot be captured with a shorter $u$.

It follows, however, that there exists $u'\in \Sigma^*$ with $|u|\le 2^{|Q|}\cdot |Q_\cD|$ such that $\delta_{\bbB}(u')=S_u$ and the run of $\cD$ on $u'$ reaches $p$. This is by e.g., removing simultaneous cycles in $\cD$ and the (Boolean) subset-construction of $\cA$.

We can now reach a contradiction as follows. On the one hand we use the lower bound on $\minweightf_D(yz,p\to F_\cD)$ in \cref{eq:no need for init and fin:D on yz lower} to get
\begin{equation}
\label{eq:no need for init and fin:D on u'yz lower}
\begin{split}
&\minweightif_\cD(u'yz,Q^0_\cD\to F_\cD)= \minweighti_\cD(u',q^0_\cD\to p)+\minweightf_\cD(yz,p\to F_\cD)\\
&\ge \minweighti_\cD(u',q^0_\cD\to p)+ \minweight_\cA(y,q\to Q)+B-M(2|Q|+1)
\end{split}
\end{equation}
On the other hand we have
\begin{equation}
\label{eq:no need for init and fin:D on u'yz upper}
\begin{split}
&\minweightif_\cD(u'yz,q^0_\cD\to F_\cD)=_{(1)} \minweightif_\cA(u'yz,Q_0\to F)\\
&\le_{(2)}  \minweighti_\cA(u',Q_0\to q)+\minweightf_\cA(yz,q\to F)\\
&\le_{(3)}  \minweighti_\cA(u',Q_0\to Q)+2^{|Q|}\cdot |Q_\cD|\cdot 2\cdot M+\minweightf_\cA(yz,q\to F)\\
&\le_{(4)} \minweighti_\cD(u',q^0_\cD\to p)+M(|Q|+1)+2^{|Q|}\cdot |Q_\cD|\cdot 2\cdot M+\minweightf_\cA(yz,q\to F)\\
&\le_{(5)} \minweighti_\cD(u',q^0_\cD\to p)+M(|Q|+1)+2^{|Q|}\cdot |Q_\cD|\cdot 2\cdot M+\minweight_\cA(y,q\to F)+M|Q|
\end{split}
\end{equation}
where
\begin{itemize}[noitemsep]
    \item[(1)] is because $\cA$ and $\cD$ are equivalent.
    \item[(2)] is because the minimal run is at most the one that goes through $q$.
    \item[(3)] is because the run to $q$ can diverge by at most $|u'|\cdot 2\cdot M$ from the minimal run.
    \item[(4)] is by \cref{prop:det iff gap: det must track min run}.
    \item[(5)] is because $|z|< |Q|$.
\end{itemize}
Combining \cref{eq:no need for init and fin:D on u'yz lower,eq:no need for init and fin:D on u'yz upper} and cancelling equal terms, we get
\[B-M(2|Q|+1)\le M(|Q|+1)+2^{|Q|}\cdot |Q_\cD|\cdot 2\cdot M+M|Q|\]
which simplifies to
\[B\le 2M(2|Q|+1+2^{|Q|}\cdot |Q_\cD|)\]
and this contradicts our choice of $B$.
We conclude that $\cB$ is determinizable (as this was our assumption by way of contradiction).

\section{$\augA_\infty^\infty$ is Determinizable if and only if $\augA$ is Determinizable}
\label{apx:aug inf inf is det iff aug A is det}
In this section we prove \cref{thm:aug inf inf is det iff aug A is det}.

We start with the trivial direction: if $\augA_\infty^\infty$ is determinizable, then let $\cD_\infty^\infty$ be an equivalent deterministic WFA. We restrict $\cD$ to the alphabet $\Gamma$ and obtain a deterministic WFA that is equivalent to $\augA$. Thus, $\augA$ is  determinizable.

The converse is the involved part. Assume $\augA_\infty^\infty$ is not determinizable. By \cref{thm:det iff bounded gap,rmk: det iff bound infinite alphabet} we have that $\augA_\infty^\infty$ has unbounded gaps. More precisely, for every $B\in \bbN$ there exist words $x,y\in (\Gamma_\infty^\infty)^*$ and a state $s\in S$  such that $\minweight(x\cdot y,s_0\to S)=\minweight(xy,s_0\runsto{x}s\runsto{y}S)$ (and is finite), but $\minweight(x,s_0\to s)-\minweight(x,s_0\to S)>B$ (where the latter difference can be infinite or finite, but $\minweight(x,s_0\to S)$ is clearly finite due to the run on $xy$).

We claim that $\augA$ also has unbounded gaps. We start by showing this for $\augA$ with jump letters, and proceed to remove jump letters afterwards.

\subsubsection*{$\augA$ with Jump Letters is Not Determinizable}
To this end, we need to show the property above holds already for words over $\Gamma^*$ (with jump letters). 
Naturally, we obtain such words using the \emph{flattening} operation (\cref{def:cactus rebase flattening}).

Consider $B\in \bbN$ and let $x,y\in (\Gamma_\infty^\infty)^*$ as above. 
Recall that for a word $u$ the value $\maxeff{u}$ is an upper-bound on the absolute value of the change in weight that a run can incur upon reading $u$.
In the following we show that $\flatten(xy \wr 2\maxeff{xy})$ can be split to two words $x'\cdot y'$  such that the gap condition above holds. 
To show this, we proceed by induction over the sequence of operations induced by the flattening inductive definition of \cref{def:cactus rebase flattening}.

\paragraph{Base case:} For the base case we already have $xy\in \Gamma^*$, so we are done.

\paragraph{Cactus case:}
For the cactus case, assume $xy=u\alpha_{S',w}v$, and let $\rho:s_0\runsto{xy}S$ be a seamless run such that $\weight(\rho)=\minweight(xy,s_0\to S)$.
Let $z=\unfold(u,\alpha_{S',w},v|2\maxeff{u\alpha_{S',w}v})$. 
By \cref{prop:unfolding cactus maintains seamlesss gaps}, there is a seamless run $\rho':s_0\runsto{z}S$ such that 
for every prefix $h$ of $xy$ there is a corresponding prefix $z'$ of $z$ such that $\weight(\rho'(z'))=\weight(\rho(h))$, and moreover $\weight(\rho')=\weight(\rho)$. 
In particular, $\weight(\rho')=\minweight(z,s_0\to S)$. Indeed, otherwise by \cref{prop:un-unfolding maintains cheap seamless runs}, there is a seamless run of lower weight also on $xy$, which contradicts the minimality of $\rho$.

Consider $\rho(x)$, then by \cref{prop:unfolding cactus maintains seamlesss gaps} there is a prefix $z'$ of $z$ that corresponds to $x$.
Since $\minweight(x,s_0\to S)<\weight(\rho(x))-B$, then again by \cref{prop:unfolding cactus maintains seamlesss gaps} there exists a seamless run $\mu:s_0\runsto{z'}S$ with 
\[
\weight(\mu)=\minweight(x,s_0\to S)<\weight(\rho(x))-B=\weight(\rho'(z'))-B
\]
which is exactly the gap constraint we wished to show (by replacing $\weight(\rho'(z'))$ with $\minweight(z',s_0\to s)$ for the last state of $\rho'(z')$).

\paragraph{Rebase case:}
In this case, assume $xy=u\beta_{S',w,s\to r}v$ and let $z=\rebaserm(u,\beta_{S',w,s\to r},v)$.

By \cref{cor:rebase removal preserves gaps}, the gaps between any two seamless runs on $z$ after reading each prefix is the same as the gap between the corresponding seamless runs on $xy$. Thus, we obtain a gap of more than $B$ after the prefix that corresponds to $x$.

To spell out this argument further, recall that $z=u\jl_{t_1\to t_p}\alpha_{S',w}\jl_{t_p\to t_2}v$ as per \cref{def:rebase removal}. 
Consider the seamless run $\rho:s_0\runsto{xy}S$ such that $\weight(\rho)=\minweight(xy,s_0\to S)$ then the run
$\rho'$ obtained from $\rho$ as per \cref{prop:rebase removal shifts runs,cor:rebase removal preserves gaps} satisfies $\weight(\rho')=\minweight(z,s_0\to S)$. Indeed, if there is a run with smaller weight, then the reverse direction of \cref{cor:rebase removal preserves gaps} would yield a run with lower weight than $\rho$ on $xy$.
Similarly, consider the minimal seamless run $\mu'$ on the prefix $x'$ of $z$ that corresponds to $x$ (which by the same argument corresponds to a minimal seamless $\mu$ run on $x$).

\cref{cor:rebase removal preserves gaps} now guarantees that there exists a constant $c$ such that 
\[\weight(\rho'(x'))-\weight(\mu'(x'))=\weight(\rho(x))+c-(\weight(\mu(x))+c)=\weight(\rho(x))-\weight(\mu(x))>B\]
where the last inequality is because both $\rho$ and $\mu$ are seamless and minimal. We thus conclude the inductive claim.

We conclude that $\flatten(xy \wr 2\maxeff{xy})$ also satisfies the gap constraints, and since $\flatten(xy \wr 2\maxeff{xy})$ is over the alphabet $\Gamma^*$ with jump letters, we conclude by \cref{thm:det iff bounded gap} that $\augA$ with jump letters is not determinizable.

\subsubsection*{$\augA$ without Jump Letters is Not Determinizable}
Having showed that $\augA$ with jump letters is not determinizable, we now turn to  remove jump letters. 
Conceptually, this is simple: jump letters do not change the weight of any run, and do not introduce nondeterminism. Thus, they can simply be cut from a run. Indeed, their only purpose is for reasoning about baseline runs.

Technically, the problem that arises from jump letters is that the baseline component might not form a legal run of $\cA$, thus preventing us from simply dropping those letters. This is illustrated in \cref{fig:jumps break runs}.
\begin{figure}[ht]
    \centering
    \includegraphics[width=0.8\linewidth]{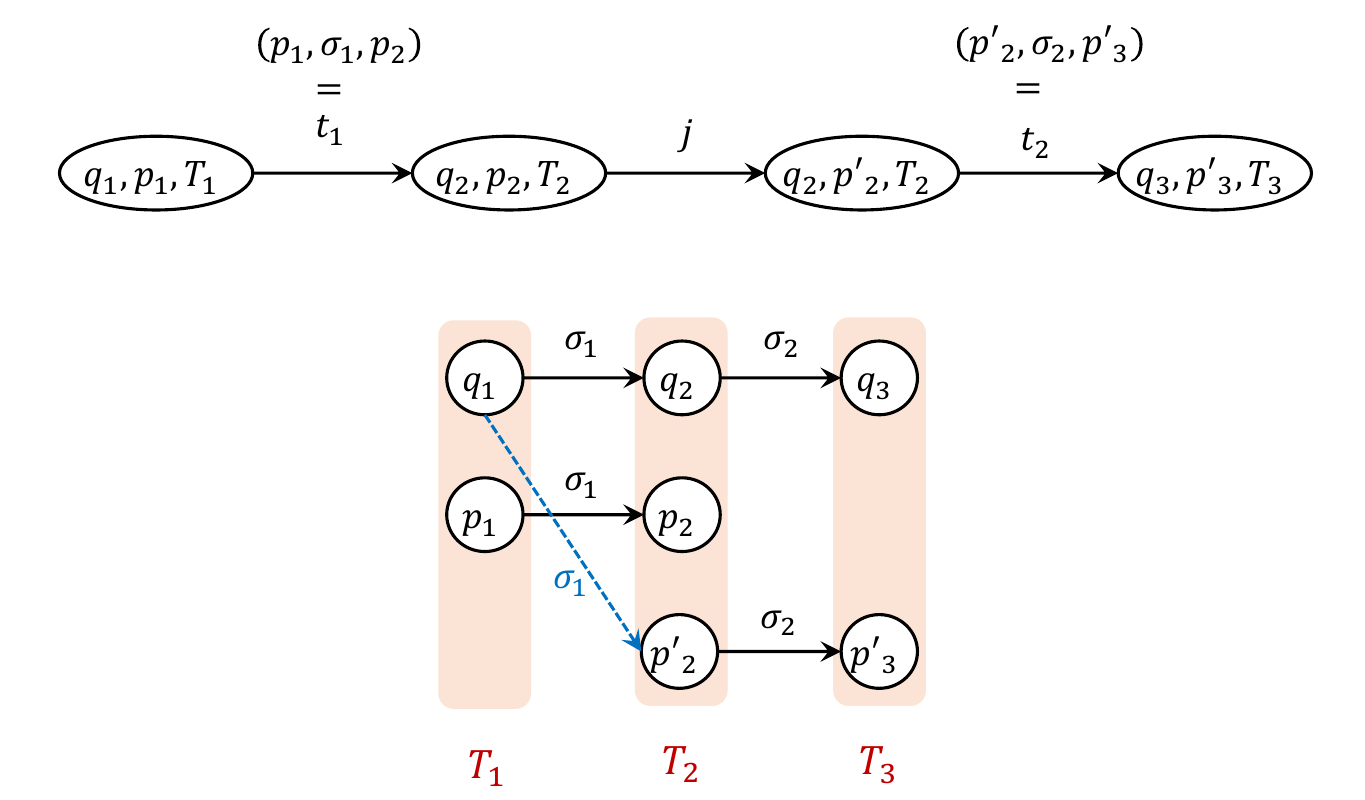}
    \caption{The jump transition from $(q_2,p_2,T_2)$ to $(q_2,p'_2,T_2)$ enables the next transition on $t_2$, which can only be taken from the baseline component $p'_2$. Thus, simply removing the jump letter (and continuing from $(q_2,p_2,T_2)$) would not yield a legal run.}
    \label{fig:jumps break runs}
\end{figure}
We therefore first shift the baseline run (akin to \cref{sec:shift POV}) so that all jump transitions do not actually change the baseline and can therefore be dropped.
As in the case of flattening, we treat jump letters inductively, one by one, and show that the gap property is maintained upon each removal.

Consider a word $x\cdot \jl_{(\cdot,p,T)\to (\cdot,p',T)}\cdot y$ where $x$ has no jump letters, and denote $x=(p_0,\sigma_1,c_1,p_1)\cdots(p_{k-1},\sigma_k,c_k,p_k)$ (as the letters in $\Gamma$ are transitions in $\cA$). 

Since $p'\in T$ (otherwise the jump letter is not defined, c.f., \cref{def:jump letters}) and $T=\booltrans(q_0,\sigma_1\cdots \sigma_k)$ (in the original WFA $\cA$), it follows that there exists a seamless run $\rho_x:q_0\runsto{\sigma_1\cdots\sigma_k} p'$. 
Denote $\rho_x=(p'_0,\sigma_1,e_1,p'_1),\ldots, (p'_{k-1},\sigma_k,e_k,p'_k)$ with $p'_k=p'$.
We now think of $\rho_x$ as a word over $\Gamma$. We show that if $x\cdot \jl_{(\cdot,p,T)\to (\cdot,p',T)}\cdot y$ is a $B$-gap witness for $\augA$ with jump letters, then $\rho_x\cdot y$ is also a $B$-gap witness (and note that $\rho_x\cdot y$ has one fewer jump letter, so we can proceed inductively).

Consider a run $\mu$ of $\augA$ on $x\cdot \jl_{(\cdot,p,T)\to (\cdot,p',T)}\cdot y$, then we can write $\mu=t_1,t_2,\ldots,t_{k},t',r_{1},\ldots,r_{|y|}$ where for $1\le i\le k$ we have
$t_i=((q_{i-1},p_{i-1},T_{i-1}),(p_{i-1},\sigma_i,c_i,p_i),f_i,(q_i,p_i,T_i))$ and $t'=((q_k,p,T),\jl_{(\cdot,p,T)\to (\cdot,p',T)},(q_k,p',T)$ (indeed, by the fact that the jump letter has defined transitions, we have that $p_k=p$). We leave $r_1,\ldots,r_{|y|}$ without particular notation.

We obtain from $\mu$ the run $\eta=t'_1,\ldots,t'_{k},r_1,\ldots,r_{|y|}$ of $\augA$ on $\rho_x\cdot y$ where for every $1\le i\le k$ we have 
$t_i=((q_{i-1},p'_{i-1},T_{i-1}),(p'_{i-1},\sigma_i,e_i,p'_i),d_i,(q_i,p'_i,T_i))$, where $d_i$ is determined by the transitions in $\augA$ (we elaborate on $d_i$ below).

Note that $\eta$ is indeed a legal run on $\rho_x\cdot y$, since the last state reached by $\eta$ after reading $\rho_x$ is $(q_k,p',T)$, which is the same state reached by $\rho$ after reading $x\cdot \jl_{(\cdot,p,T)\to (\cdot,p',T)}$.

Consider now the weight $d_i$ of the transition above. By the definition of $\augA$ in \cref{sec:augmented construction}, we have that $d_i=d'_i-e_i$, where $d'_i$ is the weight of the unique transition $(q_{i-1},\sigma_i,d'_i,q_i)$. By the same token, we also have $f_i=c_i-d'_i$ in the transitions of $\mu$.  It follows that $f_i-d_i=c_i-e_i$. In particular, the difference $f_i-d_i$ does not depend on the first component of the states visited by $\mu$ and $\eta$, and depends only on $\rho_x$ and on $\mu$.

Moreover, note that the correspondence between $\mu$ and $\eta$ is bijective. Indeed, given a run $\eta$ on $\rho_x\cdot y$, we can obtain from it a run $\mu$ on $x\cdot \jl_{(\cdot,p,T)\to (\cdot,p',T)}\cdot y$ by replacing the baseline component to that dictated by $x$, and using $\jl_{(\cdot,p,T)\to (\cdot,p',T)}$ to connect the run to $y$ (essentially retracing our construction bottom-up).

Therefore, we fix $\rho_x$ and get that this bijection between runs also shifts the weight of the runs by the same constants $f_i-d_i$. That is, for two runs $\mu_1,\mu_2$ on $x\cdot \jl_{(\cdot,p,T)\to (\cdot,p',T)}\cdot y$ and their corresponding $\eta_1,\eta_2$, for every prefix $z$ of $x$ we have $\weight(\mu_1(z))-\weight(\mu_1(z))=\weight(\eta_1(z))-\weight(\eta_2(z))$.

In particular, this result applies to prefixes of $x\cdot \jl_{(\cdot,p,T)\to (\cdot,p',T)}\cdot y$, and it follows that if $x\cdot \jl_{(\cdot,p,T)\to (\cdot,p',T)}\cdot y$ is a $B$-gap witness, then the gaps are maintained also for $\rho_x \cdot y$, and thus $\rho_x\cdot  y$ is also a $B$-gap witness.

We can thus proceed by induction until no jump letters are used, and conclude that $\augA$ is not determinizable.

\section{PSPACE-Hardness of Determinizability}
\label{apx:PSPACE hard}
We prove that the determinizability problem is PSPACE-hard by reduction from the universality problem for NFAs.

Intuitively, given an NFA $\cA$, we construct a WFA $\cB$ by giving weight $0$ to all the transitions of $\cA$, adding a new letter $\#$, and connecting all the non-accepting states of $\cA$ with $\#$ to a WFA component that is not determinizable, and gains weight above $0$. From accepting states we connect $\#$ to a $0$-sink.

Formally, given an NFA $\cA=\tup{Q,\Sigma,\Delta,Q_0,F}$ (where $\Delta\subseteq Q\times \Sigma\times Q$ is a non-weighted transition relation), 
we define a WFA $\cB=\{Q\cup \{q_a,q_b,q_\top\},\Sigma\cup \{\#,a,b\}, \Delta',Q_0,\{q_a,q_b,q_\top\}\}$ with $q_a,q_b$ being fresh states and $\#,a,b$ fresh letters, and the following transitions:
\[
\begin{split}
&\Delta'=\{(p,\sigma,0,q)\mid (p,\sigma,q)\in \Delta\}\cup\\ 
&\{(q,\#,0,q_a),(q,\#,q_b)\mid q\in Q\setminus F\}\cup \\
&\{(q,\#,0,q_\top)\mid q\in F\}\cup \\
&\{
(q_a,a,0,q_a),(q_a,b,1,q_a),(q_b,a,1,q_b),(q_b,b,0,q_b),(q_\top,a,0,q_\top),(q_\top,b,0,q_\top)
\}
\end{split}
\]

It is now easy to prove that $L(\cA)=\Sigma^*$ if and only if $\cB$ is determinizable. 
For the first direction, if $L(\cA)=\Sigma^*$, then for every $w\in \Sigma^*\cdot \#\cdot \{a+b\}^*$ we have $\cB(w)=0$, since there is a run on the $\Sigma^*$ prefix that ends in an accepting state, which proceeds to gain weight $0$ in the $q_\top$ component. 
Words in $\Sigma^*$ get weight $0$ by construction, and all other words do not have finite weight. Thus, every word gets weight either $0$ or $\infty$, and in particular has bounded gaps, and by \cref{thm:det iff bounded gap} is determinizable.

Conversely, if $L(\cA)\neq \Sigma^*$, then let $w\in \Sigma^*$ be a non-accepted word. It follows that after reading $w$ we arrive at the $q_a,q_b$ component. This component is not determinizable (which can either be seen by the gaps characterization, or directly from \cite{chatterjee2010quantitative}). It is easy to see that the same proof carries also with the $w$ prefix, so $\cB$ is not determinizable. 
\hfill \qed

\end{document}

%% file: macros.tex
\newcommand{\acctodo}[1]{\todo[disable,inline,color=red!25,size=\small]{ACCOUNTING: #1}}

\newcommand{\bs}[1]{\boldsymbol{#1}}

\newcommand{\bbN}{\mathbb{N}}
\newcommand{\bbZ}{\mathbb{Z}}
\newcommand{\bbQ}{\mathbb{Z}}
\newcommand{\bbB}{\mathbb{B}}
\newcommand{\bbNinf}{\mathbb{N}_{\infty}}
\newcommand{\bbZinf}{\mathbb{Z}_{\infty}}
\newcommand{\cA}{\mathcal{A}}
\newcommand{\cB}{\mathcal{B}}
\newcommand{\cD}{\mathcal{D}}

\newcommand{\RefStates}{\textsc{RefStates}}
\newcommand{\MinRefStates}{\textsc{MinStates}}

\newcommand{\GroundPairs}{\textsc{GrnPairs}}
\newcommand{\PlatPairs}{\textsc{PltPairs}}
\newcommand{\TethPairs}{\textsc{TthPairs}}

\newcommand{\MinGraph}{\textsc{Min\_Graph}}

\newcommand{\tup}[1]{\langle #1 \rangle}

\newcommand{\init}{\texttt{init}}
\newcommand{\fin}{\texttt{fin}}
\newcommand{\pref}{\mathsf{pref}}
\newcommand{\weight}{\mathsf{wt}}
\newcommand{\xconf}{\mathsf{next\_conf}}
\newcommand{\maxdom}{\mathsf{max\_dom}}
\newcommand{\domval}{\mathsf{dom\_val}}

\newcommand{\ST}{\text{ s.t. }}
\newcommand{\runsto}[1]{\xrightarrow{#1}}
\newcommand{\minweight}{\mathsf{mwt}}
\newcommand{\wmax}[1]{{\mathsf{max\_wt}(#1)}}
\newcommand{\maxeff}[1]{{\mathsf{max\_eff}(#1)}}

\newcommand{\augA}{\widehat{\cA}}
\newcommand{\augStates}{S}

\newcommand{\augInitState}{s_0}

\newcommand{\augTrans}{\widehat{\Delta}}

\newcommand{\rhobase}{\rho_{\mathrm{base}}}
\newcommand{\baseshift}[2]{\mathsf{bshift}[{#1}/{#2}]}
\newcommand{\augrho}{\widehat{\rho}}
\newcommand{\aug}[1]{\widehat{#1}}
\newcommand{\shift}{\mathsf{shift}}
\newcommand{\shifttoaug}{\shift^{\rhobase}_{\cA\to \augA}}
\newcommand{\shifttoorig}{\shift^{\rhobase}_{\augA\to \cA}}

\newcommand{\stab}{\mathsf{stab}}
\newcommand{\rebase}{\mathsf{rebase}}
\newcommand{\jl}{\mathsf{jump}}
\newcommand{\cost}{\mathsf{cost}}
\newcommand{\depth}{\mathsf{dep}}
\newcommand{\subcact}{\mathsf{sub}}

\newcommand{\unfold}{\mathsf{unfold}}

\newcommand{\lc}{\mathsf{lc}}
\newcommand{\precost}{\mathsf{pc}}
\newcommand{\colset}{\mathfrak{C}}

\newcommand{\booltrans}{\delta_\bbB}
\renewcommand{\vec}[1]{\boldsymbol{#1}}

\newcommand{\supp}{\mathsf{supp}}
\newcommand{\flatten}{\mathsf{flatten}}
\newcommand{\rebaserm}{\mathsf{rebase\_rm}}

\newcommand{\near}{\mathsf{near}}
\newcommand{\difftypeset}{\mathsf{DIFF\_TYPES}}
\newcommand{\gaintypeset}{\mathsf{GAIN\_TYPES}}

\newcommand{\gaintype}{\mathsf{gain\_type}}
\newcommand{\difftype}{\mathsf{diff\_type}}

\newcommand{\words}{\textsf{words}}
\newcommand{\diswords}{\lightning\!\mathsf{words}}
\newcommand{\levwords}{\rightsquigarrow\!\!\mathsf{words}}
\newcommand{\potlevwords}{{\rightsquigarrow\!\!\!\!\!\!_\pot \, \mathsf{words}}}
\newcommand{\potdiswords}{{\lightning\!^\pot\!\!\mathsf{words}}}
\newcommand{\xwords}{\mathsf{xwords}}

\newcommand{\sqrtb}{\widetilde{b}}

\newcommand{\sparse}{\mathsf{sparse}}

\newcommand{\pot}{\upphi}
\newcommand{\charge}{\uppsi}

\newcommand{\rhoss}{\rho_{\ssearrow}}

\newcommand{\bigN}{{\bs{\mathfrak{n}}}}
\newcommand{\bigM}{{\bs{\mathfrak{m}}}}

\newcommand{\interval}{\mathfrak{I}}

\newcommand{\ghostTrans}{\delta_{\mathghost}}

 \newcommand{\sparseGapBound}{16\bigM (|S|+1) (m W_{\max})^2 }
 \newcommand{\incinfGapConst}{8\bigM (|S|+1) \maxeff{xy}^2}

\newcommand{\keyicon}{%
    \includegraphics[width=1em]{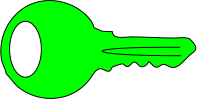}
}
\newcommand{\lightbulbicon}{%
    \raisebox{-0.4ex}{\includegraphics[height=12pt]{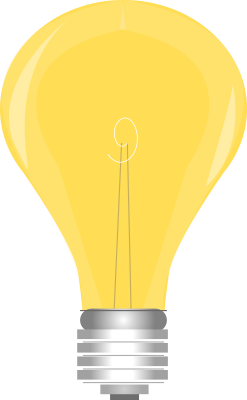}}
}



%% file: main.bbl
\begin{thebibliography}{10}

\bibitem{Almagor2020Whatsdecidableweighted}
Shaull Almagor, Udi Boker, and Orna Kupferman.
\newblock What's decidable about weighted automata?
\newblock {\em Information and Computation}, 282:104651, 2020.
\newblock Special issue on 9th International Workshop Weighted Automata: Theory and Applications (WATA 2018).

\bibitem{almagor2024determinization}
Shaull Almagor and Neta Dafni.
\newblock Determinization of integral discounted-sum automata is decidable.
\newblock In {\em International Conference on Foundations of Software Science and Computation Structures}, pages 191--211. Springer, 2024.

\bibitem{aminof2010reasoning}
Benjamin Aminof, Orna Kupferman, and Robby Lampert.
\newblock Reasoning about online algorithms with weighted automata.
\newblock {\em ACM Transactions on Algorithms (TALG)}, 6(2):1--36, 2010.

\bibitem{aminof2013rigorous}
Benjamin Aminof, Orna Kupferman, and Robby Lampert.
\newblock Rigorous approximated determinization of weighted automata.
\newblock {\em Theoretical Computer Science}, 480:104--117, 2013.

\bibitem{bell2023computing}
Jason~P Bell and Daniel Smertnig.
\newblock Computing the linear hull: Deciding deterministic? and unambiguous? for weighted automata over fields.
\newblock In {\em 2023 38th Annual ACM/IEEE Symposium on Logic in Computer Science (LICS)}, pages 1--13. IEEE, 2023.

\bibitem{bojanczyk2020languages}
Miko{\l}aj Boja{\'n}czyk.
\newblock Languages recognised by finite semigroups, and their generalisations to objects such as trees and graphs, with an emphasis on definability in monadic second-order logic.
\newblock {\em arXiv preprint arXiv:2008.11635}, 2020.

\bibitem{chatterjee2010quantitative}
Krishnendu Chatterjee, Laurent Doyen, and Thomas~A Henzinger.
\newblock Quantitative languages.
\newblock {\em ACM Transactions on Computational Logic (TOCL)}, 11(4):1--38, 2010.

\bibitem{chattopadhyay2021pumping}
Agnishom Chattopadhyay, Filip Mazowiecki, Anca Muscholl, and Cristian Riveros.
\newblock Pumping lemmas for weighted automata.
\newblock {\em Logical Methods in Computer Science}, 17, 2021.

\bibitem{choffrut1977caracterisation}
Christian Choffrut.
\newblock Une caract{\'e}risation des fonctions s{\'e}quentielles et des fonctions sous-s{\'e}quentielles en tant que relations rationnelles.
\newblock {\em Theoretical Computer Science}, 5(3):325--337, 1977.

\bibitem{colcombet2014size}
Thomas Colcombet, Laure Daviaud, and Florian Zuleger.
\newblock Size-change abstraction and max-plus automata.
\newblock In {\em International Symposium on Mathematical Foundations of Computer Science}, pages 208--219. Springer, 2014.

\bibitem{daviaud2020register}
Laure Daviaud.
\newblock Register complexity and determinisation of max-plus automata.
\newblock {\em ACM SIGLOG News}, 7(2):4--14, 2020.

\bibitem{daviaud2017degree}
Laure Daviaud, Isma{\"e}l Jecker, Pierre-Alain Reynier, and Didier Villevalois.
\newblock Degree of sequentiality of weighted automata.
\newblock In {\em International Conference on Foundations of Software Science and Computation Structures}, pages 215--230. Springer, 2017.

\bibitem{daviaud2023big}
Laure Daviaud and David Purser.
\newblock The big-o problem for max-plus automata is decidable (pspace-complete).
\newblock In {\em 2023 38th Annual ACM/IEEE Symposium on Logic in Computer Science (LICS)}, pages 1--13. IEEE, 2023.

\bibitem{daviaud2016generalised}
Laure Daviaud, Pierre-Alain Reynier, and Jean-Marc Talbot.
\newblock A generalised twinning property for minimisation of cost register automata.
\newblock In {\em Proceedings of the 31st Annual ACM/IEEE Symposium on Logic in Computer Science}, pages 857--866, 2016.

\bibitem{droste2005weighted}
Manfred Droste and Paul Gastin.
\newblock Weighted automata and weighted logics.
\newblock In {\em International Colloquium on Automata, Languages, and Programming}, pages 513--525. Springer, 2005.

\bibitem{droste2009handbook}
Manfred Droste, Werner Kuich, and Heiko Vogler.
\newblock {\em Handbook of weighted automata}.
\newblock Springer Science \& Business Media, 2009.

\bibitem{filiot2017delay}
Emmanuel Filiot, Isma{\"e}l Jecker, Nathan Lhote, Guillermo~A P{\'e}rez, and Jean-Fran{\c{c}}ois Raskin.
\newblock On delay and regret determinization of max-plus automata.
\newblock In {\em 2017 32nd Annual ACM/IEEE Symposium on Logic in Computer Science (LICS)}, pages 1--12. IEEE, 2017.

\bibitem{Has82}
K.~Hashiguchi.
\newblock Limitedness theorem on finite automata with distance functions.
\newblock {\em Journal of computer and system sciences}, 24(2):233--244, 1982.

\bibitem{Has00}
K.~Hashiguchi.
\newblock New upper bounds to the limitedness of distance automata.
\newblock {\em Theoretical Computer Science}, 233(1-2):19--32, 2000.

\bibitem{jecker2024determinisation}
Isma{\"e}l Jecker, Filip Mazowiecki, and David Purser.
\newblock Determinisation and unambiguisation of polynomially-ambiguous rational weighted automata.
\newblock In {\em Proceedings of the 39th Annual ACM/IEEE Symposium on Logic in Computer Science}, pages 1--13, 2024.

\bibitem{johnson1977efficient}
Donald~B Johnson.
\newblock Efficient algorithms for shortest paths in sparse networks.
\newblock {\em Journal of the ACM (JACM)}, 24(1):1--13, 1977.

\bibitem{kirsten2005distance}
Daniel Kirsten.
\newblock Distance desert automata and the star height problem.
\newblock {\em RAIRO-Theoretical Informatics and Applications}, 39(3):455--509, 2005.

\bibitem{kirsten2012decidability}
Daniel Kirsten.
\newblock Decidability, undecidability, and pspace-completeness of the twins property in the tropical semiring.
\newblock {\em Theoretical Computer Science}, 420:56--63, 2012.

\bibitem{kirsten2009deciding}
Daniel Kirsten and Sylvain Lombardy.
\newblock Deciding unambiguity and sequentiality of polynomially ambiguous min-plus automata.
\newblock In {\em 26th International Symposium on Theoretical Aspects of Computer Science STACS 2009}, pages 589--600. IBFI Schloss Dagstuhl, 2009.

\bibitem{klimann2004deciding}
Ines Klimann, Sylvain Lombardy, Jean Mairesse, and Christophe Prieur.
\newblock Deciding unambiguity and sequentiality from a finitely ambiguous max-plus automaton.
\newblock {\em Theoretical Computer Science}, 327(3):349--373, 2004.

\bibitem{Kro94}
D.~Krob.
\newblock The equality problem for rational series with multiplicities in the tropical semiring is undecidable.
\newblock {\em International Journal of Algebra and Computation}, 4(3):405--425, 1994.

\bibitem{LP04}
H.~Leung and V.~Podolskiy.
\newblock The limitedness problem on distance automata: Hashiguchi's method revisited.
\newblock {\em Theoretical Computer Science}, 310(1-3):147--158, 2004.

\bibitem{lombardy2006sequential}
Sylvain Lombardy and Jacques Sakarovitch.
\newblock Sequential?
\newblock {\em Theoretical Computer Science}, 356(1-2):224--244, 2006.

\bibitem{mohri1994compact}
Mehryar Mohri.
\newblock Compact representations by finite-state transducers.
\newblock In {\em 32nd Annual Meeting of the Association for Computational Linguistics}, pages 204--209, 1994.

\bibitem{mohri1997finite}
Mehryar Mohri.
\newblock Finite-state transducers in language and speech processing.
\newblock {\em Computational linguistics}, 23(2):269--311, 1997.

\bibitem{schutzenberger1961definition}
Marcel~Paul Sch{\"u}tzenberger.
\newblock On the definition of a family of automata.
\newblock {\em Inf. Control.}, 4(2-3):245--270, 1961.

\bibitem{simon1990factorization}
Imre Simon.
\newblock Factorization forests of finite height.
\newblock {\em Theoretical Computer Science}, 72(1):65--94, 1990.

\end{thebibliography}
